\documentclass{article}
\usepackage{amsthm}
\usepackage{amsmath}
\usepackage{mathrsfs}
\usepackage{amssymb}
\usepackage{graphicx}
\usepackage{tikz}
\usepackage{cases}  
\definecolor{darkgreen}{rgb}{0.1,0.5,0.1}
\usepackage{hyperref}
\hypersetup{colorlinks,linkcolor={red},citecolor={darkgreen},urlcolor={red}}

\usetikzlibrary{arrows,positioning,fit,shapes,calc}

\newtheorem{theorem}{Theorem}[section]
\numberwithin{equation}{section}

\newtheorem{lemma}[theorem]{Lemma}
\newtheorem{definition}[theorem]{Definition}
\newtheorem{proposition}[theorem]{Proposition}
\newtheorem{corollary}[theorem]{Corollary}
\newtheorem{remark}[theorem]{Remark}
\newtheorem{example}{Example}[section]

\DeclareMathAlphabet{\bi}{OML}{cmm}{b}{it}

\DeclareMathAlphabet\bfcal{OMS}{cmsy}{b}{n} 

\newcommand{\aspar}{{g}}
\newcommand{\cell}{{\Omega}} 
\newcommand{\nn}{{\nonumber}}
\newcommand{\boldcell}{{\bf \cell}} 
\newcommand{\R}{\mathbb{R}}
\newcommand{\C}{\mathbb{C}}
\newcommand{\Z}{\mathbb{Z}}
\newcommand{\N}{\mathbb{N}}
\newcommand{\NO}{{\mathbb{N}_0}}
\newcommand{\HH}{\mathbb{H}}

\newcommand{\bE}{\boldsymbol{E}}
\newcommand{\bH}{\boldsymbol{H}}

\newcommand{\bJ}{\boldsymbol{J}}
\newcommand{\bK}{\boldsymbol{\mathrm{K}}}
\newcommand{\bx}{\boldsymbol{x}}

\newcommand{\by}{\boldsymbol{y}}
\newcommand{\bk}{\boldsymbol{k}}
\newcommand{\bkg}{\boldsymbol{\kappa}}
\newcommand{\bv}{\boldsymbol{v}}
\newcommand{\bw}{\boldsymbol{w}}

\newcommand{\bu}{\boldsymbol{u}}

\newcommand{\bn}{\boldsymbol{n}}

\newcommand{\domain}{\textrm{dom}}

\newcommand{\calF}{\mathcal{F}}
\newcommand{\calH}{\mathcal{H}}

\newcommand{\calO}{\mathcal{O}}

\newcommand{\bbA}{\mathbb{A}}
\newcommand{\tbbA}{\tilde{\mathbb{A}}}
\newcommand{\bbB}{\mathbb{B}}
\newcommand{\bbC}{\mathbb{C}}

\newcommand{\bbE}{\mathbb{E}}

\newcommand{\bbP}{\mathbb{P}}
\newcommand{\bbR}{\mathbb{R}}
\newcommand{\bbS}{\mathbb{S}}

\newcommand{\bbAinfDk}{\mathbb{A}_{\infty,\bk}}

 \newcommand{\LaplDcellp}{-\Delta_{\mathrm{Dir},\cell^+}}
  \newcommand{\LaplDcellA}{-\Delta_{\mathrm{Dir},\cell^A}}

\newcommand{\curl}{\operatorname{curl}}

\newcommand{\OpcompApp}{\mathbb{V}}

\newcommand{\condS}{(\mathcal{S})}

\newcommand{\eps}{\varepsilon}

\newcommand{\rmd}{{\mathrm{d}}}
\newcommand{\rmD}{{\mathrm{D}}}
\newcommand{\rme}{{\rm e}}
\newcommand{\rmi}{{\rm i}}

\begin{document}
\title{High contrast elliptic operators in honeycomb structures}
\author{Maxence Cassier$^{a}$ and Michael I. Weinstein$^{b}$ \\
\small $^a$ Aix Marseille Univ, CNRS, Centrale Marseille, Institut Fresnel, Marseille, France,\\ \small  $^b$ Dept of Applied Physics \& Applied Mathematics, and Dept. of Mathematics, Columbia University, New-York, United States\\\small
(maxence.cassier@fresnel.fr,  miw2103@columbia.edu)}

\maketitle
%\tableofcontents

\begin{abstract}
We study the band structure of self-adjoint elliptic operators 
$\mathbb{A}_\aspar= -\nabla \cdot \sigma_{\aspar} \nabla$, where $\sigma_g$ has the symmetries of a honeycomb tiling of $\mathbb{R}^2$. 
We focus  on the case where  $\sigma_{\aspar}$ is a real-valued scalar:  $\sigma_{\aspar}=1$ within identical, disjoint  ``inclusions'', centered at vertices of a honeycomb lattice, and  
$\sigma_{\aspar}=\aspar\gg1 $ (high contrast) in the complement of the inclusion set (bulk).  Such operators govern, {\it e.g.} transverse electric (TE) modes in photonic crystal media consisting of high dielectric constant inclusions (semi-conductor pillars) within a homogeneous lower contrast bulk (air), a configuration used in many physical studies. Our approach, which  is based on monotonicity properties of the associated energy form,  extends to a class of high contrast  elliptic operators  that model heterogeneous and anisotropic honeycomb media.

Our results concern the global behavior of dispersion surfaces, and the existence of conical crossings (Dirac points)
 occurring in the lowest two energy bands as well as in bands arbitrarily high in the spectrum.   Dirac points are the source of important phenomena in fundamental and applied physics, e.g. graphene and its artificial analogues, and topological insulators.   The key hypotheses are the non-vanishing of the Dirac (Fermi) velocity
$v_D(g)$, verified numerically, and a spectral isolation condition, verified analytically in many configurations. 
Asymptotic expansions, to any order in $\aspar^{-1}$, of Dirac point eigenpairs and $v_D(g)$ are derived with error bounds.

Our study illuminates differences between the high contrast behavior of $\mathbb{A}_\aspar$ and the corresponding strong binding regime for Schroedinger operators. 
\end{abstract}

\section{Introduction and summary of the results}
\subsection{Introduction}

This article concerns the spectral properties of the second order divergence form elliptic operator
$\bbA_{\aspar}\ :=\ -\nabla \cdot \sigma_{\aspar} \nabla\quad \textrm{acting on $L^2(\bbR^2)$}$, 
 where $\sigma_{\aspar}$ is defined on $\mathbb{R}^2$ and has the symmetries of 
  a honeycomb tiling of $\R^2$. 
We focus on the case where $\sigma_{\aspar}(\bx)$ is a strictly positive, real-valued and piecewise constant scalar function of position, $\bx=(x,y)$, which is equal to $g>0$ on a set of {\it inclusions} and equal to $1$ on their complement in $\mathbb{R}^2$ (the bulk);
 see Figure \ref{fig.funcell}.
 
The interest in elliptic operators with honeycomb symmetry was catalyzed by the discovery of 2D materials, such as graphene \cite{geim2007rise,novoselov2011nobel,RMP-Graphene:09,Katsnelson:12}, and their role in the field of topological insulators. Graphene's remarkable wave propagation properties are directly related  to the presence of 
{\it Dirac points}, conical touchings of neighboring dispersion surfaces in the band structure  of the single-electron (Schroedinger) model of graphene. Dirac points have been shown to occur in generic honeycomb Schroedinger operators \cite{FW:12}; see also \cite{Ammari:2020,Berkolaiko-comech:18,C91,Grushin:09,LWZ:18}. Their implications for the dynamics of wave-packets were studied in \cite{FW:14}.

Analogous wave properties have been observed in many different physical systems with honeycomb symmetry,
 where operators of type $\bbA_{\aspar}$ arise in engineered topological materials, 
 {\it e.g.} electromagnetism for photonic graphene, acoustics, mechanics; see, for example,
  \cite{Marquardt:17,irvine:15,ozawa_etal:18,artificial-graphene:11}.  
  Such engineered honeycomb media are often called {\it artificial graphene}. 
For a discussion of operators of the type $\bbA_{\aspar}$, as they arise in the context of transverse electric (TE) modes in the 2D Maxwell equations, see Appendix \ref{TE-Max}. Elliptic operators of type $\bbA_{\aspar}$ occur as well in models for  $2D$  acoustics \cite{Kosloff1983migration} and in elasticity \cite{Pham:2017}.
 Dirac points and their dynamical consequences in photonic graphene for  the 2D Maxwell equations with smooth coefficients were studied in \cite{LWZ:18,XZ:19}. 

Typically, engineered periodic structures (honeycomb and other) are media which consist of two or more distinct materials, 
each characterized by its own constant material parameter, {\it e.g.} dielectric constant. Often the material contrast is taken to be large. {\it The goal of this article is to 
study the spectral properties of honeycomb operators $\bbA_{\aspar}$ where $\sigma_{\aspar}$ is piecewise constant.} We focus on the regime of high material contrast, corresponding to $\aspar$ large.

An analogous study of continuum honeycomb Schroedinger operators in the {\it strong binding regime} was initiated in \cite{FLW-CPAM:17}. Here, the periodic  quantum potential consists of deep atomic potential wells centered at honeycomb lattice sites.
It is shown that the low-lying (first two) dispersion surfaces, after a centering and rescaling, converge uniformly to those of
   the tight binding (discrete) model of graphene. In contrast, for $\bbA_\aspar$ we  obtain detailed information on the low-lying spectrum  and also its higher energy dispersion surfaces. For example, our results imply in particular for the case of circular inclusions that
   {\it for each 
  eigenvalue, $\tilde\delta$,  of the infinite sequence of  radial (simple) Dirichlet eigenpairs of $-\Delta$ for the single inclusion, there is a pair of  dispersion surfaces of $\bbA_\aspar$ acting in $L^2(\mathbb{R}^2)$, which meet in a Dirac point, and which converge to the constant 
  function with value equal $\tilde\delta$}, as $\aspar\uparrow\infty$, uniformly on $\mathcal{B}$ for the upper one and on any compact subsets of $\mathcal{B}\setminus\{{\bf 0}\}$ for the lower one.

 Corresponding results for Dirac points at higher energies have not yet been proved for Schroedinger operators
 in the strong binding regime.
Furthermore,  the global character of the dispersion surfaces of $\bbA_\aspar$
 is very different from that of the Schroedinger case.
   For example, general dispersion surfaces of $\bbA_\aspar$ (and, in particular, the first dispersion surface) do not converge uniformly with increasing $\aspar$ in any compact set including $\bk=0$. 
   
 The methods we use differ  in many key respects from those used in the honeycomb  Schroedinger case.
 $\bbA_\aspar$ is decomposed as a fiber integral over the subspaces $L^2_{\bf k}$, on which we use the variational characterization of eigenvalues of self-adjoint  operators. Comparison principles (Dirichlet and Neumann bracketing)
 and the monotonicity of the energy form for $\bbA_\aspar$, with respect to $\aspar$, enable verification of a key {\it spectral isolation property}, used to study the asymptotics of bands that touch in a Dirac point. 
    Finally, due to the discontinuity in the coefficients 
  of $\bbA_\aspar$,  we work with a weak formulation of the elliptic eigenvalue problem, which requires many technical 
adjustments to aspects of the analysis with parallels in \cite{FLW-CPAM:17}.

 \begin{figure}[!h]
 \includegraphics[width=1\textwidth]{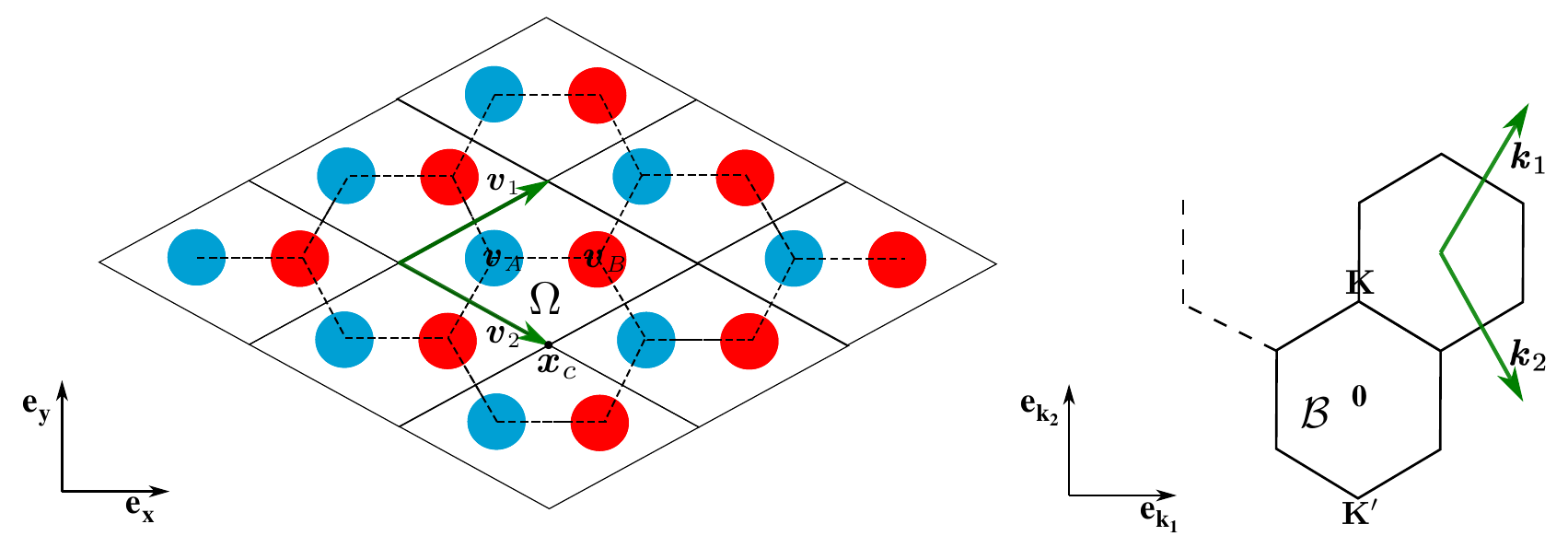}
   \caption{Left panel: Honeycomb arrangement of inclusions. Centers of  inclusions are located at vertices of the honeycomb:  $(\bv_A+\Lambda)\cup(\bv_B+\Lambda)$, where $\Lambda=\Z\bv_1\oplus\Z\bv_2$ is the equilateral triangular lattice. The unit cell $\cell$  is a diamond-shaped region  containing inclusions: $\cell^A$  centered at $\bv_A$ and $\cell^B$ centered at $\bv_B$. $\sigma_g(\bx)=1$ for $\bx\in\cell^A\cup\cell^B$ and $\sigma_g=g$ for $\bx$ outside the inclusions. $\sigma_g(\bx+\bv)=\sigma(\bx)$ for all $\bv\in\Lambda$ and all $\bx\in\R^2$. Right panel: Dual lattice $\Lambda^*=\Z\bk_1\oplus\Z\bk_2$. Brillouin zone, $\mathcal{B}$, with its two independent high symmetry quasimomenta at vertices:  ${\bf K}$ and  ${\bf K^\prime}$. }
 \label{fig.funcell}
\end{figure}

\subsection{ A brief description of the mathematical problem}
 Figure \ref{fig.funcell} (left panel) displays
 a portion of the structure defined by $\sigma_{\aspar}$. The plane is partitioned into diamond-shaped period cells with fundamental cell $\cell$. Each period cell has two disjoint identical inclusions, $\cell^A$ and $\cell^B$, centered at vertices of a honeycomb structure. The inclusion shape is required to be: $2\pi/3$ rotationally invariant and inversion symmetric about its center. The function $\sigma_g$  takes on two values
 \[
\sigma_g(\bx)\ =\ 
\begin{cases} 
1 &\textrm{for $\bx$ in the inclusions}\\
g>1 & \textrm{for $\bx$ in the bulk (outside the set of inclusions).}\\
 \end{cases}
 \]
  In this article we present results on spectral properties of 
 \begin{equation}\label{eq.defAg}
\bbA_{\aspar}\ :=\ -\nabla \cdot \sigma_{\aspar} \nabla\quad \textrm{acting on $L^2(\bbR^2)$ }.
\end{equation}

Since $\bbA_{\aspar}$ commutes with (triangular) lattice translations, its spectrum can be obtained via the family of quasi-periodic Floquet-Bloch eigenvalue problems; see Section \ref{sec-Floqth}.
 For each $\bk\in\mathcal{B}\subset(\R^2_{\bx})^*$ ($\mathcal{B}\simeq\mathbb{T}^2$), the Brillouin zone (Figure \ref{fig.funcell}, right panel), let  $ \lambda_1(g;\bk)\le  \lambda_2(g;\bk)\le\dots\lambda_n(g;\bk)\le\dots$ denote the eigenvalues, with multiplicities listed, for the eigenvalue problem: 
\[ \bbA_{\aspar}\psi=\lambda\psi\quad \textrm{subject to $\bk-$quasi-periodic boundary conditions; see \eqref{eq.defL2k}.}\]
 The functions $\bk\mapsto \lambda_n(g;\bk)$, $n\ge1$,  are Lipschitz continuous and their graphs are the {\it dispersion surfaces} of 
   $\bbA_{\aspar}$. The $L^2(\R^2)$ spectrum of 
   $\bbA_{\aspar}$ is the union of closed real intervals, that are swept out by the maps $\bk\mapsto\lambda_n(\aspar;\bk)$
    as $\bk$ varies over $\mathcal{B}\simeq\mathbb{T}^2$. The collection of all Floquet-Bloch eigenvalue / eigenfunction pairs is called 
   the {\it band structure} of   $\bbA_{\aspar}$. Section \ref{sec-Floqth} provides a more detailed discussion.

 \subsection{Summary of main results}\label{sec-mainresults}
 
We summarize our main results on the band structure of $\bbA_{\aspar}$ acting in $L^2(\R^2)$ for $g\gg1$. To keep the presentation of this introduction short, we 
outline results as they apply to the low-lying (first two) dispersion surfaces. Our results extend to higher energy bands whose high contrast limit satisfies a band spectral isolation condition $\condS$; see Definition \ref{Def.condS}. Precise formulations of results for low-lying and  higher energy dispersion surfaces are stated in Sections \ref{HL-thy} to \ref{sec.diracpointthoerem}.

 \begin{enumerate}
 \item  {\it Theorem \ref{thm.convunifbands} and Corollary \ref{cor.firstdipstcurves}: 
 Convergence of dispersion maps as $g\uparrow$ :}\  Hempel and Lienau  \cite{Hem-00} developed a variational approach for studying the convergence of the band dispersion functions as $g\uparrow\infty$; see also \cite{Frie:02}. 
   It is based on the monotonicity of the energy form, $a_{g,\bk}(u,u)=\int_{\cell}\sigma_g(\bx) |\nabla u(\bx)|^2\ d\bx$, with respect to the parameter $g$ and the min-max characterization of eigenvalues of self-adjoint operators. In Section \ref{sec.condP} we apply their approach to a study of the band structure of 
  $\bbA_{\aspar}$ for $g\gg1$.  Note that the results of  Sections \ref{HL-thy} and \ref{sec.condP}   do not require  honeycomb symmetry;  see Remark \ref{rem-noHoneycombsym}
  
To explain these results, in the context of the first two bands, note first that $\bbA_{\aspar}$ annihilates constant functions, which satisfy periodic boundary conditions ($\bk-$quasi-periodicity  with $\bk={\bf 0}$). Therefore, $\lambda_1(g;{\bf 0})=0$ for all $g$.  We prove, on the other hand, that as $g$ tends to infinity:\medskip \\
 {\it (a) the  first dispersion map, $\bk\mapsto\lambda_1(g;\bk)$ converges, uniformly on compact subsets of $\mathcal{B}\setminus\{{\bf 0}\}$ (but not on all $\mathcal{B}$), to the constant function 
of $\bk$ with value equal to the (strictly positive) $1^{st}$ Dirichlet eigenvalue  of a single inclusion, $\cell^A$
}, and\\[4pt]
{\it (b) the second dispersion map, $\bk\mapsto\lambda_2(g;\bk)$ converges, uniformly on all of $\mathcal{B}$ to the constant function in $\bk$ with value equal to the (positive) $1^{st}$ Dirichlet eigenvalue  of a single inclusion, $\cell^A$
}.\\[4pt]
{\it  (c) For $\aspar$ sufficiently large, there is a gap in the spectrum of  $\bbA_{\aspar}$ between the $2^{nd}$ and $3^{rd}$ spectral bands.}

Figure \ref{fig.3eigenvalbig} illustrates assertions (a), (b) and (c). In each panel, the $3$ displayed curves are obtained by tracking  the dispersion surfaces $\bk\mapsto\lambda_j(\aspar;\bk)$, $ j=1,2,3$, along the boundary of a symmetry-reduced Brillouin zone for the indicated value of $\aspar$.
 While in all panels $\lambda_1(\aspar,0)=0$, we see that in the complement of any neighborhood of $\bk=0$, $\lambda_1(g;\bk)$ converges uniformly to the first Dirichlet eigenvalue, $\tilde{\delta}_1>0$,  of the single inclusion. This eigenvalue is a doubly degenerate Dirichlet eigenvalue, $\delta_1=\delta_2=\tilde{\delta}_1$, for the union of two identical disjoint inclusions $\Omega^A\cup \Omega^B$.  On the other hand, $\lambda_2(g;\bk)$ is seen to converge uniformly on all $\mathcal{B}$. 
Finally, as asserted in (c),   a spectral gap opens between the first two bands and the third  band  for larger values of $g$.
\begin{figure}[h!]
\begin{center}
\includegraphics[width=5.1cm]{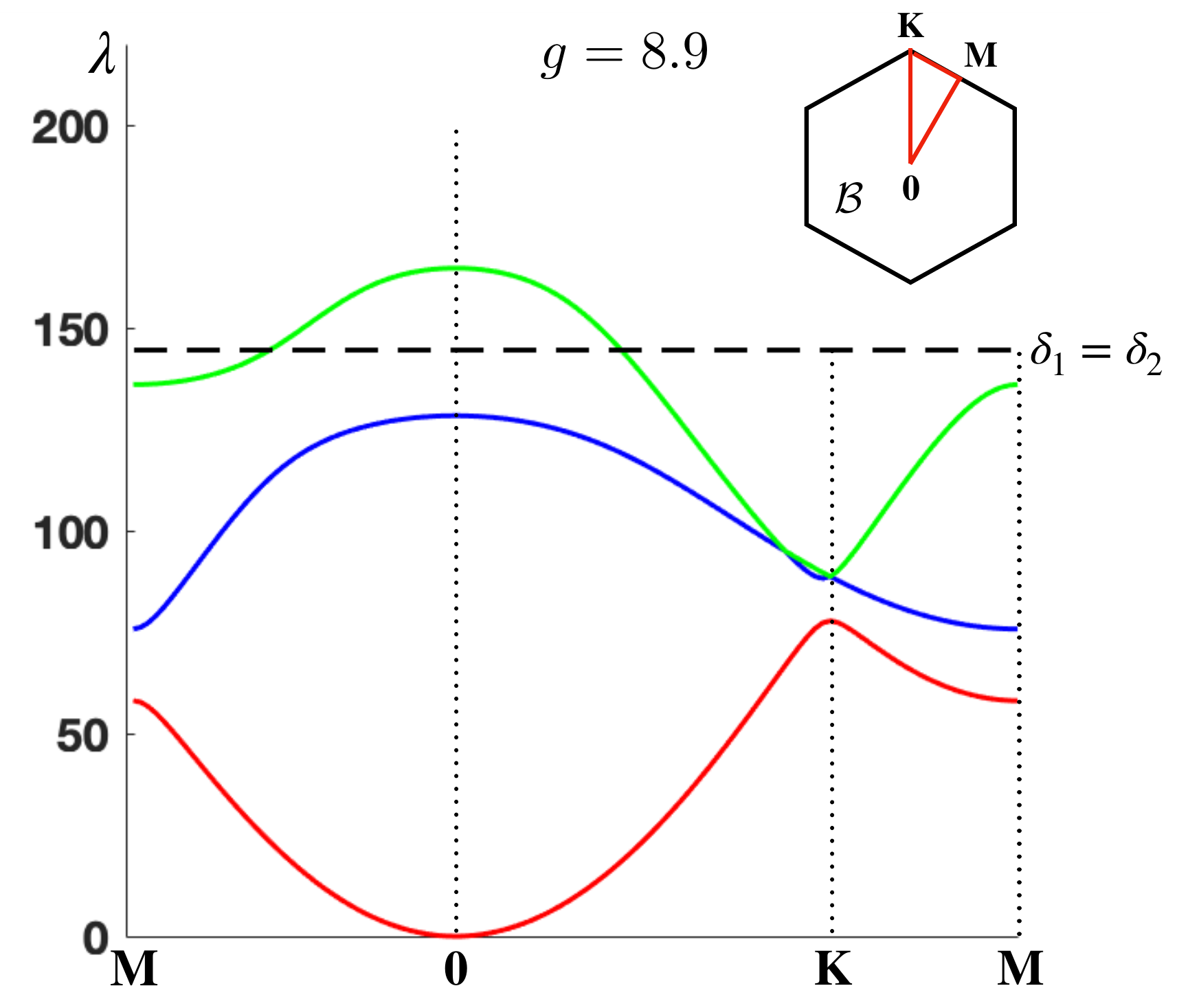}
\hspace{0.1cm}
\includegraphics[width=5.1cm]{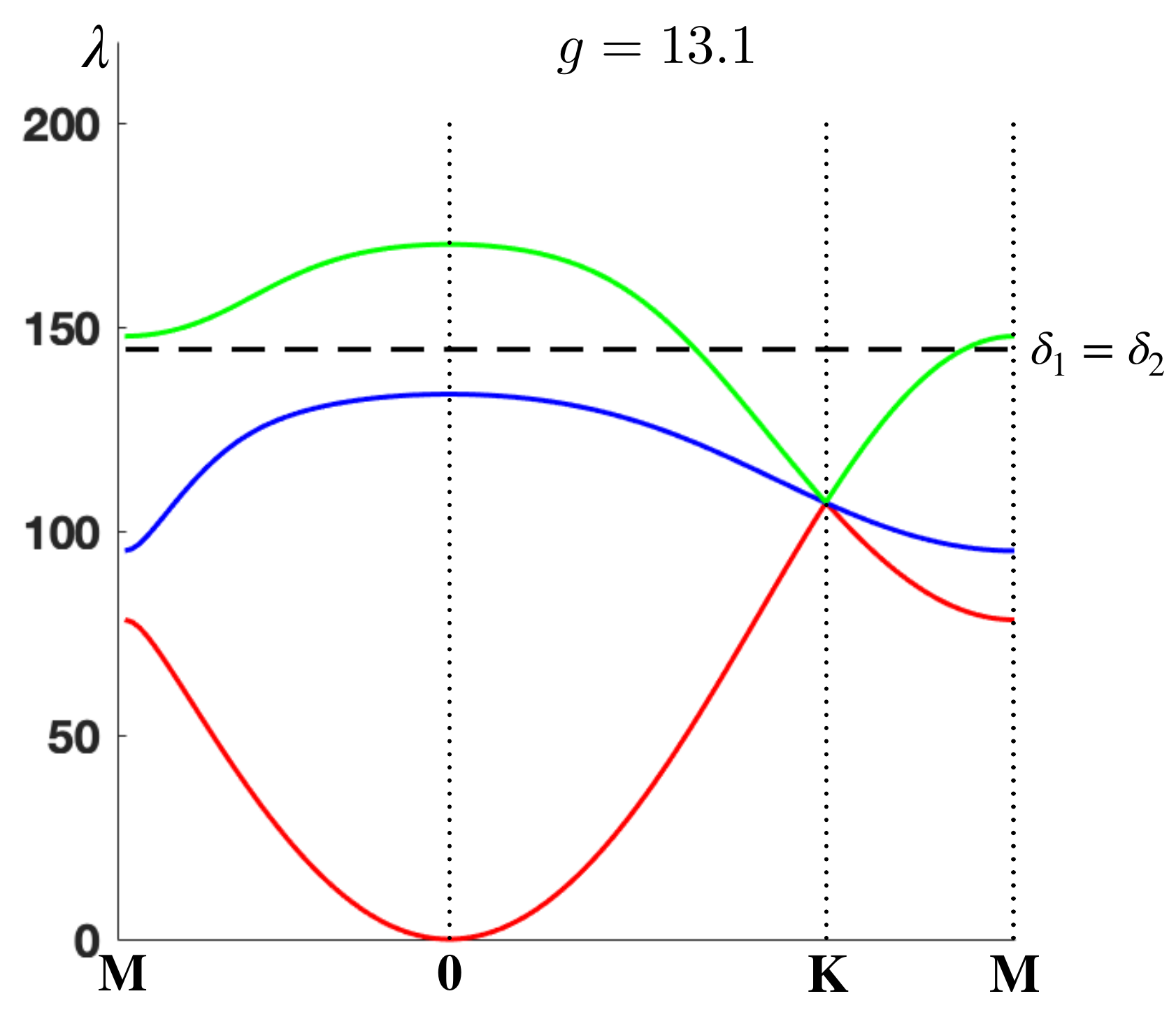}
\hspace{0.1cm}
\includegraphics[width=5.1cm]{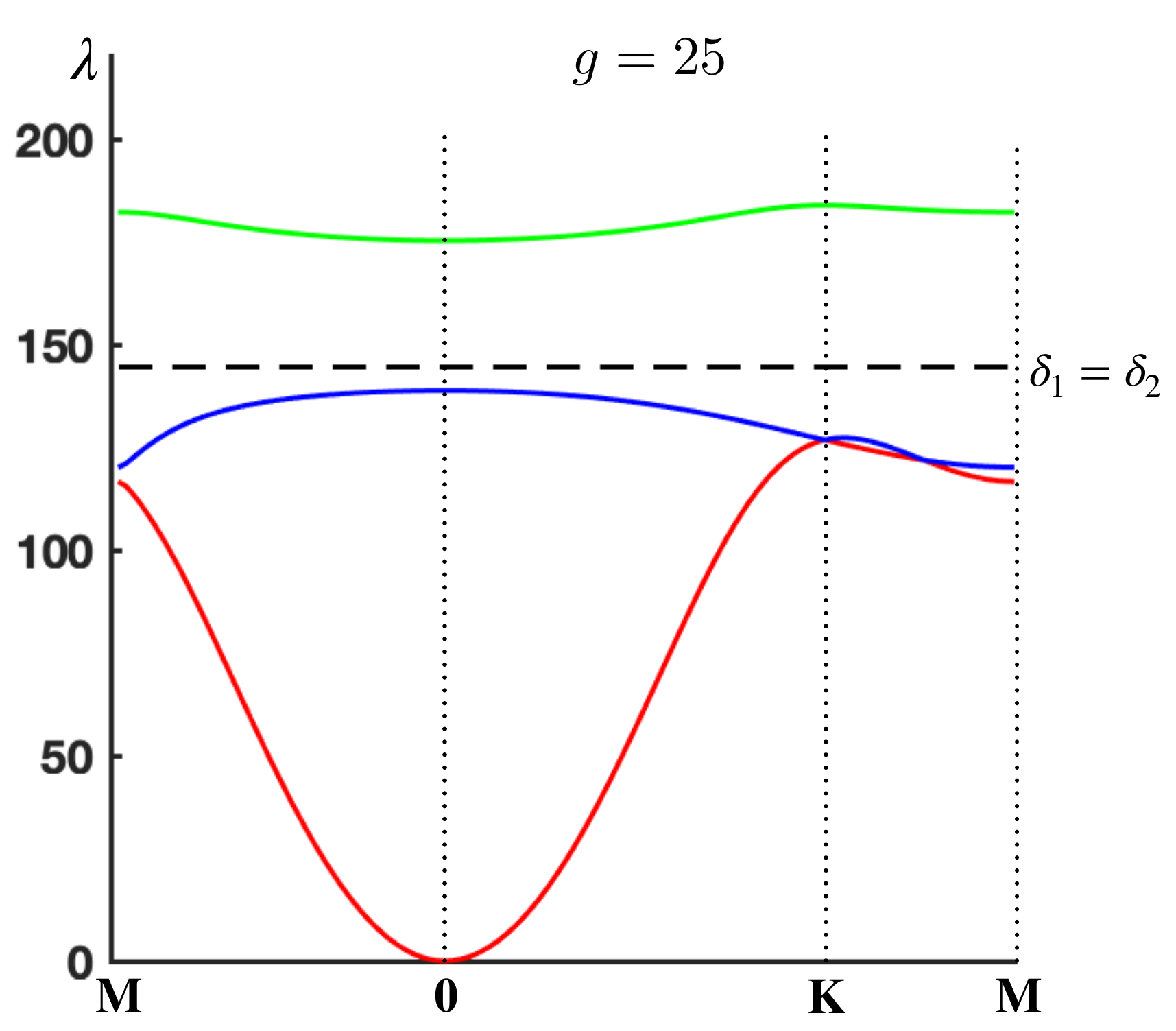}
\includegraphics[width=5.1cm]{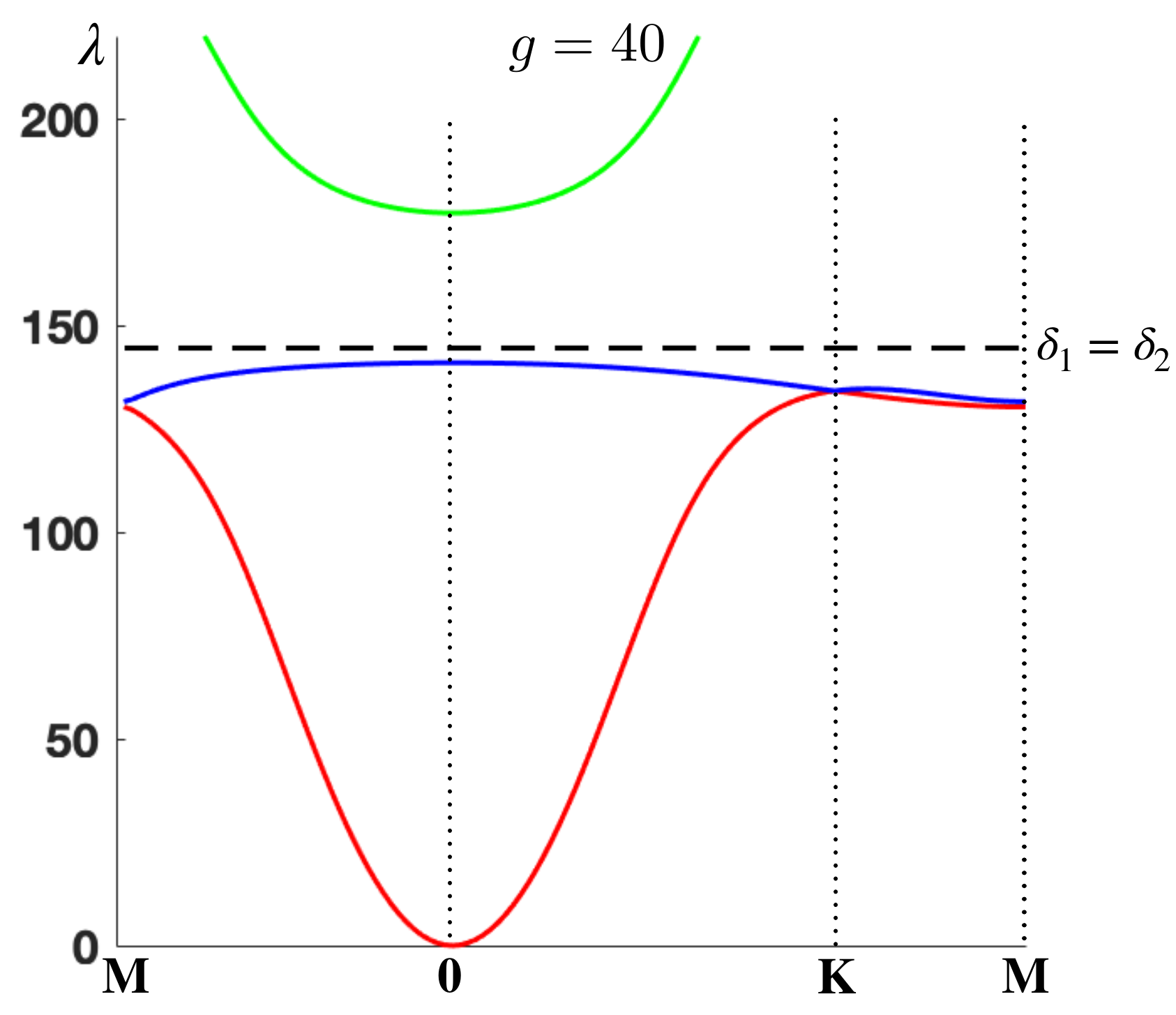}
\hspace{0.1cm}
\includegraphics[width=5.1cm]{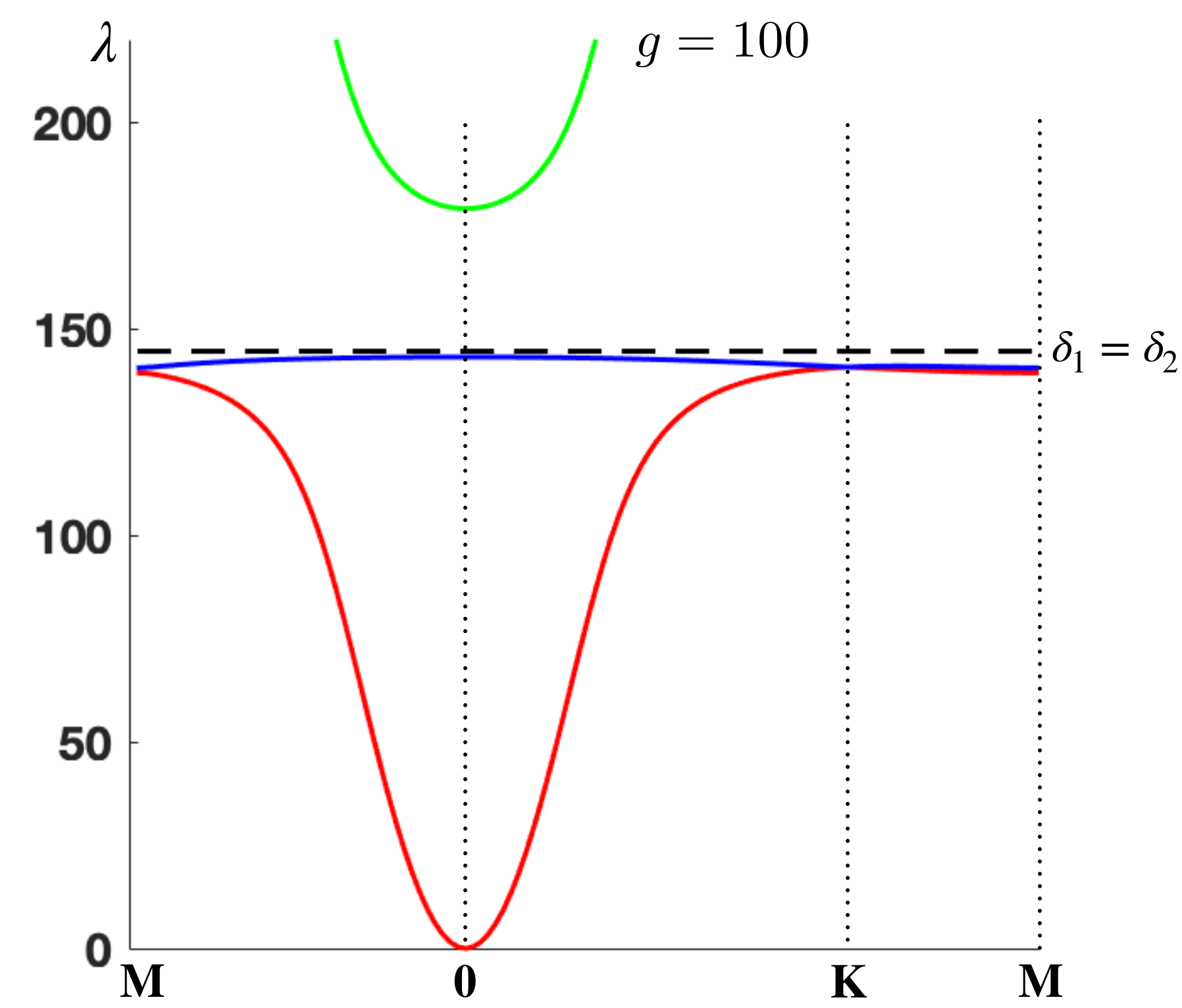}
\hspace{0.1cm}
\includegraphics[width=5.1cm]{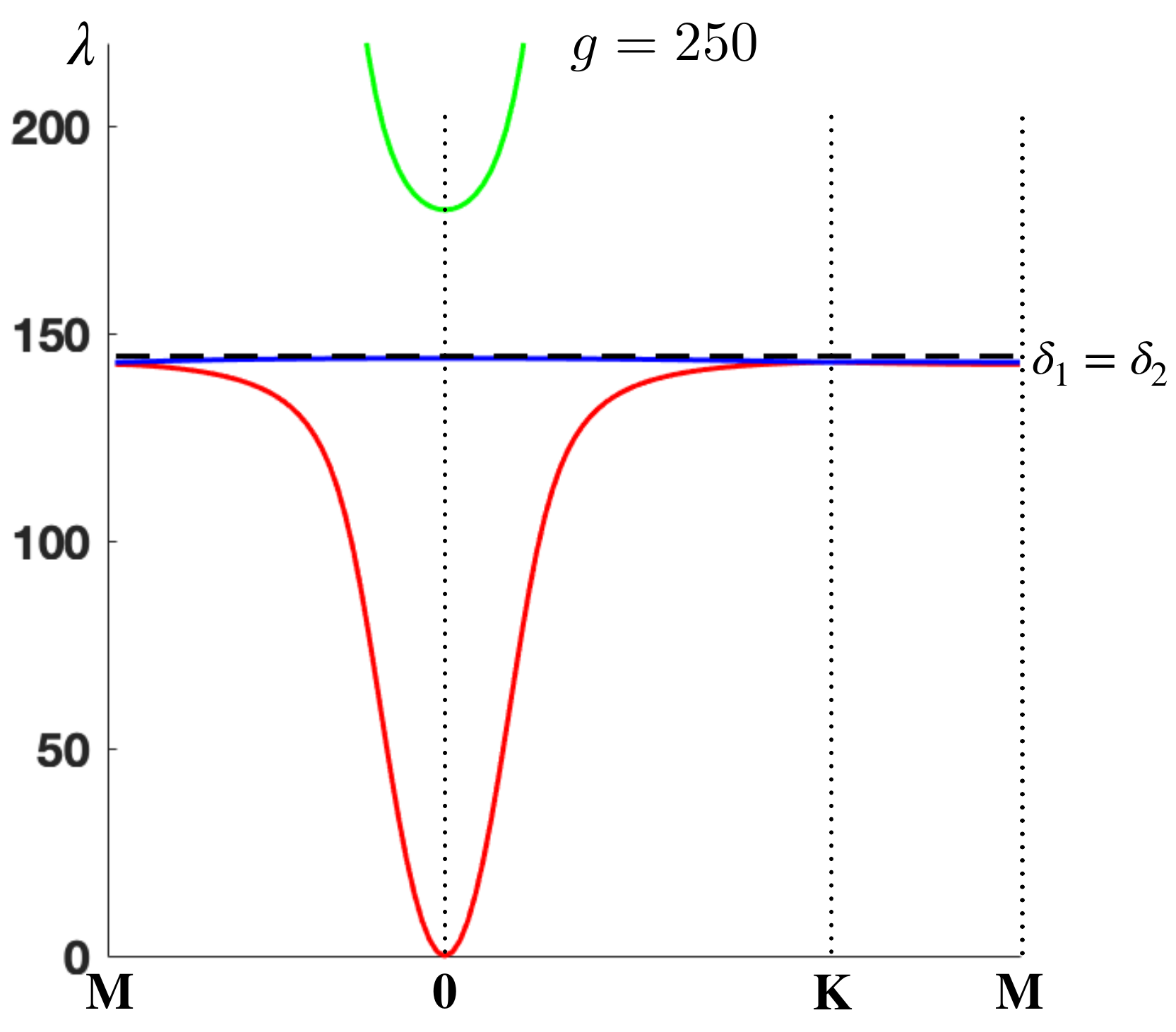}
\end{center}
\caption{Dispersion maps $\bk\mapsto \lambda_1(\aspar;\bk)$, $\lambda_2(g;\bk)$ and  $\lambda_3(g;\bk)$
 plotted along the boundary, ${\bf M}\to\bf{0}\to {\bf K}\to{\bf M}$, of the symmetry-reduced Brillouin zone for the indicated values of the contrast $\aspar$. For all $\aspar>0$, $\lambda_1(\aspar;{\bf 0})=0$.
  However, as $\aspar\uparrow\infty$, 
$\lambda_1(\aspar;\bk)$  converges,  uniformly on any compact subset of  $\mathcal{B}\setminus \{0 \}$, to the two-fold degenerate smallest  eigenvalue $\delta_1=\delta_2>0$ of the Dirichlet Laplacian   (horizontal dashed line) for the domain $\cell^A\cup\cell^B$, the union of two disc inclusions of radius $R_0=0.2$.  The second dispersion map,  $\lambda_2(g;\bk)$, converges uniformly on the \underline{full} Brillouin zone $\mathcal{B}$ to $\delta_1=\delta_2>0$ .  Here, $\lambda=\delta_1=\delta_2=\left(z_{0,1}/R_0\right)^2$, where $z_{0,1}$ denotes the first zero of the Bessel function $J_0(z)$.  For typical values of $g$, there is a Dirac point (conical intersection of bands) over the high-symmetry quasi-momentum $\bK$ (and $\bK^\prime$, not shown). For the simulated structure, this Dirac point is situated between the second and third bands for $g=8.9$ and between the first and second bands for 
 $g=25, 40, 100, 250$; see Theorem \ref{th.degeneracy}.  A (transitional) triple degeneracy occurs at $\bK$ (resp. $\bK^\prime$) for $g=13.1$. Furthermore, for large $g$ there is a gap between the $2^{nd}$ and $3^{rd}$ spectral bands; see Theorem \ref{thm.convunifbands} .}
\label{fig.3eigenvalbig}
\end{figure}

\item {\it Theorem \ref{th.degeneracy} and Corollary \ref{cor.Diracpointfirstdispcurv}, Existence of Dirac points for $g$ sufficiently large:} 
In Section \ref{DP} we prove, under a non-degeneracy condition (non-zero Dirac velocity, $ v_D(g)$),
 that for $g$ sufficiently large, $\bbA_{\aspar}$ has Dirac points $(\lambda_D(g),\bk_D)$.
These occur at intersections of the first two dispersion surfaces at an energy $\lambda_D(g)$, for   $\bk_D$  equal to any of the $6$ vertices of  $\mathcal{B}$ ($\bK$ and $\bK^\prime$, high symmetry quasi-momenta and their images by rotations of $2\pi/3$ and $4\pi/3$ centered at $0$, see section \ref{honey-sym}
and Remark \ref{rem.Diracpointsequivalence}).  By Theorem \ref{thm.convunifbands}, since $\bK_\star=\bK, \bK^\prime$ are non-zero, it follows that $\lambda_1(g;\bK_\star)=\lambda_2(g;\bK_\star)=\lambda_D(g)$ tends to the (multiplicity 2)  Dirichlet eigenvalue $\delta_1=\delta_2$ 
 of the inclusion subset $\cell^A\cup \cell^B $ of the fundamental domain $\Omega$.

 In Figure \ref{fig.3eigenvalbig}, Dirac points appear as linear (transverse) crossings of two curves. 
 For Schroedinger operators \cite{FW:12}, and divergence form elliptic operators with smooth coefficients \cite{LWZ:18})
  we have that: for all but a discrete set of values of the well-depth parameter (respectively, contrast parameter),
 Dirac points occur at least  at one energy over the $6$ vertices of $\mathcal{B}$.
 In the setting of this article, Figure  \ref{fig.3eigenvalbig} shows the emergence of Dirac points for sufficiently large  $g$ between the first two bands. For a detailed discussion of the observed transfer of the Dirac point 
from the second and third bands to the first and second bands as $\aspar$ increases, see Section \ref{numerics}, specifically the discussion of Figure \ref{fig.eigenvaltransition}.

Let $\bk_D=\bK$. The eigenvalue $\lambda_D(g)$ has a corresponding eigenspace of dimension two with associated orthonormal basis $\{\Phi_1(\aspar,\cdot),\Phi_2(\aspar,\cdot)\}$ where $\Phi_1(\aspar,\cdot)=P_{1,\bK}^A+\mathcal{O}(1/g)$ and 
$\Phi_2(\aspar,\cdot)=\pm \rme^{-2\rmi \pi/3}P_{1,\bK}^B+\mathcal{O}(1/g)$ for large $g$. Here,  $P_{1,\bK}^A$ and $P_{1,\bK}^B$
are $\bK-$quasi-periodic superpositions of single inclusion Dirichlet states.
  $P_{1,\bK}^A$ is supported on the $A-$ inclusions and $P_{1,\bK}^B$ is supported on the $B-$ inclusions.
Hence, these Dirichlet states play a role analogous to {\it atomic orbitals} in the strong-binding Schroedinger analysis \cite{FLW-CPAM:17}. Numerical simulations demonstrating this behavior are presented in Sections \ref{numerics} and \ref{sec.higherbandscirccase}.

The disjointness of the supports of the leading order terms of $\Phi_j(\cdot,\aspar)$ implies, for the Dirac velocity: 
\[ v_D(g)=v_D^{(1)}\aspar^{-1}+\mathcal{O}(\aspar^{-2}).\]
 Hence, the Dirac velocity tends to zero when $\aspar\uparrow\infty$ and the Dirac cone becomes increasing flat as $\aspar$ increases; see Figure \ref{fig.3eigenvalbig} for $\aspar=250$. This is consistent with the uniform convergence of dispersion surfaces away from $\bk={\bf 0}$. We conjecture that $v_D^{(1)}\ne0$. Numerical computations show this to be the case (see Section \ref{numerics}) and hence for $g\gg1$,  $(\lambda_D, \bk_D)$  is a Dirac (conical) point in the sense of Definition \ref{def.Diracpoints}.

Our proof of Theorem \ref{th.degeneracy} and its Corollary \ref{cor.Diracpointfirstdispcurv}, is based
on the set of sufficient conditions proved in Theorem \ref{Thm.weakLyapounovSchmidtred}. 
Analogous conditions were proved in \cite{FW:12} using a Schur complement / Lyapunov-Schmidt reduction strategy, 
  for Schroedinger operators
 and for elliptic operators with smooth coefficients \cite{LWZ:18} by the same overall strategy. But since $\bbA_{\aspar}$ has discontinuous coefficients the  proof of Theorem \ref{Thm.weakLyapounovSchmidtred} proceeds via the weak formulation of the elliptic eigenvalue problem. 
Furthermore, since we also study Dirac points arising in higher energy bands, the natural energy form is not coercive; indeed it is infinite dimensionally indefinite.
We use the notion of $T$-coercivity to transfer the problem to a coercive setting, enabling bounded invertibility of the relevant operator to obtain the reduction; see  Section \ref{sec.diracpointthoerem}.   \\
A precise formulation of the above results with extensions to Dirac points in higher energy spectral bands appears in Sections \ref{sec.condP} through \ref{sec.diracpointthoerem}.

\item {\it Theorem \ref{th.asympteigen}, Corollaries \ref{cor.defphi1} and \ref{cor.limiteigenstateH1norm}, and Theorem \ref{th.fermveloc} Asymptotic expansions of Bloch eigenmodes and Dirac velocity at $\bK_*=\bK,\, \bK'$:} 
In Section \ref{sec-asympresult} we construct asymptotic expansions to any order in $g^{-1}$  of  a) the Dirac energy  $\lambda_D(g)$, and b) an associated  orthonormal basis, $\{\Phi_1(\aspar,\cdot),\Phi_2(\aspar,\cdot)\}$, of Bloch eigenfunctions. Such approximate eigenpairs are called {\it quasi-modes}.
 By general self-adjointness principles, presented in Appendix \ref{app.Appendixquasimode} in the weak formulation,
 the existence of Dirac quasi-modes with small residual implies that the actual
  Dirac eigenpairs (established earlier) are  within a neighborhood whose size is set by the size of this residual. Thus, this justifies the expansions. These mode expansions are then used
to obtain an asymptotic expansion  of the Dirac velocity, $v_D(\aspar)$, to any order in $g^{-1}$.
 \end{enumerate}

\subsection{Connection to previous analytical works on high contrast media}

High contrast elliptic operators have been studied  in media where the contrast and the geometrical structure (e.g. inclusion length scale) are coupled. For example,  the articles \cite{Cherednichenko:2016,Fortes:11,Suslina:2013,Zhikhov:05} concern the high contrast homogenization regime, in which the size of the unit cell depends on the material contrast, and the articles \cite{Figotin:98,Gomez:2021} concern media which contain asymptotically thin structures.
In this article, inclusions defined by $\sigma_\aspar$ have {\it fixed} geometry and we take the contrast, $\aspar$, between the two material parameters to be large. We can make use of the variational methods of Hempel and Lienau  \cite{Hem-00}, who  studied the limiting behavior of the band spectrum and obtained criteria for the opening of gaps in the spectrum of $\bbA_\aspar$ for $\aspar\gg1$. Further developments along these lines concerning the density of states of such media were obtained in \cite{Frie:02}. 
The main focus in \cite{Hem-00} is the use of high contrast to open gaps in the energy spectrum. In this paper we apply these methods together the symmetries of novel structures  to study important band structure properties such as Dirac points. We believe our approach can be used to study other classes spectral degeneracies; see, e.g., \cite{kmow:18_20}.

An alternative potential theoretic formulation of the spectral problem for operators with piecewise constant coefficients,
  in terms of  boundary integral equations, is developed in \cite{Ammari2009layer}.
   But unlike, the methods of this paper, this formulation is not easily adapted
if the media has strong heterogeneities or anisotropy, {\it i.e.} if the medium is not piecewise constant. In a square lattice geometry, using such a potential theoretic approach, the authors of  \cite{Lipton2017:creating} study  the question of the existence and width of  spectral gaps  for sufficiently large contrast. They provide, in particular for disc-shaped inclusions,
 a sufficient condition with a lower bound on the contrast (in terms of inclusions radii and inclusions relative distances)  to open a gap between consecutive energy bands. Recently, high contrast elliptic operators in honeycomb structures have been studied via a potential theoretic approach for the ``inverted'' case of a honeycomb-lattice of acoustic ``Minnaert bubbles'', in which $\sigma_g$ is large within the honeycomb inclusion set and bounded outside the inclusion set \cite{Ammari:2020}.  Their results apply to the two lowest bands and, so far, have no equivalent for higher energy bands.

Concerning asymptotic expansion of  eigenpairs in high contrast media,  we  mention the works \cite{Ammari2009layer,Nazarov:2020,Hut:2014,Lipton:2017:Bloch}. In \cite{Hut:2014} the case of a bounded domain with Dirichlet boundary conditions and high contrast  coefficient within the inclusions is studied.
 Using approximation by quasi-modes, they provide an asymptotic expansion of the eigenvalues  and spectral projectors to any order. In \cite{Ammari2009layer}, 
using Riesz homomorphic functional calculus and a potential theory, the authors give the leading order term of a  Bloch simple eigenvalue  in a  square lattice at any non-zero quasimomentum. By a similar technique and in the same geometry, the authors of   \cite{Lipton:2017:Bloch}  prove the analyticity of  simple Bloch eigenvalues or of the eigenvalue group, for the case of degenerate eigenvalues.
The coefficients involved in their series expansion are defined implicitly via the Riesz holomorphic functional calculus.
Finally,  recently in \cite{Nazarov:2020}, the first terms of the asymptotic  expansion of  Bloch eigenvalues were derived for the limiting case of a single disc-shaped inclusion whose closure touches the boundary of a square unit cell. 
\medskip

We conclude  by mentioning a number of interesting natural questions to consider going forward:\\
 (1) Are there results on the scaled limiting shape 
of dispersion surfaces which intersect in a Dirac point?  Such a result was proved for the first two bands of Schroedinger operators in the strong binding regime in \cite{FLW-CPAM:17}.\\
(2) Dirac points are known to occur in generic honeycomb Schroedinger operators \cite{FW:12}. The methods of \cite{FW:12} were used to prove the analogous result for divergence form operators with smooth coefficients \cite{LWZ:18}. It would be of interest to develop a method, which through the weak formulation, can handle divergence form operators with non-smooth coefficients, $\mathbb{A}_\aspar$.  \\
(3) And finally, it would be of great interest to consider the propagation of edge modes in the current context. For honeycomb media, these have been studied (a) for the Schroedinger equation with a sharply (zigzag-) terminated honeycomb structure \cite{FW:20}, and domain wall line-defects \cite{Drou-Wein:2020,Drou:19,FLW-CPAM:17}, and for (b) Maxwell's equations in a honeycomb structure with a domain wall \cite{LWZ:18}. It would be interesting to use the techniques of the present paper to study edge states for line-defect perturbations of $\bbA_\aspar$.

\subsection{Outline of the paper}
  In Section \ref{sec-analytic-prem}, we give a precise formulation of the mathematical problem
   and review the spectral theory of  elliptic operators in periodic  structures.\\
   In Section \ref{HL-thy}, we summarize the variational theory developed by Hempel and Lineau in \cite{Hem-00}. We provide some extensions of their theory concerning  strict monotonicity Bloch eigenvalues as functions in the contrast parameter $\aspar$, and  regarding  uniform convergence of band dispersion functions  on the Brillouin zone as $\aspar\uparrow\infty$. \\
 In Section  \ref{sec.condP},  building on  \cite{Hem-00}, we introduce a spectral isolation  condition  $\condS$ 
  (Definition \ref{Def.condS}), expressed in terms of eigenvalues of the Dirichlet Laplacian  for a single inclusion. We then establish a relation between this Dirichlet spectrum  and the $L^2_{\bk}$-spectrum (for $\bk\neq 0$) of $\bbA_g$ for large $\aspar$. Furthermore, we characterize the high contrast global behavior of dispersion surfaces and show the existence of a spectral gap for dispersion surfaces whose high contrast limit is given in terms of Dirichlet eigenvalues that satisfy the condition $\condS$. Condition $\condS$ always holds for the first Dirichlet eigenvalue. Hence, these results hold  for the lowest-lying (first) dispersion surface. \\
 In Section \ref{DP} we embark on the study of honeycomb operators $\bbA_\aspar$, using general results of the previous sections.
Section \ref{honey-sym} discusses the symmetries of honeycomb operators, operators which commute 
   with $\bbA_\aspar$, and important consequences for the spectral analysis.  
In connection with this, the Appendix \ref{sec-appendixcommutation}, deals with technical questions arising due to the discontinuities of the elliptic coefficient $\sigma_{\aspar}$. The notion of a Dirac point is defined 
    in Section \ref{DP-def}, and  sufficient conditions for their existence are given in  Theorem  \ref{Thm.weakLyapounovSchmidtred}.  In Section \ref{orbitals} we construct states which capture the limiting behavior of the  Floquet Bloch eigenspace of Dirac points; these states are $\bK_*-$quasi-periodic superpositions of translates of the Dirichlet eigenstates of the single isolated inclusion. In Section \ref{DP-exist} we 
 prove  Theorem \ref{th.degeneracy} on the existence of Dirac points, associated with Dirichlet spectrum of $-\Delta$ satisfying
  $\condS$, under the non-degeneracy condition that the Dirac velocity, $v_D(g)$, is non-zero.
  Condition $\condS$ always holds for the first Dirichlet eigenvalue and hence Theorem \ref{th.degeneracy} applies to give Dirac points 
 within the two first dispersion surfaces, modulo the non-degeneracy assumption, which we address numerically.
 \\
 In Section \ref{sec-asympresult},  we prove asymptotic expansions to any order in $\aspar^{-1}$ for: the Dirac eigenvalue,
  $\lambda_D(\aspar)$,  a natural choice of orthonormal of its corresponding $2-$ dimensional eigenspace, and for the Dirac velocity $v_D(\aspar)$.  We use the weak formulation of the  quasi-mode principle, discussed in Appendix  \ref{app.Appendixquasimode}.
 \\
In  Section \ref{numerics} we present the results of numerical simulations which illustrate
 our rigorous results for the first $2$ spectral bands of $\bbA_\aspar$. \\
 In Section \ref{sec.higherbandscirccase} we show, for the case of disc-shaped inclusions,  that the spectral isolation condition  $\condS$ is satisfied for all radial Dirichlet eigenpairs, {\it i.e.} those whose eigenvalues derive  from  zeros of the Bessel function $J_0(z)$. Hence, for this special geometry, the results of Sections \ref{sec.condP}-\ref{sec-asympresult} apply to give
 Dirac points at intersections of an infinite sequence of energy  band pairs. As an illustration, we present numerical simulations for the  $11^{th}$ and $12^{th}$ spectral bands. 
 \\
 In Section  \ref{sec.diracpointthoerem}, we prove Theorem \ref{Thm.weakLyapounovSchmidtred} on sufficient conditions
  for the existence of Dirac points. The proof uses a weak formulation of Lyapunov-Schmidt / Schur complement reduction scheme
   of previous work and makes use of $T-$coercivity to define an appropriate resolvent for the reduction, which applies when Dirac points sit among spectral bands above the first two.  \\
Finally, in the Appendix \ref{sec-extensionresults}, we extend our approach and results to a class of divergence form elliptic operators  which, in electromagnetism, model inhomogeneous and anisotropic  inclusions  in an heterogeneous and anisotropic bulk.
 
 \subsection{Notations, definitions and conventions}\label{sec-not}  
 $\bullet$ We denote by  $\NO$ and $\N$  respectively the set of non-negative integers and positive integers.\\
$\bullet$ For two Banach spaces $E$ and $F$, $B(E,F)$ is the Banach algebra of bounded linear operators from $E$ into $F$. We write $B(E)=B(E,E)$ .\\
$\bullet$ $\left\langle \cdot,\cdot \right \rangle _{_{E^*,E}}$ denotes the duality product between the Banach space   $E$ and  its dual $E^*$.\\
$\bullet$ In this paper, all the considered Hilbert spaces $\mathcal{H}$ are endowed with a complex inner product. This inner product and all the considered sesquilinear forms, defined on $\mathcal{H} \times \mathcal{H}$, are  antilinear with respect to the second variable.\\
$\bullet$  Except in the context of our definition of lattice, the symbol $\oplus$ refers to both the direct sum of two closed spaces in a Banach space and the orthogonal direct sum between two closed spaces in a Hilbert space. Unless explicitly specified otherwise, the symbol $\oplus$  refers to  an orthogonal direct sum.  
\\
$\bullet$ If $\bbA:D(\bbA)\subset \mathcal{H}\to \mathcal{H}$ is a selfadjoint (resp. normal) operator on a Hilbert space $\mathcal{H}$, one denotes by $\bbE_{\bbA}(\cdot)$ its associated spectral measure defined from the Borel sets of $\R$ (resp. $\C$) into the projection of $B(\mathcal{H})$ (see \cite{Dau:85-3,RS1,Sch:12}).  \\
\noindent $\bullet$ Let $\mathcal{H}$ be a Hilbert space, $\bbA:D(\bbA)\subset \mathcal{H}\to \mathcal{H}$ a unbounded linear operator and $\bbB: \calH \to \calH$ a bounded linear (or bounded anti-linear) operator.
One says that $\bbA$ commutes with $\bbB$, if $D(\bbA)$ is stable under $\bbB$ (i.e. $\bbB(D(\bbA))\subset D(\bbA)$) and if the commutator $[\bbA,\bbB]=\bbA \bbB-\bbB \bbA $ vanishes
on $D(\bbA)$.\\ For the particular case where  $\bbA$ is a self-adjoint and $\bbB$ is linear, this definition  is equivalent to the commutation of $\bbB$ with the spectral measure  $\bbE_{\bbA}(\cdot)$ of  $\bbA$ in the sense 
that the commutator $[\bbB,\bbE_{\bbA}(I)]$ vanishes on  $\mathcal{H}$ for any Borel sets $I$ of $\R$ (see Proposition 5 p. 145 of \cite{Dau:85-3}). \\
Furthermore if $\bbB$ is a normal operator, it  is equivalent to the commutation of the spectral measures of $\bbB$ with $\bbA$, i.e $\bbE_{\bbB}(J)D(\bbA)\subset D(\bbA)$ and $[\bbA, \bbE_{\bbB}(J)]$ vanish on $D(\bbA)$ for any Borel sets $J$ of $\C$ (by combining Proposition  5.27  p. 107--108 of \cite{Sch:12}  and Proposition 5 p. 145 of \cite{Dau:85-3}). Finally, this turns out to be also equivalent to the commutation of the  two spectral measures, namely $[\bbE_{\bbA}(I),\bbE_{\bbB}(J)]=0$ on $ \calH$, for any Borel sets $I$ of $\R$ and any Borel sets $J$ of $\C$ (see Proposition  5.27  p. 107--108 of \cite{Sch:12}).\\
$\bullet$ Let $R$ denote the matrix which rotates a vector in $\R^2$  clockwise by $2\pi/3$:
\begin{equation}\label{eq.defrotmat}
R=\begin{pmatrix}\displaystyle - \frac{1}{2} & \displaystyle\frac{\sqrt{3}}{2} \\[4pt] \displaystyle- \frac{\sqrt{3}}{2}   &\displaystyle -\frac{1}{2} \end{pmatrix}.
\end{equation}
Its eigenvalues and a choice of normalized eigenvectors are given by:
\begin{equation}\label{eq.defroteigenelem}
R \, \xi= \tau \, \xi \ \ \mbox{ and } \ \ R \, \overline{ \xi}= \overline{\tau}  \, \overline{ \xi}, \ \ \mbox{ where } \tau=\rme^{2\rmi \frac{\pi}{3} } \ \mbox{ and } \ \xi=\frac{1}{\sqrt{2}}\, ( 1 \, , \rmi )^{\top}.
\end{equation}
\\
$\bullet$ Triangular lattice, $\Lambda=\Z\bv_1\oplus\Z\bv_2$, where
\[ \bv_1=\Big(\frac{\sqrt{3}}{2},  \,\frac{1}{2}\Big)^{\top} \mbox{ and } \bv_2=\Big(\frac{\sqrt{3}}{2},  \,-\frac{1}{2}\Big)^{\top}.\]
\\
$\bullet$ Honeycomb, $\mathbb{H}= \Lambda_A\cup\Lambda_B$, where $\Lambda_J=\bv_J+\Lambda$
 with base points: 
 \[ \bv_A=(0,0)^{\top} \mbox{ and }  \bv_B=\Big(\frac{1}{\sqrt{3}},0\Big)^{\top}.\]
 \\
$\bullet$ Dual lattice $\Lambda^*=\mathbb{Z}\bk_1\oplus\mathbb{Z}\bk_2 $,  where
\[
 \bk_1=2\pi \Big(\frac{\sqrt{3}}{3},  \, 1\Big)^{\top} \mbox{ and } \bk_2=2\pi \Big(\frac{\sqrt{3}}{3},  \, -1\Big)^{\top}.
\] \\
 $\bullet$ A choice of hexagonal tile center is given by
$
\bx_c=\frac{1}{2}\left( \frac{1}{\sqrt{3}}, -1\right)^{\top} \ .
$ 
It is located at a vertex of the fundamental cell $\cell$; see Figure \ref{fig.funcell}). Note that $R\bx_c=-\bv_B$. \\ [4pt]
$\bullet$ $\bK$ and $\bK'$ are the $2$ independent high-symmetry quasi-momenta, at vertices of the Brillouin zone $\mathcal{B}$ depicted in Figure \ref{fig.funcell}:
\begin{equation}\label{eq.defKK'}
\bK=\frac{1}{3}(\bk_1-\bk_2)=\left(0, \frac{4 \pi}{3}\right)^{\top}  \ \mbox{ and } \ \bK'=-\bK.
\end{equation}

\noindent{\bf Acknowledgements:} M.C. wishes to thank the Department of Applied Physics and Applied Mathematics for its hospitality during the 2017-18, when this research was initiated. M.C. and M.I.W. were supported in part by Simons Foundation Math + X Investigator Award \#376319. M.I.W. was also supported by US National Science Foundation grants DMS-1412560, DMS-1620418 and DMS-1908657. 

%%%%
\section{Analytical preliminaries}\label{sec-analytic-prem}
In this section we first introduce the equilateral triangular lattice and honeycomb structure. We then define the honeycomb medium through the piecewise constant coefficient $\sigma_g$ of the operator $\bbA_{\aspar}$.  We employ notations and conventions similar to those  in  \cite{FW:12}.

\subsection{Triangular lattice, honeycomb structure and  the periodic medium}\label{sec-detatailedform}
The equilateral triangular lattice in $\mathbb{R}^2$ is given by:
$$
\Lambda=\mathbb{Z}\bv_1\oplus\mathbb{Z}\bv_2=\{ m_1 \bv_1+ m_2 \bv_2, \, (m_1,m_2)\in \mathbb{Z}^2 \} .
$$
Given base points $\bv_A$  and $\bv_B$,
we consider  
equilateral triangular sub-lattices:  $\Lambda_A= \bv_A+\Lambda  \mbox{ and } \Lambda_B= \bv_B+\Lambda$. 
The honeycomb structure, $\HH$,  is the union of these two interpenetrating  sub-lattices:
\begin{equation*}
\HH\ =\ \Lambda_A\ \cup\ \Lambda_B .
\label{Hdef}\end{equation*} 
As a fundamental domain in $\mathbb{R}^2_{\bx}$, we choose the diamond-shaped region, 
 which contains $\bv_A$ and $\bv_B$:
\begin{equation*}
\cell=\Big(-\frac{1}{\sqrt{3}},0\Big)^{\top}+\big\{ \theta_1 \bv_1+ \theta_2 \bv_2,  0 < \theta_i < 1 ,\, i=1,2 \big\}\ ;
\label{cell}\end{equation*}
see Figure \ref{fig.funcell}.  Denote the discrete translates of $\cell$ by lattice vectors by:
\begin{equation*}
\cell_{mn}={\cell}+m\bv_1+n\bv_2,
\label{cell-mn}
\end{equation*} 
where  $\cell=\cell_{00}$. 
The family of regions $\overline{\cell_{mn}},\ (m,n)\in\Z^2$  is a tiling of  $\bbR^2$.

\subsection{The honeycomb coefficient $\sigma_g(x)$}\label{sec:sigma}

In this section we define  $\sigma_{\aspar}(x)$ to be piecewise constant on the fundamental cell, $\cell$,  and extend it to be $\Lambda-$ periodic on $\R^2$. We begin with a discussion of constraints on the  subset of $\cell$ consisting of inclusions.\\ \\
 Throughout  this article we assume:
\begin{itemize}
\item [($\Omega$.i)] The inclusion $\cell^{A}$ is a non-empty simply connected open subset of $\Omega$ with a Lipschitz boundary $\partial\Omega^A$, and with $\bv_A\in\cell^{A}$.
\item [($\Omega$.ii)] The inclusion $\cell^B$ is the translate of $\cell^A$ by $\bv_B$.
 \[\textrm{We denote the inclusion subset of $\cell$ by  $\cell^{+}:=\cell^A\cup \cell^B$.}\]
\item[($\Omega$.iii)]  The inclusions are disjoint;\ $\overline{\cell^{A}}\cap \overline{\cell^{B}}=\emptyset$.
\item[($\Omega$.iv)]  The inclusion set is uniformly bounded away from the boundary, $\partial\cell$, of the fundamental cell;\ 
${\rm dist}\big(\overline{\cell^+},\partial\cell\big)>0$ .
\end{itemize}

Next  we impose conditions, used in the construction of a honeycomb symmetric medium, $\sigma_{\aspar}(x)$,
and the associated operator $\mathbb{A}_\aspar$.
We further require the following assumptions on $\cell^{A}$ and $\cell^{B}$ linked to the honeycomb symmetries:
\begin{itemize}
\item[($\Omega$.v)] $\cell^{A}$ is invariant under the $2\pi/3-$ rotation about origin $\bv_A={\bf 0}$. That is, \[R(\cell^{A})=\cell^{A},\ \textrm{where $R$ is the clockwise  $2\pi/3-$ rotation matrix};\]
we say that $\cell^{A}$ is centered at $\bv_A$.
\item[($\Omega$.vi)]  $\cell^{A}$ is invariant under inversion with respect to $\bv_A={\bf 0}$. That is,
 \[ \cell^{A}=- \cell^{A}=\{ - \bx: \bx \in  \cell^{A} \}.\]
\end{itemize}
Note that  ($\Omega$.v) and ($\Omega$.vi) together imply that $\cell^{A}$ is also $\pi/3$ rotationally invariant about $\bv_A$.

Let $\cell^-$ denote the part of $\cell$, which is outside the inclusion set:
 \[ \cell^-=\cell\setminus \overline{\cell^+}\] 
and thus
\[ \cell=\cell^+ \cup \partial \cell^+\cup \cell^-\ ;\qquad \textrm{see Figure \ref{fig.cell}.}\] 

Our periodic partial differential operator, $\mathbb{A}_\aspar$,  is specified by a piecewise constant  constitutive law, $\sigma_\aspar=\sigma_\aspar(\bx)$. We first define $\sigma_\aspar$ on $\cell$ by
\begin{align}
\sigma_{\aspar}(\bx)\ =\ \begin{cases}
1, &\ \bx\in \cell^{+}=\cell^A\cup \cell^B\\
g, &\ \bx\in \cell^- ,
\end{cases}
\label{sig-def}
\end{align}
and then extend $\sigma_\aspar$ to be defined on all $\R^2$
 as a  $\Lambda-$ periodic function. This extension is smooth across cell interfaces but has discontinuous jumps across inclusion boundaries.
Referring to Figure \ref{fig.cell}, an admissible choice of $\sigma_\aspar$ is obtained by 
 taking  $\sigma_\aspar\equiv1$ on all inclusions and $\sigma_\aspar\equiv g$ in their complement with respect to $\R^2$.
 \begin{figure}[!h]
\centering
 \includegraphics[width=0.8\textwidth]{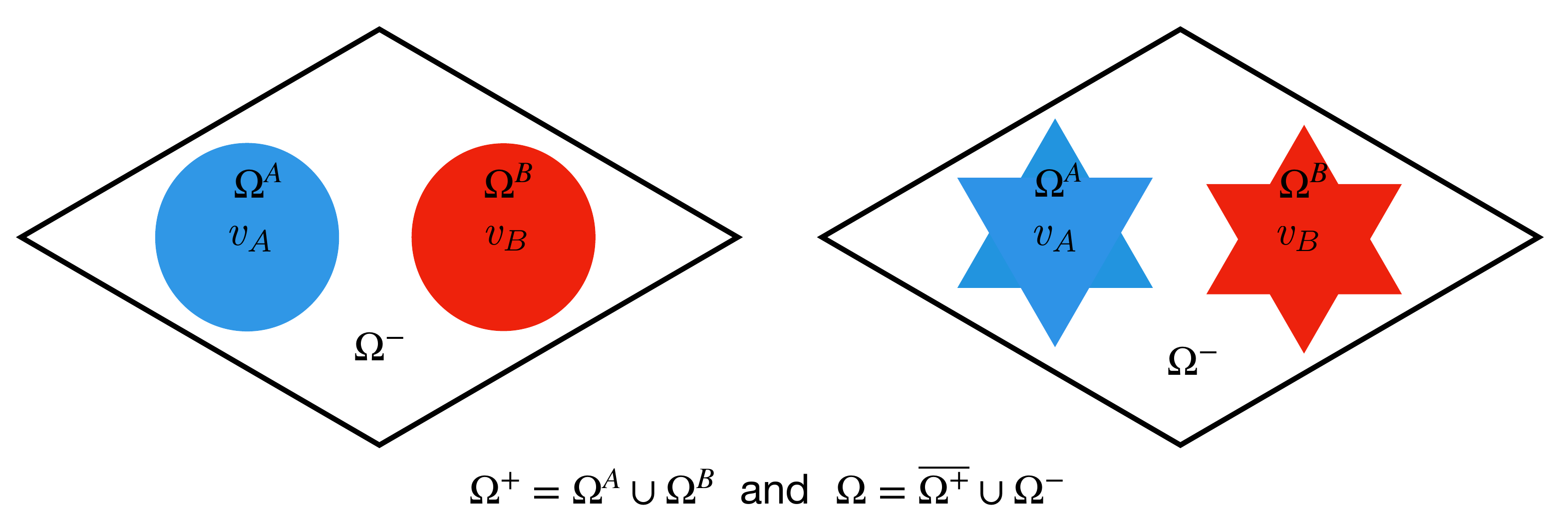}
   \caption{Examples of inclusion subsets, $\Omega^+$, of the fundamental cell, $\cell$,  satisfying conditions ($\Omega$.i)-($\Omega$.vi). Note that the boundaries of the inclusions in the right panel are Lipschitz continuous but not $C^1$.}
 \label{fig.cell}
\end{figure}

\subsection{A quick review of Floquet-Bloch theory}\label{sec-Floqth}
Let $\Lambda^*\subset (\R^2_{\bx})^*=\R^2_{\bk}$ denote the lattice which is dual to $\Lambda$:
$$
\Lambda^*=\mathbb{Z}\bk_1\oplus\mathbb{Z}\bk_2=\{ m_1 \bk_1+ m_2 \bk_2, \, (m_1,m_2)\in \mathbb{Z}^2 \} ,
\quad \textrm{where}\quad \bk_l \cdot \bv_{m}=2\pi \delta_{lm}.
$$
A choice of  fundamental cell in  $\R^2_{\bk}$,
 the {\it Brillouin zone},  is the closed regular hexagon $\mathcal{B}$  shown in Figure \ref{fig.funcell}.
The family of regions $\mathcal{B}_{mn}=\mathcal{B}+m\bk_1+n\bk_2$ where $(m,n)$ varies over $\Z^2$ tiles $\R^2_{\bk}$.

For any  $\bk \in \bbR^2$, let $L^2_{\bk}=L^2_{\bk}(\bbR^2/\Lambda)$ denote the  Hilbert space of $\bk-$ quasi-periodic functions:
\begin{equation}\label{eq.defL2k}
L^2_{\bk}:=\{ f\in L^2_{\mathrm{loc}}(\bbR^2)\mid  f(\bx+\bv)=\rme^{\rmi  \bk\cdot \bv  } f(\bx), \mbox{ for a.e } \bx\in \bbR^2 \mbox{ and all}\ \bv\in \Lambda  \},
\end{equation}
 with  inner product  $
(f,g)_{L^2_{\bk}}=\int_{\cell} f(x)  \overline{g(x)} \rmd x
$, defined  for all $f,g\in L^2_{\bk}$, 
and denote by $\| \cdot \|_{L^2_{\bk}}$ the associated norm. (When the context is clear, we omit the subscript $L^2_{\bk}$ from the norm notation.) Note that  $L^2_{\bk=0}=L^2(\R^2/ \Lambda)$, the space of $\Lambda-$ periodic functions.

Given a function $u$ in the Schwartz space $\mathcal{S}(\R^2)$, define its Floquet-Bloch transform:
\begin{equation}
\label{FB-def}
(\calF u)(\bx,\bk)=\sum_{\bv \in \Lambda}\rme^{\rmi \bk\cdot \bv }u(\bx-\bv), \ \textrm{for\ }\ \bx\ \in \bbR^2_{\bx} \mbox{ and }\ \bk\in \bbR^2_{\bk}.
\end{equation}
For a fixed $\bx \in \cell$, the expression \eqref{FB-def}  extends $\bk\mapsto (\calF u)(\bx,\bk)$ to a $\Lambda^*-$  periodic function on $\R^2_{\bk}$:
\[ (\calF u)(\bx,\bk+\bkg)\ =\ (\calF u)(\bx,\bk)\ \textrm{ for all }\ \bkg\in\Lambda^*\ .\]
Thus we may restrict the quasimomentum $\bk$ to $\mathcal{B}$.
Furthermore, for a fixed $\bk \in \mathcal{B}$,  the function is $\bx \mapsto (\calF u)(\bx,\bk)\in L^2_{\bk}$:
$$
\left(\calF u\right)(\bx+\bv,\bk)=\rme^{\rmi \bk\cdot \bv } \, \left(\calF u\right)(\bx,\bk), \  \textrm{for all}\ \bv \in \Lambda \ \textrm{ and a.e. $\bx\in \mathbb{R}^2$}.
$$
For all $u\in \mathcal{S}(\mathbb{R}^2)$, one can show the following Plancherel-type identity:
\begin{equation}
\int_{\bbR}|u(\bx)|^2 \mathrm{d}\bx=\frac{1}{|\mathcal{B}|}\int_{\mathcal{B}}\int_{\cell}|\calF{u}(\bx,\bk) |^2\mathrm{d} \bx  \, \mathrm{d} \bk ,
\label{planch}\end{equation}
where $|\mathcal{B}|=|\bk_1 \wedge  \bk_2|=2*(2\pi^2)/\sqrt{3}$ is the area of the Brillouin zone $\mathcal{B}$.
The identity \eqref{planch}  extends  by density $\mathcal{F}/{\sqrt{|\mathcal{B}|}}$ as  a linear isometry from $L^2(\mathbb{R}^2)$ to $L^2(\mathcal{B}, L^2_{\bk})$. Moreover, this isometry is unitary  (since surjective) and we have the following inversion formula;  all $u\in L^2(\bbR^2)$ have 
the following $L^2_{\bk}$ decomposition:

\begin{equation*}
u(\bx)= \frac{1}{|\mathcal{B}|}\big(\calF^{*}\mathcal{F}u\big)(\bx)=\frac{1}{|\mathcal{B}|}\int_{\bk\in\mathcal{B}}  \mathcal{F}(u)(\bx,\bk) \,  \rmd \bk, \, \mbox{ for  a.e.  $\bx \in \bbR^2$. } 
\label{F*def}
\end{equation*}

In analogy with \eqref{eq.defL2k} we introduce, for any $\bk\in \R^2$, the Sobolev spaces $H^1_{\bk}=H^1_{\bk}(\bbR^2/\Lambda)$:
$$
H^1_{\bk}:=\{ f\in H^1_{\mathrm{loc}}(\bbR^2)\mid  f(\bx+\bv)=\rme^{\rmi  \bk\cdot \bv  } f(\bx), \mbox{for  a.e. } \bx\in \bbR^2 \mbox{ and all} \  \bv\in \Lambda  \},
$$
endowed with the following inner product,
$$
(f,g)_{H^1_{\bk}}=\int_{\cell} f(\bx) \overline{g(\bx)} \rmd \bx+ \int_{\cell}  \nabla f(\bx) \cdot \overline{\nabla g(\bx)}  \, \rmd \bx
$$
and denote by $\| \cdot \|_{H^1_{\bk}}$ the associated  norm.
\begin{remark}\label{H1bk} 
We emphasize that functions $u\in H^1_{\bk}$,  are such that the trace of the map $\bx \mapsto \rme^{-i\bk\cdot\bx}u(\bx)$ has no jump  across boundaries between period cells and is $\Lambda-$ periodic. For the definition of traces (restriction operators) in terms of Sobolev spaces, we refer for e.g. to \cite{Girault:1986,Gri:85,Monk:03}.
\end{remark}

\subsection{The operator $\bbA_{\aspar}$ and its spectral theory: general properties}

Let $\sigma_{\aspar }$ denote the piecewise constant function and $\Lambda-$ periodic ($L^2_{\bk=0}=L^2(\R^2/\Lambda)$) function introduced in Section \ref{sec:sigma}. The operator
\begin{equation*}
\bbA_{\aspar}=-\nabla \cdot \sigma_{\aspar }  \nabla\ :\  {D}(\bbA_{\aspar})\subset L^2(\bbR^2) \to L^2(\bbR^2)
\label{bbA}\end{equation*}
 with domain 
\begin{equation*}
{D}(\bbA_{\aspar})=\{ u\in  H^1(\bbR^2) \mid -\nabla \cdot \sigma_{\aspar }  \nabla u \in L^2(\bbR^2) \},
\label{dom-bbA}
\end{equation*}
is a positive self-adjoint operator.
Furthermore, $\bbA_{\aspar}$ is unitarily equivalent to a direct fiber integral over self-adjoint positive operators $\bbA_{\aspar, \bk}$  (see for e.g. \cite{Kuch:01}):
\begin{equation}\label{eq.intdirec}
\bbA_{\aspar}=\frac{1}{|\mathcal{B}|}\calF^{*}\, \left( \int_{\bk\in\mathcal{B}}^{\oplus} \bbA_{\aspar, \bk} \, \rmd \bk \right) \calF\   \ \mbox{ meaning }  \ \calF \big(\bbA_{\aspar} u\big)(\cdot,\bk)= \bbA_{\aspar,\bk} \, \calF u(\cdot,\bk) \mbox{ for a.e. } \bk \in \mathcal{B},
\end{equation}
where $\calF$ denotes the Floquet-Bloch transform defined in Section \ref{sec-Floqth}. For any $\bk \in \R^2$
 \[ \bbA_{\aspar, \bk}: {D}(\bbA_{\aspar, \bk}) \to L^2_{\bk} \]
  denotes the operator $-\nabla \cdot \sigma_{\aspar }  \nabla$  with domain given by:
$$
 {D}(\bbA_{\aspar, \bk}):=\Big\{ u \in H^1_{\bk}\mid  -\nabla \cdot \sigma_{\aspar }  \nabla u \in L^2_{\bk}\Big\}.
$$

\begin{remark}\label{DomAbk} 
Let $u\in{D}(\bbA_{\aspar, \bk})$. Then,  the trace of the function $\bx \mapsto \rme^{-i\bk\cdot\bx}u(\bx)$ and its Neumann trace $\partial (\rme^{-\rmi\bk\cdot\bx}u )/\partial \bn$ are $\Lambda-$ periodic, and have no jump across the boundaries between periodic cells.  Here, the orientation of the  unit normal vector $\bn$ on $\cup_{m,n\in \mathbb{Z}} \partial \cell_{mn}$ is chosen so that $\bn$ is constant  on the sides of the periodic cells that are obtained from each other by translation of lattice vectors,  except of course at the lattice vertices where it is not defined.
\end{remark}

For $\bk \in \R^2$, $\bbA_{\aspar, \bk}$  has a compact resolvent; see  e.g. \cite{RS4,Kuch:01}.
Its spectrum consists of a sequence of non-negative eigenvalues:
 \begin{equation*} 0\le\lambda_1(\aspar;\bk)\le \lambda_2(\aspar;\bk)\le\cdots\le \lambda_b(\aspar;\bk)\le\cdots,
 \label{spec-k}\end{equation*}
 listed with multiplicity, and tending to positive infinity. For any $\kappa\in \Lambda^*$, one has  $L^2_{\bk+\bkg}=L^2_{\bk}$, $H^1_{\bk+\bkg}=H^1_{\bk}$, $D(\bbA_{\aspar,\bk+\bkg})=D(\bbA_{\aspar,\bk})$ and thus $\bbA_{\aspar,\bk+\bkg}=\bbA_{\aspar,\bk}$.
The maps 
\begin{equation}
 \bk\in \R^2\mapsto \lambda_n(\aspar;\bk) \textrm{ are real-valued, Lipschitz continuous and $\Lambda^*$-periodic functions;}  
 \label{lam-lip}
 \end{equation} 
 see, for example, \cite{avron-simon:78,Conca:1997,FW:14}.
The maps $\bk\mapsto\lambda_j(g;\bk)$ are called band dispersion maps.
Since the band  dispersion maps are $\Lambda^*$-periodic, one can restrict their study to the Brillouin zone $\mathcal{B}\simeq\mathbb{T}^2$.
Their graphs over $\mathcal{B}$ are called  \emph{the dispersion surfaces} of $\bbA_\aspar$.
By \eqref{eq.intdirec} and \eqref{lam-lip},  the spectrum of $\bbA_{\aspar}$, $\sigma(\bbA_{\aspar})$, can be obtained from the spectra of the fiber operators:
\begin{equation}\label{eq.specA}
\sigma(\bbA_{\aspar})=\bigcup_{n=1}^{\infty} \lambda_n(\aspar;\mathcal{B}) \ .
\end{equation}
The \emph{band structure} of  $\bbA_{\aspar}$ refers to the collection of dispersion (eigenvalue) maps and their associated eigenmodes.

\subsection{Eigenvalue bracketing}\label{bracket}

To estimate the location  of $\sigma(\bbA_{\aspar})$, we make use of the  Dirichlet and Neumann spectra $\bbA_{\aspar}$ on $\cell$. Introduce the operators:
\begin{align}
\bbA_{\aspar}^{Dir,\cell} &:=-\nabla \cdot \sigma_{\aspar }  \nabla, \mbox{ where }\  {D}(\bbA_{\aspar}^{Dir,\cell})=\{u \in H^1_0(\cell)\mid -\nabla \cdot \sigma_{\aspar }  \nabla u \in L^2(\cell)   \}, \label{A-Dir}
\\  
\bbA_{\aspar}^{Neu,\cell}& :=-\nabla \cdot \sigma_{\aspar }  \nabla ,\mbox{ where } {D}(\bbA_{\aspar}^{Neu,\cell})=\big\{u \in H^1(\cell)\mid -\nabla \cdot \sigma_{\aspar }  \nabla u \in L^2(\cell) \mbox{ and } \frac{\partial u}{\partial \bn }=0   \big\}, \label{A-Neu}
\end{align}
where $ \bn $ denotes unit normal vector on $\partial\Omega$, which points exterior to $\Omega$.
The operators $\bbA_{\aspar}^{Dir,\cell}$ and $\bbA_{\aspar}^{Neu,\cell}$ are  self-adjoint, positive and have compact resolvent. Moreover,  $\bbA_{\aspar}^{Dir,\cell}$ is positive definite. Their spectra consist of a sequence of real  eigenvalues of finite multiplicity and tending to infinity. We list them (with multiplicity) as:
\begin{align*}
&0<\lambda_1^{Dir,\cell}(\aspar)\leq \lambda_2^{Dir,\cell}(\aspar) \leq  \ldots \lambda_n^{Dir,\cell}(\aspar) \leq \ldots \\
& 0=\lambda_1^{Neu,\cell}(\aspar)\leq \lambda_2^{Neu,\cell}(\aspar) \leq \ldots \leq  \lambda_n^{Neu,\cell}(\aspar) \leq \ldots.\end{align*}
The following lemma provides useful upper and lower bounds on the maps of the Floquet-Bloch dispersion maps:  $\lambda_n(\bk;\aspar)$ in terms of $\lambda_n^{Dir,\cell}(\aspar)$ and $\lambda_n^{Neu,\cell}(\aspar)$.
\begin{lemma}\label{lem.inequaleigenval}
For any  fixed $\aspar>0$, one has for all $n\ge1$ and all $\bk\in\mathcal{B}$: 
\begin{equation}\label{eq.ineqvp}
 \lambda_n^{Neu,\cell}(\aspar)  \leq\lambda_n(\bk;\aspar)\leq  \lambda_n^{Dir,\cell}(\aspar)\ .
 \end{equation}
Hence,  by \eqref{eq.ineqvp} and \eqref{eq.specA}:
\begin{equation*}\label{eq.inclus}
\sigma(\bbA_{\aspar}) \subset \bigcup_{n=1}^{\infty}[\lambda^{Neu,\cell}_n(\aspar),\lambda^{Dir,\cell}_n(\aspar)].
\end{equation*}
\end{lemma}

\begin{proof}
The proof is  based on the application of the min-max principle  to the quadratic (sesquilinear) form 
associated with the operator  $-\nabla \cdot \sigma_{\aspar }\nabla u$
\begin{equation*}\label{eq.sesqatau}
a_{\aspar}(u,v)=\int_{\cell} \sigma_{\aspar }(\bx)\nabla u(\bx) \cdot  \overline{\nabla v(\bx)} \, \rmd  \bx, 
\end{equation*}
with the form domains: $ H^1_0(\cell) $, $ H^1_{\bk}$ and $ H^1(\cell)$, associated with, respectively, Dirichlet, $\bk-$quasi-periodic  and Neumann boundary conditions. By  the min-max characterization of eigenvalues \cite{Dau:85-3,Henrot:2006,RS4} of selfadjoint  positive operators with compact resolvent:
\begin{enumerate}
\item Neumann boundary conditions:
\begin{equation}\label{eq.minmaxneu}
\lambda_n^{Neu,\cell}(\aspar)=\min_{\substack{V\subset H^1(\Omega)\\ \mathrm{dim} V=n}}
 \quad \max_{\substack{u\in V\\ \|u\|_{L^2(\Omega)}=1}} a_{\aspar}(u,u) \,;
 \end{equation}
\item $\bk-$ quasi-periodic boundary conditions:
\begin{equation}\label{eq.minmaxkpseudo}
\lambda_n(\aspar;\bk)=\min_{\substack{V\subset H^1_{\bk}\\ \mathrm{dim} V=n}} \quad
 \max_{\substack{u\in V\\ \|u\|_{L^2(\Omega)}= 1}} a_{\aspar}(u,u)\,  ;
\end{equation}
\item Dirichlet boundary conditions:
\begin{equation}\label{eq.minmaxdir}
\lambda_n^{Dir,\cell}(\aspar)=\min_{\substack{V\subset H^1_{0}(\Omega)\\ \mathrm{dim} V=n}}
 \quad \max_{\substack{u\in V\\ \|u\|_{L^2(\Omega)}= 1}} a_{\aspar}(u,u) \, .
\end{equation}
\end{enumerate}
The inequality \eqref{eq.ineqvp} follows immediately from the relations: 
$H_{0}^{1}(\cell)\subset H^1_{\bk}\subset H^1(\cell)$, where we regard $H^1_{\bk}$ as the subspace $H^1(\Omega)$ consisting of $H^1_{\bk}$
 functions restricted to $\cell$.
\end{proof}

\section{High-contrast ($g\gg1$) behavior of dispersion maps }\label{HL-thy}

The article \cite{Hem-00} contains general results
on the large $\aspar$ behavior of the  band dispersion functions, $\bk\mapsto \lambda_n(\aspar;\bk)$; see also \cite{Frie:02}.
 In this section we review and extend their results  and in later sections apply them honeycomb operators $\bbA_{\aspar,\bk}$. For the results of this section we  do not require
 honeycomb symmetry.

\subsection{Limit of the  fiber operators $\bbA_{\aspar,\bk}$}

Recall that $\Omega^+$ denotes the inclusion subset within the fundamental cell, $\cell$; see Figure \ref{fig.funcell}.  
For each $(m,n)\in\Z^2$ define
\begin{align}\label{eq.defOmegapm}
&\textrm{Translates of $\Omega^+$: For $(m,n)\in\Z^2$, }\quad \cell_{mn}^{+} =\ \cell^+ +\ m\bv_1\ +\ n\bv_2\ \subset\ \cell_{mn}, \nn\\
& \textrm{Union over all translates of inclusions:}\quad \boldcell^+ =\ \cup_{(m,n)\in\Z^2}\ \cell^+_{mn}\ \subset\ \R^2,\ \ \nn  \\
& \textrm{The bulk:}\quad  \boldcell^- =\bbR^2\setminus \overline{\boldcell^+}.
\end{align}
We introduce the closed subspaces of $L^2_{\bk}$ and $H^1_{\bk}$ consisting of functions that vanish outside $\boldcell^+$:
\begin{equation}\label{eq.defL2D+}
\widetilde{L}^2_{\bk} =\ \{u\in L^{2}_{\bk} \mid u(\bx)=0 \mbox{ a.e.  on }  \overline{\boldcell^-} \}  \mbox{ and }  \widetilde{H}^{1}_{\bk}= \{u\in H_{\bk}^{1} \mid u(\bx)=0 \mbox{ a.e.  on }\overline{ \boldcell^-} \}.
 \end{equation}

Let $\bk\in \mathcal{B}$ be fixed.  In \cite{Hem-00} it is proved that  the positive operator $\bbA_{\aspar,\bk}$  converges in the norm-resolvent sense as $\aspar\to +\infty$ via a study of  the convergence of the associated sesqulinear form:
\begin{equation}\label{eq.sesquilinearformakg}
a_{\aspar, \bk}(u,v)=(\bbA^{\frac{1}{2}}_{\aspar,\bk}u, \bbA^{\frac{1}{2}}_{\aspar,\bk} v )_{L^2_{\bk}}=\int_{\cell} \sigma_{\aspar} \nabla u \cdot \overline{ \nabla v }\ \mathrm{d} \, \bx ,  \quad\textrm{for all}\  u,v \in \ D(\bbA_{\aspar,\bk}^{\frac{1}{2}})=H^1_{\bk}.
\end{equation}
We present a short outline of their reasoning, which uses a result of \cite{Sim:78} on monotone quadratic forms. 

For any fixed $u\in H^1_{\bk}$, the function $\aspar \mapsto a_{\aspar, \bk}(u,u)$ is increasing for positive $\aspar$.
Hence, we may define:
$$
a_{\infty,\bk}(u,u)=\sup_{g>0} a_{\aspar, \bk}(u,u)= \lim_{\aspar \to \infty}a_{\aspar,\bk}(u,u),
$$
with domain
$$
\domain(a_{\infty,\bk})=\Big\{ u \in H^{1}_{\bk} \mid \sup_{\aspar>0} a_{\aspar, \bk}(u,u) <\infty \Big\}.
$$
Let us now characterize $\domain(a_{\infty,\bk})$.
A function $u\in  H^{1}_{\bk}$ belongs to  $\domain(a_{\infty,\bk})$ if and only for some constant $C>0$:
$$
a_{\aspar, \bk}(u,u)=\int_{\cell} \sigma_{\aspar,\bk } |\nabla u|^2 \, \rmd  \bx\leq C ,  \quad \textrm{for all}\  \aspar>0.
$$
The latter condition implies that:
\begin{equation}\label{eq.conddom}
  \aspar \ \int_{\cell^-}|\nabla u|^2 \rmd \bx \leq C, \quad  \textrm{for all}\   \aspar>0,
  \quad\textrm{and hence}\ \nabla u=0\quad \textrm{a.e. in}\quad  \cell^-.
\end{equation}
Since $\cell^-$ is open and connected (recall $\cell^A$ is simply connected), it follows that $u(\bx)=$ a constant a.e. in  $\cell^-$.
Furthermore, since $H^1_{\bk}$ functions are $\bk$ quasi-periodic and belong to $H^1_{loc}(\R^2)$ (their trace is continuous across the boundary of the cells and the boundary of the inclusions), we conclude that $u(\bx)$ is a.e. equal to a constant on $\overline{\boldcell^{-}}$. There are now two cases: $\bk\ne0$ and $\bk=0$. 

\subsubsection{The limiting operator $\bbAinfDk$ for  $\bk\neq 0$}

 For $\bk\in \mathcal{B}\setminus \{ 0\}$, non-zero constant functions on $\overline{\boldcell^{-}}$ do not belongs to $H^1_{\bk}$. Since $u$ is a.e. constant on $\overline{\boldcell^{-}}$, it follows  that $u(\bx)=0$ for  a.e. $\bx$ in $\overline{\boldcell^{-}}$. Hence, $\domain(a_{\infty,\bk})=  \widetilde{H}^{1}_{\bk}$; see \eqref{eq.defL2D+}.

Thus, for all $\bk \in \mathcal{B}$, the limiting sesqulininear form is given by:
$$
a_{\infty, \bk}(u,u)=\lim_{\aspar \to \infty}a_{\aspar,\bk}(u,u)=\int_{\cell^+}  |\nabla u|^2 \, \rmd  \bx \, , \quad  \textrm{for all}\  u \in \domain(a_{\infty,k})= \widetilde{H}^{1}_{\bk}.
$$
For all $\bk\ne0$ we may now associate,  to the limiting form $a_{\infty, \bk}$,   the self-adjoint positive definite operator $\bbAinfDk: \domain(\bbAinfDk) \subset  \widetilde{L}^{2}_{\bk} \to  \widetilde{L}^{2}_{\bk}$:
\begin{align}\label{eq.defDIropinf}
\left(\ \bbAinfDk u\ \right)(\bx)\ =\ 
\begin{cases} -\Delta u(\bx), &  \mbox{ for  a. e. } \bx\in{\boldcell^+}\\ 
\ 0\ & \mbox{ for  a. e. }  \bx\in  \overline{ \boldcell^-}=\mathbb{R}^2 \setminus \boldcell^+,
\end{cases}
\end{align}
with domain:
$$
\domain(\bbAinfDk)=\{ u \in \domain(a_{\infty,\bk})= \widetilde{H}^1_{\bk} \mid   \bbAinfDk u \in  \widetilde{L}^2_{\bk} \}. 
$$
The operator $\bbAinfDk$ is  a positive definite self-adjoint operator with a compact resolvent. 
It has a discrete set of real and strictly positive eigenvalues:
 \begin{equation*}
 0<\delta_1\leq \delta_2 \ldots  \leq \delta_n \leq \ldots,
 \label{delta_n}
 \end{equation*} 
 listed with multiplicity and tending to $+\infty$. 
By \eqref{eq.defDIropinf},  the sequence $(\delta_n)_{n\geq 1}$ is independent of $\bk$ since  
it coincides with the sequence of eigenvalues (counted also with multiplicity) of the Dirichlet Laplacian:
\begin{equation}\label{eq.defADinfty}
\LaplDcellp:=-\Delta :\ \domain(\LaplDcellp)\subset L^2(\cell^+) \to  L^2(\cell^+)
\end{equation}
where its domain ${D}(\LaplDcellp):=\{u\in H^1_0( \cell^+)\mid  \Delta u \in L^2(\cell^+)\}$ with $\cell^+=\cell^A\cup \cell^B$.
 The following result on norm-resolvent convergence of $\bbA_{\aspar,\bk}$  as $\aspar\to +\infty$ is proved in
\cite{Hem-00}:
\begin{theorem}[Norm resolvent convergence of $\bbA_{\aspar,\bk}$ to $\bbAinfDk$] \label{thmsaympAtau}
Let $\bk\in \mathcal{B}\setminus \{ {\bf 0}\}$. Then,  for all $\zeta \in \bbC\setminus \bbR$:
\begin{equation}\label{eq.strongresconv}
R_{\bbA_{\aspar,\bk}}(\zeta)=\left(\bbA_{\aspar,\bk}-\zeta\right)^{-1} \longrightarrow R_{\bbAinfDk }(\zeta)= \left( \bbAinfDk-\zeta\right)^{-1}  , \mbox{ as } \aspar \to \infty,  \mbox{ in }B( L^2_{\bk}).
\end{equation}  
It follows that the eigenvalues of $\bbA_{\aspar,\bk}$ converge to those of $\bbAinfDk$: Fix $\bk\ne0$,
$$
\textrm{for any $n>0$},\qquad \lambda_n(\aspar;\bk) \to \delta_n \ \mbox{ as }  \ \aspar \to +\infty.
$$
\end{theorem}

\begin{remark}\label{rem.extensionres}
The resolvent $R_{\bbAinfDk }(\zeta)$ acts on a closed proper subspace  $\widetilde{L}^2_{\bk}$ of $L^2_{\bk}$.
Thus, to discuss norm resolvent convergence  in \eqref{eq.strongresconv}, one needs to extend  the resolvent $R_{\bbAinfDk}( \zeta)$ to the whole space $L^2_{\bk}=\widetilde{L}^2_{\bk}{\small\oplus} (\widetilde{L}^2_{\bk})^{\perp}$. 
This is done (see \cite{Hem-00,Sim:78}) by setting $R_{\bbAinfDk}(\zeta) u=0$  for $u\in (\widetilde{L}^2_{\bk})^{\perp}=\{u \in  L^2_{\bk} \mid u=0 \mbox{ a.e. on } \boldcell^+ \} $.
\end{remark}

\subsubsection{The limiting operator $\bbAinfDk$ for  $\bk= 0$}
For $\bk={ \bf 0}$ (periodic-case), the constant functions  belong to $H^1_{\bf 0 }=H^1(\R^2/ \Lambda)$. Thus, in contrast to the case when $\bk\ne0$,  we can only conclude from \eqref{eq.conddom} that $u$ is constant a.e. on $\overline{\boldcell^{-}}$.  Thus, $\domain(a_{\infty,{ \bf 0}})=  \widetilde{H}^{1}_{\bf 0 }\oplus \operatorname{span}({\bf 1})$ where ${\bf 1}$ is the constant function equal to one on all $\bbR^2$.  Hence, for $\bk=0$ the limiting form $a_{\infty,{ \bf 0}}$ is given by:
$$
a_{\infty,{ \bf 0}}(u,u)=\lim_{\aspar \to \infty}a_{\aspar,\bk}(u,v)=\int_{\cell^+}  |\nabla u|^2 \, \rmd  \bx \, , \quad  \forall u \in \domain(a_{\infty,{\bf 0}})=  \widetilde{H}^{1}_{\bf 0 }\oplus \operatorname{span}({\bf 1}).
$$ 
We note that the latter direct sum is not orthogonal.
There exists a unique limiting positive self-adjoint operator  $\bbA_{\infty,{ \bf 0}}$  associated with $a_{\infty, {\bf 0}}$. It acts on the closed subspace  $\widetilde{L}^{2}_{\bf 0}\oplus \operatorname{span}({\bf 1})$  of $L^2_{{ \bf 0}}$  and its domain is given by $\domain(\bbA_{\infty,{ \bf 0}}):=\{ u \in \widetilde{H}^{1}_{\bf 0 }\mid \tilde{A}_0 u\in  \widetilde{L}^{2}_{\bf 0} \}\oplus \operatorname{span}({\bf 1})$ where $\tilde{A}_0 u$ is defined by $\tilde{A}_0 u=-\Delta u$ a.e  on $\boldcell^+$ and $\tilde{A}_0 u=0$ a.e on  $\overline{\boldcell^-}$.
The description of the limiting operator requires much more technical study than for the case $\bk\neq 0$. This is carried out  in  \cite{Hem-00}. We do not provide the details here.
  $\bbA_{\infty,{ \bf 0}}$ has a compact resolvent and  its spectrum 
consists  of a non-negative sequence $(\nu_n)$ of eigenvalues:  
\begin{equation*}
0\leq\nu_1\leq \nu_2 \leq \ldots \leq \nu_n \leq  \ldots ,\label{nu_n}
\end{equation*}  
listed with multiplicity and tending to $+ \infty$. The following result is proved in \cite{Hem-00}.
\begin{theorem}[$\bbA_{\aspar,{\bf 0}}\to\bbA_{\infty, {\bf 0}}$ in the norm resolvent sense]\label{thmsaympANinf}
Let $\bk={\bf 0}$. Then we have for all $\zeta \in \bbC\setminus \bbR$:
\begin{equation}\label{eq.strongresconvn}
R_{\bbA_{\aspar,{\bf 0}}}(\zeta)=\left(\bbA_{\aspar,{\bf 0}}-\zeta\right)^{-1} \longrightarrow R_{\bbA_{\infty,{\bf 0}}
 }(\zeta)= \left( \bbA_{\infty,{\bf 0}}-\zeta\right)^{-1}  , \mbox{ as } \aspar \to \infty,  \mbox{ in }B( L^2_{{\bf0}}).
\end{equation}  
Hence, 
$$
\textrm{For all\quad  $n>0$},\qquad \lambda_n(\aspar;{\bf 0}) \to \nu_n \ \mbox{ as }  \ \aspar \to +\infty.
$$
\end{theorem}
Remark \ref{rem.extensionres} applies to  Theorem  \ref{thmsaympANinf}; one defines resolvent norm convergence \eqref{eq.strongresconvn} by extending the resolvent $R_{\bbA_{\infty,{\bf 0}}}(\zeta)$ by $0$ on the orthogonal complement  of $\big(\widetilde{L}^{2}_{\bf 0}\oplus \operatorname{span}({\bf 1})\big)$ in $L^2_{{ \bf 0}}$.

We conclude this section by recalling a further result of  \cite{Hem-00}, proved by similar techniques to those above, 
 showing   that  limiting eigenvalues of  $\bbAinfDk$, for $\bk\ne0$,  and $\bbA_{\infty,{\bf 0}}$
also arise as limits of the eigenvalues of the operators $\bbA^{Dir,\cell}_\aspar$ and  $\bbA^{Neu,\cell}_\aspar$,  
  introduced in \eqref{A-Dir}-\eqref{A-Neu}. By the min-max characterization of eigenvalues, \eqref{eq.minmaxdir} and \eqref{eq.minmaxneu}, for any $n\ge1$,   $\aspar \mapsto \lambda^{Dir,\cell}_n(\aspar)$ and $ \aspar \mapsto \lambda^{Neu,\cell}_n(\aspar)$ are increasing functions. Hence they either converge or diverge to $+\infty $ as  $\aspar\to+\infty$. In \cite{Hem-00}  the following is proved: 
  \begin{theorem}\label{thm.dirsneumlimitpec}
The spectrum of the operators $\bbA_{\aspar}^{Dir,\cell}$ and  $\bbA_{\aspar}^{Neu, \cell}$ converge respectively to the spectra of $\bbAinfDk$ for $\bk\in \mathcal{B}\setminus \{ {\bf 0}\}$ and $\bbA_{\infty,{\bf 0}}$ in the sense that for all $n>0$:
$$
\lambda_n^{Dir,\cell}(\aspar) \to \delta_n  \quad \mbox{ and } \quad  \lambda_n^{Neu,\cell}(\aspar) \to \nu_n,  \mbox{ as } \aspar \to \infty.
$$
\end{theorem}

\subsection{Limiting behavior of dispersion surfaces and the opening of spectral gaps}
We study the uniform convergence of the dispersions surfaces  $\bk \mapsto \lambda_n(\aspar; \bk)$ defined on the Brillouin-zone $\mathcal{B}$ and also on the existence of a criterion for the existence of gaps between these surfaces. Although not explicitly stated in \cite{Hem-00}, a direct consequence is the  uniform convergence of  the dispersion surfaces $\bk\mapsto \lambda_n(\aspar; \bk)$ on any compact subset  of the Brillouin zone $\mathcal{B}$ which does not contain ${\bf 0}$. We precede the formulation of this result with a lemma on the strict monotonicty of the dispersion maps with respect to the contrast parameter $\aspar$. This property is not discussed in \cite{Hem-00}. 
\begin{lemma}\label{lem.monoticitydispcurv}
\begin{enumerate}
\item Let  $n=1$.  For all $\aspar>0$, $\lambda_1(\aspar;{\bf 0})=0$. Furthermore, for fixed  $\bk \neq 0$, $\aspar\in\R_+ \mapsto \lambda_1(\aspar;\bk)$  is  strictly increasing.  
\item  Let $n\ge2$ and fix $\bk \in \mathcal{B}$. Then, the function $\aspar\in\R_+ \mapsto \lambda_n(\aspar;\bk)$  is  strictly increasing.
\end{enumerate}
\end{lemma}
\begin{proof} Note that for any fixed $u\in H^1_{\bk}$, $\aspar\mapsto a_{\aspar}(u,u)$ defined for $\aspar>0$ is increasing. Hence, by the min-max characterization \eqref{eq.minmaxkpseudo}, for any fixed $n\ge1$ and $\bk\in\mathcal{B}$,  $\aspar\mapsto\lambda_n(\aspar;\bk)$ are increasing.
 For the particular case $n=1$ and $\bk =0$, $\bbA_{\aspar,{\bf 0}} u=0$ implies $u$ is equal to a constant on $\cell$.  Thus, $\lambda_1(\aspar;{\bf 0})=0$   for all $\aspar>0$ and  furthermore this eigenvalue is simple.

We  prove  now by contradiction, that $\lambda_n(\cdot;\bk)$ is strictly increasing for $n\geq 2$ or  for $n=1$  if $\bk \neq 0$. This strict monotonicity of the dispersion maps is not mentioned in \cite{Hem-00}.  Assume that there exist $\aspar_1,\aspar_2>0$ such that  $\aspar_1<\aspar_2$ and $\lambda_n(\aspar_1;\bk)=\lambda_n( \aspar_2;\bk)$. Then by the min-max theory, there exists a subspace $V_{j}$, for $ j=1,2$, of dimension $n$ with $V_{j}=\operatorname{span}(u_{1,j}, \ldots u_{n,j} )$ where $ u_{m,j}\in D(\bbA_{\aspar_j},\bk)$, $\|u_{m,j}\|_{L^2(\Omega)}= 1$ and $A_{\aspar_j, \bk} u_{m,j}=\lambda_{m}(\aspar_j;\bk)\, u_{m,j}$  for  $m=1,\ldots,n$ such that:
\begin{eqnarray}\label{eq.minmax2}
\lambda_n(\aspar_j;\bk)&=& a_{\aspar_j}(u_{n,j},u_{n,j})\nonumber \\
&=&\max_{\substack{u\in V_j, \, \|u\|_{L^2(\Omega)}= 1}} a_{\aspar_j}(u,u) =\min_{\substack{V\subset H^1_{\bk} \,, \mathrm{dim} V=n}} \quad
 \max_{\substack{u\in V, \, \|u\|_{L^2(\Omega)}= 1}} a_{\aspar_j}(u,u) . 
 \end{eqnarray}
 Thus, it follows that:
 \begin{eqnarray*}
 a_{\aspar_1}(u_{n,1},u_{n,1})&\leq& \max_{ u \in V_2, \, \|u\|_{L^2(\Omega)}= 1 } a_{\aspar_1}(u,u) \\
  &\leq& \max_{ u \in V_2, \, \|u\|_{L^2(\Omega)}= 1 } a_{g_2}(u,u) =a_{g_2}(u_{n,2},u_{n,2})=a_{g_1}(u_{n,1},u_{n,1})
 \end{eqnarray*}
 (where the first inequality is due to the min-max formula \eqref{eq.minmax2} for $j=1$ and the second to the increasing of the functions $\aspar\mapsto a_{\aspar}(u,u)$ defined for $\aspar>0$). Hence $V_2$ is a subspace of dimension $n$ for which the minimum is reached in \eqref{eq.minmax2} for both $j=1$ and  $j=2$. Thus, by min-max theory, the set $V_2$ consists  of the linear combination of eigenfunctions associated to the $n-th$ first eigenvalues  of the operator $A_{\aspar_1,\bk}$ and  the function $u_{n,2}\in V_2$ which reaches the maximum  over $V_2$ is an eigenfunction associated to the eigenvalue $\lambda_{n}(\aspar_1; \bk)(=\lambda_{n}(\aspar_2; \bk)$).
Thus, one has $A_{\bk,\aspar_1} u_{n,2}=\lambda_{n}(\aspar_1; \bk) u_{n,2}$ and $A_{\bk,\aspar_2} u_{n,2}=\lambda_{n}(\aspar_1; \bk) u_{n,2}$. This leads immediately to $-\Delta u=(\lambda_{n}(\aspar_1; \bk)/\aspar_1) \, u_{n,2}=(\lambda_{n}(\aspar_1; \bk)/\aspar_2 )\, u_{n,2}$ on $\cell^-$ and thus to $u_{n,2}=0$ on $\cell^-$ because $\lambda_{n}(\aspar_1; \bk)\neq 0$ for $n\geq 2$  or $\bk\neq 0$. Using the continuity of the traces of $u_{n,2}$ and of $\sigma_{\aspar_2}\partial u_{n,2}/\partial \bn$ across $\partial \cell^+$, it implies  that  $u_{n,2}$ and  $\partial u_{n,2}/ \partial \bn$ are also continuous and  thus they  both vanish on $\partial \cell^+$. Hence, one has $\Delta u_{n,2}=\lambda_{n}(\aspar_1; \bk) u_{n,2}$ on $\cell$ and $u_{n,2}=0$ on $\cell^-$ which leads  by the unique continuation principle (see e.g. Lemma 2.2 of \cite{Bour:10}) to the contradiction $u_{n,2}=0$ on $\cell$. 
\end{proof}

\begin{proposition}\label{prop.convunif}
Let $n\ge1$. (a) For $\bk\in\mathcal{B}\setminus\{{\bf0}\}$, the dispersion map $\bk\mapsto\lambda_n(\aspar;\bk)$ converges to the constant function with value $\delta_n$ as $\aspar\to \infty$.  This convergence is uniform on compact subsets of $\mathcal{B}\setminus\{{\bf0}\}$.
(b) The convergence in (a) is uniform on  all of $\mathcal{B}$  if and only if $\nu_n=\delta_n$.
\end{proposition}
\begin{proof} Choose $J$, a compact subset of $\mathcal{B}\setminus\{{\bf0}\}$.
For $\aspar>0$, the functions $\bk\mapsto \lambda_n(\aspar;\bk)$  are continuous  and converge pointwise  to the constant function $\delta_n$; see Theorem \ref{thmsaympAtau}. Furthermore, for all fixed $\bk\in J$,  the functions $\aspar \mapsto \lambda_n(\aspar;\bk)$ are increasing  by Lemma \ref{lem.monoticitydispcurv}. Therefore, by Dini's theorem, the convergence is uniform on $J$. 
If $\nu_n\neq \delta_n$,  then  $\lambda_n(\aspar;{\bf0})\to \nu_n\ne\delta_n=\lim_{\aspar\to\infty}\lambda_n(\aspar,\bk\ne0)$. In this case the convergence is not uniform on $\mathcal{B}$ since $\bk\mapsto\lambda_n(\aspar;\bk)$ is continuous and the limiting function  is  discontinuous. On the other hand,  if $\nu_n=\delta_n$, take $J=\mathcal{B}$ and then Dini's theorem implies uniform convergence on $\mathcal{B}$.
\end{proof}

\begin{proposition}[Interlacing of limiting Dirichlet and Neumann eigenvalues]\label{Prop.interlace}
The eigenvalues $ \nu_n$ and $\delta_n$ interlace in the following sense:
\begin{equation}\label{eq.interlacing}
\nu_n \leq \delta_n \leq \nu_{n+1}, \ \mbox{ for }\  n\geq 1.
 \end{equation}
 \end{proposition}
\begin{proof} The  first inequality $\nu_n \leq \delta_n$ in \eqref{eq.interlacing} follows from the limit $\aspar \to \infty$ of the inequality \eqref{eq.ineqvp}: $\lambda_n^{Neu,\cell}(\aspar)  \leq  \lambda_n^{Dir,\cell}(\aspar)$ and Theorem \ref{thm.dirsneumlimitpec}. The second 
inequality $\delta_n \leq \nu_{n+1}$ is much more delicate; it is proved by a min-max argument in \cite[Proposition 3.3]{Hem-00}.
\end{proof}

That the dispersion surfaces ``collapse'' onto asymptotic sets determined by the limiting Dirichlet and Neumann spectra $\{\nu_n, \, n \in \mathbb{N}\}$  and $\{\delta_n, \, n \in \mathbb{N}\}$ provides a means for identifying gaps in the spectrum of $\bbA_{\aspar}$
acting in $L^2(\R^2)$ .
The following result of \cite{Hem-00}  gives a condition  for the opening of a gap, for $\aspar$ sufficiently large, between the $n^{th}$ and $(n+1)^{st}$ dispersion surfaces.
\begin{proposition}[Condition for gap opening and location of spectral bands for $\aspar\gg1$]\label{Prop.nundeltan}
{\ }
\begin{enumerate}
\item Suppose $\delta_n< \nu_{n+1}$. Then,  for sufficiently large $\aspar$, there is a  gap in the spectrum of the periodic operator $\bbA_{\aspar}$.
More precisely, for all $\eta$ sufficiently small ($\eta \in (0,\nu_{n+1}-\delta_n)$), there exists $\aspar_{\eta}>0$ such that if  $\aspar >\aspar_{\eta}$, then  
\begin{equation}\label{eq.gapopen}
\sigma(\bbA_{\aspar}) \cap [\delta_n,\nu_{n+1}-\eta]=\emptyset.
\end{equation}
Hence,  there is spectral gap  located between the $n^{th}$ and $(n+1)^{st}$ spectral bands, 
which contains the interval $[\delta_n,\nu_{n+1}-\eta]$. 
\item  Suppose $ \nu_{n}< \delta_n$. Then, for any  sufficiently small $\eta$  (i.e. $\eta \in( 0,\delta_n-\nu_n)$), and $g>g_\eta$ sufficiently large,
the $n^{th}$ spectral band, $\lambda_{n}(\mathcal{B};\aspar)$,  contains the interval $[\nu_n,\delta_n-\eta]$. Hence, the $n-$th band ``does not get flat'' for large contrast.  
\end{enumerate}
\end{proposition}
\begin{proof}
Figure \ref{fig.gapcond} (resp. figure \ref{fig.gapcond2}) serves as a clarifying schematic of the point 1 of Proposition \ref{Prop.nundeltan} (resp. of the point 2). 
\begin{figure}[!h]
\vspace{-0.3cm}
\centering
 \includegraphics[width=0.9\textwidth]{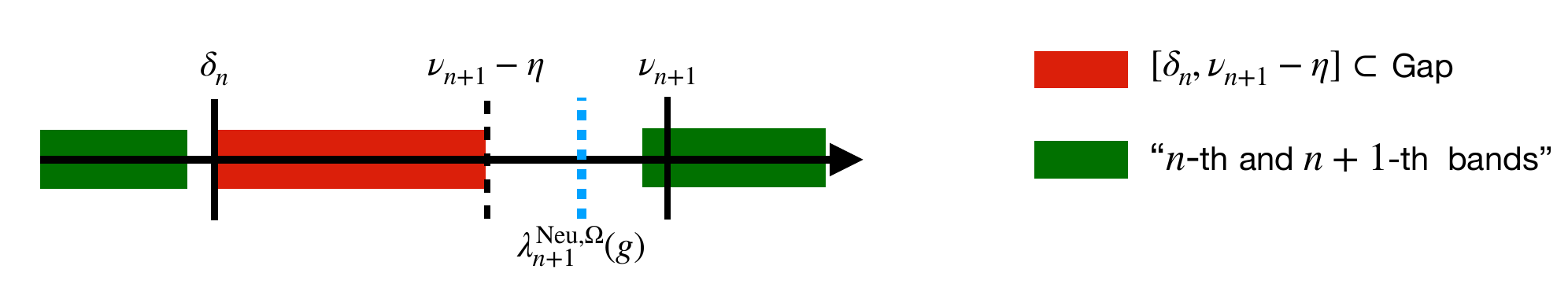}
   \caption{Bands and gaps of $\bbA_{\aspar}$ for $\aspar$ large when $\delta_n<\nu_{n+1}$ (case 1 of Proposition \ref{Prop.nundeltan}).   }
 \label{fig.gapcond}
\end{figure}

We first prove part (1).  Let $\eta\in(0,\nu_{n+1}-\delta_n)$ be fixed. One has by Theorem \ref{thm.dirsneumlimitpec} that  $\lambda_{n+1}^{Neu, \cell }(\aspar)\to \nu_{n+1}$, as $\aspar\to +\infty$. Therefore, there exists $\aspar_{\eta}>0$  such that if $ \aspar  >\aspar_{\eta}$ , then $ \nu_{n+1}-\eta< \lambda_{n+1}^{Neu, \cell }(\aspar)$.

Therefore, by the eigenvalue bracketing inequality \eqref{eq.ineqvp}, one has for $ \aspar  >\aspar_{\eta}$ and  all $  \bk \in \mathcal{B}$, $ \nu_{n+1}-\eta < \lambda_{n+1}^{Neu, \cell }(\aspar) \leq \lambda_{n+1}(\aspar;\bk)$. Hence the $(n+1)^{st}$- spectral band ``is above''  the energy $ \nu_{n+1}-\eta$. Furthermore, $\lambda_n(\cdot;\bk)$ is a strictly increasing function for $\aspar>0$  for $\bk\neq 0$ or $n>1$ (see lemma \ref{lem.monoticitydispcurv}) and it tends to $\delta_n$ (by Theorem  \ref{thm.dirsneumlimitpec}), as $\aspar\to +\infty$. Therefore one has $ \lambda_{n}(\aspar;\bk)< \delta_n$. For the particular case, $n=1$ and $\bk=0$, $\lambda_1(\aspar;{\bf 0})=0<\delta_1$ for $\aspar>0$. Thus, the $n-$th dispersion curve is always ``strictly below'' $\delta_n$ which yields to \eqref{eq.gapopen}.

\begin{figure}[!h]
\vspace{-0.3cm}
\centering
 \includegraphics[width=0.8\textwidth]{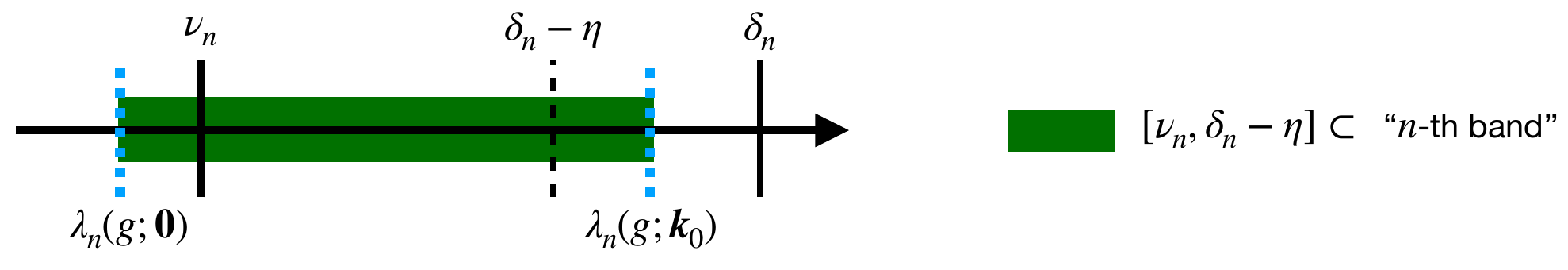}
   \caption{$n$-th bands for $\bbA_{\aspar}$ for $\aspar$ large when $\nu_n<\delta_n$ (case 2 of Proposition \ref{Prop.nundeltan}).}
 \label{fig.gapcond2}
\end{figure}
We now prove part (2). Let $\eta\in(0,\delta_n-\nu_n)$ be fixed. By Lemma \ref{lem.monoticitydispcurv},   $\lambda_n (\cdot;{\bf 0})$ is an increasing function for $\aspar >0$ and,  by Theorem \ref{thmsaympANinf}, it tends to $\nu_n$, as $\aspar \to \infty$. Thus, one has  $\lambda_n( \aspar;{\bf 0})\leq \nu_n$ for all $\aspar>0$. Let $\bk_0 \in \mathcal{B}\setminus\{{\bf 0} \}$, as $\lambda_{n}(\aspar;\bk_0)\to \delta_n$, as $\aspar\to +\infty$. One obtains that it  exists $\aspar_{\eta}>0$  such that if $ \aspar  >\aspar_{\eta}$ , $ \delta_n-\eta< \lambda_{n}(\aspar;\bk_0)$. One deduces from the convexity  of $\mathcal{B}$ and the continuity of  $\lambda_{n}(\aspar;\cdot)$ over $\mathcal{B}$ that $h: t\to \lambda_{n}(\aspar;t\bk_0)$ is well-defined  and continuous on $[0,1]$ and therefore by the intermediate value theorem $  [\nu_n,\delta_n-\eta]\subset h([0,1])\subset  \lambda_{n}(\aspar;\mathcal{B})$ for all $\aspar  >\aspar_{\eta}$. Hence it implies with \eqref{eq.specA} that for all $\aspar>\aspar_{\eta}$,
$ [\nu_n,\delta_n-\eta]\subset    \lambda_{n}(\aspar;\mathcal{B})\subset \sigma(\bbA_{\aspar})$.
\end{proof}

Thus, when the strict inequality $\nu_n<\delta_n<\nu_{n+1}$ is satisfied, one gets information on the large contrast asymptotics  of the $n^{th}$ and $(n+1)^{st}-$ spectral bands. In particular when $\delta_n<\nu_{n+1}$, there is a spectral gap between these two bands.

 In  \cite{Hem-00}, the authors provide the following useful condition on the eigenfunctions of  $\LaplDcellp$ (defined by \eqref{eq.defADinfty}) to ensure that $\delta_n<\nu_{n+1}$, and hence the existence of a spectral gap for $\aspar\gg1$.

\begin{proposition}\label{prop.condspec}(Conditions for a  gap in the spectrum  for large $\aspar$)\\
Let $(\delta_n)_{n\geq1}$ be the limiting Dirichlet eigenvalues
 (Theorem  \ref{thm.dirsneumlimitpec}), and  
formally set  $\delta_{0}=-\infty$. Assume ${\rm dim}\ \operatorname{ker}(\LaplDcellp-\delta_\star\, \mathrm{Id})=m+1$,
 with $\delta_\star=\delta_j=\dots= \delta_{j+m}$ so that 
 $$
 \delta_{j-1}<  \delta_\star< \delta_{j+m+1}.
$$
Then, we have the  following two distinct scenarios:
\begin{itemize}
\item[(A)] (Band separation)  there exists $u\in \operatorname{ker}(\LaplDcellp-\delta_\star \mathrm{Id}) $ such that
$
\displaystyle \int_{\cell^{+}} u(\bx)  \, \rmd \bx\neq 0,
$
$$
\textrm{and in this case,}\quad \nu_j < \delta_j \ \mbox{ and  }  \  \delta_{j+m}   <  \nu_{j+m+1}.
$$
\item[(B)]  for all $u\in \operatorname{ker}(\LaplDcellp-\delta_* \mathrm{Id}) $, one has
$
\displaystyle \int_{\cell^{+}} u (\bx) \, \rmd  \bx= 0,
$
$$
\textrm{and  then either}\quad \nu_j= \delta_j \  \mbox{ or } \ \delta_{j+m}  =  \nu_{j+m+1}.
$$
\end{itemize}
\end{proposition}

\section{Isolation and limit behavior of two degenerate spectral bands}\label{sec.condP}

Our strategy to prove the existence of Dirac points
in the band structure of honeycomb operators, involves an asymptotic reduction of the full spectral 
 problem for $\mathbb{A}_\aspar$  to a problem localized about two degenerate (touching) bands. This (Lyapunov-Schmidt / Schur complement) reduction scheme requires that for $\aspar\gg1$
the two bands that touch in a Dirac point are separated from the remainder of the spectrum.  Since, the dispersion maps collapse onto asymptotic Dirichlet or asymptotic Neumann eigenvalues (Proposition \ref{prop.convunif}), the  condition in part (A) of 
Proposition \ref{prop.condspec} may be used to obtain such band separation. Indeed, suppose the scenario of part (A) holds, {\it i.e.}
  $u\in \operatorname{ker}(\LaplDcellp-\delta_\star \mathrm{Id}) $ and
$ \int_{\cell^{+}} u(\bx)  \, \rmd \bx\neq 0.$ Then, 
\begin{equation}
 \delta_{j-1}\le\nu_j<\delta_j=\delta_*=\dots =\delta_{j+m}<\nu_{j+m+1}\le\delta_{j+m+1}.
 \label{isolate-m}
 \end{equation}
Hence, for sufficiently large $\aspar$: bands $r$, with $r\le j-1$ and $r\ge j+m+1$ are uniformly bounded away from the $m+1$ bands with indices: $r=j,\dots, j+m$.

For the construction of Dirac points we require two spectral bands which touch to be isolated from other bands of spectrum; hence we require \eqref{isolate-m} with $m=1$. 
Since $\cell^+=\cell^A\cup\cell^B$ is the union of disjoint translates of a connected set, we may express this condition  in terms of the Dirichlet eigenvalues of a single inclusion, $\cell^A$.

We introduce  the eigenvalues (listed with multiplicity)
$$0<\tilde{\delta}_1< \tilde{\delta}_2 \ldots  \leq \tilde{\delta}_n \leq \ldots\ $$ 
of the single inclusion Dirichlet  Laplacian $\LaplDcellA$
and  recall that $(\delta_n)_{n\geq 1}$ is the sequence  of eigenvalues  (listed also with multiplicity) of $\LaplDcellp$ (and  of $\mathbb{A}_{\infty,\bk}$, for $\bk\in \mathcal{B}\setminus\{0\}$).
Since the single inclusion $\cell^A$ is connected,  $\tilde{\delta}_1$ is simple \cite[Theorems 1.2.5 and 1.3.2]{Henrot:2006}.
  The following proposition relates the spectra of 
 $\LaplDcellp$ and $\LaplDcellA$. We omit its elementary proof, which uses  \cite[Proposition 3 p.  269]{RS4}.
\begin{proposition}\label{1incl2incl}
The spectra of $\LaplDcellp$ and $\LaplDcellA$ satisfy the following relations:
\begin{enumerate}
\item $\sigma(\LaplDcellp)=\sigma(\LaplDcellA)$,
\item for all $n\ge1$: $\delta_{2n-1} =\delta_{2n}=\tilde{\delta}_n$,
\item   for all $n\ge1$:$$ \operatorname{dim}\operatorname{ker}(\LaplDcellp- \delta_{2n} \, \mathrm{I}d)\ =\ 2\times \operatorname{dim}\operatorname{ker}(\LaplDcellp- \tilde{\delta}_{n} \, \mathrm{I}d).$$
 \item for all $n\geq 1$: 
 \begin{align*}
  &\textrm{there exists}\quad u\in \operatorname{ker}(\LaplDcellp- \delta_{n} \, \mathrm{I}d)\ \textrm{such that }\  \displaystyle \int_{\cell^+} u(\bx) \rmd \bx \neq 0 \\ &\ \qquad\qquad \Longleftrightarrow \\
  &\textrm{there exists}\quad  v\in \operatorname{ker}(\LaplDcellA- \tilde{\delta}_{n} \, \mathrm{I}d)\ \textrm{such that }\  \displaystyle \int_{\cell_A} v(\bx) \rmd \bx \neq 0.
  \end{align*}
 \end{enumerate}
\end{proposition}

 The band separation scenario (A) of Proposition \ref{prop.condspec} with $m=1$ and  Proposition \ref{1incl2incl} lead naturally to the following definition on the eigenvalues of $\LaplDcellA$ referred as  a spectral isolation condition.
\begin{definition}(Spectral isolation condition, $\condS$)\label{Def.condS}
Let  $n\geq 1$ and $\tilde{\delta}_{n} $  be an eigenvalue of the single inclusion Dirichlet  Laplacian, $\LaplDcellA$.
We say that
$\tilde{\delta}_{n} $, for some $n\geq 1$, satisfies the spectral separation condition $\condS$ if the  two following properties hold:
\begin{enumerate}
\item[(a)]  $\tilde{\delta}_{n} $ is a simple eigenvalue of $\LaplDcellA$, i.e. $\operatorname{dim}\operatorname{ker}(\LaplDcellA-\tilde{\delta}_n \mathrm{I}d)=1$, and 
\item[(b)]  There exists an eigenfunction $v\in \operatorname{ker}(\LaplDcellA-\tilde{\delta}_n \mathrm{I}d)$  such that
\begin{equation*}\label{eq.refintnonzero}
\int_{\cell^A}v(\bx) \rmd \bx \neq 0.
\end{equation*}
\end{enumerate}
\end{definition}

Propositions \ref{prop.condspec} and \ref{1incl2incl} imply
\begin{proposition}\label{prop.limiteigenval}
 If the  the condition $\condS$  holds for $\tilde{\delta}_n$ with $n\geq 1$ then Dirichlet and Neumann eigenvalues satisfy the following inequalities 
$$
\nu_{2n-1} < \delta_{2n-1}=\nu_{2n}=\delta_{2n} < \nu_{2n+1}\leq \delta_{2n+1}.
$$
\end{proposition}
\begin{proof}
Since $\tilde{\delta}_{n}$ is a simple eigenvalue of $\LaplDcellA$,   Proposition \ref{1incl2incl} implies that $\delta_{2n-1}$ is an eigenvalue of multiplicity $2$ with  $\delta_{2n-1}=\delta_{2n}=\tilde{\delta}_{n}< \delta_{2n+1} $.  Furthermore, the point (b) of Definition  \ref{Def.condS}  implies, via  part 4 of  Proposition \ref{1incl2incl}, that $\delta_\star=\delta_{2n-1}$ satisfies the band separation condition (A) of Proposition \ref{prop.condspec} with $m=1$. Thus it follows that:  $\nu_{2n-1}<\delta_{2n-1}=\delta_{2n}<\nu_{2n+1}$.
Finally, the interlacing inequality \eqref{eq.interlacing} yields: $\nu_{2n-1} < \delta_{2n-1}=\nu_{2n}=\delta_{2n} < \nu_{2n+1}\leq \delta_{2n+1}$.
\end{proof}

From the isolation of two degenerate asymptotic Dirichlet eigenvalues we next conclude  properties
of the corresponding dispersion surfaces of $\bbA_\aspar$ for sufficiently large $\aspar$.
\begin{theorem}\label{thm.convunifbands}
If $\tilde{\delta}_{n}$ satisfies  the spectral separation condition $\condS$, then the dispersion maps $\lambda_{2n-1}( \aspar; \cdot)$ and $\lambda_{2n}(\aspar;\cdot)$ satisfy  the following  (global) properties:
\begin{enumerate}
\item The dispersion map $\lambda_{2n-1}(\aspar;\cdot)$ converges uniformly to the constant function $\delta_{2n}=\tilde{\delta}_n$ as $\aspar\to +\infty$ on  any compact set of $\mathcal{B}\setminus \{0 \}$ whereas $\lambda_{2n-1}(\aspar;{\bf 0} )\to \nu_{2n-1}\neq \delta_{2n} $ as $\aspar\to + \infty$. In this case, the band (dispersion surface) does not flatten as $\aspar\uparrow\infty$.
\item $\lambda_{2n}(\aspar;\cdot)$ converges uniformly to $\delta_{2n}$ as $\aspar\to +\infty$ on the whole Brillouin zone $\mathcal{B}$. In this case, the band  ``becomes increasingly flat'' as $\aspar$ becomes large.

\item For $\aspar$ sufficiently large, there exists a gap between the $(2n)^{th}$ and the $(2n+1)^{st}$ dispersion maps. More precisely, one has $\lambda_{2n}( \aspar;\mathcal{B}) \subset [0,\delta_{2n})$ for all $\aspar >0$ and for any $0<\eta<\nu_{2n+1}-\delta_{2n}$, there exists $\aspar_{\eta}>0,$ such that $\forall \aspar>\aspar_{\eta} $,
$\lambda_{2n+1}(\aspar;\mathcal{B})\subset (\delta_{2n}+\eta,+\infty)$. 
\end{enumerate}
\end{theorem}
\begin{proof}
By Proposition \ref{prop.limiteigenval}, the limit Dirichlet and Neumann eigenvalues satisfy  $\nu_{2n-1} < \delta_{2n-1}=\nu_{2n}=\delta_{2n} < \nu_{2n+1}\leq \delta_{2n+1}$.
First note that $\lambda_{2n-1}(\aspar;{\bf 0} )\to \nu_{2n-1} \neq \delta_{2n}$, as $\aspar \to +\infty$  (see Theorem \ref{thmsaympANinf}). 
By Proposition \ref{prop.convunif} and the inequality $\nu_{2n-1}<\delta_{2n-1}=\delta_{2n}$, 
 $\lambda_{2n-1}(\aspar;\bk )$ converges uniformly to $ \delta_{2n}$ on compact sets of $\mathcal{B}\setminus \{0 \}$.
Furthermore, since $\nu_{2n}=\delta_{2n}$, part (b) of Proposition \ref{prop.convunif} implies that
 $\lambda_{2n}(\aspar;\bk )$ converges uniformly to $\delta_{2n}$ on all of  $\mathcal{B}$.
Finally, as $\delta_{2n}<\nu_{2n+1}$,  the third point is a direct application of \eqref{eq.gapopen} (with $2n$ replacing $n$ and $\nu_{2n+1}-\delta_{2n}-\eta$ replacing $\eta$) in Proposition \ref{Prop.nundeltan}. This  yields a  gap between the $2n$-th and the $2n+1$-th band.
More precisely, one has $\lambda_{2n}( \aspar;\mathcal{B}) \subset [0,\delta_{2n})$ for all $\aspar >0$. On the other hand, by Proposition \ref{Prop.nundeltan}, for any $0<\eta<\nu_{2n+1}-\delta_{2n}$, there exists  $\aspar_{\eta}>0,$ such that $\forall \aspar>\aspar_{\eta}$, $\lambda_{2n+1}(\aspar;\mathcal{B})\in (\delta_{2n}+\eta,+\infty)$.
\end{proof}

Note, in particular, that the spectral isolation condition $\condS$ of Definition \ref{Def.condS}  applies to  to the smallest eigenvalue of $\LaplDcellA$. Indeed,  the smallest eigenvalue of the Dirichlet Laplacian $\LaplDcellA$ is simple and admits an eigenfunction $v$ that is almost everywhere positive (see, \cite[Theorems 1.2.5 and 1.3.2]{Henrot:2006} or \cite[Theorem 4.1]{Arendt:2020}). 
Thus, the assumptions of Theorem \ref{thm.convunifbands} hold for $n=1$ since $\tilde{\delta}_{1}$ satisfies the condition $\condS$.
\begin{corollary}\label{cor.firstdipstcurves}
Points $1,\, 2$ and $3$ of Theorem \ref{thm.convunifbands} hold for $n=1$, where for point $1$,  we have 
$\lambda_{1}(\aspar;{\bf 0})=\nu_1 = 0$ for all $\aspar >0$.
\end{corollary}
\begin{proof}
As $\tilde{\delta}_1$ satisfies  $\condS$,  the points $1,\, 2$ and 3 of  Theorem \ref{thm.convunifbands} hold for $n=1$. The fact that  $\lambda_{1}(\aspar;{\bf 0})= \nu_1=0$ for all $\aspar >0$ follows from Lemma \ref{lem.monoticitydispcurv} and Theorem \ref{thmsaympANinf} applied for $n=1$.
\end{proof}
\begin{remark}
For honeycomb Schroedinger operators in the strong binding regime, the lowest two dispersion surfaces (after centering and rescaling) 
converge uniformly on the whole Brillouin zone to a tight binding model; \cite{FLW-CPAM:17}. Here uniform convergence of the surface to the first Dirichlet eigenvalue $\delta_1$ holds only away from a neighborhood of $\bk=0$.
\end{remark}
\begin{remark}\label{rem-noHoneycombsym}
We point out that untill now we have not required the honeycomb symmetry of $\sigma_\aspar$.
Thus so far we have only used that $\Omega=\cell^+\cup \cell^-\cup \partial \Omega$ contains  two identical open  simply connected inclusions $\cell^A$ and $\cell^B$ with Lipschitz boundary that are disjoint and have a positive distance from $\partial \Omega$. 
\end{remark}
\section{High contrast honeycomb structures and Dirac points}\label{DP}

In the present and subsequent sections, we use extensively the symmetries of the honeycomb structure.

\subsection{Symmetries of Honeycomb structures and their implications}\label{honey-sym}

 We begin by recalling here symmetry  properties of  honeycomb media; see \cite{FW:12}.
Let $\bK$ and $\bK'$ be the  two vertices of  $\mathcal{B}$ defined by  \eqref{eq.defKK'} (see  also figure \ref{fig.funcell}). Then, the six vertices of $\mathcal{B}$ are  generated from $\bK$ and $\bK'$ by applying the $2\pi/3$ clockwise  rotation matrix  $R$ (defined in \eqref{eq.defrotmat}). Thus,  the six vertices  fall into two groups:
\begin{itemize}
\item $\bK$ type-points: $\bK$, $R\, \bK=\bK+\bk_2$, $R^2\, \bK=\bK-\bk_1$,
\item $\bK'$ type-points: $\bK'$, $R\, \bK'=\bK'-\bk_2$, $R^2\, \bK'=\bK'+\bk_1$.
\end{itemize}
For any vertices $\bK_*$ of the Brillouin zone, one introduces also the rotation operator $\mathcal{R}[f]:L^2_{\bK_{*}} \to L^2_{\bK_{*}}$ with respect to the reference point $\bx_c$ (see section \ref{sec-not} for the definition of $\bx_c$) given by 
\begin{equation}\label{ed.defrotop}
\mathcal{R}[f](\bx)=f(\bx_c+R^{*}(\bx-\bx_c)) , \quad  f \in L^2_{\bK_{*}}.
\end{equation}
One first checks easily that $\mathcal{R}$ is well-defined. Indeed,  for any $f \in L^2_{\bK_*}$ and a.e. $\bx \in  \R^2$, one has for $\bv\in \Lambda$:
$
\mathcal{R}[f](\bx+\bv)=f(\bx_c+R^{*}(\bx-\bx_c)+R^{*}\bv ) 
$
and since $R^{*} \bv \in \Lambda$, we obtain
\begin{eqnarray}\label{eq.justfRop}
\mathcal{R}[f](\bx+\bv)&=&\rme^{i \bK_* \cdot R^{*}\bv}f(\bx_c+R^{*}(\bx-\bx_c)) \nonumber\\
&=&\rme^{i   R \bK_* \cdot \bv}f(\bx_c+R^{*}(\bx-\bx_c) ) \nonumber\\
&=&\rme^{i   \bK_* \cdot \bv}  \mathcal{R}[f](\bx) \quad  \mbox{ (since $R \bK_*-\bK_* \in \Lambda^*)$}.
\end{eqnarray}
Moreover, $\mathcal{R}$  is a unitary operator and  one can check that its (essential)  spectrum consists of three eigenvalues  $1, \tau, \overline{\tau}$ with $\tau=\exp(2\pi \rmi /3)$ with associated eigenspaces: 
\begin{equation}
L^2_{\bK_*,\nu}=\{ g\in L^2_{\bK_*} \mid \mathcal{R}g=\nu g\},
\label{L2Ks}\end{equation}
for $\nu=1, \tau, \overline{\tau}$.
Since $\mathcal{R}$ acting in $L^2_{\bK_{*}}$ is a normal operator, the spectral theorem implies that $L^2_{\bK_{*}}$ has the orthogonal decomposition:
\begin{equation}\label{eq.refL2K*decomp}
L^2_{\bK_*}=L^2_{\bK_*,1} \oplus L^2_{\bK_*,\tau} \oplus L^2_{\bK_*,\overline{\tau}}.
\end{equation}

Introduce, $\mathcal{P}$,  the inversion operator with respect to $\bx_c$, and complex conjugation, $\mathcal{C}$:
$$
\mathcal{P}[f](\bx)=f(2\bx_c-\bx),\qquad \mathcal{C}[f](\bx)=\overline{f(\bx)}.
$$
Their composition operator $\mathcal{P}\mathcal{C}:L^2_{\bK_{*}}\to L^2_{\bK_{*}}$  is given by
\begin{equation}\label{eq.definvoppseudoper}
\mathcal{P}\mathcal{C} [f] (\bx)=\mathcal{C} [f](2\bx_c-\bx)=\overline{f(2\bx_c-\bx)}.
\end{equation}
Furthermore, it  is easily verified (see Proposition 7.2 of \cite{FLW-CPAM:17})  that $\mathcal{P}\mathcal{C}$ is well-defined and  is an anti-linear involution  that  satisfies 
\[ \mathcal{P}\mathcal{C}(L^2_{\bK_*,\tau})=L^2_{\bK_*,\overline{\tau}} .\]
The vertices of $\mathcal{B}$ are {\it high symmetry quasi-momenta} in the following sense:
\begin{proposition}\label{prop.commutAk}
For any vertex $\bK_*$ of the Brillouin zone $\mathcal{B}$,  
$[\mathcal{R},\bbA_{\aspar, \bK_*}]=0$ and  $[\mathcal{PC},\bbA_{\aspar, \bK_*}]=0$; see Section \ref{sec-not}.
\end{proposition}
\begin{proof}
The proof is given in Proposition \ref{prop.commutoprper}.
\end{proof}
One denotes  by $\operatorname{ker}_{\nu}(\bbA_{\aspar, \bK_*}-\lambda \mathrm{I}d)$ for $\lambda\in \R$ and $\nu=1, \tau, \, \overline{\tau}$ the space defined by
\begin{equation}\label{eq.defKersigma}
\operatorname{ker}_{\nu}(\bbA_{\aspar, \bK_*}-\lambda \mathrm{I}d):=\operatorname{ker}(\bbA_{\aspar, \bK_*}-\lambda \mathrm{I}d)\cap L^2_{\bK_*,\nu}
\end{equation}
and by $ m_{\aspar,\bK_*,\nu}=\operatorname{dim}(\operatorname{ker}_{\nu}(\bbA_{\aspar, \bK_*}-\lambda \mathrm{I}d))$.
Finally, we denote by $\sigma_{\nu}(\bbA_{\aspar,\bK_*})$ the $L^2_{\bK_*,\nu}$ spectrum of $\bbA_{\aspar,\bK_*}$:
\begin{equation}\label{def.tauspec}
\sigma_{\nu}(\bbA_{\aspar,\bK_*}):=\{ \lambda\in \bbR \mid  m_{\aspar,\bK_*,\nu}(\lambda)>0 \} \quad   \mbox{ for  }\nu=1\, ,\tau,\, \overline{\tau} .
\end{equation}
The commutation relations  with the symmetry operators imply:
\begin{corollary}\label{cor.commut}
Let $\lambda\in \R$ and $\bK_*$ be a vertex of $\mathcal{B}$. Then, one has
\begin{align}\label{eq.propcommutRspec}
\sigma(\bbA_{\aspar,\bK_*})&=\bigcup_{\nu=1,\tau,\overline{\tau}}\sigma_{\nu}(\bbA_{\aspar,\bK_*})\, ,\\
\operatorname{ker}\big(\bbA_{\aspar, \bK_*}-\lambda \mathrm{I}d\big)&= {\small{ \small \bigoplus_{\nu=1,\tau, \overline{\tau}}}}  \operatorname{ker}_{\nu}\big(\bbA_{\aspar, \bK_*}-\lambda \mathrm{I}d\big).
\label{eq.propcommutR}\end{align}
Thus, solving the eigenvalue problem: $\bbA_{\aspar, \bK_*} u=\lambda u$ is equivalent to solving the three eigenvalue problems: $\bbA_{\aspar, \bK_*} u=\lambda  u$ in $L^2_{\bK_{*},\nu}$ for $\nu=1,\tau,\overline{\tau}$.
 Moreover, 
\begin{equation}\label{eq.propcommutS}
\mathcal{P}\mathcal{C}\big(\operatorname{ker}_{\nu}\big(\bbA_{\aspar, \bK_*}-\lambda \mathrm{I}d\big)\Big)= \operatorname{ker}_{\overline{\nu}}\big(\bbA_{\aspar, \bK_*}-\lambda \mathrm{I}d\big) \
\mbox{ for } \ \nu=1,\tau, \overline{\tau}.
\end{equation}
Finally, we have the  following relations on the dimension of eigenspaces:
\begin{equation}\label{eq.propcommutdim}
m_{\aspar,\bK_*,\tau}(\lambda) =m_{\aspar,\bK_*,\overline{\tau}}(\lambda ) \mbox{ and } \operatorname{dim}\operatorname{ker}(\bbA_{\aspar, \bK_*}-\lambda \mathrm{I}d)=2\, m_{\aspar,\bK_*,\tau}(\lambda) +m_{\aspar,\bK_*,1}(\lambda),
\end{equation}
and on the spectra:
$\sigma_{\tau}(\bbA_{\aspar,\bK_*})=\sigma_{\overline{\tau}}(\bbA_{\aspar,\bK_*})$ and 
$\sigma(\bbA_{\aspar,\bK_*})=\sigma_{1}(\bbA_{\aspar,\bK_*})\cup \sigma_{\tau}(\bbA_{\aspar,\bK_*}) $.
\end{corollary}
\begin{proof}
We prove only \eqref{eq.propcommutRspec} and \eqref{eq.propcommutR} since the other relations are proved in a similar way. By  \eqref{eq.refL2K*decomp} and  \eqref{eq.defKersigma}, the spaces $\operatorname{ker}_{\nu}(\bbA_{\aspar, \bK_*}-\lambda \mathrm{I}d)$ are orthogonal and their orthogonal sum is included in $\operatorname{ker}(\bbA_{\aspar, \bK_*}-\lambda\mathrm{I}d)$. We show now the other inclusion. Let $u\in \operatorname{ker}(\bbA_{\aspar, \bK_*}-\lambda \mathrm{I}d)$. By virtue of \eqref{L2Ks} and \eqref{eq.refL2K*decomp}, $u$ admits the following orthogonal decomposition
\begin{equation}\label{eq.orthogonalrotdecomp}
u=u_1+u_{\tau}+u_{\overline{\tau}},   \mbox{ with } u_{\nu}=\bbE_{\mathcal{R}}(\{ \nu\})u\in L^2_{\bK_*,\nu} \  \mbox{ for } \nu=1,\tau,\overline{\tau},
\end{equation}
where $\bbE_{\mathcal{R}}(\{ \nu\})$ is the spectral projector of $\mathcal{R}$ associated to the eigenvalue $\nu$. As $\mathcal{R}$ is a bounded normal operator which commutes with the self-adjoint operator $\bbA_{\aspar, \bK_*} $ (see Proposition \ref{prop.commutAk}), it implies that its spectral measure: $\bbE_{\mathcal{R}}( \cdot)$ commutes also with $\bbA_{\aspar, \bK_*}$ (see section \ref{sec-not}). Hence, one has  $u_{\nu}=\bbE_{\mathcal{R}}(\{ \nu\})u\in D(\bbA_{\aspar, \bK_*} )$ (since $u\in   \operatorname{ker}(\bbA_{\aspar, \bK_*}-\lambda \mathrm{I}d)\subset  D(\bbA_{\aspar, \bK_*} )$ and $D(\bbA_{\aspar, \bK_*} )$ is stable by $\bbE_{\mathcal{R}}(\{ \nu\})$) and
$$
\bbA_{\aspar, \bK_*}u_{\nu}=\bbE_{\mathcal{R}}(\{ \nu\} )\bbA_{\aspar, \bK_*}u=\lambda  \bbE_{\mathcal{R}}(\{ \nu\})u =\lambda \, u_{\nu} \quad \mbox{ for } \nu=1,\tau,\overline{\tau}.$$
Thus,  $u_{\nu}\in \operatorname{ker}_{\nu}(\bbA_{\aspar, \bK_*}-\lambda \mathrm{I}d)$ and with \eqref{eq.orthogonalrotdecomp}, $u$ belongs to the orthogonal sum of  the three spaces: $\operatorname{ker}_{\nu}\big(\bbA_{\aspar, \bK_*}-\lambda \mathrm{I}d\big)$ for  $\nu=1,\tau,\overline{\tau}$. This proves \eqref{eq.propcommutR}. \eqref{eq.propcommutRspec} follows immediately  from  \eqref{def.tauspec} and \eqref{eq.propcommutR}.
\end{proof}

\subsection{Dirac points}\label{DP-def}
We recall the precise definition of a Dirac point for a divergence form  elliptic operator $\bbA_{\aspar}$; 
 see \cite{FW:12,LWZ:18}. 
\begin{definition}[Dirac Points]\label{def.Diracpoints}
Fix $\aspar>0$. The ``energy / quasimomentum''
pair $(\lambda_D(\aspar), \bk_D)\in \R^+\times \mathcal{B}$ is called a Dirac point of the operator $\bbA_{\aspar}$ if there exists $n\geq 1$
such that:
\begin{enumerate}
\item $\lambda_n(\aspar;\bk_D)=\lambda_{n+1}(\aspar;\bk_D)=\lambda_D(\aspar)$ is an eigenvalue of multiplicity $2$ of the operator  $\bbA_{\aspar,\bk_D}$;
\item The  dispersion maps $\lambda_n(\aspar;\cdot)$ and $ \lambda_{n+1}(\aspar;\cdot)$ touch in isotropic cones at $\bk_D$, i. e. for some $v_D(\aspar)>0$:
\begin{eqnarray*}
\lambda_{n+1}(\aspar;\bk)&=& \lambda_D(\aspar)+ v_D(\aspar) \, |\bk-\bk_D|+o(|\bk-\bk_D|);\\
\lambda_{n}(\aspar;\bk)&=& \lambda_D(\aspar)- v_D(\aspar) \, |\bk-\bk_D|+o(|\bk-\bk_D|).
\end{eqnarray*}
\end{enumerate}
\end{definition}
The following theorem gives sufficient conditions for the existence of a Dirac points at any vertex $\bK_*$ of the Brillouin zone $\mathcal{B}$. Its analogue was proved for Schroedinger operators
in \cite{FW:12} and for elliptic operators with smooth coefficients \cite{LWZ:18}. 
 Our proof  is given in Section 
\ref{sec.diracpointthoerem}.
 Since the coefficient of the operator $\bbA_\aspar$, $\sigma_{\aspar}$, is discontinuous, the proof is significantly different from that in previous works.   
\begin{theorem}[Sufficient condition for the existence of Dirac points]\label{Thm.weakLyapounovSchmidtred}
Let  $\aspar$  the positive contrast parameter be fixed, $\bK_*$ be any vertex of the Brillouin zone $\mathcal{B}$ and $\lambda_D(\aspar)\in \R^+$. Assume that:
\begin{enumerate}
\item $m_{\bK_*,\aspar,\tau}(\lambda_D(\aspar))=\mathrm{dim} \operatorname{ker}_{\tau}(\bbA_{\aspar, \bK_*}-\lambda_D(\aspar) \, \mathrm{I}d)=1$ (i.e. the $L^2_{\tau}$ eigenvalue problem $\bbA_{\aspar, \bK_*}u=\lambda_D(\aspar) u$ has a one dimensional space of solutions $u$).  Let $\Phi_1(\aspar,\cdot)$ be a normalized eigenfunction of $\operatorname{ker}_{\tau}(\bbA_{\aspar, \bK_*}-\lambda_D(\aspar) \, \mathrm{I}d)$ and  $\Phi_2(\aspar,\cdot)=\mathcal{P}\mathcal{C}\Phi_1(\aspar,\cdot)\in \operatorname{ker}_{\overline{\tau}}(\bbA_{\aspar, \bK_*}-\lambda_D(\aspar ) \, \mathrm{I}d)$ (by \eqref{eq.propcommutS}).
\item $m_{\bK_*,\aspar,1}(\lambda_D(\aspar))=0$ (i.e the $L^2_{\bK_*,1}$ eigenvalue $\bbA_{\aspar, \bK_*}u= \lambda_D(\aspar)u$ admits only $u=0$ as solution). 
\item (Non-vanishing of the Dirac velocity)
\begin{equation}\label{eq.fermyveloc}
 v_D(\aspar)=\Big| \int_{D} \sigma_{\aspar}\Phi_1(\aspar,\bx) \, \overline{\nabla \Phi_2(\aspar,\bx) }\, \rmd \bx \cdot (1,-\rmi)^{\top}\Big| \neq 0\ .
\end{equation}
\end{enumerate}
Then, $(\bK_*,\lambda_D(\aspar))$ is a Dirac point in the sense of Definition \ref{def.Diracpoints}.
\end{theorem}

\begin{remark}\label{rem.Diracpointsequivalence}
By symmetry, it is sufficient to prove Theorem  \ref{Thm.weakLyapounovSchmidtred} for one of the 6 vertices of  the Brillouin zone $\mathcal{B}$ \cite{FLW-CPAM:17,LWZ:18}. This relies on the following symmetry properties: 
\begin{enumerate}
\item If $\bK_*$ is a $\bK$ type-points then $L^2_{\bK}=L^2_{\bK_*}$ and $\bbA_{\aspar,\bK}=\bbA_{\aspar,\bK_*}$ for all $\aspar>0$ since $\bK$ and $\bK_*$ yields to the same quasi-periodic conditions (as their quasimomenta differ from  a dual lattice vector). Thus, $\bbA_{\aspar,\bK}$ and $\bbA_{\aspar,\bK_*}$ have the same eigenelements. Hence, conditions of Theorem   \ref{Thm.weakLyapounovSchmidtred}  are satisfied for $(\bK_*,\lambda_D(\aspar))$ if and only if they are satisfied for $(\bK,\lambda_D(\aspar))$ and in particular one has $v^{\bK_*}_{F}(\aspar)=v^{\bK}_{F}(\aspar)$. If $\bK_*$ is a $\bK'$ type-points, the same property holds (by replacing $\bK$ by $\bK'$)  since  $\bbA_{\aspar,\bK'}=\bbA_{\aspar,\bK_*}$ .
\item As $\bK'=-\bK$, one checks easily that $(\lambda, \Phi)$ is  an eigenpair of  $\bbA_{\aspar,\bK}$ if and only if $(\lambda, \mathcal{C}{\Phi}=\overline{\phi})$ is an eigenpair of  $\bbA_{\aspar,\bK'}$. Furthermore, one shows easily that for real $\lambda$ and $\sigma=1, \tau,\overline{\tau}$: 
\begin{equation}\label{eq.conj} \mathcal{C}(L^2_{\pm \bK, \sigma})=L^2_{\mp \bK, \overline{\sigma}}  \ \mbox{ and } \ \mathcal{C}(\operatorname{ker}_{\sigma}\big(\bbA_{\aspar,\pm \bK}-\lambda \mathrm{I}d\big))=\operatorname{ker}_{\overline{\sigma}}\big(\bbA_{\aspar, \mp \bK}-\lambda \mathrm{I}d\big)).
\end{equation}
Since $\mathcal{C}$ (an anti-linear involution) does not change the subspace dimensionality, we have by \eqref{eq.propcommutdim} and \eqref{eq.conj}:
$$m_{\aspar,\bK,\sigma}(\lambda) =m_{\aspar,\bK',\overline{\sigma}}(\lambda)=m_{\aspar,\bK',\sigma}(\lambda) \, \mbox{ for } \sigma=1, \tau,\overline{\tau} \mbox{ and } \lambda\in \R.$$ 
Thus,  Theorem \ref{Thm.weakLyapounovSchmidtred} holds for $(\bK, \lambda_D(\aspar))$  with associated $\Phi_1^{\bK}(\aspar,\cdot)$, $\Phi_2^{\bK}(\aspar,\cdot)=\mathcal{P}\mathcal{C} \Phi_1^{\bK}(\aspar,\cdot)$ and $v_{F}^{\bK}(\aspar)$ if and only if it is satisfied for $(\bK', \lambda_D(\aspar))$ with associated $\Phi_1^{\bK'}(\aspar,\cdot)=\mathcal{C}\, \Phi_2^{\bK}(\aspar,\cdot)$, $\Phi_2^{\bK'}(\aspar,\cdot)=\mathcal{P}\,\mathcal{C} \Phi_1^{\bK'}(\aspar,\cdot)=\mathcal{P} \Phi_2^{\bK}(\aspar,\cdot)= ( \mathcal{P}^2\circ \mathcal{C} )\Phi_1^{\bK}(\aspar,\cdot)=  \mathcal{C}  \Phi_1^{\bK}(\aspar,\cdot)$ and $v_{F}^{\bK'}(\aspar)=v_{F}^{\bK}(\aspar)$. 
\end{enumerate}
\end{remark}
\begin{remark}
 In Theorem \ref{Thm.weakLyapounovSchmidtred}, the normalized vector  $\Phi_1(\aspar)$ is uniquely defined up to a complex phase. However, the  Dirac  velocity  $v_D(\aspar)$,  defined by  \eqref{eq.fermyveloc}, is independent of this choice of phase.
\end{remark}

\subsection{Construction of limiting eigenstates at high symmetry quasi-momenta}\label{orbitals}

In this section we build up approximations for a basis of the degenerate eigenspace associated to a Dirac point.
The idea is that as $\aspar\uparrow\infty$, the $L^2_{\bK_\star}$  eigenstates of $\mathbb{A}_{\aspar,\bK_\star}$ converge to eigenstates of the Dirichlet Laplacian for an isolated inclusion. It is therefore natural, when $\aspar$ is large, to seek Floquet-Bloch states
  which are quasi-periodic superpositions of translates of Dirichlet eigenstates. 
  In the context of quantum chemistry, this idea 
  is known as LCAO: the linear combination of atomic orbitals. By analogy we refer to the translated Dirichlet eigenstates {\it Dirichlet orbitals}. 
 
\subsubsection{Dirichlet orbitals}
Introduce the operators associate with  inversion, complex conjugation and $2\pi/3$ rotation on the  single inclusion $\cell^A$:
\begin{equation}
\label{PCR-def}
\mathcal{P}_{\cell^A} \,f(\bx)=f(-\bx),\quad 
\quad \mathcal{C}f(\bx)=\overline{f(\bx)},\quad \mathcal{R}_{\cell^A} \,f(\bx)=f(R^* \, \bx).
\end{equation}
Here, $R$ denotes the $2\times2$  matrix, which rotates a vector in the plane about $\bv_A=0$ by $2\pi/3$ clockwise.
Since $R(\cell^A)=\cell^A$ and $\mathcal{P}(\cell^A)=\cell^A$ (assumptions ($\Omega$.v) and ($\Omega$.vi) of Section \ref{sec:sigma}), the operators $\mathcal{P}_{\cell^A}, \mathcal{R}_{\cell^A}$ and $\mathcal{C}$ map $L^2(\cell^A)$ to itself. $\mathcal{P}_{\cell^A}$ and $\mathcal{R}_{\cell^A} $ are unitary
 and $\mathcal{C}$ is anti-unitary. Furthermore, 
 \begin{equation}
  [\mathcal{P}_{\cell^A}, \LaplDcellA]=0,\quad [\mathcal{R}_{\cell^A},\LaplDcellA]=0,\  [\mathcal{C}, \LaplDcellA]=0]. 
  \label{commute}\end{equation}
The commutation with the conjugation operator is obvious, for the two other commutation relations, see  \ref{sec-not},  in particular, Proposition \ref{prop.comutappend} for more details.

\begin{proposition}\label{Prop-honeycombsym}
Let $\tilde{\delta}_n$, $n\geq 1$  be an eigenvalue  of $\LaplDcellA$ satisfying the  spectral isolation condition $\condS$ of Definition \ref{Def.condS}. 
Let $v\in \operatorname{ker}(\LaplDcellA-\tilde{\delta}_{n} \mathrm{I}d)$, then
one has 
\begin{equation}\label{eq.honeycombsym}
\mathcal{R}_{\cell^A}[v](\bx)=\mathcal{R}^*_{\cell^A}[v](\bx)=v(\bx).
\end{equation}
Moreover,  there exists  a unique (up to a  factor of $-1$) normalized eigenfunction $p_n\in \operatorname{ker}(\LaplDcellA-\tilde{\delta}_{n} \mathrm{I}d)$ such that  for almost all $\bx\in \cell^A $:
\begin{align}\label{eq.honeycombsymCP}
\mathcal{C}[p_n](\bx) &= p_n(\bx),\quad {\rm and}\\
{\rm either}\quad  \mathcal{P}_{\cell^A}[p_n](\bx) &=p_n(\bx) \quad {\rm    or}\quad   \mathcal{P}_{\cell^A}[p_n](\bx) = -p_n(\bx) \ .\nonumber
\end{align}
\end{proposition}
\begin{proof}
Let $v$ be an eigenfunction associated with $\tilde{\delta}_n$.   Since  $\LaplDcellA$ commutes with $\mathcal{C}_{\cell^A}$,  $\overline{v}$ is also an eigenfunction of $\LaplDcellA$  associated with $\tilde{\delta}_n$.
Therefore, $\operatorname{Re}(v)=(v+\overline{v})/2$ and  $\operatorname{Im}(v)= (v-\overline{v})/(2\rmi)\in \operatorname{ker}(\LaplDcellA-\tilde{\delta}_n \mathrm{I}d)$ and these two functions cannot be simultaneously equal to the zero function since $v=\operatorname{Re}(v)+\rmi \operatorname{Im}(v) \neq 0$. Thus,  there exists  a real-valued eigenfunction associated to $\tilde{\delta}_n$. Moreover, since $\tilde{\delta}_n$ is a simple eigenvalue, there exists a unique (up to a multiplication by $-1$) real-valued  normalized eigenfunction, which we denote   $p_n\in \operatorname{ker}(\LaplDcellA-\ \tilde{\delta}_n \mathrm{Id})$.
By the commutation relations \eqref{commute},  the one dimensional space $\operatorname{ker}(\LaplDcellA-\tilde{\delta}_n \mathrm{I}d)=\operatorname{span}\{p_n\}$ is invariant under  $\mathcal{R}_{\cell^A}$ and $\mathcal{P}_{\cell^A}$. In addition, $\mathcal{R}_{\cell^A} \,p_n$ and $\mathcal{P}_{\cell^A}\,p_n$ are real-valued, thus there exist  $\alpha, \beta\in \R$ such that $(\mathcal{R}_{\cell^A} \,p_n)(\bx)=p_n(R^* \bx)=\alpha \, p_n(\bx)$ and  $\mathcal{P}_{\cell^A} p_n(\bx)  = \beta \, p_n(\bx) $ for $ \bx\in\cell^A$. But as $\mathcal{R}_{\cell^A} \,p_n $
and $\mathcal{P}_{\cell^A} \,p_n $ are normalized  (since $p_n$ is normalized), and since $\mathcal{R}_{\cell^A}$ is  unitary and $\mathcal{P}_{\cell^A}$ is  anti-unitary, we have $\alpha=\pm 1$ and $\beta=\pm 1$. Furthermore, since $\mathcal{R}_{\cell^A}^3=\mathrm{Id}$, we have $\alpha^3=1$ and thus $\alpha=1$.  Finally, since $\operatorname{ker}(\LaplDcellA-\tilde{\delta}_n \mathrm{I}d)=\operatorname{span}\{p_n\}$ and $\mathcal{R}_{\cell^A}p_n=p_n$, it follows that any $v\in \operatorname{ker}(\LaplDcellA-\tilde{\delta}_n \mathrm{I}d)$ satisfies   \eqref{eq.honeycombsym}.
\end{proof}

Let $\tilde{\delta}_n$, $n\geq 1$  be an eigenvalue  of $\LaplDcellA$ satisfying the band separation condition $\condS$ and let $p_n\in \operatorname{ker}(\LaplDcellA-\tilde{\delta}_n \mathrm{I}d)$  denote the normalized  eigenfunction (unique up to a factor of $-1$), guaranteed by Proposition \ref{Prop-honeycombsym} which satisfies
  the symmetry relations \eqref{eq.honeycombsym} and \eqref{eq.honeycombsymCP}.
   We extend $p_n$ to be defined on all $\mathbb{R}^2$ by setting it equal to zero on $\mathbb{R}^2\setminus\cell^A$.  We continue to denote this extension by $p_n$ and observe that for a.e. $\bx\in\mathbb{R}^2$ we have:
   \begin{align}\label{eq.RsymR2}
p_n(\bx) &=p_n(R^*\bx)=p_n(R\bx)\\
\overline{p_n(\bx)} &= p_n(\bx),\quad {\rm and} \label{eq.CsymR2}\\
\label{eq.CPsymR2}
{\rm either}\quad  p_n(-\bx) &=p_n(\bx) \quad {\rm    or}\quad  p_n(-\bx) = -p_n(\bx) .
\end{align}

The following proposition is an immediate consequence of Proposition \ref{1incl2incl}. 
\begin{proposition}\label{lem.eigenDirichlettwoinclusions}
Assume that $\tilde\delta_n$  is an eigenvalue of $\LaplDcellA$  for which the spectral isolation condition $\condS$ of Definition \ref{Def.condS} holds, with corresponding normalized eigenfunction  $p_n$ that satisfies  \eqref{eq.RsymR2}, \eqref{eq.CsymR2} and \eqref{eq.CPsymR2}.
Then,  $\delta_{2n-1}=\delta_{2n}= \tilde{\delta}_{n}$ is an eigenvalue of $\LaplDcellp$ of multiplicity $2$, i. e.
$$
\delta_{2n-1}=\delta_{2n}= \tilde{\delta}_{n}<\delta_{2n+1}= \tilde{\delta}_{n+1}   \textrm{ for $n\ge1$ }  \mbox{ and } \  \tilde{\delta}_{n-1}=\delta_{2n-2}<\delta_{2n-1} \textrm{ for $n>1$.}
$$
Furthermore, $\{p_{n}(\bx-\bv_A),p_{n}(\bx-\bv_B)\}$ is an orthonormal basis for 
$\operatorname{ker}(\LaplDcellp-\delta_{2n} \mathrm{I}d)$.
\end{proposition}

\begin{example}[$\cell^A$ and $\cell^B$, circular inclusions]
If  $\cell^A$ and $\cell^B$ are two discs of radius $R_0$, the three first eigenvalues of $\LaplDcellp$ are given by
$
\delta_1=\delta_2=\tilde{\delta}_1=\left(z_{0,1}/R_0\right)^2< \delta_3=\tilde{\delta}_2= \left(z_{1,1}/R_0\right)^2$
where $z_{p,q}$ denotes the $q^{th}$ positive zero of the Bessel function $J_p$ ($p\in \N_0$). 
Moreover, the normalized eigenfunction $p_1$  associated to $\tilde{\delta}_{1}$ is (up to $-1$ factor) given by
\begin{equation*}\label{eq.defeigenfunctions}
p_1(\bx)=\frac{{ \bf1}_{_{|\bx|\le R_0}}(\bx)}{\sqrt{\pi}}\frac{J_0\big(\sqrt{\tilde{\delta}_1}|\bx|\big)}{|J_0'(z_{0,1})| R_0} ;
\end{equation*}
note that $|J_0'(z_{0,1})|\neq 0$ as the zeros  of $J_0$ are simple; see \cite{Wat:66}. 
For the normalization of $p_1$, we use the identity $\int_0^{R_0} J_0^2(\sqrt{\tilde{\delta}_1} r) \, r \, \mathrm{d}r =R_0^2  \, J_0'(z_{0,1})^2/2$ (see  chapter 5, formula 11 page 135 of \cite{Wat:66}).
Finally, note that in this particular case $p_1$ is an  even function.
\end{example}

\subsubsection{Pseudo-periodic superposition of Dirichlet orbitals}
In this section, we prove the existence of Dirac points over the Brillouin zone vertices for $\aspar$ sufficiently large. By  Remark \ref{rem.Diracpointsequivalence}, it suffices to work at the single vertex $\bK$ of $\mathcal{B}$.  
First, one construct the eigenstates of the limit problem as $\aspar \to+ \infty$. We  assume that the Dirichlet eigenvalue $\tilde{\delta}_n$, $n\geq 1$ satisfies the spectral isolation condition $\condS$ of Definition \ref{Def.condS}.
 Let $p_n$ be the unique (up to a factor $-1$) normalized eigenfunction  associated to $\tilde{\delta}_n$ that  satisfies  \eqref{eq.RsymR2}, \eqref{eq.CsymR2} and \eqref{eq.CPsymR2}. For each of the two triangular sublattices, $\Lambda+\bv_J$, $J=A,B$ which  comprise
  the honeycomb, we associate the function
\begin{equation}\label{eq.defPka}
P_{n,\bK}^J(\bx):=\sum_{\bv\in \Lambda} \rme^{\rmi \bK\cdot \bv} p_{n}(\bx-\bv-\bv_J), \ \ \mbox{ for a. e. }    \bx \in \mathbb{R}^2,\qquad  J=A,B.
\end{equation}
The supports of  each summand are disjoint, so the above series is trivially convergent; for each $\bx\in\mathbb{R}^2$, 
  at most one of its term is nonzero.  Furthermore, note that $P_{n,\bK}^J(\cdot)$ is the Floquet-Bloch transform at quasi-momentum $\bK$ of
  $p_{n}(\cdot-\bv_J)$ (see \eqref{FB-def}) and is therefore $\bK-$quasi-periodic. Hence, $P_{n,\bK}^J\in L^2_{\bK}$.
\begin{lemma}\label{lem.pkapkb}
Assume that the eigenvalue $\tilde{\delta}_n$ of $\LaplDcellA$  satisfies the condition $\condS$ of Definition \ref{Def.condS}. Then, 
\begin{enumerate}
\item
\begin{equation}\label{eq.PkAPKB}
P_{n,\bK}^A\in L_{\bK,\tau}^2 \mbox{ and } P_{n,\bK}^B\in L_{\bK, \overline{\tau}}^2,\quad \mathcal{P}\mathcal{C}P_{n,\bK}^A=\pm \rme^{- \rmi  \frac{2\pi}{3}} P^B_{n,\bK},
\end{equation}
where the choice of sign, $\pm$, is  $+$ (resp. $-$) if  the  single inclusion Dirichlet eigenfunction $p_n$ with corresponding eigenvalue  $\tilde{\delta}_n$ is even (resp. odd).
\item  the two-dimensional eigenspace of $\bbA_{\infty, \bK}$  associated with  the eigenvalue
  $\tilde{\delta}_n=\delta_{2n-1}=\delta_{2n}$ admits $\{P_{n,\bK}^A,P_{n,\bK}^B\}$ as an orthonormal basis.
\end{enumerate}
\end{lemma}
\begin{proof}
Definitions \eqref{eq.defDIropinf} and \eqref{eq.defADinfty} imply that  $\bbA_{\infty, \bK}$ and $\LaplDcellp$ 
have the same eigenvalues  (with multiplicity). Furthermore, the eigenfunctions of $\bbA_{\infty, \bK}$ are obtained from those of $\LaplDcellp$   by first extending them to be identically zero on $\cell\setminus\cell^+$, and then extending this function on $\cell$ to be $\bk-$quasi-periodic on  $\R^2$.
Thus,  by Proposition \ref{lem.eigenDirichlettwoinclusions},  $\{P_{n,\bK}^{A},P_{n,\bK}^{B}\} $ is an orthonormal basis of the $2$-dimensional space $\operatorname{ker}(\bbA_{\infty, \bK}-\delta_{2n}\mathrm{I}_d)$.

We next prove that $P_{n,\bK}^A\in L_{\bK,\tau}^2$. This proof is similar to the one of \cite[Lemma 10.4]{LWZ:18}. We recall that $\bv_A=(0,0)$. For almost all $ \bx\in \mathbb{R}^2$:
\begin{eqnarray*}
\mathcal{R}[P_{n,\bK}^{A}](\bx)
&=&\sum_{\bv \in \Lambda} \rme^{\rmi \bK \cdot \bv} p_{n}(\bx_c+R^{*}(\bx-\bx_c)-\bv)\\
&=&\sum_{\bv \in \Lambda}\rme^{\rmi \bK \cdot \bv} p_{n}(R\, \bx_c+(\bx-\bx_c)-R\, \bv)\, ,
\end{eqnarray*}
where the last equality holds since $p_n(R\, \by)=p_n(\by)$ for all $ \by\in \mathbb{R}^2$ (Proposition
 \ref{Prop-honeycombsym}). Using  that
$\bv_2=\bx_c-R\bx_c$ (since $\bx_c=\bv_2-\bv_B$ and $R\bx_c=-\bv_B$) and  that $R$ is unitary we have
\begin{eqnarray*}
\mathcal{R}[P_{n,\bK}^{A}](\bx)&=&\sum_{\bv \in \Lambda}\rme^{\rmi R \bK \cdot R \bv} p_{n}\big(\bx-(R\, \bv+\bv_2)\big)\\
&=& \rme^{-\rmi R \bK\cdot \bv_2} \sum_{\bv \in \Lambda}\rme^{\rmi R \bK \cdot (R \bv+\bv_2)} p_{n}\big(\bx-(R\, \bv+\bv_2)\big),\\
&=&  \rme^{-\rmi  \bK\cdot \bv_2} \sum_{\bv \in \Lambda}\rme^{\rmi \bK \cdot (R \bv+\bv_2)} p_{n}\big(\bx-(R\, \bv+\bv_2)\big) \quad  (\mbox{ since }R\, \bK=\bk_2+\bK), \\
&=&\tau  \sum_{\bw \in \Lambda}\rme^{\rmi \bK \cdot \bw} p_{n}\big(\bx-\bw\big). \end{eqnarray*}
The last equality follows since $\bv\mapsto R\bv+\bv_2$ maps $\Lambda$ to itself
 and and $\rme^{-\rmi  \bK\cdot \bv_2}=\rme^{-\frac{\rmi}{3}(\bk_1-\bk_2)\cdot\bv_2}=\rme^{-\frac{\rmi}{3}(-2\pi)}=\tau$.
Thus, $\mathcal{R}[P_{n,\bK}^A](\bx) =\tau \, P_{n,\bK}^A(\bx)$ and  $P_{n,\bK}^A\in L_{\bK,\tau}^2$.

Finally, we prove the equality in \eqref{eq.PkAPKB}. This will imply that $P_{\bK,B}$ is in $L^2_{\bK,\overline{\tau}}$ since $P_{\bK,A}\in L_{\bK,\tau}^2$ and $\mathcal{P}\mathcal{C}$ maps $ L_{\bK,\tau}^2$ into  $L_{\bK,\overline{\tau}}^2$.   By definition of $\mathcal{P}\mathcal{C}$, one has
\begin{eqnarray*}
\mathcal{P}\mathcal{C}[P_{n,\bK}^{A}](\bx)&=&\overline{P_{n,\bK}^{A}(2\bx_c-\bx)}\\
&=& \sum_{v\in \Lambda} \rme^{- \rmi \bK\cdot \bv } p_{n}(2\bx_c- \bx- \bv) \, \quad  \mbox{(since $p_n$ is real-valued),}\\
&=& \pm  \sum_{v\in \Lambda} \rme^{- \rmi \bK\cdot \bv } p_{n}(\bx+\bv -2\bx_c)  \quad  \mbox{(since $p_n$ is even ($+$) resp. odd ($-$)),} \\
&=&\pm  \sum_{v\in \Lambda} \rme^{- \rmi \bK\cdot \bv } p_{n}(\bx+\bv -\bv_2+\bv_1-\bv_B) \quad \mbox{(since $2\bx_c=\bv_2-\bv_1+\bv_B$),}\\
&=& \pm \rme^{ \rmi \bK\cdot (\bv_1-\bv_2) } \sum_{\bw\in \Lambda} \rme^{ \rmi \bK\cdot \bw } p_{n}(\bx- \bw-\bv_B)\quad \mbox{($\bw=-\bv+\bv_2-\bv_1$),}\\
&=& \pm  \rme^{-\rmi  \frac{2 \pi}{3}} P_{n,\bK}^{B}(\bx)  \quad  \mbox{(since $\rme^{ \rmi \bK\cdot (\bv_1-\bv_2) }=\rme^{-\rmi \frac{2\pi}{3}}$)}.
\end{eqnarray*}
\end{proof}

\subsection{Existence of Dirac points for high contrast}\label{DP-exist}

 In this section, building upon the sufficient conditions of Theorem \ref{Thm.weakLyapounovSchmidtred},
   we prove that for $\aspar$ sufficiently large, the existence of Dirac points is reduced
 to the non-vanishing of the Dirac velocity $v_D(g)$. 
 For this result we require the following classical result of spectral theory on spectral measures and norm resolvent convergence. Its proof can be found in \cite[Theorem VIII.23 p. 289-290]{RS1}.
\begin{lemma}\label{lem.convspecproj}
Let $(\bbA_n)_{n\in \N}$ be a sequence of unbounded  self-adjoint operators on a Hilbert space $\mathcal{H}$ that converges to a self-adjoint operator $\bbA$  on  $\mathcal{H}$ in  the norm resolvent sense. That is, for any $\zeta\in \C\setminus \R$, $\|R_{\bbA_n}(\zeta)-R_{\bbA}(\zeta) \|_{B(\mathcal{H})}\to 0 $, as  $n\to+\infty$. Then, for any real $a,b\notin \sigma(\bbA)$, one has the following norm convergence of the spectral  projectors onto the real interval $(a,b)$:
$$
\|\bbE_{\bbA_n}\big((a,b)\big)- \bbE_{\bbA}\big((a,b)\big) \|_{B(\mathcal{H})} \to 0 \quad \mbox{ as }n \to + \infty.
$$
\end{lemma}

Let $\bbE_{{\bbA}_{\aspar,\bK}}(\cdot)$ and $\bbE_{\bbA_{\infty,\bK}}(\cdot)$ denote the spectral measures associated with the self-adjoint operators $\bbA_{\aspar,\bK}$ and $\bbA_{\infty,\bK}$. 
\begin{remark}
From Theorem \ref{thmsaympAtau},  one knows that  $\bbA_{\aspar,\bK}$ tends to $\bbA_{\infty,\bK}$ in the norm resolvent. But $\bbA_{\infty,\bK}$ acts on a strict closed subspace $\widetilde{L}^2_{\bK}$ of $L^2_{\bK}$, the Hilbert space associated to $\bbA_{\infty,\bK}$.  Thus, to give a meaning to  Lemma \ref{lem.convspecproj} in our setting, one needs as in Remark \ref{rem.extensionres} for the resolvent operator $R_{\bbA_{\infty,\bK}}(\zeta)$ to extend the definition of the spectral measure of $\bbE_{\bbA_{\infty,\bK}}(\cdot)$ to the space  $L^2_{\bK}=\widetilde{L}^2_{\bK}\oplus(\widetilde{L}^2_{\bK})^{\perp}$ by setting $\bbE_{\bbA_{\infty,\bK}}(I)u=0$ for any $u\in (\tilde{L}^2_{\bK})^{\perp}$ and any Borel sets $I$ of $\R$, for more details (see \cite{Hem-00,Sim:78}).
\end{remark}

\begin{theorem}[$\condS$  and $v_D(g)\ne0\implies$ existence of Dirac points]\label{th.degeneracy}
Assume that the eigenvalue $\tilde{\delta}_n$ of $\LaplDcellA$  satisfies the spectral isolation condition $\condS$ 
 of Definition \ref{Def.condS}. Let $\{P_{n,\bK}^{A},P_{n,\bK}^{B}\}$ be as in \eqref{eq.defPka}. Then, there exists $\aspar_*>0$ such that for all $\aspar>\aspar_*$ the following conditions hold:
\begin{enumerate}
\item $\lambda_{2n-1}(\aspar;\bK)=\lambda_{2n}(\aspar;\bK)=\lambda_D(\aspar)$ is a multiplicity $2$ eigenvalue of  $\bbA_{\aspar, \bK}$.
\item  $m_{\aspar,\bK,\tau}(\lambda_D(\aspar)=1$ and a normalized eigenfunction of $\operatorname{ker}_{\tau}(\bbA_{\aspar, \bK}-\lambda_D(\aspar) \, \mathrm{I}d)$ is given by
\begin{equation}\label{eq.defphi1}
\Phi_1(\aspar,\cdot):=\frac{\mathbb{E}_{\bbA_{\aspar,\bK}}(\{ \lambda_D(\aspar) \})\, P_{n,\bK}^A}{\|\mathbb{E}_{\bbA_{\aspar,\bK}}(\{ \lambda_D(\aspar) \})\, P_{n,\bK}^A\|_{L^2_{\bK}}},
\end{equation} 
 In addition, one has  $\Phi_1(\aspar,\cdot) \to P_{n,\bK}^A$ as $\aspar\to +\infty$ in the $L^2_{\bK}-$norm.
\item  $m_{\bK_*,\aspar,\overline{\tau}}(\lambda_D(\aspar))=1$ and a normalized eigenfunction in $\operatorname{ker}_{\overline{\tau}}(\bbA_{\aspar, \bK}-\lambda_D(\aspar) \, \mathrm{I}d)$
 is given by $\Phi_{2}(\aspar,\cdot)=\mathcal{P}\mathcal{C} \Phi_{1}(\aspar,\cdot)$. Moreover, this eigenfunction satisfies 
\begin{equation}\label{eq.defphi2}
\Phi_{2}(\aspar,\cdot)=\pm \rme^{-i  \frac{2\pi}{3}} \, \frac{\mathbb{E}_{\bbA_{\aspar,\bK}}(\{\lambda_D(\aspar)\})\, P_{n,\bK}^B}{\|\mathbb{E}_{\bbA_{\aspar,\bK}}(\{ \lambda_D(\aspar) \})\, P_{n,\bK}^B\|_{L^2_{\bK}}}.
\end{equation}
Equation \eqref{eq.defphi2} holds with the $+$ sign (resp. $-$) if  the eigenfunction $p_n$ is even (resp. odd) in Definition \eqref{eq.defPka}. Furthermore, $\Phi_2(\aspar,\cdot) \to \pm \rme^{-i  \frac{2\pi}{3}} \,  P_{n,\bK}^B$ as $\aspar\to +\infty$  in the $L^2_{\bK}-$norm.

\item $m_{\bK_*,\aspar,1}(\lambda_D(\aspar))=0$. Thus, $\bbA_{\aspar,\bK}\, u =\lambda_D(\aspar) \, u$  admits only  zero as solution in $L_{\bK,1}^2$. 

\item $\{\Phi_{1}(\aspar,\cdot),\Phi_{2}(\aspar,\cdot)\}$ is  an orthonormal basis of $\operatorname{ker}(\bbA_{\aspar,\bK}-\lambda_D(\aspar)\, \mathrm{I}d)$ in $L^2_{\bK}$. 
\end{enumerate}
Finally, if the Dirac velocity $v_D(\aspar)$  satisfies the non degeneracy condition:
\begin{equation}\label{eq.fermvelocnond}
v_D(\aspar)\equiv \Big| \int_{D} \sigma_{\aspar}\Phi_1(\aspar,\bx) \, \overline{\nabla \Phi_2(\aspar,\bx) }\, \rmd \bx \cdot \overline{(1,\rmi)}^{\top}\Big| \neq 0,
\end{equation}
then $(\bK,\lambda_D(\aspar))$ is a Dirac point   in the sense of Definition \ref{def.Diracpoints}. It follows 
 (see Remark \ref{rem.Diracpointsequivalence}) that  $(\bK_*,\lambda_D(\aspar))$ is a Dirac point for any vertex $\bK_*$ of the Brillouin zone $\mathcal{B}$.
\end{theorem}

Since the first eigenvalue $\tilde{\delta}_1$ of $\LaplDcellA$  (with corresponding positive eigenfunction) satisfies the spectral isolation condition $\condS$ of  Definition \ref{Def.condS}, we have:
\begin{corollary}\label{cor.Diracpointfirstdispcurv}
Let  $\aspar\ge\aspar_\star$ be sufficiently large.  Assume the non-degeneracy condition \eqref{eq.fermvelocnond} on $v_D(\aspar)$.
Then, there exists are Dirac points between the two first dispersion surfaces of $\mathbb{A}_\aspar$, at each vertex of the Brillouin zone $\mathcal{B}$.
\end{corollary}

\begin{proof} We prove now  Theorem \ref{th.degeneracy}. As we work at a fixed quasimomentum $\bK$, we omit in the proof the $\bK$  dependence of the eigenvalues and write for e.g. $\lambda_{2n-1}(\aspar)$ for  $\lambda_{2n-1}(\aspar;\bK)$.\\ 

{\bf Step 1}: {\it Localization of the eigenvalue  for high contrast.} $\tilde{\delta}_n$  satisfies the spectral isolation condition $\condS$ of  Definition \ref{Def.condS}. Therefore,  by  Theorem \ref{thm.convunifbands}, one has on one hand that  $\lambda_{2n-1}(\aspar)$ and $\lambda_{2n}(\aspar)$ tend to $\delta_{2n}$ as $\aspar \to +\infty$ and on the other  that there exists  a gap between the $2n^{th}$ and the $(2n+1)^{st}$ bands.  More precisely, by virtue of Theorem  \ref{thm.convunifbands}, for a fixed $\eta$ satisfying $0<\eta<\min(\delta_{2n},\nu_{2n+1}-\delta_{2n})$, there exists $\aspar_*>0$  such that for  $\aspar>\aspar_*$,
$\lambda_{2n-1}(\aspar), \lambda_{2n}(\aspar) \in (\delta_{2n}-\eta, \delta_{2n}+\eta)$ and $\lambda_{2n+1}(\aspar)\notin [\delta_{2n}-\eta, \delta_{2n}+\eta]$.
If $n>1$, we also need  that $\lambda_{2n-2}(\aspar)\notin [\delta_{2n}-\eta, \delta_{2n}+\eta]$. Indeed,  one knows from  Theorem \ref{thmsaympAtau} that $\lambda_{2n-2}(\aspar)\to \delta_{2n-2}$ as $\aspar\to +\infty $ and $\aspar \mapsto \lambda_{2n-2}(\aspar)$ is strictly increasing by Lemma \ref{lem.monoticitydispcurv}. Hence, $\lambda_{2n-2}(\aspar)<\delta_{2n-2}$. Furthermore,  by  the condition $\condS$, $\tilde{\delta}_n$ is simple and thus $\delta_{2n-2}=\tilde{\delta}_{n-1}< \delta_{2n}=\tilde{\delta}_{n}$. Thus, if $n>1$ for any fixed $\eta$ satisfying $0<\eta<\min(\delta_{2n}-\delta_{2n-2},\nu_{2n+1}-\delta_{2n})$, there exists $\aspar_*>0$ such that $\aspar>\aspar_*$ implies $\lambda_{2n-2}(\aspar),\lambda_{2n+1}(\aspar) \notin [\delta_{2n}-\eta, \delta_{2n}+\eta]$ and $\lambda_{2n-1}(\aspar), \lambda_{2n}(\aspar) \in (\delta_{2n}-\eta, \delta_{2n}+\eta)$. Hence, for $n\geq1$, the eigenvalue $\lambda_{2n-1}(\aspar)$ is  of multiplicity  at most $2$ for $\aspar>\aspar_*$. \\
 
{\bf Step 2}:  {\it Proof of point 1. of Theorem \ref{th.degeneracy}.}\  We prove this by  contradiction. Suppose $\lambda_{2n-1}(\aspar)=\lambda_{2n-2}( \aspar)$ does not hold for $\aspar$ large enough. Then, there is a sequence $(\aspar_m)$ with $\aspar_m\to+ \infty$ as $m\to \infty$ such that $\lambda_{2n-1}(\aspar_m)\neq\lambda_{2n}(\aspar_m)$. Step 1 enables us to localize and isolate these eigenvalues to an interval about $\delta_{2n}$; for  any fixed $\eta$, satisfying $0<\eta<\min(\delta_{2},\nu_{3}-\delta_{2})$ if $n=1$ or $0<\eta<\min(\delta_{2n}-\delta_{2n-2},\nu_{2n+1}-\delta_{2n})$ if $n>1$, there exists  $\aspar_*>0$ such that  for $\aspar_m>\aspar_*$:
\begin{equation}\label{eq.speclocl}
\sigma(\bbA_{\aspar_m,\bK})\cap (\delta_{2n}-\eta, \delta_{2n}+\eta)=\{\lambda_{2n-1}(\aspar_m), \lambda_{2n}(\aspar_m)\}.
\end{equation}
This yields the following relation on the spectral projectors of $\bbA_{\aspar_m,\bK}$:  
\begin{equation}\label{eq.specprojequality}
\bbE_{\bbA_{\aspar_m,\bK}} \big((\delta_{2n}-\eta, \delta_{2n}+\eta)\big)=\bbE_{\bbA_{\aspar_m,\bK}}(\{\lambda_{2n-1}(\aspar_m), \lambda_{2n}(\aspar_m)\}).
\end{equation}
We define 
\begin{equation}\label{eq.nonzeroeigenfunc}
 u_{\aspar_m}= \bbE_{\bbA_{\aspar_m,\bK}}(\{\lambda_{2n-1}(\aspar_m), \lambda_{2n}(\aspar_m)\})P_{n,\bK}^A,
 \end{equation}
 and we prove that $u_{\aspar_m}\neq 0$ for all $m$ large enough.
 The resolvent norm convergence of $\bbA_{\aspar_m,\bK}$ to $\bbA_{\infty,\bK}$ (see Lemma \ref{lem.convspecproj}) implies the strong norm convergence of  the spectral projectors associated to an open interval whose endpoints are not in $\sigma(\bbA_{\infty,\bK})$. 
Thus,  using   \eqref{eq.specprojequality} and the fact  $\delta_{2n}\pm \eta \notin \sigma(\bbA_{\infty,\bK})$, one has the following limit in the $L^2_{\bK}$ norm:
$$
u_{\aspar_m}= \bbE_{\bbA_{\aspar_m,\bK}} \big((\delta_{2n}-\eta, \delta_{2n}+\eta)\big) P_{n,\bK}^{A} \longrightarrow \bbE_{\bbA_{\infty,\bK}} \big((\delta_{2n}-\eta, \delta_{2n}+\eta)\big) P_{n,\bK}^{A}  
$$
as $m\to \infty$ with 
\begin{equation}\label{eq.nozeroeigenfunc}
\bbE_{\bbA_{\infty,\bK}} \big((\delta_{2n}-\eta, \delta_{2n}+\eta)\big) P_{n,\bK}^{A} =\bbE_{\bbA_{\infty,\bK}}(\{\delta_{2n}\}) P_{n,\bK}^{A} = P_{n,\bK}^{A}  \neq 0.
\end{equation}
The  first equality in \eqref{eq.nozeroeigenfunc} and the property $\delta_{2n}\pm \eta \notin \sigma(\bbA_{\infty,\bK})$ rely on the following point. One has $\delta_{2n}=\delta_{2n-1}$, $\delta_{2n+1}\notin (\delta_{2n}-\eta,\delta_{2n}+\eta)$ (since $\nu_{2n+1}\leq \delta_{2n+1}$ by Proposition \ref{Prop.nundeltan}) and $\delta_{2n-2}\notin (\delta_{2n}-\eta,\delta_{2n}+\eta)$ for $n>1$. Hence,  $\bbE_{\bbA_{\infty,\bK}} \big((\delta_{2n}-\eta, \delta_{2n}+\eta)\big) =\bbE_{\bbA_{\infty,\bK}}(\{\delta_{2n}\}) $. Further, since $P_{n,\bK}^A$ is a normalized eigenfunction of $\bbA_{\infty,\bK}$ associated to $\delta_{2n}$ (Lemma \ref{lem.pkapkb}), the second  equality of \eqref{eq.nozeroeigenfunc} holds.

Consider now the  orthogonal decomposition of the non-zero vector $u_{\aspar_m}$ defined by \eqref{eq.nonzeroeigenfunc}:
\begin{equation}\label{eq.orthogdecompeigenfunc}
u_{\aspar_m}=\bbE_{\bbA_{\aspar_m,\bK}}(\{\lambda_{2n-1}(\aspar_m)\}) P_{n,\bK}^A+ \bbE_{\bbA_{\aspar_m,\bK}}(\{\lambda_{2n}(\aspar_m)\})P_{n,\bK}^A .
 \end{equation}
By \eqref{eq.nozeroeigenfunc} at least one  term of the right hand side of \eqref{eq.orthogdecompeigenfunc} does not vanish for $m$ large enough. Up to a subsequence extraction  on the sequence $(\aspar_m)$, one can  assume
without a loss of generality that $\bbE_{\bbA_{\aspar_m,\bK}}(\{\lambda_{2n-1}(\aspar_m)\}) P_{n,\bK}^A\neq 0$ for $m$ large enough. Thus, we can define the normalized vector $\Phi_{1}(\aspar_m,\cdot)$ as in \eqref{eq.defphi1} (by replacing  $\lambda_D(\aspar)$ by $\lambda_{2n-1}(\aspar_m)$ in  \eqref{eq.defphi1}).
Recall by Lemma \ref{lem.pkapkb},  that $P_{n,\bK}^A\in  L_{\bK,\tau}$. Since $\mathcal{R}$ commutes with  $\bbA_{\aspar_m,\bK}$ (see Proposition \ref{prop.commutAk}), $\mathcal{R}$ commutes with its associated spectral measure (see Section \ref{sec-not}). Therefore,  $\bbE_{\bbA_{\aspar_m,\bK}}(\{\lambda_{2n-1}(\aspar_m)\})\ L_{\bK,\tau}^2\subset L_{\bK,\tau}^2 $ and thus $\Phi_{1}(\aspar_m,\cdot)\in \operatorname{ker}_{\tau}(\bbA_{\aspar_m, \bK}-\lambda_{2n-1}(\aspar_m) \, \mathrm{I}d) $ is a  $L_{\bK,\tau}^2$  normalized eigenfunction for the eigenvalue $\lambda_{2n-1}(\aspar_m)$. We now  define $\Phi_{2}(\aspar_m,\cdot)=\mathcal{P}\mathcal{C} \Phi_{1}(\aspar_m,\cdot)$ (which is normalized since $\mathcal{P}\mathcal{C}$ is anti-unitary). By Corollary \ref{cor.commut}, $\Phi_{2}(\aspar_m,\cdot)\in \operatorname{ker}_{\overline{\tau}}(\bbA_{\aspar_m, \bK}-\lambda_{2n-1}(\aspar_m) \, \mathrm{I}d) $
is a normalized  $L^2_{\bK,\overline{\tau}}$ eigenfunction associated to $\lambda_{2n-1}(\aspar_m)$.
Moreover, as $L^2_{\bK,\tau}$ is orthogonal to $L^2_{\bK,\overline{\tau}}$,  $\{\Phi_{1}(\aspar_m,\cdot),\Phi_{2}(\aspar_m,\cdot)\}$ is an orthonormal set of two eigenfunctions associated to the eigenvalue $\lambda_{2n-1}(\aspar_m)$. But $\lambda_{2n-1}(\aspar_m)$ is a simple eigenvalue. Indeed, by assumption $\lambda_{2n-1}(\aspar_m)\neq \lambda_{2n}(\aspar_m)$,  and  from Step 1, one has  $\lambda_{2n-2}(\aspar_m)<\lambda_{2n-1}(\aspar_m)$ for $m$ large enough if $n>1$.
(Of course the same contradiction is reached if we were to assume that, up to a subsequence extraction,  $\mathbb{E}_{\bbA_{\aspar_m,\bK}}(\{\lambda_{2n}(\aspar_m)\})\, P_{n,\bK}^A\neq 0$ for $m$ large enough.) 
Hence,  for $\aspar$ large enough , one has $\lambda_{2n-1}(\aspar)=\lambda_{2n}(\aspar)=:\lambda_D(\aspar)$. Step 1 implies $\lambda_D(\aspar)$ is of multiplicity at most $2$.  Hence, $\lambda_D(\aspar)$ is of multiplicity $2$. This completes the proof of point $1$. \\

{\bf Step 3}: {\it Proof of points 2-5 in Theorem \ref{th.degeneracy}:} 
The reasoning of Step 2 applied to $\bbE_{\bbA_{\aspar;\bK}}\big(\{\lambda_D(\aspar;\bK)\}\big)$  shows  that the normalized eigenfunction $\Phi_{1}(\aspar,\cdot)$  defined by \eqref{eq.defphi1}  belongs to $\operatorname{ker}_{\tau}(\bbA_{\aspar, \bK}-\lambda_D(\aspar) \, \mathrm{I}d)$ and thus $m_{\bK_*,\aspar,\tau}(\lambda_D(\aspar)\geq 1$ . Moreover the strong convergence of the spectral projector implies:
$$
\Phi_{1}(\aspar,\cdot) \rightarrow P_{n,\bK}^A \mbox{ as } \aspar \to \infty \  \mbox{ since} \ \bbE_{\bbA_{\infty,\bK}}(\{\delta_n\})   P_{n,\bK}^A =P_{n, \bK}^A \mbox{ and } \|P_{n,\bK}^A\|_{L^2_{\bK}}=1.
$$
In the same way, we have that $\Phi_{2}(\aspar,\cdot)=\mathcal{P}\mathcal{C} \Phi_{1}(\aspar,\cdot)\in \operatorname{ker}_{\overline{\tau}}(\bbA_{\aspar, \bK}-\lambda_D(\aspar) \, \mathrm{I}d)$ is a  normalized eigenfunction associated to  $\lambda_{D}(\aspar)$ and thus $m_{\bK_*,\aspar,\overline{\tau}}(\lambda_D(\aspar))\geq 1$. Hence,
$\{\Phi_{1}(\aspar,\cdot),\Phi_{2}(\aspar,\cdot)\}$ is an orthonormal set of eigenfunctions associated to $\lambda_{D}(\aspar)$. Since $\lambda_{D}(\aspar)$ is  of multiplicity $2$, it follows from  the relation \eqref{eq.propcommutdim} that $$m_{\bK_*,\aspar,\tau}(\lambda_D(\aspar))= m_{\bK_*,\aspar,\overline{\tau}}(\lambda_D(\aspar))=1 \ \mbox{ and  } \ m_{\bK_*,\aspar,1}(\lambda_D(\aspar))=0.$$  This proves points  2, 4 and 5. 

To prove relation \eqref{eq.defphi2} of the point 3,   
apply $\mathcal{P}\mathcal{C}$ to \eqref{eq.defphi1} and use  the identity \eqref{eq.PkAPKB} and the facts that $\mathcal{P}\mathcal{C}$ commutes with $\bbE\big(\{\lambda_D(\aspar)\}\big)$ and  preserves the norm.

Finally, if  the non-degeneracy condition \eqref{eq.fermvelocnond} on $v_D(\aspar)$ holds, all assumptions of Theorem \ref{Thm.weakLyapounovSchmidtred} hold and we conclude that $(\bK,\lambda_D(\aspar))$ is a Dirac point in the sense of Definition \ref{def.Diracpoints}.
\end{proof}

\section{High contrast asymptotic analysis of Bloch eigenelements}\label{sec-asympresult}

\subsection{Expansion in powers of $\aspar^{-1}$ and hierarchy of PDEs in $L^2_{\bK,\tau}\cap H^1_{\bK}$}

In this section, we derive asymptotic expansions for a large contrast parameter, $\aspar$,  of Bloch eigenvalues, eigenfunctions at vertices of $\mathcal{B}$, and the Dirac velocity $v_D(\aspar)$. The validity of these asymptotic expansions is proved in Section \ref{sec.quasimode}.  These expansions  hold for eigenvalues whose limiting behavior  is given by a  Dirichlet  eigenvalue $\tilde{\delta}_n$  for $n\geq 1$ that satisfies the spectral isolation condition
 $\condS$ of  Definition \ref{Def.condS} (e.g. for the two first bands when $n=1$, or higher energy bands as it is illustrated in  section \ref{sec.higherbandscirccase}). We focus on the vertex $\bK$ of $\mathcal{B}$, all results 
apply by symmetry to the other vertices;
 see  Remark  \ref{rem.Diracpointsequivalence}.

If $\tilde{\delta}_n$ satisfies the condition $\condS$ of  Definition \ref{Def.condS}, the conclusions of Theorem \ref{th.degeneracy} hold.
In particular, for $\aspar$ large enough, the eigenvalue $\lambda_{2n-1}(\aspar;\bK)=\lambda_{2n}(\aspar;\bK)=\lambda_D(\aspar)$ is of multiplicity $2$  and  there exists an orthonormal basis $\{\Phi_1(\aspar,\cdot),\Phi_2(\aspar,\cdot)\}$   of $\operatorname{ker}(\bbA_{\aspar,\bK}-\lambda_D(\aspar) \mathrm{I}d)$, where $\Phi_1(\aspar,\cdot)\in L^2_{\bK,\tau}$ and $\Phi_2(\aspar,\cdot)=\mathcal{P}\mathcal{C}\Phi_1(\aspar,\cdot) \in L^2_{\bK,\overline{\tau}}$.

We expand the eigenelements  $\lambda_D(\aspar)$ and $\Phi_1(\cdot,\aspar)$ in powers of  $\aspar^{-1}$ for $\aspar$ large. By Theorems \ref{thm.convunifbands} and \ref{th.degeneracy} we have $ \lambda_D(\aspar)\to \tilde{\delta}_n $ and we can choose  a normalized eigenfunction
$
\Phi_1(\aspar,\cdot)  \to \Phi_1^{(0)}=P_{n,\bK}^A \mbox{ in } L^2_{\bK,\tau},
$
as  $\aspar \to +\infty$. Here, $P_{n,\bK}^A$ is defined by \eqref{eq.defphi1}. Thus, we formally expand, for any $M\ge1$:
\begin{align}\label{eq.ansatz1}
\lambda_D(\aspar) &= \sum_{m=0}^{M} \aspar^{-m}  \lambda_D^{(m)} +O(\aspar^{-(M+1)} ) \\
   \Phi_1( \aspar,\cdot) &= \sum_{m=0}^{M} \aspar^{-m}  \Phi_1^{(m)} +O(\aspar^{-(M+1)}) \mbox{ with } \Phi_1^{(m)} \in  L^2_{\bK,\tau} \cap  H^1_{\bK}
\nonumber\end{align}
with  $\lambda_D^{(0)} =\tilde{\delta}_n$ and $\Phi_1^{(0)}= P_{n,\bK}^A $.
We remark that if $H^1_{\bK}$ convergence holds for \eqref{eq.ansatz1}, then 
$$
\Phi_2(\aspar,\cdot)=\mathcal{P}\mathcal{C}\Phi_1(\aspar,\cdot)=\sum_{m=0}^{M} g^{-m} \, \mathcal{P}\mathcal{C}\Phi_1^{(m)} +O\big(g^{-(M+1)} \big) \mbox{ with }  \mathcal{P}\mathcal{C} \Phi_1^{(m)} \in  L^2_{\bK,\overline{\tau}} \cap  H^1_{\bK} .
$$

We next derive recursion relations  for $\Phi_1^{(m)}(\bx)$ and $\lambda_D^{(m)}$,  $m\ge1$.
We substitute the expansion \eqref{eq.ansatz1}  into the eigenvalue problem:
\begin{equation}\label{eq.eigenvalK}
\bbA_{\aspar,\bK}  \Phi_1(\aspar,\cdot )=\lambda_D(\aspar)   \Phi_1(\aspar, \cdot ) \mbox{ with } \Phi_1(\aspar,\cdot )\in \rmD(\bbA_{\aspar,\bK}).
\end{equation}
Thus $\Phi_1(\aspar,\cdot )$ is the solution of the following cell problem on $\cell$:
\begin{equation*}\label{eq.transmission}
\left\{ \begin{array}{ll}
 -  \Delta \Phi_1(\aspar,\cdot ) = \lambda_D(\aspar) \Phi_1(\aspar, \cdot )  & \mbox{ in } \cell^{+}, \\[6pt]
[\Phi_1(\aspar, \cdot)]=0 \mbox{ and } \displaystyle \Big[\sigma_\aspar \frac{\partial \Phi_1(\aspar, \cdot )}{\partial \bn} \Big]=0  & \mbox{ on } \partial{\cell}^{+}, \\[6pt]
- \aspar\, \Delta \Phi_1(\aspar, \cdot) = \lambda_D(\aspar) \Phi_1(\aspar, \cdot)  & \mbox{ in }  \cell^{-},\\[6pt]
\Phi_1(\aspar, \cdot) \ \mbox{ and } \ \displaystyle  \frac{\partial \Phi_1(\aspar, \cdot)}{\partial \bn}&  \hspace{-2cm}  \mbox{ $\bK-$quasi-periodic} \mbox{ on } \partial{\cell}.
\end{array} \right.
\end{equation*}
The jump of  traces  on $ \partial{\cell}^{+}$ is here defined by $[f]=f^- - f^+$ where $f^{\pm}$ is the trace of $f$  on $\partial \cell^{+}$ from the domain $\cell^{\pm}$. $\bn$ is  the outward unit normal vector oriented from $\cell^{+}$ to $\cell^{-}$ so that with the same convention the jump of Neumann  traces on  $\partial{\cell}^{+}$ is defined by $[\partial f/\partial \bn]=[\partial f/\partial \bn]^- - [\partial f/\partial \bn]^+$.
Inserting the expansions \eqref{eq.ansatz1} in \eqref{eq.eigenvalK} gives,  for all $M\ge0$:
{\small{
\begin{equation}\label{eq.transmission1}
\left\{ \begin{array}{ll} \displaystyle
\sum_{m=0}^{M} -   \aspar^{-m}  \Delta  \Phi_1^{(m)}+O(\aspar^{-(M+1)}) = \Big( \sum_{m=0}^{M} \aspar^{-m} \lambda_D^{(m)} +O(g^{-(M+1)})\Big)  \Big(\sum_{m=0}^{M} \aspar^{-m} \Phi_1^{(m)} +O(\aspar^{-(M+1)})\Big)\ \  \mbox{in }  \cell^{+}, \\[15pt]
 \displaystyle \Big[ \sum_{m=0}^{M} \aspar^{-m} \,   \Phi_1^{(m)} +O(\aspar^{-(M+1)}) \Big]^{-}=\Big[  \sum_{m=0}^{M} \aspar^{-m}  \Phi_1^{(m)}  +O(\aspar^{-(M+1)})\Big]^{+} \mbox{on } \partial{\cell}^{+}, \\[15pt] 
 \displaystyle\Big[\sum_{m=0}^{M}  \aspar^{-m+1} \, \frac{\partial  \Phi_1^{(m)}}{\partial \bn}+O(\aspar^{-M}) \Big]^{-}= \Big[\sum_{m=0}^{M}  \aspar^{-m}\frac{\partial  \Phi_1^{(m)}}{\partial \bn}+O(\aspar^{-(M+1)}) \Big]^{+} \mbox{on } \partial{\cell}^{+},\\[15pt]
 \displaystyle \sum_{m=0}^{M} -\aspar^{-m+1} \Delta\, \Phi_1^{(m)} +O(\aspar^{-M} )= \Big( \sum_{m=0}^{M} \aspar^{-m} \, \lambda_D^{(m)} +O(g^{-(M+1)}\Big) \, \Big(\sum_{m=0}^{M} \aspar^{-m} \, \Phi_1^{(m)} +O\big(\aspar^{-(M+1)} \big)\Big)\ \    \mbox{in }  \cell^{-},\\[6pt]
\Phi_1^{(m)} \mbox{ and } \displaystyle  \frac{\partial \Phi_1^{(m)}}{\partial \bn}   \mbox{  $\bK-$quasi-periodic } \mbox{ for }m=0,1,\ldots M \mbox{ on } \partial{\cell}.
\end{array} \right. \ .
\end{equation}
}}To simplify the notation, in the following we write $\partial f^{\pm}/\partial \bn$ for  the Neumann trace of $f$ on $\partial \cell^{+}$  from the domain $\cell^{\pm}$.
\subsection{Determination of  $\Phi_1^{(0)}\in L^2_{\bK,\tau}\cap H^1_{\bK}$ and $\lambda_D^{(0)}\in\R$}

In \eqref{eq.transmission1}, the order $\aspar^0$ term in $\cell^{+}$ and the order $\aspar$ term in $\cell^{-}$ yield the following equations:
\begin{equation}\label{eq.transmissionzeroorder}
\left\{ \begin{array}{ll}
 -  \Delta \Phi_1^{(0)}= \lambda_D^{(0)}   \Phi_1^{(0)}  & \mbox{ in }  \cell^{+}, \\[5pt]
 \Phi_1^{(0),+}= \Phi_1^{(0),-} \mbox{ and } \displaystyle \frac{\partial \Phi_1^{(0),-}}{\partial \bn}=0   & \mbox{ on } \partial{\cell}^+ \\[5pt]
- \Delta  \Phi_1^{(0)} =0  & \mbox{ in }  \cell^{-},\\[5pt]
 \Phi_1^{(0)} \mbox{ and } \displaystyle  \frac{\partial  \Phi_1^{(0)}}{\partial \bn}&  \hspace{-3cm}  \mbox{   $\bK-$quasi-periodic } \mbox{ on } \partial{\cell}.
\end{array} \right.
\end{equation}
We first solve for $ \Phi_1^{(0)}$ on ${\bf \cell^{-}}$. By the Lax-Milgram theorem, $ \Phi_1^{(0)}=0$ is the unique solution of the corresponding PDE problem in $H^{1}_{\bK}({\bf \cell^{-}})$ (see Appendix   \ref{sec-appendixcommutation} for the definition of the Sobolev space $H^{1}_{\bK}({\bf \cell^{-}})$). 
By the boundary 
conditions along $\partial\cell^+$, we have that $\Phi_1^{(0),+}=\Phi_1^{(0),-}=0$. Thus taking $\lambda_D^{(0)}=\tilde{\delta}_n$ and $\Phi_1^{(0)}$ on $\cell^+=\cell^A\cup\cell^B$ to be a Dirichlet eigenfunction on $\cell^A$ vanishing on $\cell^B$ solves the PDE problem in $\cell^+$.
Hence, the  $\bK$-quasi-periodic 
 extension to all $\R^2$, $\Phi_1^{(0)}=P_{n,\bK}^A\in L^2_{\bK,\tau}\cap H^1_{\bK}$ , satisfies all conditions of 
\eqref{eq.transmissionzeroorder} (Lemma \ref{lem.pkapkb}).

Starting with this choice of  $\lambda_D^{(0)}$  and  $\Phi_1^{(0)}$, we next define  $\lambda_D^{(m)}$ and $\Phi_1^{(m)}$ recursively  for all $m\geq 1$. The Sobolev spaces used in the following discussion are defined in  Appendix   \ref{sec-appendixcommutation}.
The symmetry operators $\mathcal{R}_{\bf \cell^{\pm}}$, $\mathcal{P}\mathcal{C}_{\bf \cell^{\pm}}$  and  $\mathcal{R}_{ \partial {\bf  \Omega^{+}}}$ and $\mathcal{P}\mathcal{C}_{ {\partial \bf  \Omega^{+}}} $ which are the operators $\mathcal{R}$ and $\mathcal{P}\mathcal{C}$ but defined on  the sets ${\bf \cell^{\pm}}$ and ${\partial \bf  \Omega^{+}} $ are also defined in this section.

\subsection{Determination of $\Phi_1^{(1)}\in L^2_{\bK,\tau}\cap H^1_{\bK}$ and $\lambda_D^{(1)}\in\R$}
In \eqref{eq.transmission1} the order $\aspar^{-1}$  term in $\cell^{+}$  and the order $\aspar^{0}$ term in $\cell^{-}$ yield:
\begin{equation}\label{eq.transmissionfirstorder}
\left\{ \begin{array}{ll}
 -  \Delta \Phi_1^{(1)}= \lambda_D^{(0)} \Phi_1^{(1)}+\lambda_D^{(1)} \Phi_1^{(0)}  & \mbox{ in }  \cell^{+}, \\[5pt]
\Phi_1^{(1),-}=\Phi_1^{(1),+} \mbox{ and }  \displaystyle  \frac{\partial \Phi_1^{(1),-}}{\partial \bn}=\frac{\partial P_{n,\bK}^{A,+} }{\partial \bn}  & \mbox{ on } \partial{\cell}^{+}, \\[10pt]
- \Delta \Phi_1^{(1)} =\lambda_D^{(0)}  \Phi_1^{(0)}  =0  & \mbox{ in }  \cell^{-},\\[10pt]
\Phi_1^{(1)} \mbox{ and } \displaystyle  \frac{\partial \Phi_1^{(1)}}{\partial \bn}&  \hspace{-3.7cm}  \mbox{  $\bK-$quasi-periodic } \mbox{ on } \partial{\cell}.
\end{array} \right.
\end{equation}
The latter two equations of \eqref{eq.transmissionfirstorder} imply that   $\Phi_1^{(1)}$  on the domain ${\bf \cell^{-}}$
 is determined by the boundary value problem: 
\begin{equation}\label{eq.phi1D-}
- \Delta \Phi_1^{(1)}=0  \mbox{ on } \cell^- \mbox{ and }  \displaystyle \frac{\partial \Phi_1^{(1),-}}{\partial \bn}= \frac{\partial P_{n,\bK}^{A,+}}{\partial \bn} \mbox{ on } \partial \cell^+, 
\end{equation}
with  $\bK-$quasi-periodic boundary conditions for $\Phi_1^{(1)} \mbox{ and }  \partial \Phi_1^{(1)} / \partial \bn \mbox{ on } \partial{\cell}$. By the Lax-Milgram theorem, 
this problem admits a unique solution in $H^{1}_{\bK}({\bf \cell^{-}})$ that belongs to $H^{1}_{\bK,\Delta}({\bf \cell^{-}})$.

To eventually establish that $\Phi_1^{(1)}\in L^2_{\bK,\tau}$, we first verify that $\Phi_1^{(1)}$ on $\cell^-$ inherits
that symmetries of $P_{n,\bK}^{A}$. 
Since  $P_{n,\bK}^{A} \in H^1_{\bK,\Delta}( {\bf  \Omega^{+}})$,
 using Lemma \ref{lem.comquasiper2} and the fact $P_{n,\bK}^A\in L^2_{\bK,\tau}$ yields that $\mathcal{R}_{ {\bf  \Omega^{-}}} \Phi_1^{(1)}\in H^1_{\bK}( {\bf  \Omega^{-}})$ satisfies
\begin{equation*}\label{eq.symrelation}
- \Delta \mathcal{R}_{ {\bf  \Omega^{-}}} \Phi_1^{(1)} =0  \mbox{ on } \cell^- \mbox{ and }  \displaystyle \Big[ \frac{\partial \mathcal{R}_{ \bf \Omega^-} \Phi_1^{(1)}}{\partial \bn}\Big]^{-}= \tau  \frac{\partial P_{n, \bK}^{A,+}}{\partial \bn} ,
\end{equation*}
with $\bK-$quasi-periodic boundary conditions   for $  \mathcal{R}_{ {\bf  \Omega^{-}}} \Phi_1^{(1)}  \mbox{ and }  \partial  \mathcal{R}_{ {\bf  \Omega^{-}}} \Phi_1^{(1)}  / \partial \bn \mbox{ on } \partial{\cell}$.
Therefore, $\mathcal{R}_{\bf \Omega_-} \Phi_1^{(1)} $ and $\tau \Phi_1^{(1)}$ satisfies the  same  boundary value problem which admits a unique solution  in $H^1_{\bK}( {\bf  \Omega^{-}})$ and it follows that
\begin{equation}\label{eq.rotationOmegam}
\mathcal{R}_{\bf  \Omega^{-}} \Phi_1^{(1)}=\tau \, \Phi_1^{(1)} \mbox{ on  } {\bf  \Omega^{-}} \ \mbox{ and  } \ [ \mathcal{R}_{\bf  \Omega^{-}} \Phi_1^{(1)}]^{-}=\tau\, \Phi_1^{(1),-} \mbox{ on  } \partial {\bf  \Omega^{-}}.
\end{equation}

We  next construct   $\Phi_1^{(1)}\in H^{1}_{\bK}({\bf  \Omega^{+}})$ and $\lambda_D^{(1)}$ satisfying
\begin{equation}\label{eq.compatcondi1}
 -  \Delta \Phi_1^{(1)}=\lambda_D^{(0)} \ \Phi_1^{(1)}+\lambda_D^{(1)}  \Phi_1^{(0)}  \ \mbox{ in } \cell^+ \ \mbox{ and } \Phi_1^{(1),+}=\Phi_1^{(1),-} \mbox{ on } \partial \cell^+,
\end{equation}
and compatible with the goal of obtaining $\Phi_1^{(1)}\in L^2_{\bK,\tau}\cap H^{1}_{\bK}$.

We  claim that   $\lambda_D^{(1)}$ can be chosen so that $\Phi_1^{(1)}$  is  unique  in the space 
\begin{equation}\label{eq.defW1}
\mathcal{W}_{1}=\Big\{ u\in H^1_{\bK}( {\bf  \Omega^{+}})  \, \mid \left(u,P_{n,\bK}^A\right)_{L^2_{\bK}({\bf  \Omega^{+}})} =0 \mbox{ and } \left(u,P_{n,\bK}^B\right)_{L^2_{\bK}{(\bf  \Omega^{+}})}=0 \Big\}.
\end{equation}

Seek $\Phi_1^{(1)}\in \mathcal{W}_{1} $, such that $\Phi^{(1)}_1=u_1+v_1$ where $u_1\in \mathcal{W}_{1} $ is an extension, to the region $\cell^+$ of  the boundary values: $\Phi_1^{(1),-}\in H^{\frac{1}{2}}_{\bK}( {\bf \partial \cell^{+}})$. 
Let us construct this extension. By standard elliptic theory, there exists a unique $\tilde{u}_1 \in H_{\bK}^{1}(\bf{ \cell^{+}})$ such that  $ \Delta \tilde{u}_1=0$ on $\cell^+$ and
$\tilde{u}_1=\Phi_1^{(1),-}$ on $\partial \cell^{+}$.  Furthermore,  \eqref{eq.trace} and \eqref{eq.rotationOmegam}  imply  $\mathcal{R}_{\bf  \partial \cell^+}\Phi_1^{(1),-}=[\mathcal{R}_{\bf  \Omega^{-}} \Phi_1^{(1)}]^{-}=\tau \Phi_1^{(1),-} $. Thus, it follows from Lemma \ref{lem.comquasiper} that $\mathcal{R}_{\bf  \cell^+} \tilde{u}_1=\tau \, \tilde{u}_1$ since  $\mathcal{R}_{\bf  \cell^+} \tilde{u}_1$ and $ \tau \tilde{u}_1$ are both solutions of the elliptic boundary value problem:
$$
\Delta u=0 \ \mbox{ on } \cell^+ \ \mbox{ and } \ u= \tau \, \Phi_1^{(1),-}  \mbox{ on } \partial \cell^+,
$$
which admits a unique solution in  $H_{\bK}^{1}({\bf  \Omega^{+}})$. Now, we set 
$$u_1=\tilde{u}_1-(\tilde{u}_1,P_{n,\bK}^A )_{L^2_{\bK}({\bf  \Omega^{+}} )}\, P_{n,\bK}^A.$$  
Note that $\mathcal{R}_{\bf  \cell^+} u_1=\tau u_1$ and $\mathcal{R}_{\bf  \cell^+}  P_{n,\bK}^B=\overline{\tau} P_{n,\bK}^B$. Therefore,  
$({u}_1,P_{\bK,B} )_{L^2_{\bK}({\bf  \Omega^{+})}}=0$,
  and thus $u_1\in  \mathcal{W}_{1}$.
On the other hand,  $u_1=\Phi_1^{(1),-}$ on ${\partial \bf \cell^{+}}$ since $P_{n,\bK}^A=0$ on ${\partial \bf \cell^{+}}$. Furthermore, one has   $u_1\in H^1_{\bK,\Delta}( {\bf \cell^{+}})$ (since $ \Delta \tilde{u}_1=0$ and $P_{n,\bK}^A\in H^1_{\bK,\Delta}( {\bf \cell^{+}} )$).

We now construct $v_1$. Since $\Phi^{(1)}_1=u_1+v_1$, equation \eqref{eq.compatcondi1} can be rewritten as
\begin{equation}\label{eq.v_1}
 \big(-  \Delta -\lambda_D^{(0)} \big)  v_1= \big(\Delta +\lambda_D^{(0)}\big) u_1+ \lambda_D^{(1)} \Phi_1^{(0)}  \mbox{ in }  \cell^+ \  \mbox{ and }  \ v_1=0  \mbox{ on } \partial \cell^+ .
\end{equation}
By the Fredholm alternative (see e.g.  \cite{Mclean:2000}),  \eqref{eq.v_1}   admits a unique solution $v_1\in \mathcal{W}_{1}$  if and only if the right hand side is orthogonal  $\mathrm{Ker}(\bbA_{\infty,\bK}-\lambda_D^{(0)}\mathrm{I}d)=\mathrm{span}\{P_{n,\bK}^A, P_{n,\bK}^B\}$ restricted to ${\bf  \Omega^{+}}$. Moreover, such a solution $v_1$ will satisfy  the additional regularity $v_1\in  H^1_{\bK,\Delta}( {\bf \cell^{+}} ) $  inherited from the equation \eqref{eq.v_1}.

We first impose orthogonality to $P_{n,\bK}^A$.
Using  Green's identity, relations \eqref{eq.transmissionzeroorder} and \eqref{eq.transmissionfirstorder} and the relations $\lambda_D^{(0)} =\tilde{\delta}_n$, $\Phi_1^{(0)}= P_{n,\bK}^A $ and $u_1=\Phi_1^{(1),+}=\Phi_1^{(1),-}$ on  ${\partial \bf \Omega^{+}}$ we have:
\begin{eqnarray*}
0 &=& \int_{\cell^{+}} [(\Delta +\lambda_D^{(0)}) u_1+ \lambda_D^{(1)}  \Phi_1^{(0)} ] \cdot \overline{P_{n,\bK}^A}  \, \rmd \bx\\ 
&=&\Big\langle  \frac{\partial  u_1}{\partial \bn},\overline{ P_{n,\bK}^A}\Big\rangle _{ H^{-1/2}_{\bK},H^{1/2}_{\bK}}-\Big\langle  \overline{ \frac{\partial P_{n,\bK}^{A,+}  }{\partial \bn} } , u_1 \Big\rangle _{ H^{-1/2}_{\bK},H^{1/2}_{\bK}}+ \lambda_D^{(1)} \\
&=&-\Big\langle  \overline{ \frac{\partial P_{n,\bK}^{A,+}  }{\partial \bn} } , \Phi_1^{(1),-}\Big\rangle _{ H^{-1/2}_{\bK},H^{1/2}_{\bK}}+ \lambda_D^{(1)},
\end{eqnarray*}
where  $\left\langle \cdot,\cdot \right \rangle _{  H^{-1/2}_{\bK},H^{1/2}_{\bK}}$ stands for the duality product between the Sobolev spaces  $H^{-1/2}_{\bK}({\bf \partial \Omega^{+}})$ and $H^{1/2}_{\bK}({\bf \partial \Omega^{+}})$ (see Appendix   \ref{sec-appendixcommutation} for more details). Therefore,  it yields
\begin{equation*}\label{eq.compatcondi}
\lambda_D^{(1)}=\Big\langle  \overline{ \frac{\partial P_{n,\bK}^{A,+}  }{\partial \bn} } , \Phi_1^{(1),-} \Big \rangle _{ H^{-1/2}_{\bK},H^{1/2}_{\bK}}.
\end{equation*}

Note that since the eigenvalue $\aspar \mapsto  \lambda_D(\aspar)=\lambda_{2n}(\aspar;\bk)$ is increasing and approaches $\tilde{\delta}_n$ (see Lemma \ref{lem.monoticitydispcurv} and Theorem \ref{thm.convunifbands}), we must have that $\lambda_D^{(1)}\leq 0$. Indeed, this can be explicitly displayed.
Using the equation \eqref{eq.phi1D-} and applying the Green's  identity  in $\cell^{-}$, one obtains that
\begin{eqnarray}\label{eq.deflmabda1}
\lambda_D^{(1)}&=& \Big\langle  \overline{ \frac{\partial \Phi^{(1),-}}{\partial \bn} }  , \Phi^{(1),-} \Big \rangle _{ H^{-1/2}_{\bK},H^{1/2}_{\bK}} \ \ (\mbox{where $-\bn$ is  here the outward normal to $\cell^{-}$}),  \nonumber\\
&=&\int_{\cell^{-}}- \Phi^{(1)}_1  \, \overline{ \Delta \Phi_1^{(1)}} - | \nabla \Phi^{(1)} |^2 \, \rmd \bx  \  \quad   \nonumber\\
&=&- \int_{\cell^{-}}  | \nabla \Phi^{(1)} |^2 \, \rmd \bx <0.
\end{eqnarray} 
This previous expression of $\lambda_D^{(1)}$ was obtained with a different approach by \cite{Ammari2009layer} for the case of a simple Bloch eigenvalue  and a  square lattice at any non-zero quasimomentum.

We now verify orthogonality  of the right hand side of \eqref{eq.v_1} to $P_{n, \bK}^B$.  Since $P_{n,\bK}^A=\Phi_1^{(0)}$ and $P_{n,\bK}^B$  are orthogonal in $L^2_{\bK}({\bf  \Omega^{+}})$:
\begin{eqnarray*}
\int_{\cell^{+}} [(\Delta +\lambda_D^{(0)}) u_1+ \lambda_D^{(1)}   \Phi_1^{(0)}  ] \cdot \overline{P_{n, \bK}^B }  \, \rmd \bx&=&\Big\langle  \frac{\partial  u_1}{\partial \bn}, \overline{P_{n,\bK}^B}\Big\rangle _{ H^{-1/2}_{\bK},H^{1/2}_{\bK}}-\Big\langle  \overline{ \frac{\partial P_{n,\bK}^{B,+}  }{\partial \bn} } , u_1 \Big \rangle _{ H^{-1/2}_{\bK},H^{1/2}_{\bK}}
 \\
&=&-   \Big\langle  \overline{  \frac{\partial  P^{B,+}_{n,\bK} }{\partial \bn} } ,  \Phi^{(1),-}_1 \Big \rangle_{ H^{-1/2}_{\bK},H^{1/2}_{\bK}}.
\end{eqnarray*}
As $\mathcal{R}_{\partial \bf \cell^+}$ is unitary in $H^{1/2}_{\bk}(\partial \bf \cell^+)$, one deduces from the definition \eqref{eq.defrotneumtrace} of the operator $\mathcal{R}_{\partial \bf \cell^+}$ in  $H^{-1/2}_{\bk}(\partial \bf \cell^+)$ that
$$
\Big\langle \frac{\partial  P^{B,+}_{n,\bK} }{\partial \bn }, \overline{ \Phi^{(1),-}_1} \Big \rangle_{ H^{-1/2}_{\bK},H^{1/2}_{\bK}} =\Big\langle \mathcal{R}_{\partial \bf \cell^+}\frac{\partial  P^{B,+}_{n,\bK} }{\partial \bn }, \overline{\mathcal{R}_ {\partial \bf \cell^+} \Phi^{(1),-}_1} \Big \rangle_{ H^{-1/2}_{\bK},H^{1/2}_{\bK}}.
$$
Hence, from  the relations \eqref{eq.rotationOmegam},  \eqref{eq.trace}, \eqref{eq.tracenormal}, and the fact that $ P_{n,\bK}^B\in L_{\bK,\overline{\tau}}$, it follows that
\begin{eqnarray*}
\Big\langle \frac{\partial  P^{B,+}_{n,\bK} }{\partial \bn }, \overline{ \Phi^{(1),-}_1} \Big \rangle_{ H^{-1/2}_{\bK},H^{1/2}_{\bK}}&=&\Big\langle \Big[ \frac{\partial \mathcal{R}_{\bf \cell^+} P^{B}_{n,\bK} }{\partial \bn }\Big]^{+}, \overline{\mathcal{R}_ {\partial \bf \cell^+} \Phi^{(1),-}_1} \Big \rangle_{ H^{-1/2}_{\bK},H^{1/2}_{\bK}}\\
&=&\overline{\tau}^2 \Big\langle \frac{\partial  P^{B,+}_{n,\bK} }{\partial \bn }, \overline{ \Phi^{(1),-}_1} \Big \rangle_{ H^{-1/2}_{\bK},H^{1/2}_{\bK}}.
\end{eqnarray*}
As $\overline{\tau}^2=\tau\neq 1$, one concludes that:
$$
\Big\langle  \overline{  \frac{\partial  P^{B,+}_{n,\bK} }{\partial \bn} } ,  \Phi^{(1),-}_1 \Big \rangle_{ H^{-1/2}_{\bK},H^{1/2}_{\bK}}=\overline{ \Big\langle \frac{\partial  P^{B,+}_{n,\bK} }{\partial \bn }, \overline{ \Phi^{(1),-}_1} }\Big \rangle_{ H^{-1/2}_{\bK},H^{1/2}_{\bK}}=0.
$$
Thus, the second compatibility condition holds automatically by symmetry arguments.

Hence if $\lambda_D^{(1)}$ is given by $\eqref{eq.deflmabda1}$, then  there is a unique solution $v_1\in \mathcal{W}_{1}$ of \eqref{eq.v_1} and one concludes that $\Phi^{(1)}_1=u_1+v_1$  is the unique solution of \eqref{eq.compatcondi1} in $\mathcal{W}_{1} $ (with the additional regularity $\Phi^{(1)}_1\in  H^1_{\bK,\Delta}( {\bf \cell^{+}} ) $ inherited from the equation \eqref{eq.compatcondi1}).
Moreover, using that  $\mathcal{R} \Delta u_1= \Delta \mathcal{R} u_1 $, $\mathcal{R}_{\bf \Omega^+} u_1=\tau u_1$, $\mathcal{R}_{\bf \Omega^+}  P_{n,\bK}^A=\tau  P_{n,\bK}^A$ on ${\bf \cell^+}$,  one easily checks that
$$
\mathcal{R}_{\bf \Omega^+}[(\Delta +\delta_1) u_1+ \lambda_D^{(1)}  P_{n,\bK}^A]=\tau \big( (\Delta +\delta_1) u_1+ \lambda_D^{(1)}  P_{n,\bK}^A \big)\mbox{ on }{ \bf \cell^+}.
$$ 
Furthermore, from the definition \eqref{eq.defW1} of  $\mathcal{W}^{1}$, one deduces immediately with Lemma \ref{lem.comquasiper}  that $\mathcal{W}^{1}$ is stable by $\mathcal{R}_{\bf \Omega^+}$ and thus  $\mathcal{R}_{\bf \Omega^+}   v_1\in \mathcal{W}^{1}$.
Thus, from \eqref{eq.v_1}, one obtains that $\mathcal{R}_{\bf \Omega^+}   v_1$ and $\tau v_1$ satisfies the same boundary value problem which  admits a unique solution in $\mathcal{W}^{1}$. Hence, one has $\mathcal{R}_{\bf \Omega^+}   v_1=\tau v_1 \mbox{ on } {\bf  \cell^+}$ and therefore  $\mathcal{R}_{\bf \Omega^+}\Phi_1^{(1)}=\tau \Phi_1^{(1)} \mbox{ on } {\bf \cell^+}$. One concludes finally with \eqref{eq.rotationOmegam} that $\Phi_1^{(1)}\in  L^2_{\bK,\tau}$.

To sum up, $\Phi_1^{(1)}\in  L^2_{\bK,\tau}$ is the unique solution of \eqref{eq.transmissionfirstorder} in $H^1_{\bK}$  orthogonal to $P_{n,\bK}^A$ and $P_{n,\bK}^B$ in $L^2_{\bK}$, with the additional regularity $\Phi^{(1)}_1\in H^1_{\bK,\Delta}( {\bf \cell^{\pm}} )$ inherited from the equation \eqref{eq.transmissionfirstorder}.

\subsection{Determination of  $\Phi_1^{(m)}\in L^2_{\bK,\tau}\cap H^1_{\bK}$ and $\lambda_D^{(m)}$, for  $m>1$}
We now generalize our construction at order $\aspar^{-1}$ to all orders $\aspar^{-m}$, for $m> 1$.
Identifying in \eqref{eq.transmission1} the order $\aspar^{-m}$  terms in $\cell^{+}$  and the order $g^{-m+1}$ terms in $\cell^{-}$ leads to:
\begin{equation}\label{eq.transmissionordern}
\left\{ \begin{array}{ll}
\displaystyle  -  \Delta \Phi_1^{(m)}= \sum_{p=0}^{m} \lambda_D^{(m-p)} \Phi_1^{(p)}& \mbox{ in }  \cell^{+}, \\[5pt]
\Phi_1^{(m),-}=\Phi_1^{(m),+} \mbox{ and }  \displaystyle \frac{\partial \Phi_1^{(m),-}}{\partial \bn}=  \frac{\partial \Phi_1^{(m-1),+}  }{\partial \bn}  & \mbox{ on } \partial{\cell}^{+} ,\\[10pt]
- \Delta \Phi_1^{(m)} = \displaystyle  \sum_{p=0}^{m-1} \lambda_D^{(m-1-p)} \Phi_1^{(p)}   & \mbox{ in }  \cell^{-},\\[15pt]
\Phi_1^{(m)} \mbox{ and } \displaystyle  \frac{\partial \Phi_1^{(m)}}{\partial \bn}&  \hspace{-3.5cm}  \mbox{$\bK-$quasi-periodic} \mbox{ on } \partial{\cell}.
\end{array} \right.
\end{equation}
The system \eqref{eq.transmissionordern} reduces to \eqref{eq.transmissionfirstorder} when $m=1$. For $m>1$,
the functions $\Phi_1^{(m)}$ and the scalars $\lambda_D^{(m)}$  for $m > 1$ are defined recursively. 
Assume  for $p=1,2,\ldots, m-1$, that $\Phi_1^{(p)}\in L^2_{\bK,\tau}\cap H^1_{\bK}$ and $\lambda_D^{(p)}\in\mathbb{C}$ are defined to uniquely solve \eqref{eq.transmissionordern} in $H^1_{\bK}\times \mathbb{C} $ and such that $\Phi_1^{(p)}$  is   $L^2_{\bK}-$ orthogonal to $P_{n,\bK}^A$ and $P_{n,\bK}^B$, and such  that $\Phi_1^{(p)}\in L^2_{\bK,\tau}$, with the additional regularity $\Phi^{(p)}_1\in H^1_{\bK,\Delta}( {\bf \cell^{\pm}} )$ inherited from the equations \eqref{eq.transmissionordern}. 
We proceed to construct $ \lambda_D^{(m)}$ and $\Phi_1^{(m)}$ satisfying these same properties. 

Consider \eqref{eq.transmissionordern} on $\cell^-$, which states
\begin{equation*}
- \Delta \Phi_1^{(m)} = \displaystyle  \sum_{p=0}^{m-1} \lambda_D^{(m-1-p)} \Phi_1^{(p)} \ \   \mbox{ in }   \ \ \cell^{-} \quad  \mbox{ and } \quad \frac{\partial \Phi_1^{(m),-}}{\partial \bn} =  \frac{\partial \Phi_1^{(m-1),+}  }{\partial \bn} \ \   \mbox{ on }  \partial\cell^+,
\end{equation*}
where $\Phi_1^{(m)}$  and $ \partial \Phi_1^{(m)}/\partial \bn$ are required to satisfy $\bK-$quasi-periodic boundary conditions on $\partial{\cell}$.
As in our analysis for $m=1$,  by the Lax-Milgram theorem, this problem admits a unique solution in $H^{1}_{\bK}({\bf \cell^{-}})$ (that is also in $H^{1}_{\bK,\Delta}({\bf \cell^{-}})$). Furthermore,  as for $m=1$, by applying Lemma \ref{lem.comquasiper2} with the fact that $\sum_{p=0}^{m-1} \lambda_D^{(m-1-p)} \Phi_1^{(p)}\in L^{2}_{\bK,\tau} $ (since $\Phi_1^{(p)} \in L^{2}_{\bK,\tau}$ for $m=0,\ldots, m-1$) leads to $\mathcal{R}_{{\bf \cell^{-}}}\Phi_1^{(m)}=\tau \Phi_1^{(m)}$ in ${\bf \cell^-}$. Therefore, one deduces that $[\mathcal{R}_{{\bf \cell^{-}}}\Phi_1^{(m)}]^-=\tau \Phi_1^{(m),-} \mbox{ on } {\partial \bf \cell^+}$.

Turning to $\Phi_1^{(m)}$ on $\cell^+$, we seek $\lambda_D^{(m)}\in\C$, and  $\Phi_1^{(m)}\in H^{1}_{\bK}({\bf \cell^+})$ such that
\begin{equation}\label{eq.compatcondi2}
 (- \Delta-\lambda^{(0)}_D) \, \Phi_1^{(m)} = \sum_{p=0}^{m-1} \lambda_D^{(m-p)} \Phi_1^{(p)}\ \ \textrm{in}\ \    \cell^{+}  \quad \mbox{ and } \quad \Phi_1^{(m),+} =\Phi_1^{(m),-}   \mbox{ on }   \partial{\cell}^{+} .
\end{equation}
For the construction we again use the decomposition: $\Phi^{(m)}=u_m+v_m$, where $u_m\in \mathcal{W}_{1} $ is an extension to ${\bf\cell^+}$ of $\Phi_1^{(m),-}\in H^{\frac{1}{2}}_{\bK}( {\bf \partial \cell^{+}})$. The function $u_m\in  \mathcal{W}_{1}$ is constructed in a manner analogous to that for orders $\aspar^{-p}$, $p=1,\dots,m-1$. Thus, the resulting $u_m$ satisfies 
the symmetry and regularity properties: $ \mathcal{R}_{\bf \cell^+ }u_m=\tau u_m$ (with $ u_m$ also in  $H^{1}_{\bK,\Delta}({\bf \cell^{+}})$). 

Setting $\Phi^{(m)}=u_m+v_m$ in \eqref{eq.compatcondi2} leads to 
\begin{equation}\label{eq.compatcondi3}
 (- \Delta-\lambda^{(0)}_D) \, v_{m}=(\Delta+\lambda^{(0)}_D) \, u_{m} +  \lambda^{(m)}_D \, \Phi_{1}^{(0)}+  \sum_{p=1}^{m-1} \lambda_D^{(m-p)} \Phi_1^{(p)}  \mbox{ in }  \cell^{+}  \ \mbox{ and }  \ v_m=0 \mbox{ on }   \partial{\cell}^{+} .
\end{equation}

As for the case $m=1$, we determine  $\lambda_D^{(m)}$ from  the compatibility conditions for solvability for 
  $\Phi^{(m)}\in \mathcal{W}_1$, i.e. that $\Phi^{(m)}$  is orthogonal to $P_{n,\bK}^A$ and $P_{n,\bK}^B$. By  induction hypothesis,  $\Phi_1^{(p)}$ is orthogonal to $P_{n,\bK}^A$ for $p=1,\ldots, m-1$ and $\Phi_1^{(0)}=P_{n,\bK}^A$. Therefore,  using \eqref{eq.compatcondi3}) we have that
$$
\lambda_D^{(m)}= \int_{\cell^{+}}\big(-  \Delta - \lambda_D^{(0)} \big) u_{m}\,  \overline{P_{n,\bK}^A } \, \rmd \bx.
$$
Applying Green's identity  yields:
\begin{equation*}
 \lambda_D^{(m)}=\big\langle  \overline{ \frac{\partial P_{n,\bK}^{A,+}}{\partial \bn} } ,  \Phi_1^{(m),-}  \big \rangle_{ H^{-1/2}_{\bK},H^{1/2}_{\bK}}  .
\end{equation*}
Continuing, as in the case  $m=1$, by Green's identity applied in $\cell^-$ and  equation \eqref{eq.phi1D-}, one obtains:
\begin{equation}\label{eq.deflambdan}
 \lambda_D^{(m)}=-\int_{\cell^{+}}  \nabla \Phi_1^{(m)} \ \overline{  \nabla \Phi_1^{(1)}  } \,  \rmd \bx.
\end{equation}
Together  with \eqref{eq.deflmabda1}, for the case $m=1$, we have \eqref{eq.deflambdan} for $m\geq 1$.

The other compatibility condition,  orthogonality of the right hand side of \eqref{eq.compatcondi3} to $P_{n,\bK}^B$, follows using that $\Phi_1^{(m)}$ is orthogonal in $L_{\bK}^2({\bf \cell ^+})$ to $P_{n,\bK}^B$ for $p=0,\ldots, m-1$ and then by reproducing  the reasoning done for $m=1$, using  symmetry relation $\mathcal{R} \Phi_1^{(m),-}= \tau  \Phi_1^{(m),-} \mbox{ on } \partial \bf \cell^+$ and Lemmas \ref{lem.comquasiper} and \ref{lem.comquasiper2}. 
Hence, for  $\lambda_D^{(m)}$ given by  \eqref{eq.deflambdan}, equation \eqref{eq.compatcondi3} admits a unique solution $v_m \in \mathcal{W}_1$ (and also in  $H^{1}_{\bK,\Delta}({\bf \cell^{+}})$). Therefore,  $\Phi^{(m)}_1=u_m+v_m$  is the unique solution of \eqref{eq.compatcondi2} in $\mathcal{W}_{1} $ (with the additional regularity $\Phi^{(m)}_1\in H^{1}_{\bK,\Delta}({\bf \cell^{+}})$).
Furthermore, using  that   $\Phi_1^{(p)}\in L^2_{\bK,\tau}$ for $p=0,\ldots,m-1$,  one proves easily by mimicking the  reasoning done for $m=1$  that $\mathcal{R} \Phi_1^{(m)}= \tau  \Phi_1^{(m)} \mbox{ in } {\bf \cell^+}$. 

Thus, one concludes that  $\Phi_1^{(m)}\in L^{2}_{\bK,\tau} $ and that  with $\lambda_D^{(m)}\in \mathbb{C}$ as defined, 
 $\Phi_1^{(m)}$ is the unique solution of \eqref{eq.transmissionordern} in $H^{1}_{\bK}$ such that $\Phi_1^{(m)}$ is orthogonal to $P_{n,\bK}^A$ and $P_{n,\bK}^B$ (with the additional regularity $\Phi^{(m)}_1\in H^1_{\bK,\Delta}( {\bf \cell^{\pm}} )$ inherited from the equation \eqref{eq.transmissionordern}). 
 
\subsection{Asymptotic expansions}\label{sec.quasimode}
In the previous section we developed a formal procedure for computing  approximate $L^2_{\bK,\tau}$ eigenpairs  of $\bbA_{\aspar,\bK}$ to any order in $\aspar^{-1}$.  Such approximations are often called {\it quasi-modes}. Our goal in this section is to prove that these approximate eigenpairs approximate genuine eigenpairs of
$\bbA_{\aspar,\bK}$. To show this we use general  principles of self-adjoint operators described in Appendix \ref{app.Appendixquasimode}.
\begin{remark}\label{quasi}
Since the operator  $\bbA_{\aspar,\bK}$ has discontinuous coefficients, its domain $\rmD(\bbA_{\aspar,\bK})$  depends on the asymptotic parameter $\aspar$. Therefore, we use here a weak formulation of the quasi-modes approach  outlined in Appendix \ref{app.Appendixquasimode}, which permits the extension of  the notion of quasi-mode to functions with less regularity (in particular, functions that do not belong to $\rmD(\bbA_{\aspar,\bK})$ but  belong rather  to  $\rmD(\bbA_{\aspar,\bK}^{1/2})=H^1_{\bK}$; the latter is independent of $\aspar$). Furthermore, this approach yields an asymptotic expansion of the Bloch eigenfunctions $\Phi_j(\aspar,\cdot)$ in a norm which is stronger than the $H^1$-norm which, in particular, allows us to obtain  an asymptotic expansion of the Dirac velocity $v_D(\aspar)$. Note that the expression for $v_D(\aspar)$, see \eqref{eq.fermyveloc}, depends both on  $\Phi_j(\aspar,\cdot)$ and  $\nabla\Phi_j(\aspar,\cdot)$. Results related to this weak formulation of quasi-modes  expansions are summarized in Appendix  \ref{app.Appendixquasimode}.
\end{remark}

Introduce the inner product defined on $ H^{1}_{\bK}$ by 
$$
(u,v)_{a_{\aspar,\bK}}=a_{\aspar,\bK}(u,v), \quad \textrm{for all}\  u,v \in H^{1}_{\bK}.
$$
Here, $a_{\aspar,\bK}$ is the  sesqulinear form  defined in  \eqref{eq.sesquilinearformakg}.
 We denote  by $\|\cdot \|_{ a_{\aspar,\bK}}$ the norm  associated to  $(\cdot,\cdot)_{a_{\aspar,\bK}}$. By a Poincar\'{e} type inequality, this  norm dominates  the  norm $\| \cdot \|_{H^{1}_{\bK}}$ with a constant $C$ independent of  $\aspar$ for $\aspar\geq 1$.

\noindent For $M \geq 0$, introduce $\lambda_D^M(\aspar)$ and $\Phi_1^{M}(\aspar)$, the formal approximations
 of the previous section:
\begin{equation}\label{eq.defquasimode}
\lambda_D^M(\aspar)=\sum_{m=0}^M   \lambda_D^{(m)}  \, \aspar^{-m} \ \mbox{ and } \ \Phi_{1}^M(\aspar,\cdot)=\sum_{m=0}^M   \Phi_1^{(m)}    \aspar^{-m}.
\end{equation}
Here, $\lambda_D^{(0)}=\tilde{\delta}_n$, $\Phi_1^{(0)}=P_{n,\bK}^A$ and for $m\geq 1$, $\lambda_D^{(m)}$ is defined  by \eqref{eq.deflambdan},  and $\Phi_1^{(m)}\in H^{1}_{\bK}\cap L^2_{\bK,\tau}$ is the unique solution of  \eqref{eq.transmissionordern} for $m\geq 1$ in $H^{1}_{\bK}$, which is orthogonal to $\ \mathrm{span}\{P_{n,\bK}^A, P_{n,\bK}^B\}$.

We shall use the following proposition to justify the asymptotic expansion  \eqref{eq.ansatz1} of the eigenvalue $\lambda_D(\aspar)= \lambda_{2n-1}(\aspar;\bK)= \lambda_{2n}(\aspar;\bK)$ for $g$ sufficiently large.

\begin{proposition}\label{prop.quasiestm}
For any $M\in \NO$, there exists $C>0$ such that for all $v\in H^{1}_{\bK}$ and  all $\aspar>1$:
\begin{equation}\label{eq.estimationquasimode}
\Big|a_{\aspar,\bK}\big(\Phi_{1}^{M+1}(\aspar,\cdot),v\big)- \lambda_D^M(\aspar) \ \big(\Phi_{1}^{M+1}(\aspar,\cdot),v\big)\Big|\leq C \|v \|_{ a_{\aspar,\bK}}  \  \aspar^{-(M+1)},
\end{equation}
where $\lambda_D^M(\aspar) $ and $\Phi_{1}^M(\aspar,\cdot)$ are defined by \eqref{eq.defquasimode}.
\end{proposition}
\begin{proof}
Let $v\in H^{1}_{\bK}$.
We begin by splitting the difference on the left hand side of \eqref{eq.estimationquasimode} into two parts, terms of order $\aspar^{-m}$ with $m\le M$  and terms of order $\aspar^{-(M+1)}$:
\begin{equation}\label{eq.decompositiona}
a_{\aspar,\bK}\big(\Phi_{1}^{M+1}(\aspar,\cdot),v\big)- \lambda_D^M(\aspar) \ \big(\Phi_{1}^{M+1}(\aspar,\cdot),v\big)=s_{\aspar}(v)+r_{\aspar}(v)
\end{equation}
with 
\begin{equation}\label{eq.defsv}
s_{\aspar}(v)=\sum_{m=0}^{M} \frac{1}{g^{m}}\int_{\cell^-} \nabla \, \Phi_1^{(m+1)}\,  \overline{\nabla v} \, \rmd \bx+\sum_{m=0}^{M} \frac{1}{g^{m}}\int_{\cell^+} \nabla \, \Phi^{(m)}\, \overline{\nabla v}  \,\rmd \bx -\sum_{m=0}^{M}\ \frac{1}{g^{m}} \,\sum_{p=0}^{m}\lambda_D^{(p)} \,  (\Phi^{(m-p)}_1,v) 
\end{equation}
and
$$
r_{\aspar}(v)=\frac{1}{g^{M+1}}\int_{\cell^+} \nabla \, \Phi_1^{(M+1)} \, \overline{ \nabla v} \, \rmd \bx-\sum_{m=M+1}^{2M+1}\ \frac{1}{g^{m}} \,\sum_{p=0}^{M}\lambda_D^{(p)} \,  (\Phi^{(m-p)}_1,v),
$$
(we used here in particular that $\Phi_1^{(0)}=P_{n,\bK}^A=0$ on $\cell^-$).
By the Cauchy-Schwarz inequality $ r_{\aspar}(v)$ satisfies the bound:
\begin{equation*}
|r_{\aspar}(v)|\leq   \aspar^{-(M+1)}  \| \nabla \Phi_1^{(M+1)} \|_{L^2(\cell^{+})} \,  \| \nabla  v \|_{L^2(\cell^{+})}+\Big(\sum_{m=M+1}^{2M+1} \aspar^{-m} \sum_{p=0}^{M} |\lambda_D^{(p)} |  \ \|\Phi_1^{(m-p)}\| \Big)    \  \|  v \|.
\end{equation*}
By a Poincar\'{e} type inequality, it follows that there exists $C>0$, independent of $\aspar\ge1$,  such that 
\begin{equation}\label{eq.ineqrv}
|r_{\aspar}(v)|\leq   C  \,  \aspar^{-(M+1)}  \|  v \|_{ a_{\aspar,\bK}} .
\end{equation}

To complete the proof, we  claim that $s_{\aspar}(v)=0$ by the definitions of  $\lambda_D^{(m)}$ and $\Phi_1^{(m)}$. 
Using the Green identity in $\cell^{\pm}$ and the fact that $\pm \bn$ is the outward normal of $\cell^{\pm}$ leads to
$$
 \sum_{m=0}^{M} \frac{1}{g^{m}}\int_{\cell^ -} \nabla \, \Phi_1^{(m+1)}\,  \overline{ \nabla v} \, \rmd \bx=\sum_{m=0}^{M} \frac{1}{g^{m}}\Big[\int_{\cell^-}- \Delta \, \Phi_1^{(m+1)}\,  \overline{  v} \, \rmd \bx-\big\langle\frac{\partial \Phi_1^{(m+1),-}}{\partial \bn}  , \overline{  v} \big \rangle_{ H^{-1/2}_{\bK},H^{1/2}_{\bK}}\Big]
$$
and
$$
\sum_{m=0}^{M} \frac{1}{g^{m}}\int_{\cell^+} \nabla \, \Phi_1^{(m)}\,  \overline{ \nabla v} \, \rmd \bx=\sum_{m=0}^{M} \frac{1}{g^{m}}\Big[\int_{\cell^+}- \Delta \, \Phi_1^{(m)}\,  \overline{ v} \, \rmd \bx+
\big\langle  \frac{\partial \Phi_1^{(m),+}}{\partial \bn}  , \overline{  v} \big \rangle_{ H^{-1/2}_{\bK},H^{1/2}_{\bK}}\Big].
$$
Adding the last expressions and using that $ \partial \Phi_1^{(m+1),-}/\partial \bn= \partial \Phi_1^{(m),+}/\partial \bn$ on $\partial \cell^+$ (by  \eqref{eq.transmissionfirstorder} and \eqref{eq.transmissionordern}), we have that the duality products on $\partial \cell^+$ cancel  and $s_{\aspar}$ defined by \eqref{eq.defsv} can be rewritten as
\begin{equation*}\label{eq.defsvzero}
  s_{\aspar}(v)=  \sum_{m=0}^{M} \frac{1}{g^{m}} \Big[ \int_{\cell^-}\big(- \Delta \, \Phi_1^{(m+1)} 
 -\sum_{p=0}^{m}\lambda_D^{(p)} \,  \Phi_1^{(m-p)}\big)\,  \overline{  v} \, \rmd \bx+\int_{\cell^+}\big(- \Delta \, \Phi_1^{(m)} -\sum_{p=0}^{m}\lambda_D^{(p)} \,  \Phi_1^{(m-p)}\big)\,  \overline{  v} \, \rmd \bx \Big] .
\end{equation*}
Hence,  $s_{\aspar}(v)=0$, since $\Phi_1^{(0)}$ and $\Phi_1^{(m)}$ for $m\geq  1$ satisfy respectively  \eqref{eq.transmissionzeroorder} and  \eqref{eq.transmissionordern}  in $\cell^{\pm}$.
Therefore, with \eqref{eq.decompositiona} and \eqref{eq.ineqrv}, one obtains immediately the inequality \eqref{eq.estimationquasimode}.
\end{proof}

Using the preceding Lemma, we can now prove the asymptotic expansion  \eqref{eq.ansatz1}.
\begin{theorem}\label{th.asympteigen}(Asymptotic expansion of the Dirac eigenvalue, $\lambda_D(\aspar)$, for large $\aspar$)\\
Assume $\tilde{\delta}_n$ satisfies condition $\condS$ of Definition \ref{Def.condS} for a fixed $n\geq 1$. Then,  for $M\ge0$ and $\aspar$ large enough, the {\it Dirac eigenvalues}: $\lambda_{2n-1}(\aspar;\bK)= \lambda_{2n}(\aspar;\bK)=\lambda_D(\aspar)$ admit the asymptotic expansion 
\begin{equation}\label{eq.asympeigenval}
 \lambda_D(\aspar)= \lambda_D^M(\aspar)+O\big(g^{-(M+1)} \big)=\sum_{m=0}^M \lambda_D^{(m)} \aspar^{-m} +O\big(g^{-(M+1)} \big) ,
\end{equation}
where $\lambda_D^M(\aspar)$ is defined by the $M-$ term expansion \eqref{eq.defquasimode}.
\end{theorem}
\begin{remark}
A consequence of \eqref{eq.asympeigenval} is that $\lambda_D^{(m)}$ and $\lambda_D^M(\aspar)$ are real for all $m, M\in \NO$ and  $\aspar>0$.
\end{remark}
 
\begin{proof}
 Let $M\in \NO$. First,  by Theorem \ref{th.degeneracy} we have for $\aspar$ sufficiently large that
  $\lambda_{D}(\aspar)=\lambda_{2n-1}(\aspar;\bK)=\lambda_{2n}(\aspar;\bK)$, and  is of multiplicity $2$.  
We will  show \eqref{eq.asympeigenval}, by applying the Corollary \ref{cor.eigvalest} to the approximate eigenpair (quasi-mode): $\lambda= \lambda_D^M(\aspar)$ and $u=\Phi_{1}^{M+1}(\aspar,\cdot)$ defined  in \eqref{eq.defquasimode},
using the bound \eqref{eq.estimationquasimode} of Proposition \ref{prop.quasiestm}. 

We next verify the assumptions of Corollary \ref{cor.eigvalest}. By  \eqref{eq.defquasimode}, we have
 $\lambda_D^M(\aspar)\to \tilde{\delta}_{n}>0$ as $\aspar \to +\infty $ which implies that $\operatorname{Re}(\lambda_D^M(\aspar))>0$ for $\aspar$ large enough. Moreover,  as $\Phi^{(0)}_1=P_{n,\bK}^A$ vanishes on $\cell^-$ and $P_{n,\bK}^A  \neq 0$, one gets
\begin{equation}\label{eq.termnormphiNlim}
\|\Phi_{1}^{M+1}(\aspar, \cdot)\|_{a_{\aspar,\bK}} \to \|\nabla  P_{n,\bK}^A\| _{L^2 (\cell^+)}>0, \mbox{ as  } \aspar \to + \infty.
\end{equation}
Thus, it follows that
$$(|\lambda_D^M(\aspar)|+1)^{-1}  \|\Phi_{1}^{M+1}(\aspar, \cdot)\|_{a_{\aspar,\bK}}  \to (\tilde{\delta}_{n}+1)^{-1} \|\nabla  P_{n,\bK}^A\|_{L^2 (\Omega^+)} >0.$$ 
We deduce that for $\aspar$ large enough, $C/\aspar^{M+1} < (|\lambda_D^M(\aspar)|+1)^{-1}  \|\Phi_{1}^{M+1}(\aspar, \cdot)\|_{a_{\aspar,\bK}}  $, wherer $C>0$ (independent of $\aspar$) is the constant in the estimate \eqref{eq.estimationquasimode}. Therefore, the bound  \eqref{eq.estimationquasimode} and Corollary \ref{cor.eigvalest} imply that for $\aspar$ large enough, there exists $\tilde{\lambda}(\aspar,\bK)\in \sigma(\bbA(\aspar,\bK))$
such that:
$$
|\lambda_D^M(\aspar)-\tilde{\lambda}(\aspar,\bK)|\leq  \frac{C}{\|\Phi_{1}^{M+1}(\aspar, \cdot)\|_{a_{\aspar,\bK}} \  \aspar^{M+1}} \ (|\lambda_D^M(\aspar)|+1).
$$
Therefore, there is a constant $\widetilde{C}>0$, independent  of $\aspar$, such that for $\aspar$ large enough:
\begin{equation}\label{eq.estimateeigen}
|\lambda_D^M(\aspar,\bK)-\tilde{\lambda}(\aspar,\bK)|\leq   \widetilde{C} \aspar^{-(M+1)}.
\end{equation}
Since  $\lambda_D^M(\aspar,\bK) \to \tilde{\delta}_n=\delta_{2n}$ as $\aspar\to +\infty$, it follows form \eqref{eq.estimateeigen} that $\tilde{\lambda}(\aspar,\bK) \to \delta_{2n}$. Moreover, by Theorem \ref{thmsaympAtau} and  Proposition \ref{lem.eigenDirichlettwoinclusions}, one has $\lambda_{2n-1}(\aspar;\bK)=\lambda_{2n}(\aspar;\bK)=\lambda_{D}(\aspar)\to \delta_{2n}$,
$\lambda_{2n+1}(\aspar;\bK)\to \delta_{2n+1}>\delta_{2n}$  if $n\geq 1$ and if $n>1$, $\lambda_{2n-2}(\aspar;\bK)\to \delta_{2n-2}<\delta_{2n}$ for $\aspar\to +\infty$. Thus, one has necessarily  $\tilde{\lambda}(\aspar,\bK)=\lambda_D(\aspar)$ for  all $\aspar $ sufficiently large.

Finally, we show that $\lambda_D^{(n)}$ and $\lambda_D^N(\aspar,\bK)$ are real-valued.  Indeed, we know that
$\lambda_D^{0}(\aspar)=\lambda^{(0)}_D=\tilde{ \delta}_n\in \mathbb{R}$ and also from \eqref{eq.deflmabda1} that $\lambda_D^{(1)}<0$. Therefore,  $\lambda_{D}^1(\aspar)=\lambda^{(0)}_D+\aspar^{-1} \lambda_D^{(1)}\in \mathbb{R}$ for all $\aspar>0$. Unfortunately,  that $\lambda^{(m)}_D$ and $\lambda_D^{M}(\aspar)$ are real-valued is not easily deduced from formula  \eqref{eq.deflambdan} for $m, M>1$. However,  we can  straightforwardly verify this from \eqref{eq.asympeigenval} by induction. 
Indeed, assume that for all $0\leq m\leq M$, $\lambda^{(m)}_D$ and $\lambda_D^{m}(\aspar)$ are real-valued for any $\aspar>0$.
Using the relation \eqref{eq.asympeigenval} at the order $M+1$ leads to
$$
\lim_{\aspar \to \infty} \aspar^{M+1} \big( \lambda_D(\aspar) - \lambda_D^{M}(\aspar) \big)= \lambda^{(M+1)}_D 
$$
where $ \lambda_D(\aspar)\in \sigma(\bbA_{\aspar,\bK})\subset \R$ and $\lambda_D^{M}(\aspar)$ is real-valued by induction. Hence,  taking the latter limit, one deduced  that  $\lambda^{(M+1)}_D\in \R$.  Thus $\lambda^{M+1}_D(\aspar)=\lambda_D^{M}(\aspar)+\lambda^{(M+1)}/\aspar^{M+1}$ is also real-valued for any $\aspar>0$.
\end{proof}

We next address bounds on the truncation error for our asymptotic expansions of the eigenfunctions.
\begin{theorem}\label{th.eigenaprox}(Asymptotic expansion of the eigenfunction  in the $\|\cdot \|_{a_{\aspar}}$ norm)\\
Assume that  $\tilde{\delta}_n$ satisfies the condition $\condS$ for a fixed $n\geq 1$ and let $M\in \NO$. Then, for $\aspar$ large enough, there exists an eigenfunction $\Psi^{(1)}(\aspar,\cdot)\in \operatorname{ker}(\bbA_{\aspar,\bK}-\lambda_D(\aspar) \mathrm{I}_d)$ (with $\lambda_D(\aspar) =\lambda_{2n}(\aspar;\bK)=\lambda_{2n-1}(\aspar;\bK)$)  which satisfies $\|\Psi^{(1)}(\aspar,\cdot)  \|=1$ (where $\| \cdot\|$ is the $L^2_{\bK}$-norm) and $C>0$ such that:
\begin{equation}\label{eq.asympeigenfunc}
\Big\| \Psi^{(1)}(\aspar,\cdot)- \frac{ \Phi_{1}^{M+1}(\aspar,\cdot)}{\| \Phi_{1}^{M+1}(\aspar,\cdot)\|} \Big\|_{a_{\aspar,\bK}} \leq \frac{C}{\aspar^{M+1}}.
\end{equation}
\end{theorem}
\begin{proof}To prove this result,
we apply the Corollary \ref{cor.eigenaprox} to $u_{\aspar}=\Phi_{1}^{M+1}(\aspar,\cdot)/\| \Phi_{1}^{M+1}(\aspar,\cdot)\|$.
We first  establish an estimate of the form \eqref{eq.quasimodeweakestimateeigenfunc1} for  $u_\aspar$ and the eigenvalue  $\lambda_D(\aspar)$ of $\bbA_{\aspar,\bK}$. 
For all $v\in H^{1}_{\bK}$,  one has
\begin{eqnarray*}
&&  \Big|a_{\aspar,\bK}\big(u_{\aspar},v\big)- \lambda_D(\aspar) \ \big(u_{\aspar},v\big)\Big|  \\[4pt]
 && \leq  \| \Phi_{1}^{M+1}(\aspar,\cdot)\|^{-1} \big( |a_{\aspar,\bK}\big(\Phi_{1}^{M+1}(\aspar,\cdot),v\big)-\  \lambda_D^M(\aspar) \, (\Phi_{1}^{M+1}(\aspar,\cdot),v)| \big) \\ [4pt]
  && \ \ +   |   \lambda_D(\aspar)- \lambda_D^M(\aspar)|  \ \|v \|_{ a_{\aspar,\bK}}.
\end{eqnarray*}
By  estimate  \eqref{eq.estimationquasimode}, the asymptotic relation \eqref{eq.asympeigenval} and the fact that $\|\Phi_{1}^{M+1}(\aspar,\cdot)\|  \to \| P_{n,\bK}^A \|=1 $ as $\aspar\to +\infty$,  there exists $C_1>0$ such that for all $\aspar$ sufficiently large:
$$
\Big|a_{\aspar,\bK}\big(u_{\aspar},v\big)- \lambda_D(\aspar) \ \big( u_{\aspar},v\big)\Big| \leq C_1 \|v \|_{ a_{\aspar,\bK}}  \,  \aspar^{-(M+1)}, \quad  \ \textrm{for all}\ \  v\in H^{1}_{\bK}.
$$
This is precisely an estimate of the form \eqref{eq.quasimodeweakestimateeigenfunc1}. 

We next need require, for large $\aspar$, a lower bound on $ \operatorname{dis}(\lambda_D(\aspar)^{-1},\sigma(\bbA_{\aspar,\bK}^{-1})\setminus\{\lambda_D(\aspar)^{-1}\})$; see  Corollary \ref{cor.eigenaprox} . By Theorem \ref{th.degeneracy},   for $\aspar$ large: $\lambda_{2n-1}(\aspar;\bK)=\lambda_{2n}(\aspar;\bK)=\lambda_D(\aspar)$ is of multiplicity $2$.  Thus,
\begin{equation} \label{eq.defdistancespec}
\hspace{-0.1cm} \operatorname{dis}(\lambda_D(\aspar)^{-1},\sigma(\bbA_{\aspar,\bK}^{-1})\setminus\{\lambda_D(\aspar)^{-1}\})=\ 
  \hspace{-0.1cm} \begin{cases} 
  &  \hspace{-0.35cm} \lambda_D(\aspar)^{-1}-  \lambda_{3}(\aspar;\bK)^{-1}  \textrm{ if $n=1$}\\[5pt]
&  \hspace{-0.35cm}  \min \big( \lambda_{2n-2}(\aspar;\bK)^{-1}- \lambda_D(\aspar)^{-1}, \lambda_D(\aspar)^{-1}-\lambda_{2n+1}(\aspar;\bK)^{-1} \big) \textrm{ if $n>1$}.
 \end{cases}
\end{equation}
We deduce, using Theorem \ref{thmsaympAtau}, \eqref{eq.termnormphiNlim}, $\lambda_D(\aspar)\to \delta_{2n}$  and $\|\Phi_{1}^{M+1}(\aspar,\cdot)\|  \to 1$, that
$$
b(\aspar)=1/2 \, \lambda_D(\aspar) \, \|u_{\aspar}\|_{a_{\aspar,\bK}} \operatorname{dis}(\lambda_D(\aspar)^{-1},\sigma(\bbA_{\aspar,\bK}^{-1})\setminus\{\lambda_D(\aspar)^{-1}\}) \to C_2>0 , \mbox{ as } \aspar \to +\infty,
$$
with 
\begin{equation*}
 C_2=\ 
\begin{cases} 
  & \displaystyle  \frac{ \delta_2}{2} \, \|\nabla  P_{n,\bK}^A\| _{L^2 (\cell^+)} (\delta_{2}^{-1}- \delta_{3}^{-1})  \mbox{ for } n=1 \\[5pt]
&    \displaystyle \frac{ \delta_{2n}}{2}  \, \|\nabla  P_{n,\bK}^A\| _{L^2 (\cell^+)} \min( \delta_{2n-2}^{-1}- \delta_{2n}^{-1}, \delta_{2n}^{-1}- \delta_{2n+1}^{-1} ) \ \mbox{ for } n>1.
 \end{cases}
\end{equation*}
Therefore, for $\aspar$ sufficiently large, $0<C_1\, \aspar^{-(M+1)}<\min( b(\aspar),1) $. It follows using Corollary \ref{cor.eigenaprox} that for $\aspar$ large, there exists $\Psi^{(1)}(\aspar,\cdot) \in L^2_{\bK}$ an eigenfunction of $\bbA_{\aspar,\bK}$ associated to the eigenvalue $\lambda_D(\aspar)$ which satisfies  and $\|\Psi^{(1)}(\aspar,\cdot)  \|=1$,
and  such that
\begin{equation} \label{eq.defdconstanteconv}
\Big\| \Psi^{(1)}(\aspar,\cdot)- \frac{ \Phi_{1}^{M+1}(\aspar,\cdot)}{\| \Phi_{1}^{M+1}(\aspar,\cdot)\|} \Big\|_{a_{\aspar,\bK}} \leq  \frac{\widetilde{C}(\aspar) \, C_1}{\aspar^{M+1}} 
\end{equation}
where
$$
 \widetilde{C}(\aspar)=\widetilde{C}_1(\aspar) +\lambda_{D}(\aspar)^{-\frac{1}{2}}+ \lambda_{D}(\aspar)^{\frac{1}{2}}   \lambda_1(\aspar;\bK)^{-\frac{1}{2}},
\mbox{ with } \widetilde{C}_1(\aspar)=\frac{4 \, \lambda_D(\aspar)^{-1}}{  \operatorname{dis}(\lambda_D(\aspar)^{-1},\sigma(\bbA_{\aspar,\bK}^{-1}\setminus\{\lambda_D(\aspar)^{-1}\})}.$$
The expression of $\widetilde{C}(\aspar)$ comes from \eqref{eq.constante}. Using   \eqref{eq.defdistancespec},  \eqref{eq.defdconstanteconv} and Theorem \ref{thmsaympAtau}, one gets that  $\widetilde{C}(\aspar)\to C_3>0$, as $\aspar\to+\infty$ where the constant $C_3$  can be easily made explicit.
We conclude with \eqref{eq.defdconstanteconv} that for $\aspar$ large enough, there exists $C_4>0$ (independent of $\aspar$) such that
$$
\Big\| \Psi^{(1)}(\aspar,\cdot)- \frac{ \Phi_{1}^{M+1}(\aspar,\cdot)}{\| \Phi_{1}^{M+1}(\aspar,\cdot)\|} \Big\|_{a_{\aspar,\bK}} \leq  \frac{ C_4}{\aspar^{M+1}}.
$$
\end{proof}
In the following corollary we construct (more explicitly  than  in Theorem \ref{th.eigenaprox}) an eigenfunction of $\operatorname{ker}(\bbA_{\aspar,\bK}-\lambda_D(\aspar) \mathrm{I}_d)$
that is approximated by $\Phi_{1}^{M+1}(\aspar,\cdot)/\| \Phi_{1}^{M+1}(\aspar,\cdot)\|$.
\begin{corollary}\label{cor.defphi1}
Assume that  $\tilde{\delta}_n$ satisfies the condition $\condS$ for a fixed $n\geq 1$ and let $M\in \NO$. Then,  for  $\aspar $ sufficiently large, $\{\Phi_1(\aspar, \cdot),\Phi_2(\aspar, \cdot)\}$,  defined by 
\begin{equation}\label{eq.defPhi1proj}
\Phi_1(\aspar, \cdot)=\frac{\mathbb{E}_{\bbA_{\aspar,\bK}}(\{\lambda_D(\aspar)\})\,\Phi_{1}^{M+1}(\aspar,\cdot)}{\|\mathbb{E}_{\bbA_{\aspar,\bK}}(\{\lambda_D(\aspar)\})\, \Phi_{1}^{M+1}(\aspar,\cdot)\|} \in L^2_{\bK,\tau} \ \mbox{ and  }  \ \Phi_2(\aspar, \cdot)=\mathcal{P} \mathcal{C}\Phi_1(\aspar, \cdot)\in L^2_{\bK,\overline{\tau}}\ ,
\end{equation}
is an orthonormal basis of $\operatorname{ker}(\bbA_{\aspar,\bK}-\lambda_D(\aspar) \mathrm{I}_d)$  with $\lambda_D(\aspar) =\lambda_{2n}(\aspar;\bK)=\lambda_{2n-1}(\aspar;\bK)$
and  there exists $C>0$ such that:
\begin{equation}\label{eq.asympeigenfunc2}
\Big\|\Phi_1(\aspar, \cdot) - \frac{ \Phi_{1}^{M+1}(\aspar,\cdot)}{\| \Phi_{1}^{M+1}(\aspar,\cdot)\|} \Big\|_{a_{\aspar,\bK}} \leq \frac{C}{\aspar^{M+1}} \ \mbox{ and }  \ \Big\|\Phi_2(\aspar, \cdot) -\frac{  \mathcal{P} \mathcal{C} \, \Phi_{1}^{M+1}(\aspar,\cdot)}{\| \Phi_{1}^{M+1}(\aspar,\cdot)\|}  \Big\|_{a_{\aspar,\bK}} \leq \frac{C}{\aspar^{M+1}}.
\end{equation}
\end{corollary}
\begin{proof}
We first prove that \eqref{eq.defPhi1proj} is well-defined for $\aspar$ large enough.
By Theorem \ref{th.degeneracy},  $\lambda_{2n}(\aspar;\bK)=\lambda_{2n-1}(\aspar;\bK)=\lambda_D(\aspar)$ is an eigenvalue of multiplicity $2$ of $\bbA_{\aspar,\bK}$.
Thus, from steps 1 and 2 of the proof of Theorem 
\ref{th.degeneracy}, we know that for any fixed $\eta$ satisfying $0<\eta<\min(\delta_{2},\nu_{3}-\delta_{2})$ if $n=1$ or $0<\eta<\min(\delta_{2n}-\delta_{2n-2},\nu_{2n+1}-\delta_{2n})$ if $n>1$, the relations \eqref{eq.speclocl} and  \eqref{eq.specprojequality} hold (for $\aspar$ large enough) with   $\{\lambda_{2n}(\aspar;\bK),\lambda_{2n-1}(\aspar;\bK)\}=\{\lambda_D(\aspar)\}$ (and $\bbA_{\aspar,\bK}$ instead of $\bbA_{\aspar_m,\bK}$) and 
 $\bbE_{\bbA_{\infty,\bK}}((\delta_{2n}-\eta,\delta_{2n}+\eta))=\bbE_{\bbA_{\infty,\bK}}(\{\delta_{2n}\}) $. Hence, by Lemma \ref{lem.convspecproj},  one has:
$$
\bbE_{\bbA_{\aspar,\bK}}((\delta_{2n}-\eta,\delta_{2n}+\eta))=\mathbb{E}_{\bbA_{\aspar,\bK}}(\{\lambda_D(\aspar)\}) \to \bbE_{\bbA_{\aspar,\bK}}(\delta_{2n}-\eta,\delta_{2n}+\eta)= \bbE_{\bbA_{\infty,\bK}}(\{\delta_{2n}\}), \ \mbox{ as } \ \aspar \to +\infty.
$$
Furthermore, as $\Phi_{1}^{M+1}(\aspar,\cdot) \to P_{n,\bK}^A$ in $L^2_{\bK}$ with $P_{n,\bK}^A\in \operatorname{ker}(\bbA_{\infty,\bK}-  \delta_{2n}\mathrm{I}_d)$, it follows that
$$
\mathbb{E}_{\bbA_{\aspar,\bK}}(\{\lambda_D(\aspar)\}) \, \Phi_{1}^{M+1}(\aspar,\cdot) \to  \bbE_{\bbA_{\infty,\bK}}(\{\delta_{2n}\}) P_{n,\bK}^A= P_{n,\bK}^A \neq 0, \quad \aspar \to +\infty.
$$
Hence, $\|\mathbb{E}_{\bbA_{\aspar,\bK}}(\{\lambda_D(\aspar)\}) \Phi_{1}^{M+1}(\aspar,\cdot)\|\neq 0$ for $\aspar$ large enough and thus $\Phi_1(\aspar, \cdot)$ is well-defined by \eqref{eq.defPhi1proj} as a normalized function of $\operatorname{ker}(\bbA_{\aspar,\bK}-\lambda_D(\aspar) \mathrm{I}d)$. Moreover, as $\mathbb{E}_{\bbA_{\aspar,\bK}}$ commutes with $\mathcal{R}$ (since $\mathbb{E}_{\bbA_{\aspar,\bK}}$ commutes with $\bbA_{\aspar,\bK}$) and $L^2_{\bK,\tau}$ is an eigenspace of $\mathcal{R}$, $L^2_{\bK,\tau}$ is stable under $\mathbb{E}_{\bbA_{\aspar,\bK}}$. Hence, as $\Phi_{1}^{M+1}(\aspar,\cdot) \in L^2_{\bK,\tau}$, one deduces that $\Phi_1(\aspar, \cdot)\in L^2_{\bK,\tau}$. Thus,  \eqref{eq.propcommutS} and  $\|\mathcal{P} \mathcal{C}\Phi_1(\aspar, \cdot)\|=\| \Phi_1(\aspar, \cdot)\|=1$  give that $\Phi_2(\aspar, \cdot)=\mathcal{P}\mathcal{C}\Phi_1(\aspar, \cdot)$ is a normalized function of $\operatorname{ker}_{\overline{\tau}}(\bbA_{\aspar,\bK}-\lambda_D(\aspar) \mathrm{I}d)$. Since $\lambda_D(\aspar)$ is of multiplicity $2$, we conclude using \eqref{eq.propcommutR} that $\{\Phi_1(\aspar, \cdot), \Phi_2(\aspar, \cdot)\}$ is an orthonormal basis of $\operatorname{ker}(\bbA_{\aspar,\bK}-\lambda_D(\aspar) \mathrm{I}_d)$.

We  now prove the estimate \eqref{eq.asympeigenfunc2}. For the remainder of the proof we use the compressed notation $\mathbb{E}_{\aspar}$  for the projection $\mathbb{E}_{\bbA_{\aspar,\bK}}(\{\lambda_D(\aspar)\})$ and $\|\cdot\|_{a_{\aspar}}$ for the norm $\|\cdot\|_{a_{\aspar,\bK}}$.
 First, one has:
\begin{equation}\label{eq.triangineqphi1}
\Big\|\Phi_1(\aspar, \cdot) - \frac{ \Phi_{1}^{M+1}(\aspar,\cdot)}{\| \Phi_{1}^{M+1}(\aspar,\cdot)\|} \Big\|_{a_{\aspar}} 
 \leq  \Big\| \Phi_1(\aspar, \cdot) -  \frac{\mathbb{E}_{\aspar} \Phi_{1}^{M+1}(\aspar,\cdot)}{\| \Phi_{1}^{M+1}(\aspar,\cdot)\|}\Big\|_{a_{\aspar}} + \Big\| \frac{  \mathbb{E}_{\aspar}  \Phi_{1}^{M+1}(\aspar,\cdot)}{\| \Phi_{1}^{M+1}(\aspar,\cdot)\|}- \frac{ \Phi_{1}^{M+1}(\aspar,\cdot)}{\| \Phi_{1}^{M+1}(\aspar,\cdot)\|}  \Big\|_{a_{\aspar}}.
\end{equation}
Concerning the first term of the right hand side of \eqref{eq.triangineqphi1}, one observes with \eqref{eq.defPhi1proj} that:
\begin{equation}\label{eq.firsttermproj}
 \Big\|\Phi_1(\aspar, \cdot) - \frac{  \mathbb{E}_{\aspar} \Phi_{1}^{M+1}(\aspar,\cdot)}{\| \Phi_{1}^{M+1}(\aspar,\cdot)\|} \Big\|_{a_{\aspar}} =\| \mathbb{E}_{\aspar} \Phi_{1}^{M+1}(\aspar,\cdot)\|_{a_{\aspar}}  \ \Big| \,  \| \mathbb{E}_{\aspar} \Phi_{1}^{M+1}(\aspar,\cdot)\|^{-1}- \| \Phi_{1}^{M+1}(\aspar,\cdot)\|^{-1} \Big| .
\end{equation}
By  \eqref{eq.termnormphiNlim}, 
\begin{equation}\label{eq.normagproj}
\|  \mathbb{E}_{\aspar} \, \Phi_{1}^{M+1}(\aspar,\cdot)\|_{a_{\aspar}} \leq \|  \Phi_{1}^{M+1}(\aspar,\cdot)\|_{a_{\aspar}} \to\| \nabla P_{n, \bK}^A \|>0, 
\end{equation}
where the inequality holds since the spectral projector $ \mathbb{E}_{\aspar}$ is an orthogonal projection on  $\operatorname{ker}(\bbA_{\aspar,\bK}-\lambda_D(\aspar) \mathrm{I}_d)$ when one considers $H^1_{\bK}$  endowed here with the Hilbert norm $\| \cdot\|_{a_{\aspar}}$ and thus in this functional framework its operator norm in $B(H^1_{\bK},H^1_{\bK})$  is $1$. This last point is easily shown by using that for any $u\in H^1_{\bK}=D(\bbA_{\aspar,\bK}^{1/2})$: $\|   u \|_{a_{\aspar}}^2=\|\bbA_{\aspar,\bK}^{1/2}  u\|^{2}$ (see relation \eqref{eq.sesquilinearformakg}) and  by decomposing $u$ via the spectral Theorem on an orthornormal basis of eigenfunctions of $ \bbA_{\aspar,\bK}$.
Moreover, one has
\begin{equation}\label{eq.ineqinterm}
\Big| \  \| \Phi_{1}^{M+1}(\aspar,\cdot)\|^{-1}-\| \mathbb{E}_{\aspar} \Phi_{1}^{M+1}(\aspar,\cdot)\|^{-1} \Big|  \leq  \|   \mathbb{E}_{\aspar} \Phi_{1}^{M+1}(\aspar,\cdot)\|^{-1}  \, \frac{ \| \Phi_{1}^{M+1}(\aspar,\cdot)-   \mathbb{E}_{\aspar} \Phi_{1}^{M+1}(\aspar,\cdot)\|}{\|\Phi_{1}^{M+1}(\aspar,\cdot) \|},
\end{equation}
with $\| \mathbb{E}_{\aspar} \Phi_{1}^{M+1}(\aspar,\cdot)\|^{-1}\to \| P_{n,\bK}^A\|^{-1}=1$  as $\aspar\to +\infty$. Now, using the Theorem \ref{th.eigenaprox}, one knows that for $\aspar$ large enough, there exists a normalized eigenfunction $\Psi^{(1)}(\aspar,\cdot)$  of $\bbA_{\aspar,\bK}$ associated to the eigenvalue $\lambda_D(\aspar)$ such that \eqref{eq.asympeigenfunc} holds. Furthermore, $\bbE_{\aspar}$ is the orthogonal projection on the subspace $\operatorname{ker}(\bbA_{\aspar,\bK}-\lambda_D(\aspar)\ \mathbb{I} \mathrm{d})$ and
 the $L^2_{\bK}$- norm is dominated by the $\small\| \cdot\small\|_{a_\aspar}$ (with a constant  independent of $\aspar$ for $\aspar\geq1$).  Thus, it follows from \eqref{eq.asympeigenfunc}   that for $\aspar$ large enough:
\begin{equation}\label{ineqL2proj}
\Big\| \frac{ \Phi_{1}^{M+1}(\aspar,\cdot)}{\| \Phi_{1}^{M+1}(\aspar,\cdot)\|}  -  \frac{ \mathbb{E}_{\aspar} \Phi_{1}^{M+1}(\aspar,\cdot)}{\| \Phi_{1}^{M+1}(\aspar,\cdot)\|} \Big\|\leq  \Big\|
\frac{\Phi_{1}^{M+1}(\aspar,\cdot)}{\| \Phi_{1}^{M+1}(\aspar,\cdot)\|} -\Psi_1(g,\cdot)\Big\|
 \leq \frac{C}{\aspar^{M+1}},
\end{equation}
for some $C>0$.
Hence, combining \eqref{eq.firsttermproj}, \eqref{eq.normagproj}, \eqref{eq.ineqinterm} and  \eqref{ineqL2proj}, one concludes that there exists $C>0$ such that for $\aspar$ large enough:
\begin{equation}\label{eq.firsttermphi1}
\Big\| \Phi_1(\aspar, \cdot) -   \frac{\mathbb{E}_{\aspar} \Phi_{1}^{M+1}(\aspar,\cdot)}{\| \Phi_{1}^{M+1}(\aspar,\cdot)\|}\Big\|_{a_{\aspar}} \leq \frac{C}{\aspar^{M+1}}.
\end{equation}
We now bound  the second term of the right hand side of \eqref{eq.triangineqphi1}. 
One uses that the  first inequality \eqref{ineqL2proj} holds with the $L^2_{\bK}$-norm, $\small\| \cdot \small\|$, replaced by the $\small\| \cdot \small\|_{a_\aspar}$ (since the spectral projector $ \mathbb{E}_{\aspar}\in B(H^1_{\bK},H^1_{\bK})$ is an orthogonal projection on  $\operatorname{ker}(\bbA_{\aspar,\bK}-\lambda_D(\aspar) \mathrm{I}_d)$ when one considers $H^1_{\bK}$  endowed with $\| \cdot\|_{a_{\aspar}}$).
Moreover, the second inequality of \eqref{ineqL2proj} holds in the norm $\small\| \cdot \small\|_{a_\aspar}$ by Theorem \ref{th.eigenaprox}. Hence,
one gets that:
\begin{equation}\label{eq.firsttermphi2}
\Big\|   \frac{  \mathbb{E}_{\aspar} \Phi_{1}^{M+1}(\aspar,\cdot)}{\| \Phi_{1}^{M+1}(\aspar,\cdot)\|}- \frac{ \Phi_{1}^{M+1}(\aspar,\cdot)}{\| \Phi_{1}^{M+1}(\aspar,\cdot)\|}  \Big\|_{a_{\aspar}} \leq \frac{C}{\aspar^{M+1}}.
\end{equation}
Combining  \eqref{eq.triangineqphi1}, \eqref{eq.firsttermphi1} and \eqref{eq.firsttermphi2} leads to  the first estimate of \eqref{eq.asympeigenfunc2}. Finally, using \eqref{eq.commut3}, one gets that $\mathcal{P}\mathcal{C}$ preserves the norm  $\|\cdot\|_{a_\aspar}$ and thus the second estimate of \eqref{eq.asympeigenfunc2} follows from the first one.  
\end{proof}
The following Corollary gives an approximation  of the eigenstates in the $H^1_{\bK}$-norm. The difference with  Corollary \ref{cor.defphi1} relies on the fact that  for $\|\cdot\|_{a_{\aspar,\bK}}-$norm, we require  a quasi-mode function of order $M+1$ to obtain a remainder of order  $M+1$. Here, the quasi-mode function  of order $M$ is sufficient. Indeed, one requires an approximation of order $M+1$ (at least in the domain $\bf {\cell^-}$, see Proposition \ref{prop.quasiestm}) for the norm  $\|\cdot\|_{a_{\aspar,\bK}}$ because of the coefficient $\sigma_{\aspar}$ that appears as a weight in this norm and induced a multiplication by $\aspar$ of the asymptotic expansion in  the domain  $\bf {\cell^-}$.

\begin{corollary}(Asymptotic expansion of the eigenfunction in the $\|\cdot \|_{H^1_{\bK}}$ norm)\label{cor.limiteigenstateH1norm}
Assume that  $\tilde{\delta}_n$ satisfies the condition $\condS$ for a fixed $n\geq 1$ and let $M\in \NO$. For $\aspar$ large enough, we define an orthonormal basis $\{\Phi_1(\aspar, \cdot), \Phi_2(\aspar, \cdot)\}$ of $\operatorname{ker}(\bbA_{\aspar,\bK}-\lambda_D(\aspar) \mathrm{I}_d)$  (with $\lambda_D(\aspar) =\lambda_{2n}(\aspar;\bK)=\lambda_{2n-1}(\aspar;\bK)$)   by:
\begin{equation}\label{eq.defPhi1proj2}
\Phi_1(\aspar, \cdot)=\frac{\mathbb{E}_{\bbA_{\aspar,\bK}}(\{\lambda_D(\aspar)\})\,\Phi_{1}^{M}(\aspar,\cdot)}{\|\mathbb{E}_{\bbA_{\aspar,\bK}}(\{\lambda_D(\aspar)\})\, \Phi_{1}^{M}(\aspar,\cdot)\|} \in L^2_{\bK,\tau} \ \mbox{ and  }  \ \Phi_2(\aspar, \cdot)=\mathcal{P} \mathcal{C}\Phi_1(\aspar, \cdot)\in L^2_{\bK,\overline{\tau}}.
\end{equation}
Then, there exists $C>0$  such that for $\aspar$ large enough:
\begin{equation}\label{eq.asympeigenfunc1}
\Big\|\Phi_1(\aspar, \cdot) - \frac{ \Phi_{1}^{M}(\aspar,\cdot)}{\| \Phi_{1}^{M}(\aspar,\cdot)\|} \Big\|_{H^1_{\bK}} \leq \frac{C}{\aspar^{M+1}} \ \mbox{ and } \ \Big\|\Phi_2(\aspar, \cdot) -\frac{  \mathcal{P}\mathcal{C} \, \Phi_{1}^{M}(\aspar,\cdot)}{\| \Phi_{1}^{M}(\aspar,\cdot)\|}  \Big\|_{H^1_{\bK}}  \leq \frac{C}{\aspar^{M+1}}.
\end{equation}
\end{corollary}

\begin{proof}
The existence of the orthonormal basis $\{\Phi_1(\aspar, \cdot), \Phi_2(\aspar, \cdot)\}$  of $\operatorname{ker}(\bbA_{\aspar,\bK}-\lambda_D(\aspar) \mathrm{I}_d)$  for $\aspar$ large enough has been proved in Corollary \ref{cor.defphi1} for $M\geq 1$ and in Theorem \ref{th.degeneracy} for $M=0$.
We proceed with the notations:  $\mathbb{E}_{\aspar}$ for $\mathbb{E}_{\bbA_{\aspar,\bK}}(\{\lambda_D(\aspar)\})$,  $\|\cdot \|_{H^1}$ for $\|\cdot \|_{H^1_{\bK}}$ and $f(\aspar)$ for a function of the form $f(\aspar,\cdot)$.

Let $\widetilde{\Phi}_1(\aspar)=\mathbb{E}_{\aspar} \Phi_{1}^{M+1}/\|\bbE_{\aspar} \Phi_{1}^{M+1}(\aspar)\|$ and $\widetilde{\Phi}_2(\aspar)=\mathcal{PC}\widetilde{\Phi}_1(\aspar)$.  The functions $\widetilde{\Phi}_j(\aspar)$  for $j=1,2$ are defined as $\Phi_j(\aspar, \cdot)$ but with the index $M+1\geq 1$ replacing $M$. For $\aspar$ large enough:
\begin{equation}\label{eq.convnormHK}
\Big\|\Phi_1(\aspar) - \frac{ \Phi_{1}^{M}(\aspar)}{\| \Phi_{1}^{M}(\aspar)\|}   \Big\|_{H^1} \leq  \big\|\Phi_1(\aspar) -\widetilde{\Phi}_1(\aspar) \big\|_{H^1} 
+ \Big\|\widetilde{\Phi}_1(\aspar) -\frac{ \Phi_{1}^{M}(\aspar)}{\| \Phi_{1}^{M}(\aspar)\|}  \Big\|_{H^1} .
\end{equation}
We deal first  with the second term of the right hand side of \eqref{eq.convnormHK}:
\begin{equation}\label{eq.convnormHK2}
\Big\|\widetilde{\Phi}_1(\aspar) -\frac{ \Phi_{1}^{M}(\aspar)}{\| \Phi_{1}^{M}(\aspar)\|} \Big\|_{H^1} \leq  \Big\| \widetilde{\Phi}_1(\aspar)- \frac{ \Phi_{1}^{M+1}(\aspar)}{\| \Phi_{1}^{M+1}(\aspar)\|} \Big\|_{H^1} 
+  \Big\|\frac{ \Phi_{1}^{M+1}(\aspar)}{\| \Phi_{1}^{M+1}(\aspar)\|} -\frac{ \Phi_{1}^{M}(\aspar)}{\| \Phi_{1}^{M}(\aspar)\|}  \Big\|_{H^1} .
\end{equation}
By virtue of Corollary \ref{cor.defphi1}, the first term of the right hand side of  \eqref{eq.convnormHK2} is bounded by $C / \aspar^{M+1}$ (since the $H^1_{\bK}$-norm is dominated by the norm $\small\| \cdot \small\|_{a_{\aspar,\bK}}$ with a constant  independent of  $\aspar$ for $\aspar\geq 1$).

We estimate now  the second term of  the right hand side of \eqref{eq.convnormHK2}.  From the definition \eqref{eq.defquasimode} of $\Phi^{M}_1(\aspar)$, the fact that $\| \Phi_{1}^{M}(\aspar)\|_{H^1_{\bK}} $ tends to $\| P_{n,\bK}^A\|_{H^1_{\bK}}>0$  and  both $\| \Phi_{1}^{M}(\aspar) \|$ and  $\| \Phi_{1}^{M+1}(\aspar) \|$ tend to $\| P_{n,\bK}^A\|=1$ (as $\aspar \to +\infty$), one has  for $\aspar$ large enough
\begin{eqnarray*}
\hspace{-0.2cm}\Big\|\frac{ \Phi_{1}^{M+1}(\aspar)}{\| \Phi_{1}^{M+1}(\aspar)\|}-\frac{ \Phi_{1}^{M}(\aspar)}{\| \Phi_{1}^{M}(\aspar)\|}  \Big\|_{H^1}  &\leq & \frac{\| \Phi_{1}^{M+1}(\aspar)-\Phi_{1}^{M}(\aspar) \|_{H^1}} {\| \Phi_{1}^{M+1}(\aspar)\| }+  \| \Phi_{1}^{M}(\aspar)\|_{H^1} \Big( \frac{| \, \| \Phi_{1}^{M+1}(\aspar)\|-\| \Phi_{1}^{M}(\aspar)\|\, |}{\| \Phi_{1}^{M}(\aspar)\| \, \| \Phi_{1}^{M+1}(\aspar)\|}\Big) \nonumber \\[3pt]
 & \leq &   C  \|\Phi_{1}^{(M+1)}\|_{H^1} \, \aspar^{-(M+1)} +C \, \| \Phi_{1}^{(M+1)}\| \,  \aspar^{-(M+1)} \nonumber \\
&\leq  &    C\, \aspar^{-(M+1)}.
\end{eqnarray*}
Thus, the second term of the right hand side of \eqref{eq.convnormHK} satisfies:
\begin{equation}\label{eq.convnormHK2bis}
\Big\|\widetilde{\Phi}_1(\aspar) -\frac{ \Phi_{1}^{M}(\aspar)}{\| \Phi_{1}^{M}(\aspar)\|} \Big\|_{H^1} \leq  C\, \aspar^{-(M+1)}.
\end{equation}
We estimate now the first term of  the right hand side of \eqref{eq.convnormHK}. To this aim, we first bound  the norm of $\bbE_{g}$ in $B(H^1_{\bK},H^1_{\bK})$ (with $H^1_{\bK}$ is endowed with the  standard $H^1_{\bK}-$norm). For $\aspar$  large enough, $\{ \widetilde{\Phi}_1(\aspar), \widetilde{\Phi}_2(\aspar)\}$  is an orthonormal basis of $\operatorname{ker}(\bbA_{\aspar,\bK}-\lambda_D(\aspar) \mathrm{I}_d)$. Therefore, one has for all $u \in H^1_{\bK}$:
\begin{eqnarray}\label{eq.boundprojH1norm}
\|\bbE_{g} u\|_{H^1} &=& \| (u,  \widetilde{\Phi}_1(\aspar))  \widetilde{\Phi}_1( \aspar)+ (u,  \widetilde{\Phi}_2(\aspar))  \widetilde{\Phi}_2\|_{H^1} \nonumber\\
& \leq&   2  \| \widetilde{\Phi}_1( \aspar)\|  \, \| u\|  \,  \| \widetilde{\Phi}_1( \aspar)\|_{H^1} \nonumber \\
& \leq& 2     \| \widetilde{\Phi}_1( \aspar)\|_{H^1}^2   \| u\|_{H^1},
\end{eqnarray}
(where we used for the first inequality that $\| \widetilde{\Phi}_1\|_{H^1}=\| \widetilde{\Phi}_2\|_{H^1}$ since  by virtue of \eqref{eq.commut3},  the  $\|\cdot\|_{H^1}$ norm is  preserved by $\mathcal{P}\mathcal{C}$). 
From \eqref{eq.convnormHK2bis}, the definition \eqref{eq.defquasimode} of $\Phi^{M}_1(\aspar)$ and the fact that $\|P_{n,\bK}^A\|=1$, one obtains: $\| \widetilde{\Phi}_1( \aspar)\|_{H^1} \to \|P_{n,\bK}^A \|_{H_1}$ as $\aspar\to +\infty$. Thus, it leads with \eqref{eq.boundprojH1norm} that $\|\bbE_{g} \|_{B(H^1_{\bK},H^1_{\bK})}$ is bounded by a constant independent of $\aspar$ for large $\aspar$.
Therefore, there exists $C>0$ such that for $\aspar$ large enough:
$$
\big\|\Phi_1(\aspar) -\widetilde{\Phi}_1(\aspar) \big\|_{H^1}  \leq C \Big\| \frac{\Phi_{1}^{M+1}}{\|\bbE_{\aspar} \Phi_{1}^{M+1}(\aspar) \|}- \frac{\Phi_{1}^{M}}{\|\bbE_{\aspar} \Phi_{1}^{M}(\aspar) \|}   \Big\|_{H^1}.
$$
Hence, it yields:
\begin{eqnarray}\label{eq.convnormHK1bis}
 \big\|\Phi_1(\aspar) -\widetilde{\Phi}_1(\aspar) \big\|_{H^1} & \leq & C\, \Big( \frac{\|  \Phi_{1}^{M+1}- \Phi_{1}^{M}  \|_{H^1}}{\|\bbE_{\aspar} \Phi_{1}^{M+1}(\aspar) \| }+ \| \Phi_{1}^{M} \|_{H^1}  \big|  \|\bbE_{\aspar} \Phi_{1}^{M+1}(\aspar) \|^{-1}- \|\bbE_{\aspar} \Phi_{1}^{M}(\aspar)\|^{-1}\big| \Big) \nonumber \\ 
&\leq & C_1 \, \aspar^{-(M+1)}+  C_2 \,  \frac{ \Big|  \|  \bbE_{\aspar} \Phi_{1}^{M+1}(\aspar) \| -  \|  \bbE_{\aspar} \Phi_{1}^{M}(\aspar) \| \Big|}{\|\bbE_{\aspar} \Phi_{1}^{M+1}(\aspar) \| \,\|\bbE_{\aspar} \Phi_{1}^{M}(\aspar)\|} \nonumber \\ [4pt]
& \leq &C_1 \,\aspar^{-(M+1)}+ C_3 \|  \Phi_{1}^{M+1}(\aspar)- \Phi_{1}^{M}(\aspar)  \| \nonumber \\ 
& \leq & C_4\,  \aspar^{-(M+1)},
\end{eqnarray}
(where we use that  $\|\bbE_{\aspar}\|=1$ in $B(L^2_{\bK},L^2_{\bK})$ and that $\|  \bbE_{\aspar} \Phi_{1}^{M+1}(\aspar) \|, \,   \|  \bbE_{\aspar} \Phi_{1}^{M}(\aspar) \|\to \| P_{n,\bK}^A\|=1 $ as $\aspar \to +\infty$).
Combining  \eqref{eq.convnormHK}, \eqref{eq.convnormHK2bis} and \eqref{eq.convnormHK1bis}  yields immediately the first  estimate of \eqref{eq.asympeigenfunc1}.
Finally, as $\mathcal{P}\mathcal{C}$ preserves the norm $\|\cdot\|_{H^1_{\bK}}$, the second estimate of \eqref{eq.asympeigenfunc1} follows from the first one. 
\end{proof}

\subsection{Asymptotic expansion  of the {\it Dirac velocity} , $v_D(\aspar)$, for $\aspar\gg1$}
We now present an asymptotic expansion of the Dirac velocity, $v_D(g)$, for $\aspar$ large. In contrast to the case of  honeycomb Schroedinger  operators in the strong binding regime \cite{FLW-CPAM:17}, where  the asymptotic parameter is an emergent {\it hopping coefficient}, which is exponentially small in the well-depth parameter, here the asymptotic expansion is  in powers of $1/\aspar$; the convergence is therefore slower. Comparing further,  the Bloch functions $\Phi_j(\cdot,\aspar)$ for $j=1,2$ (defined  formula \eqref{eq.defphi1}) have the following  limit behavior (see Corollary \ref{cor.limiteigenstateH1norm}) in the $H^1-$norm: $\Phi_1(\aspar,\cdot)=P_{n,\bK}^A+O(1/g)$ and $\Phi_2(\aspar,\cdot)=\pm \rme^{- 2\pi\rmi/3} P_{n,\bK}^B+O(1/g)$.
$P_{n,\bK}^A$ and $P_{n,\bK}^B$  are $\bK-$quasi-periodic superpositions of single inclusion Dirichlet states; $P_{n,\bK}^A$ is supported on the $A-$ inclusions and $P_{n,\bK}^B$ is supported on the $B-$ inclusions and, very roughly speaking, play the role of the {\it atomic orbitals} of the Schroedinger analysis; see \cite{FLW-CPAM:17}. In contrast to ground state quantum atomic orbitals, these approximations are compactly and disjointly supported, rather than exponentially localized in the well-depth parameter. The disjointness of the supports of $\Phi_j(\cdot,\aspar)$, $j=1,2$ at leading order implies, via \eqref{eq.fermyveloc}, that $v_D(\aspar)=v_D^{(1)}\aspar^{-1}+O(\aspar^{-2})$. Hence, order zero term, $v_D^{(0)}$, vanishes. A rigorous proof that $v_D^{(1)}\neq 0$ is an open question.  In section \ref{numerics},  we  numerically   observe
 that  $v_D^{(1)}\neq 0$ for $n=1$. Thus, for $g$  large enough,  $(\lambda_D(\aspar), \bK)$  is a  (non-degenerate) conical  / Dirac point for the two first dispersion surfaces.
\begin{theorem}\label{th.fermveloc}
Assume that  $\tilde{\delta}_n$ satisfies the condition $\condS$ for a fixed $n\geq 1$ and let $M\in \NO$.
 Then, for $\aspar$ sufficiently large,  the Dirac velocity $v_D(\aspar)$, defined  by \eqref{eq.fermyveloc}, in terms of the eigenstates  $\Phi_j(\aspar, \cdot)$, $j=1,2$ (given by \eqref{eq.defPhi1proj2}) , associated with the eigenvalue $\lambda_D(\aspar)=\lambda_{2n-1}(\aspar;\bK)=\lambda_{2n}(\aspar;\bK)$, has the following asymptotic expansion:
\begin{equation}\label{eq.asymptVD}
v_D(\aspar)=\sum_{m=1}^M v_D^{(m)} \, \aspar^{-m} + O(\aspar^{-(M+1)}) \  \mbox{ with } \  v_D^{(m)}\in \R \ \mbox{ for } \  m=1,\ldots, M.
\end{equation}
Furthermore, the two first  coefficients are explicitly  given by:
\begin{equation*}\label{eq.asymptVDleadingorder}
 v_D^{(0)}=0  \  \mbox{ and } \ v_D^{(1)}=|2 \int_{\cell^B} \Phi_1^{(1)} \, \overline{ \nabla \mathcal{P}\mathcal{C}  P_{n,\bK}^A  }  \, \rmd \bx +\int_{\cell^-}  \Phi^{(1)}_1 \,  \overline{ \nabla \mathcal{P} \mathcal{C}\Phi^{(1)}_1 } \rmd \bx \Big] \cdot (1,-\rmi)^{\top}|.
\end{equation*} 
\end{theorem}
\begin{proof}
The proof is a straightforward application of the formula  for $v_D(\aspar)$, \eqref{eq.fermyveloc},   and the asymptotic expansions of  $\Phi_j(\aspar,\cdot)$ for $j=1,2$ in the $H^1-$norm given by Corollary \ref{cor.limiteigenstateH1norm}. We omit the detailed calculations.
\end{proof}

\section{Transfer of Dirac points from the $2^{nd}$ and $3^{rd}$ bands to the $1^{st}$ and $2^{nd}$ bands as $\aspar\uparrow$; numerical results}\label{numerics}

In this section we corroborate  the results of Sections \ref{sec.condP} through \ref{sec-asympresult} with numerical simulations. The computations are implemented with Free Fem++ \cite{Hech-12}  and displayed either with Matlab or  with ParaView.  For these numerical experiments the fundamental  periodic cell $\cell$ contains  two disc-shaped inclusions $\cell^A$ and $\cell^B$ of radius $R_0=0.2$. The results displayed in Figures \ref{fig.3eigenvalbig}, \ref{fig.eigenvaltransition}, \ref{fig.vecg23} and \ref{fig.vecg12}
are obtained using $P_2$ periodic Lagrange finite elements on the same mesh of $\cell$.

\begin{figure}[h!]
\begin{center}
\includegraphics[width=8.5cm]{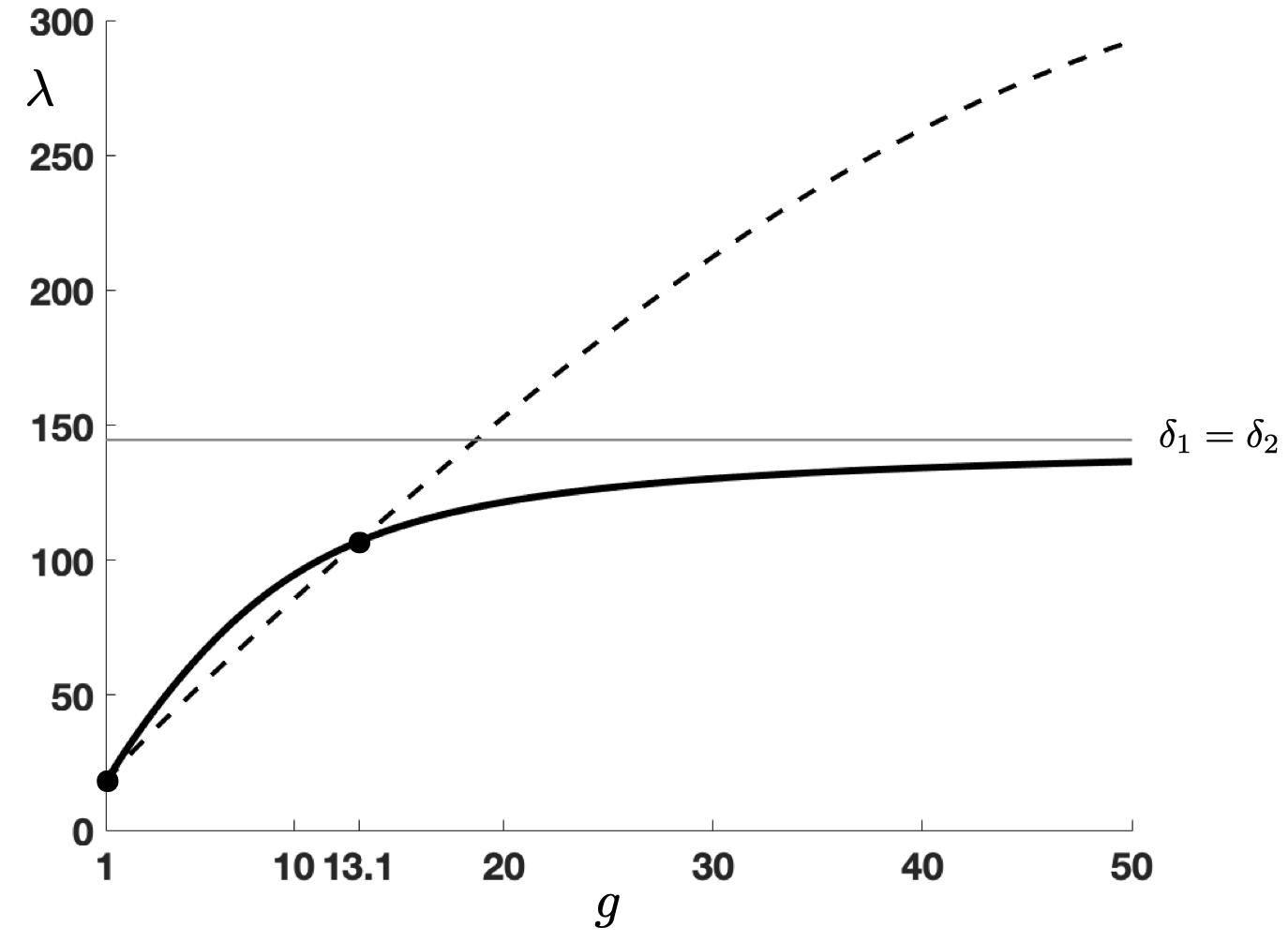}
\end{center}
\caption{
Transfer of Dirac cone with vertex at Dirac point, $(\bK,\lambda_D(g))$, from $2^{nd}$ and $3^{rd}$ bands for $1<\aspar<\aspar_{c}\approx13.1$,
 to the $1^{st}$ and $2^{nd}$ bands for $\aspar>1$. 
 Solid curve: Dirac point energy,  $g\mapsto\lambda_D(g)$, a double eigenvalue with corresponding eigenspace 
 $\subset L^2_{\bK,\tau}\oplus L^2_{\bK,\overline{\tau}}$. Dashed curve: 
 simple eigenvalue with  $1-$ dimensional eigenspace $\subset L^2_{\bK,1}$. As $\aspar\uparrow$, $\lambda_D(g)$ converges from below to the lowest Dirichlet eigenvalue of the domain $\Omega^A\cup\Omega^B$.}
\label{fig.eigenvaltransition}
\end{figure}

\begin{figure}[h!]
\begin{center}
\includegraphics[width=7.25cm]{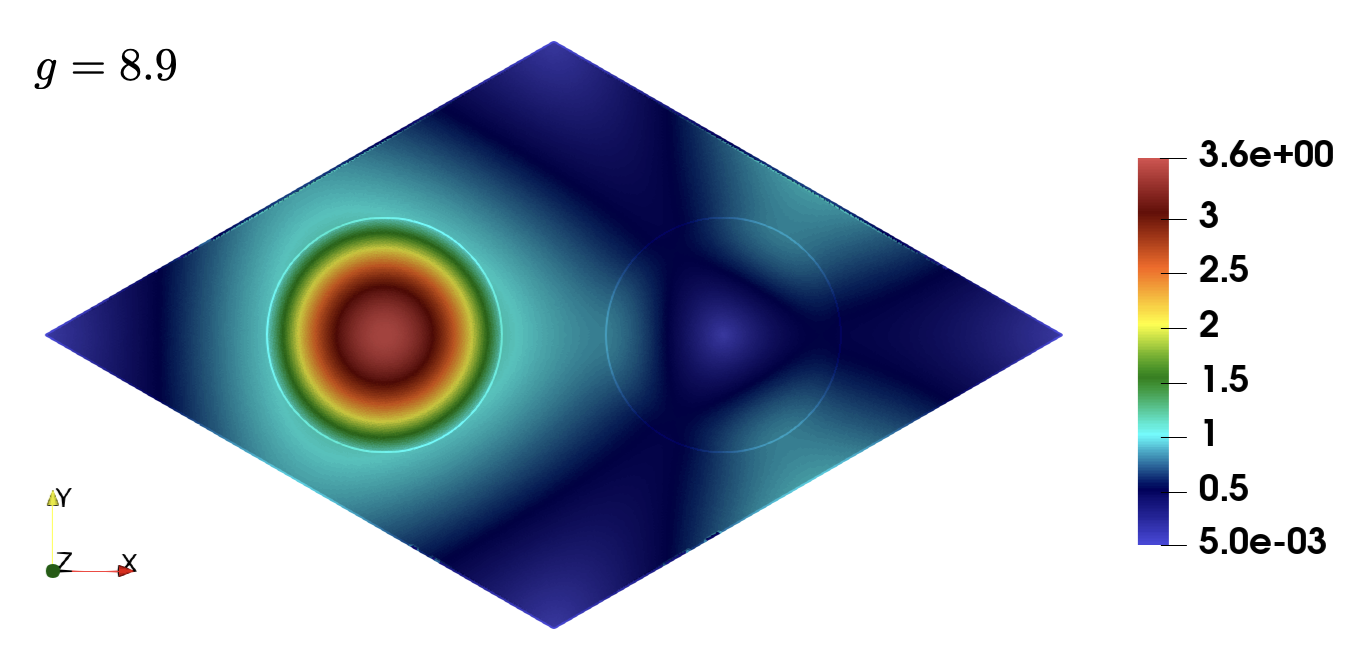}
\includegraphics[width=7.25cm]{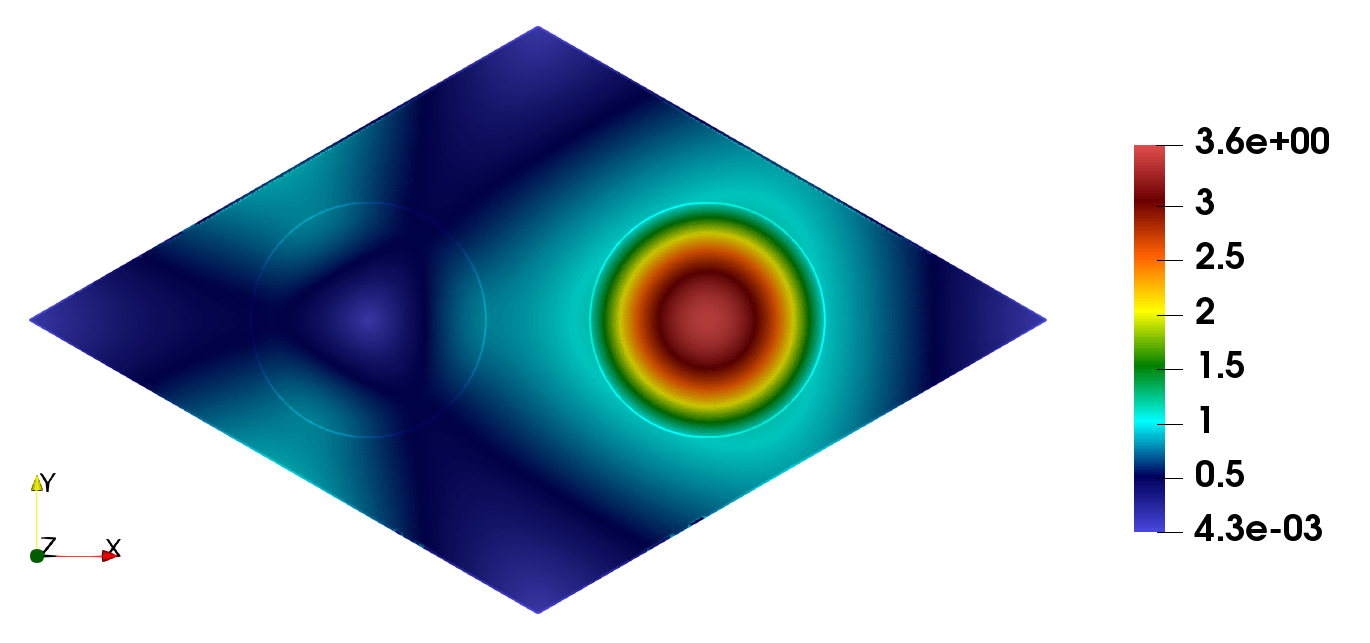}

\includegraphics[width=7.25cm]{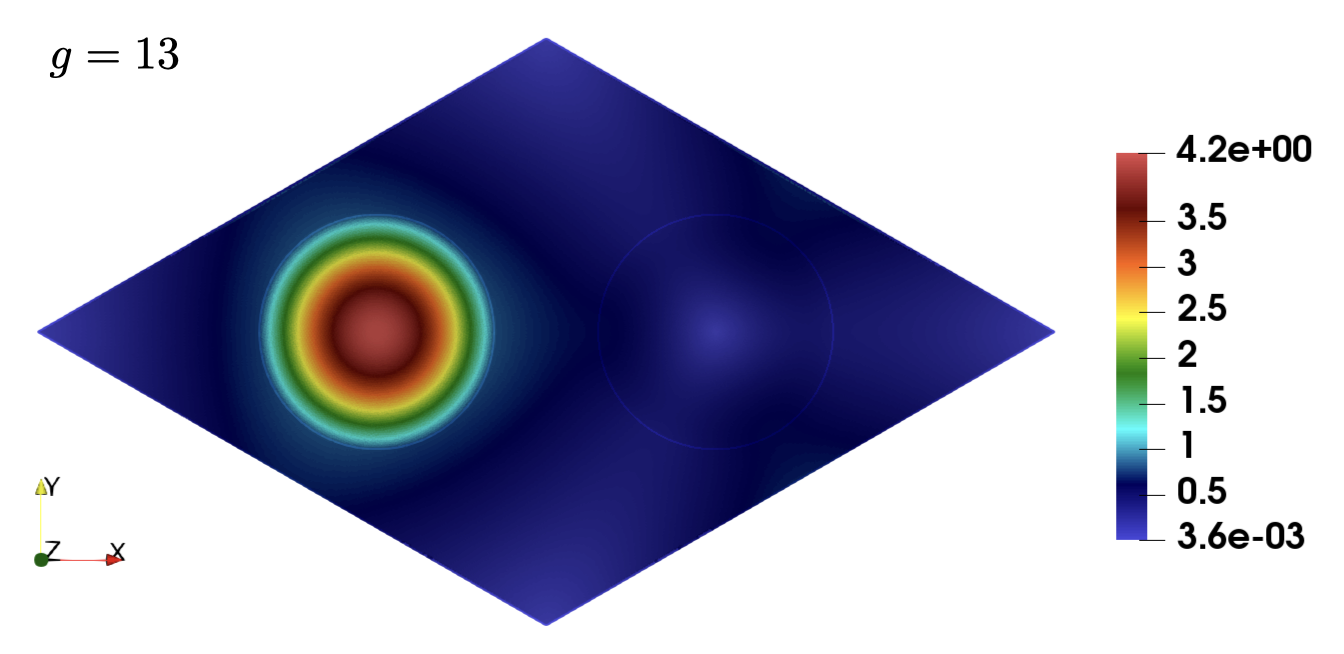}
\includegraphics[width=7.25cm]{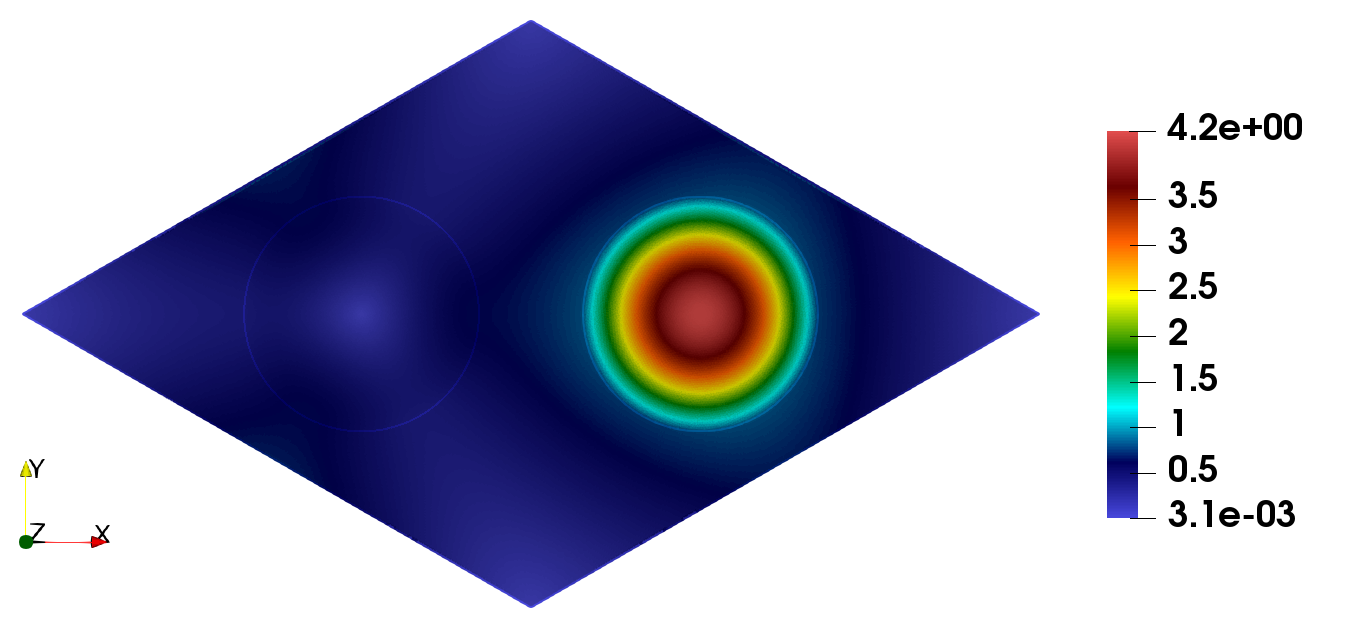}
\end{center}
\caption{$|\Phi_1(\aspar,\cdot)|$ (left)  and $|\Phi_2(\aspar,\cdot)|$  (right) computed with formula  \eqref{eq.defphi1} and \eqref{eq.defphi2} for  $\lambda_D(\aspar)=\lambda_2(\aspar;\bK)=\lambda_3(\aspar;\bK)$ in the case of disc-shaped inclusions of radius $R_0=0.2$  and for $\aspar=8.9$ (first row) and $13$ (second row). }
\label{fig.vecg23}
\end{figure}

\begin{figure}[h!]
\begin{center}
\includegraphics[width=7.25cm]{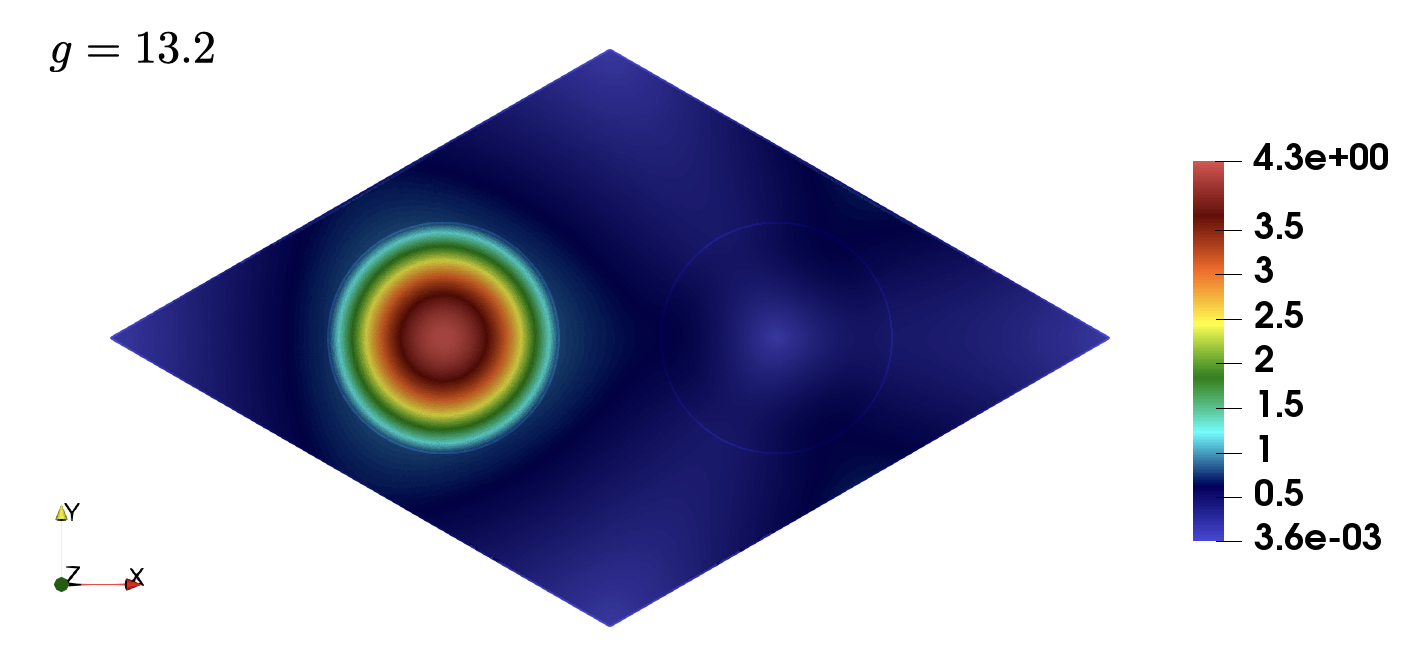}
\includegraphics[width=7.25cm]{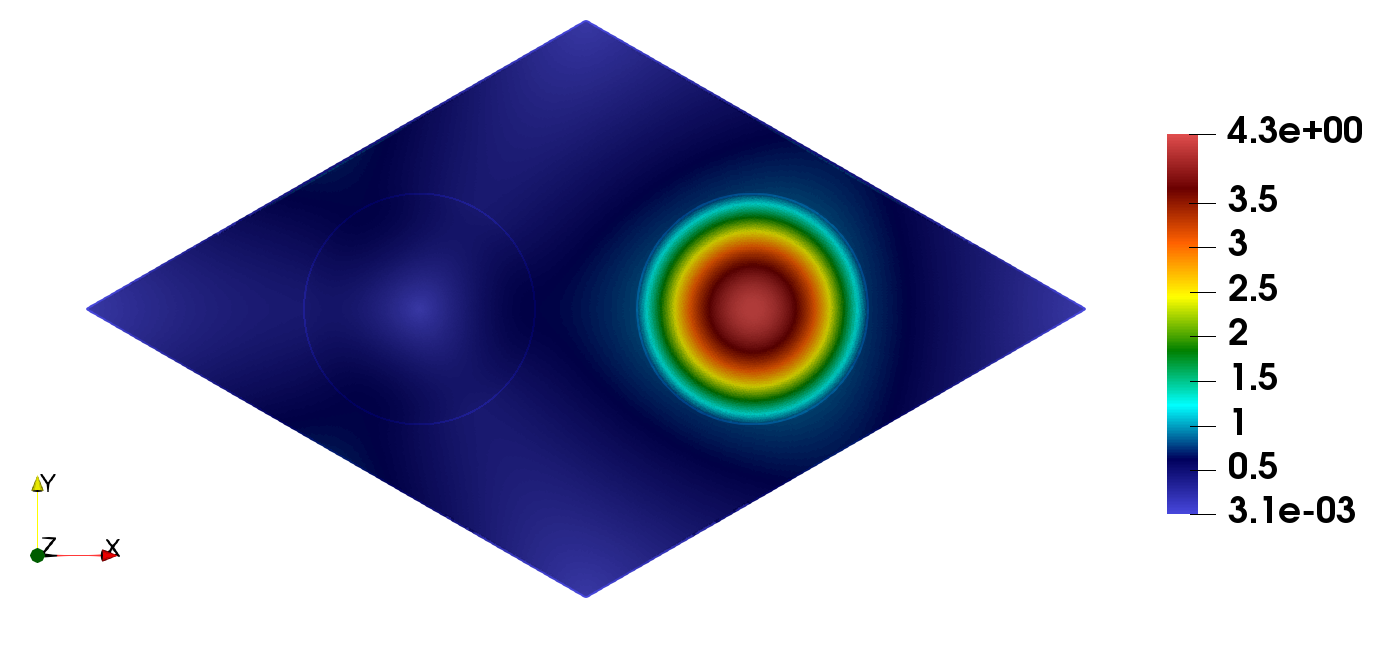}

\includegraphics[width=7.25cm]{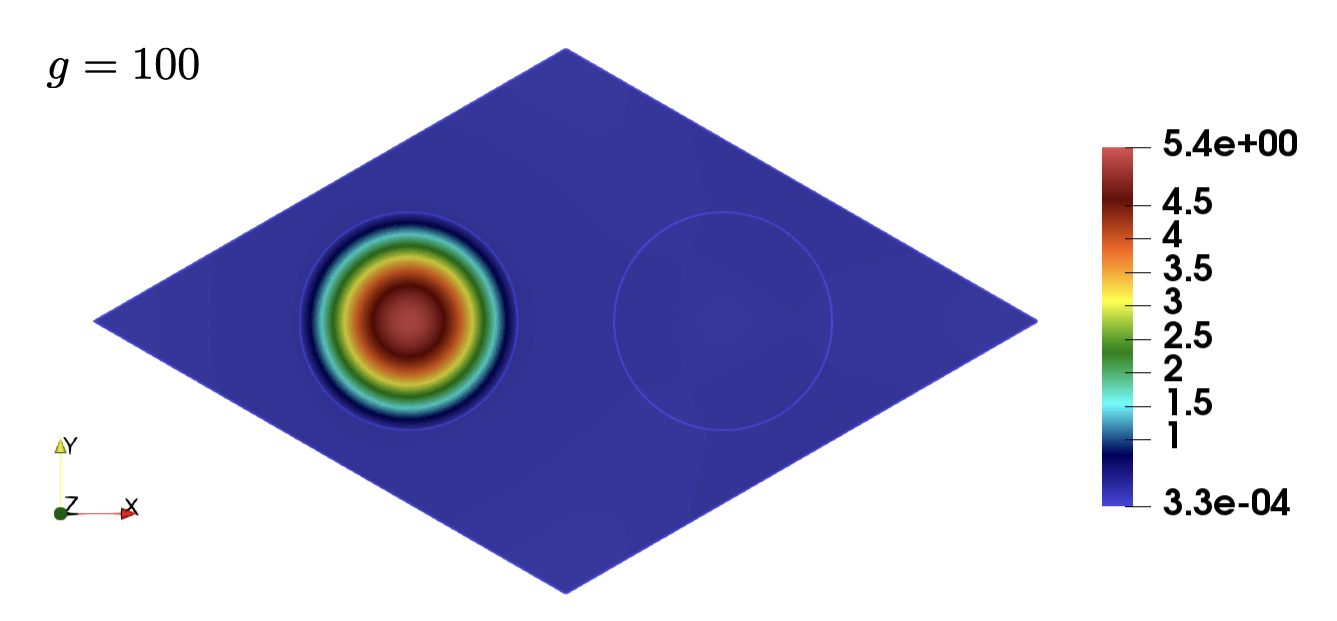}
\includegraphics[width=7.25cm]{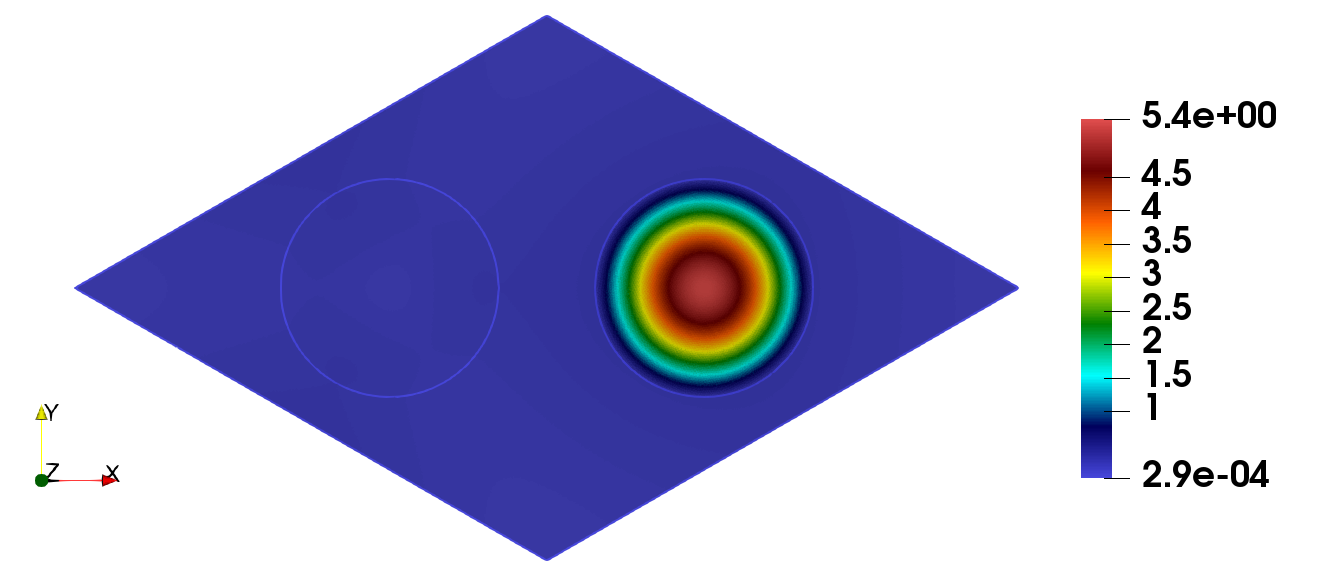}
\end{center}
\caption{$|\Phi_1(\aspar,\cdot)|$ (left)  and $|\Phi_2(\aspar,\cdot)|$  (right) computed with formula \eqref{eq.defphi1} and \eqref{eq.defphi2} for  $\lambda_D(\aspar)=\lambda_1(\aspar;\bK)=\lambda_2(\aspar;\bK)$ in the case of disc-shaped inclusions of radius $R_0=0.2$  and for $\aspar=13.2$ (first row) and $100$ (second row). }
\label{fig.vecg12}
\end{figure}

The behavior of the  first $3$ dispersion surfaces as the contrast parameter, $\aspar$, is varied is described  in Figure \ref{fig.3eigenvalbig}  of Section \ref{sec-mainresults}. Numerically illustrated are: the global behavior on $\mathcal{B}$ of these  dispersion surfaces: uniform convergence, existence of a gap, ... (Corollary  \ref{cor.firstdipstcurves}), and their  local behavior in a neighborhood of the vertices of $\mathcal{B}$: degeneracy, existence of Dirac points (Corollary \ref{cor.Diracpointfirstdispcurv}).

Figure \ref{fig.eigenvaltransition} displays a transfer of Dirac points from the $2^{nd}$ and $3^{rd}$ bands
 to the $1^{st}$ and $2^{nd}$ bands as the contrast parameter $\aspar$ is increased by tracking the Dirac cone vertex 
  at $(\bK,\lambda_D(g))$ as $\aspar$ varies; see Definition \ref{def.Diracpoints}.
For $\aspar=1$, the operator $\bbA_{1,\bK}$ coincides with minus Laplacian (since $\sigma_{\aspar}=1$) with $\bK-$quasi-periodic conditions. As shown in \cite{FW:12}, $\lambda_1(1;\bK)= \lambda_2(1;\bK)=\lambda_3(1;\bK)$ is a triple eigenvalue.
 For $1<\aspar<\aspar_c\approx13.1$ (in particular, $\aspar=8.9$, $\aspar=13$),
  a Dirac point occurs between the $2^{nd}$ and $3^{rd}$ bands  and we observe:
  $\lambda_1(\aspar;\bK)< \lambda_2(\aspar;\bK)=\lambda_3(\aspar;\bK)\equiv \lambda_D(g)$; the (solid curve) graphs of  $\aspar\mapsto \lambda_2(\aspar;\bK)$ and $\aspar \mapsto \lambda_3(\aspar;\bK)$ coincide in Figure  \ref{fig.eigenvaltransition}. It follows from the second  relation in \eqref{eq.propcommutdim} (a consequence of symmetry) that for $\aspar<\aspar_c$,  $g\mapsto  \lambda_1(\aspar;\bK)$ is (dashed) curve of simple $L^2_{\bK}-$ eigenvalues of  
  $\sigma(\bbA_{\aspar,\bK})$. Thus,  by Corollary \ref{cor.commut}, $\lambda_1(\bK,\aspar)$ belongs to $\sigma_1(\bbA_{\bK})$ for $1<\aspar<\aspar_c$. That the Dirac point is located between
  the $2^{nd}$ and the $3^{rd}$ bands indicates that  $\aspar$ in the parameter range $\aspar<\aspar_c$ is not sufficiently large for our theorems to apply;
   indeed for $\aspar$ sufficiently large we know from  Corollary \ref{cor.Diracpointfirstdispcurv}  and  Corollary  \ref{cor.firstdipstcurves} that a Dirac point appears between the first and second bands and that 
    there exists a gap between  the second  and the third band. Increasing $\aspar$, we find for $\aspar=\aspar_{c}\approx 13.1$, the first $3$ dispersion surfaces touch over $\bK$ (and, by symmetry, over all other vertices of $\mathcal{B}$)  
 at a triple eigenvalue $\lambda_1(\aspar_c;\bK)= \lambda_2(\aspar_c;\bK)=\lambda_3(\aspar_c;\bK)$, and that for $\aspar>\aspar_c$ the Dirac point {\it transfers} to an intersection of the first two bands:
 $\lambda_D(g)\equiv \lambda_1(\aspar;\bK)= \lambda_2(\aspar;\bK)<\lambda_3(\aspar;\bK)$ (since the graphs of  $\aspar\mapsto \lambda_1(\aspar;\bK)$ and $\aspar \mapsto \lambda_2(\aspar;\bK)$ coincide on Figure  \ref{fig.eigenvaltransition}).

 Figure \ref{fig.vecg23}  depicts, for  $\aspar=8.9,\, 13$,  the Bloch eigenfunctions
 $\Phi_1(\aspar,\cdot)$ and $\Phi_2(\aspar,\cdot)=\mathcal{P} \mathcal{C}\Phi_{1}(\aspar,\cdot)$,  which are  associated with the degenerate (multiplicity $2$) Dirac eigenvalue $\lambda_D(\aspar)= \lambda_2(\aspar;\bK)=\lambda_3(\aspar;\bK)$. $\Phi_1(\aspar,\cdot)$ and $\Phi_2(\aspar,\cdot)$), given by formulae \eqref{eq.defphi1} and \eqref{eq.defphi2}), are the ``normalized projections'' of  $P_{1,\bK}^A$ (resp. $ \mathcal{P}\mathcal{C} P_{1,\bK}^A=\rme^{-  2\pi \mathrm{i}/3}P_{1,\bK}^B $) on the  $2-$ dimensional kernel of  $\bbA(\aspar,\bK)-\lambda_D(\aspar)I$.

The contrast $\aspar=8.9$ and the inclusion radius $R_0=0.2$ were chosen equal to parameters used in  simulations of Floquet-Bloch modes in \cite[Figure 4 page 71]{Joannopoulos} for the first $2$ TE bands  for a square lattice of circular allumina inclusions at a non-zero quasimomentum. 
Consistent with  observations in \cite{Joannopoulos}  we find, for this range of $\aspar$, that $\Phi_1$ and $\Phi_2$ are not well  approximated by superposition  of our {\it Dirichlet orbitals}.  However, once $\aspar$ is increased to around  $13$ or beyond,   the modes $\Phi_1$ (resp. $\Phi_2$) become   localized on the sublattice, $\Lambda_A$ (resp.  $\Lambda_B$) and begin to look like superpositions of Dirichlet orbitals.

Figure \ref{fig.vecg12}, which depicts the behavior of the Bloch eigenfunctions for higher contrasts: 
  $\aspar=13.2,\, 100>\aspar_c$, demonstrates the behavior predicted by Theorem \ref{th.degeneracy} and its  Corollary \ref{cor.Diracpointfirstdispcurv}. In particular,  $\lambda_D(\aspar)$   is   now situated between the two first bands, and  $\Phi_1$ given by  by formula \eqref{eq.defphi1}  (resp. $\Phi_2$ given by \eqref{eq.defphi2})  are spatially localized  around the $A-$sites (resp. the $B-$sites).  This localization is more pronounced for $\aspar=100$ as predicted by Corollary \ref{cor.Diracpointfirstdispcurv}; $\Phi_1(\aspar,\cdot)\to P_{1,\bK}^A$ and $\Phi_2(\aspar,\cdot) \to 
  \rme^{-2\pi \rmi /3} \,P_{1,\bK}^B$ as $\aspar\to \infty$, where the function $ P_{\bK,1}^A$ (resp. $ P_{1,\bK}^B$), defined by \eqref{eq.defPka}, has support equal to the A- inclusions (resp. on the B-inclusions).

\begin{figure}[h!]
\begin{center}
\includegraphics[width=5.26cm]{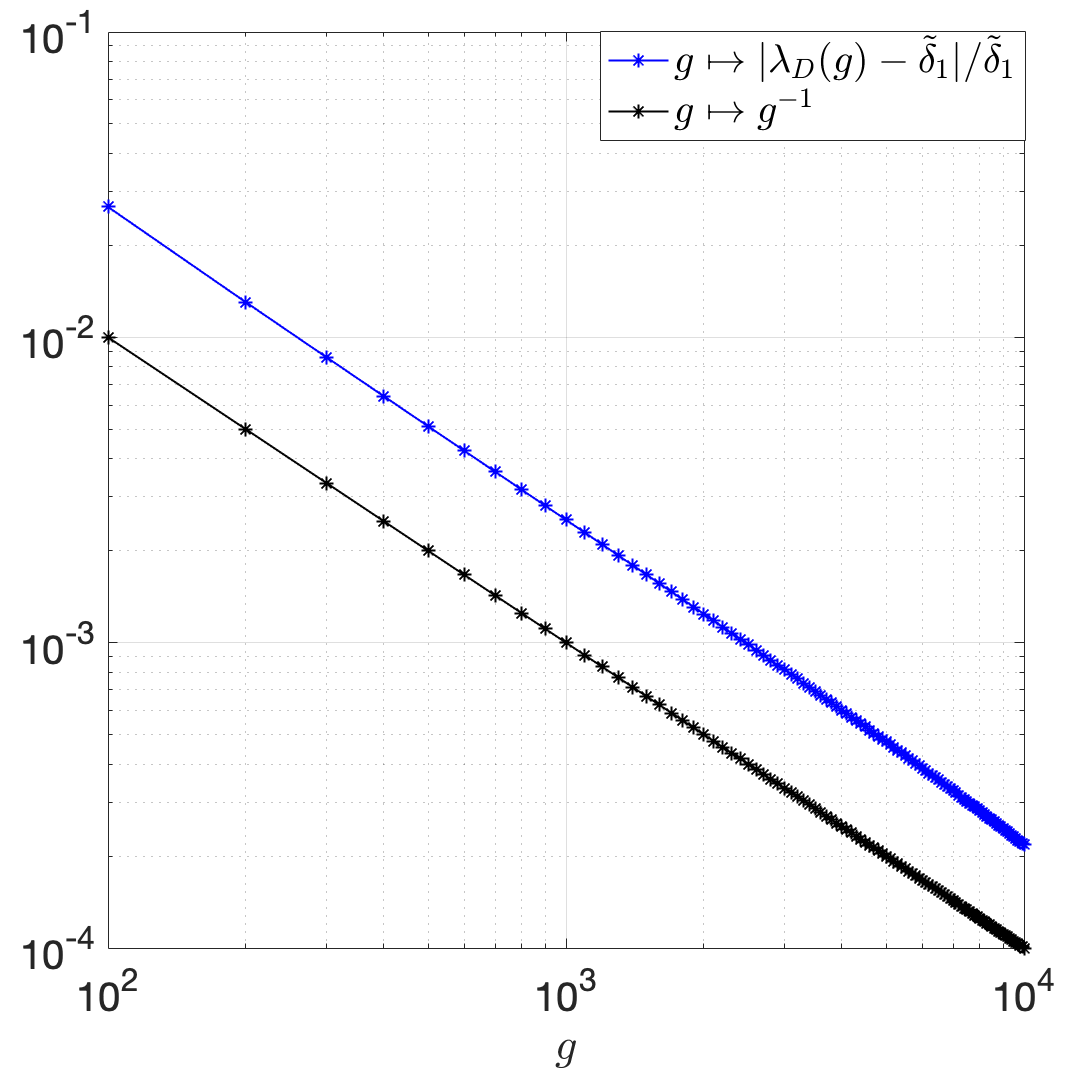}
\includegraphics[width=5.26cm]{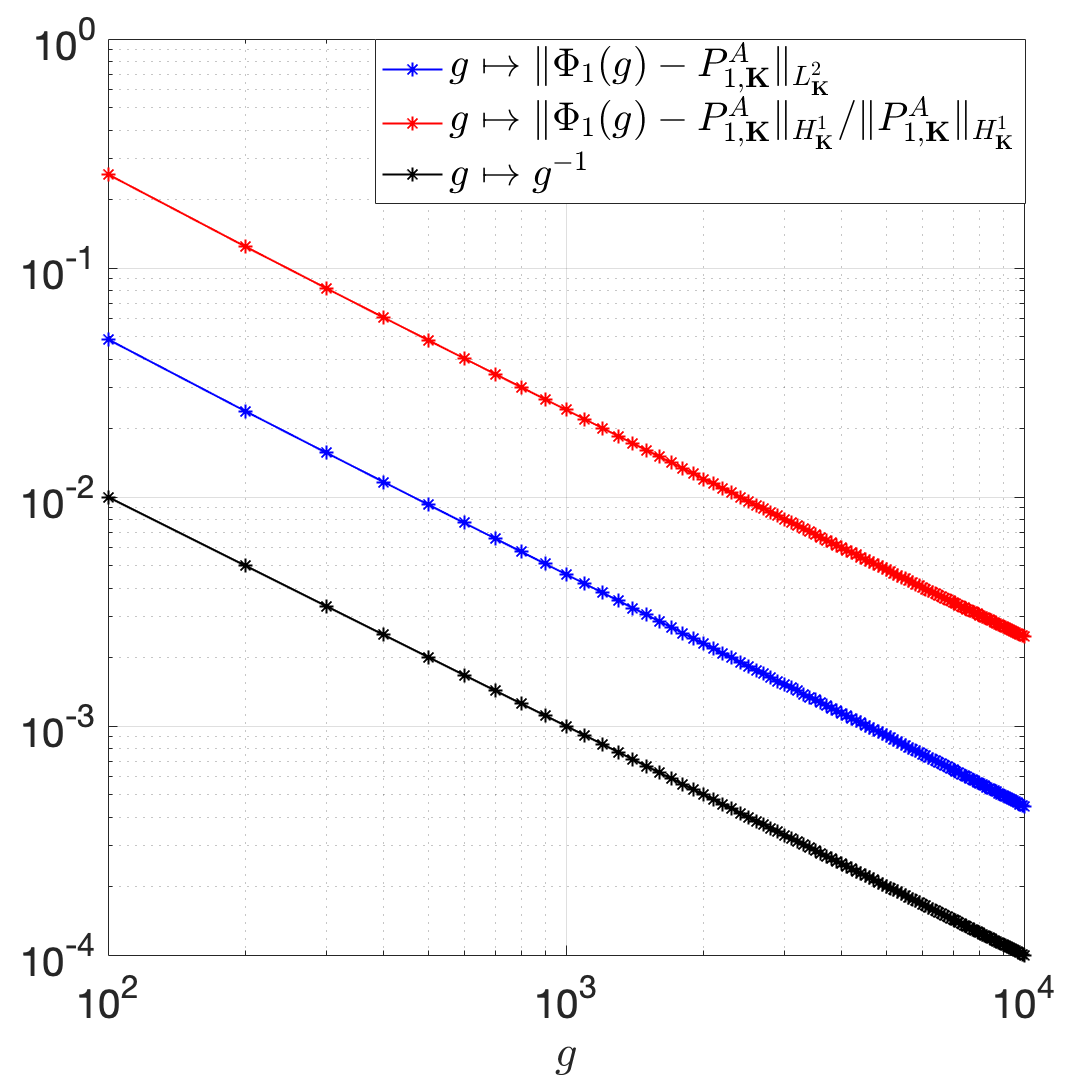}
\includegraphics[width=5.2cm]{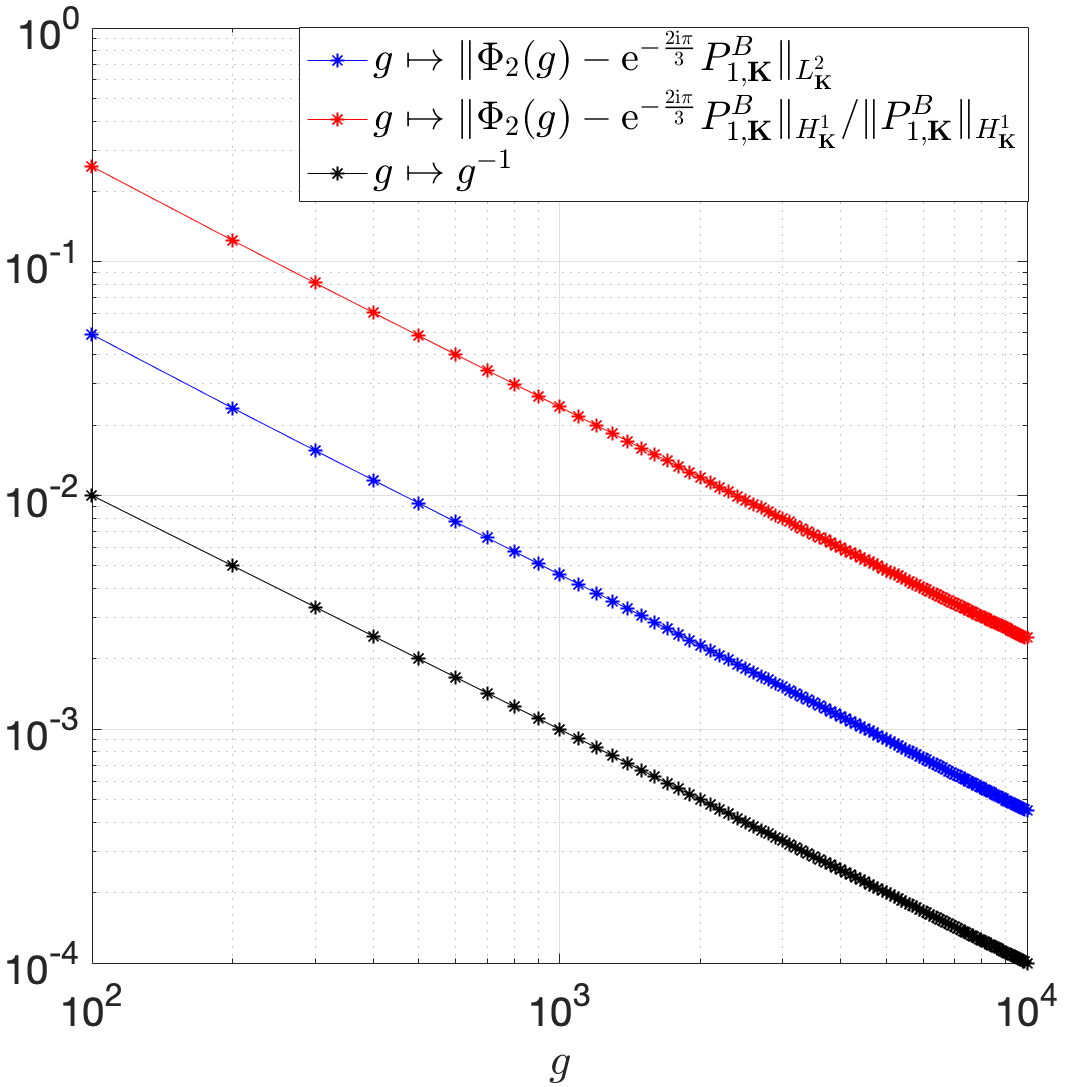}
\end{center}
\caption{All graphs are represented  in a logarithmic scale.
In the left figure, one represents  in blue $\aspar\mapsto |\lambda_D(\aspar)-\tilde{\delta}_1|/\tilde{\delta}_1$. In the center, one plots in blue $\aspar\mapsto \|\Phi_1(\aspar,\cdot)-P_{\bK,A}\|_{L^2_{\bK}}$  and in red $\aspar\mapsto \|\Phi_1(\aspar,\cdot)-P_{\bK,A}\|_{H^1_{\bK}}/\|P_{\bK,A}\|_{H^1_{\bK}}$ . At  the right, one  plots in blue  $\aspar\mapsto \|\Phi_2(\aspar,\cdot)-\rme^{- \, 2\pi \rmi/3} \,P_{1,\bK}^B\|_{L^2_{\bK}}$ and in red  $\aspar\mapsto \|\Phi_2(\aspar,\cdot)-\rme^{-\rmi \, 2\pi/3} \,P_{1,\bK}^B\|_{H^1_{\bK}}/\|P_{1,\bK}^B\|_{H^1_{\bK}}$. In each figure,  $\aspar\mapsto g^{-1}$ is plot in black as a reference graph.
}
\label{fig.conveig}
\end{figure}

\begin{figure}[h!]
\begin{center}
\includegraphics[width=6.2cm]{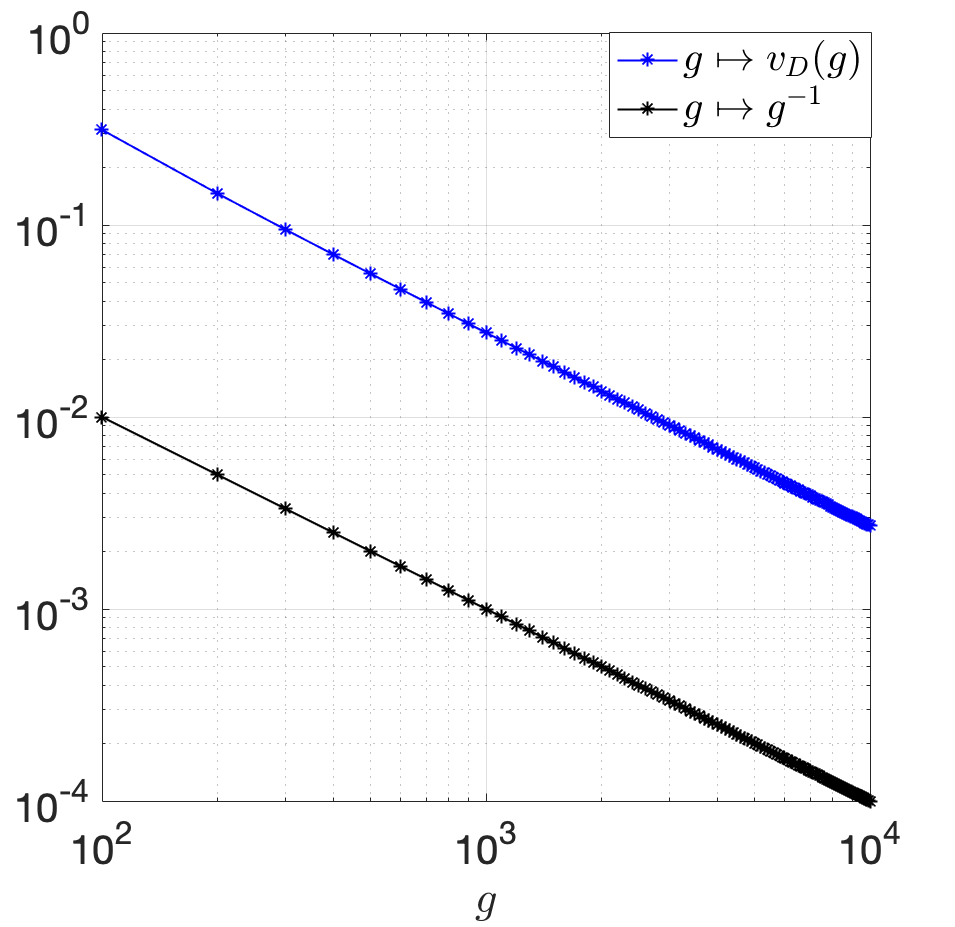} 
\hspace{1cm}
\includegraphics[width=6.2cm]{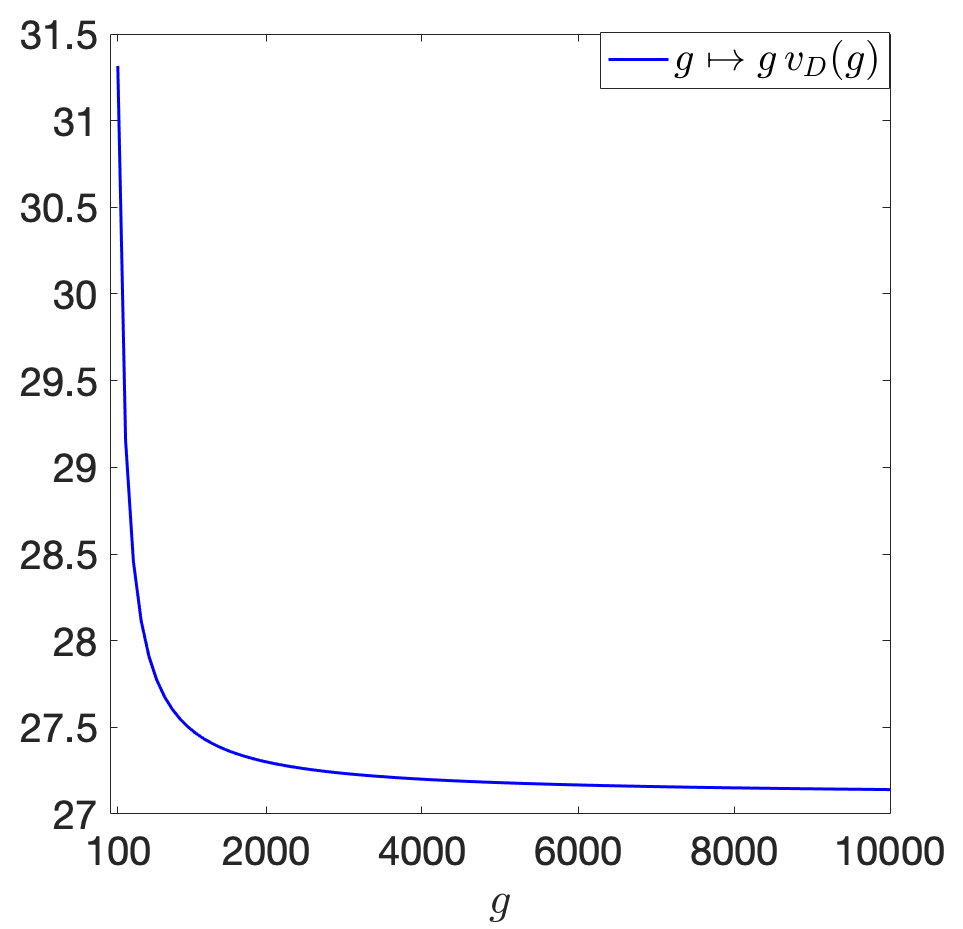}
\end{center}
\caption{The left graph sketches $\aspar\mapsto v_D(\aspar)$ (in blue) $\aspar\mapsto \aspar^{-1}$ (in black, as a reference curve) both in a logarithmic scale whereas the right graph represents $\aspar\mapsto \aspar \, v_D(\aspar)$.}
\label{fig.convgroupvelocity}
\end{figure}

Figures  \ref{fig.conveig} and \ref{fig.convgroupvelocity} are produced using a very fine mesh with $P_2$ periodic Lagrange  finite elements to simulate very large values of the contrast $\aspar$. These figures  concern the first two bands and are performed for $\aspar$ varying from $\aspar=10^2$ to $10^4$ with a constant  step size $\Delta g=100$. Here,  $\aspar$ is large enough to be in the regime of validity of Corollary \ref{cor.Diracpointfirstdispcurv}; $\lambda_D(\aspar)=\lambda_1(\aspar;\bK)= \lambda_2(\aspar;\bK)$ is of multiplicity $2$ and   $\Phi_1(\aspar,\cdot)$ and $\Phi_2(\aspar,\cdot)$ are given by  formulae \eqref{eq.defphi1} and \eqref{eq.defphi2}).
The $1/g$ rate of convergence  of the eigenpairs $(\lambda_D(\aspar),\Phi_j(\aspar,\cdot) ), \, j=1,2$ at $\bK$, predicted by the asymptotic analysis of the Section \ref{sec-asympresult} (Corollary \ref{cor.limiteigenstateH1norm} and  Theorem \ref{th.asympteigen}), is displayed in Figure \ref{fig.conveig}. In the left figure, one observes that the relative error: $\aspar\mapsto |\lambda_D(\aspar)-\tilde{\delta}_1|/\tilde{\delta}_1$ converges to zero, linearly in $1/g$. The center and right figures make clear the $1/g$ convergence of the relative error of $\Phi_1(\aspar,\cdot)-P_{1,\bK}^A$ and $\Phi_2(\aspar,\cdot)- \rme^{-2\rmi \pi /3 } P_{1,\bK}^B$ in the $L^2$ and $H^1$ norms.

The left panel of Figure  \ref{fig.convgroupvelocity}  illustrates that the Dirac velocity $v_D(g)$ (defined by  \eqref{eq.fermyveloc})  tends to $0$ at the rate $1/g$. 
The right panel displays  $g \mapsto g \, v_D(g)$, which indicates convergence to a positive constant $\approx27.1$. The asymptotic expansion \eqref{eq.asymptVD} of Theorem \ref{th.fermveloc} implies $v_D(g)= v_D^{(1)}\aspar^{-1}+\mathcal{O}(\aspar^{-2})$ for $\aspar\gg1$. 
Thus, numerically we observe $v_D^{(1)}>0$. This provides, for $\aspar$ sufficiently large,  a numerical verification of the non-degeneracy condition \eqref{eq.fermvelocnond}  on the Dirac velocity,  $v_D(\aspar)$,   associated with the energy quasi-momentum pair $(\bK,\lambda_D)$; see Corollary \ref{cor.Diracpointfirstdispcurv} and Theorem \ref{th.degeneracy}.
 With this verification, Corollary \ref{cor.Diracpointfirstdispcurv} ensures, for $\aspar$ sufficiently large, the existence of Dirac points at a touching of the first two bands over the $6$ vertices of the Brillouin zone.

\section{Higher energy bands for disc-shaped inclusions}\label{sec.higherbandscirccase}
\subsection{Eigenelements of the limiting operator ($\bbA_{\infty,\bk}$) for disc-shaped inclusions}
 Sections \ref{sec.condP} through \ref{sec-asympresult} discuss results on  the convergence, as  $\aspar\to\infty$, of band dispersion functions,  the existence of Dirac points, asymptotic expansions of the Floquet-Bloch eigenelements and the Dirac velocity. These results require that  $\tilde{\delta}_n$, the $n^{th}$ Dirichlet eigenvalue of  $\LaplDcellA$ for the single inclusion $\cell^A$, satisfies the spectral isolation condition $\condS$
 of Definition \ref{Def.condS}.  As we have seen, condition  $\condS$ holds for the first eigenvalue, $\tilde{\delta}_1$,  for any inclusion shape $\cell^A$ satisfying the assumptions of  Section \ref{sec:sigma}. In this section we demonstrate that condition $\condS$ can be verified in many cases for $n>1$.  This enables us, for disc-shaped inclusions, to obtain results on Dirac points occurring at energies deep within the band spectrum,  infinitely many such.

In this section, we take  $\cell^A=B(\bv_A,R_0)$ where  $B(\bv_A,R_0)$ is the open ball of radius $R_0>0$ centered at $\bv_A$. The radius $R_0$ is chosen so that the geometrical assumptions of Section \ref{sec:sigma} hold, namely $\overline{\cell^{A}}\cap \overline{\cell^{B}}=\emptyset$ and  $\overline{\cell^+}\cap \partial \cell=\emptyset$;
 in other words,  the $2$ disc-shaped inclusions do not overlap  and  do not intersect the boundary of the unit cell.
For this geometry, the spectrum $\sigma(\LaplDcellA)$  is expressible in terms of  zeros of the Bessel functions $J_p(z)$ for $p\in \NO$. 
We first recall some properties of these zeros; see \cite[chapter XV]{Wat:66} and  Figure \ref{fig.zerosBessel}:

\begin{itemize}
\item For each $p\in \NO$, the zeros of the Bessel functions $J_p(z)$ are real. Moreover,  $J_p(z)$ has an infinite number of positive zeros $z_{p,q}$ arranged in an ascending order of magnitude and indexed by $q\geq 1$.  These zeros are simple, {\it i.e.} $J_p'(z_{p,q})\neq 0$ for $q\geq 1$.
\item The zeros of any $2$ different Bessel functions, $J_{p_1}(z)$ and $J_{p_2}(z)$ where  $p_1\neq p_2$ are distinct. Moreover,  for  consecutive Bessel functions: $J_p(z)$ and $J_{p+1}(z)$,  they interlace:
\begin{equation*}\label{eq.interlac}
0<z_{p,1}< z_{p+1,1}< z_{p,2}<z_{p+1,2}<z_{p,3}<\ldots
\end{equation*}
\end{itemize}
\begin{figure}[!h]
\vspace{-0.3cm}
\centering
 \includegraphics[width=0.9\textwidth]{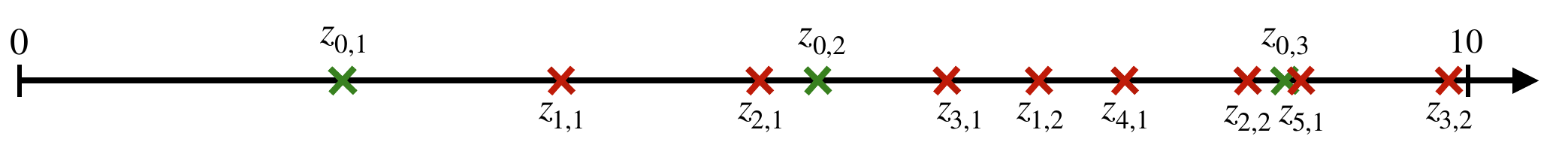}
   \caption{Positive zeros $z_{p,q}$ whose magnitude are below  $10$.}
 \label{fig.zerosBessel}
\end{figure}
Using the above properties, one obtains that the discrete spectrum of $\LaplDcellA$  is given by
\begin{equation*}\label{eq.speccirc}
\sigma( \LaplDcellA)=\left\{ \left(\frac{z_{p,q}}{R_0}\right)^2 \ :\  (p,q) \in \NO \times \N \right\}=\{ \tilde{\delta}_n,   \mbox{ for } n\in \N \},
\end{equation*} 
where the listing of $\tilde{\delta}_n$ is with multiplicity.
The corresponding eigenfunctions are given in terms of  Bessel functions. Indeed,
if  the $n^{th}$ Dirichlet eigenvalue of $\cell^A$, $ \tilde{\delta}_n =\left(z_{0,q}/R\right)^2$ for some $q \in \N$,
 then $\operatorname{ker}(\LaplDcellA-\tilde{\delta}_n \mathrm{Id})$ is $1$-dimensional and spanned by the normalized function $p_n$ with 
\begin{equation}\label{eq.defJoeigenfunc}
 p_n(\bx)=\frac{J_0\big(z_{0,q}\,|\bx| \, R_0^{-1} \big)}{\sqrt{\pi} \, |J_0'(z_{0,q})| R_0},  \quad\bx\in \cell^A.
\end{equation}
However, if  $ \tilde{\delta}_n =z_{p,q}^2/R_0^2 $ for some $p,q\geq1$, then  $\operatorname{ker}(\LaplDcellA-\tilde{\delta}_n \mathrm{Id})$ is $2$-dimensional  and spanned by the orthonormal set consisting of the functions:
\begin{equation}\label{eq.defJpeigenfunc}
p_{n,s}(\bx)= \frac{\sqrt{2}\,  J_p\big( z_{p,q} |\bx| \, R_0^{-1} \big)}{\sqrt{\pi} |J_p'(z_{p,q})| R_0}  \, \sin(p\, \theta) \ \mbox{ and }\ p_{n,c}(\bx)=  \frac{\sqrt{2}\, J_p\big( z_{p,q} |\bx| \, R_0^{-1}\big)}{\sqrt{\pi} |J_p'(z_{p,q})| R_0}  \, \cos(p\, \theta),  \quad \bx\in \cell^A.
\end{equation}
To obtain the normalization  in \eqref{eq.defJoeigenfunc} and  \eqref{eq.defJpeigenfunc}, we use  \cite[formula (11) page 135, chapter 5]{Wat:66}: $$\int_0^{R_0} J_p^2(z_{p,q} \, r/R_0) \, r \, \mathrm{d}r =R_0^2  \, J_p'(z_{p,q})^2/2.$$ This normalization is well-defined since $J_p'(z_{p,q})\neq 0$. 

The next proposition states  that the (infinitely many) Dirichlet  eigenvalues of $\cell^A$, $ \tilde{\delta}_n $, 
 which arise from zeros of the zeroth order Bessel function, $J_0(z)$, satisfy the spectral isolation condition $\condS$.

\begin{proposition}\label{prop.higherband}
Let $ \tilde{\delta}_n =\left(z_{0,q}/R_0\right)^2 $ for some $n,q \in \N$. Then, $\tilde{\delta}_n$ satisfies condition $\condS$ of Definition \ref{Def.condS}. Furthermore, the corresponding normalized eigenfunction $p_n$, defined by \eqref{eq.defJoeigenfunc} (and extended by $0$ on $\bbR^2\setminus \cell^A$), satisfies  the symmetry relations \eqref{eq.RsymR2} and \eqref{eq.CsymR2} and is even. Thus all the results from Sections \ref{sec.condP} to \ref{sec-asympresult} hold for such values of $n$. 
\end{proposition}
\begin{proof}
If   $ \tilde{\delta}_n =\left(z_{0,q}/R_0\right)^2 $ for some $q \in \N$, then by \eqref{eq.defJoeigenfunc} we have that $\tilde{\delta}_n$ is a simple eigenvalue with $p_n$, defined by  \eqref{eq.defJoeigenfunc}, as an associated normalized eigenfunction. Moreover, by the identities: $(x\, J_1(x))'=x \, J_0(x)$ and $J_1(x)=-J_0'(x)$ for $x\in \R$  (see \cite{Wat:66}, equations (5) and (7) page 18), it follows that
\begin{equation*}
\int_{\cell^A}p_{n}(\bx) \rmd \bx= \frac{-2 \sqrt{\pi}  \,R_0\, J_0'(z_{0,q})}{ z_{0,q} |J_0'(z_{0,q})|}=-2  \sqrt{\pi} \, \frac{R_0} {z_{0,q}} \operatorname{sgn}(J_0'(z_{0,q})) \neq 0,
\end{equation*}
where $ \operatorname{sgn}$ stands for the sign function. 
Thus,  $\tilde{\delta}_n$  satisfies the condition $\condS$. Morever, it is clear that $p_n$ satisfies the symmetry relations \eqref{eq.RsymR2} and \eqref{eq.CsymR2} and is even.  
\end{proof}

\begin{remark}
The eigenvalues   $\tilde{\delta}_n=\left(z_{p,q}/R_0\right)^2$ with $p,q\geq 1$  satisfy neither of the two properties of the condition $\condS$. Specifically,  (a) $\tilde{\delta}_n$ is of multiplicity $2$, and (b)  $(p_{n,s},p_{n,c})$ is an orthonormal basis of $\operatorname{ker}(\LaplDcellA-\tilde{\delta}_n \mathrm{Id})$ and both functions  have zero mean. Hence, 
$ \int_{\cell^A}u(\bx) \rmd \bx=0$ for all  $ u\in \operatorname{ker}(\LaplDcellA-\tilde{\delta}_n \mathrm{Id})$.
 Hence, the results  of Sections \ref{sec.condP} to \ref{sec-asympresult} do not hold in such cases. The questions of  the existence of Dirac points and  the asymptotic expansions of the Bloch eigenelements at  the quasimomenta  $\bK_*$ in the  high contrast  regime  are not treated in this paper; only our general results on the convergence of the band functions of Section \ref{HL-thy} apply in this setting.
\end{remark}

\subsection{Numerical results for high energy bands}

\begin{figure}[h!]
\begin{center}
\includegraphics[width=5.1cm]{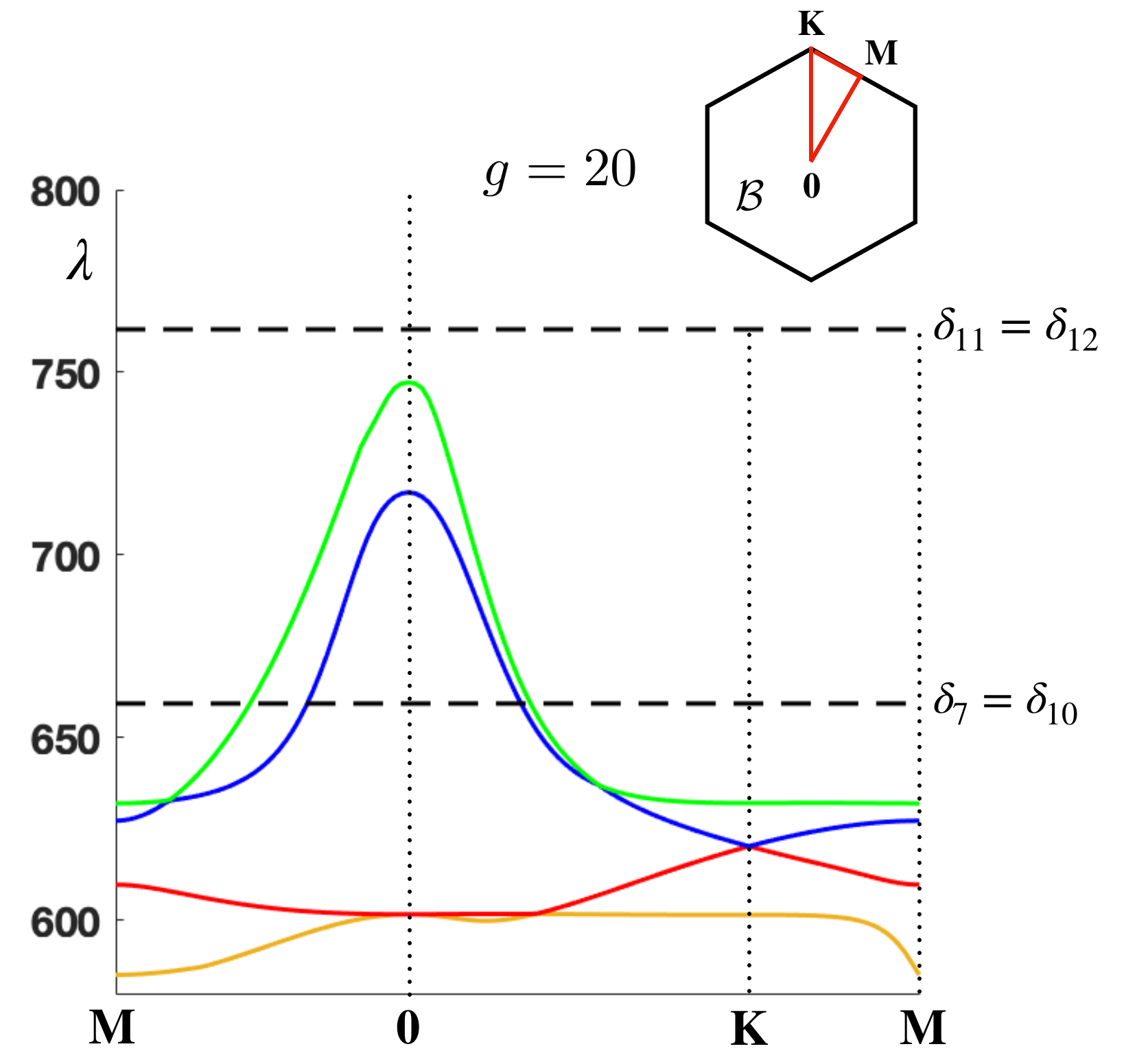}
\hspace{0.1cm}
\includegraphics[width=5.1cm]{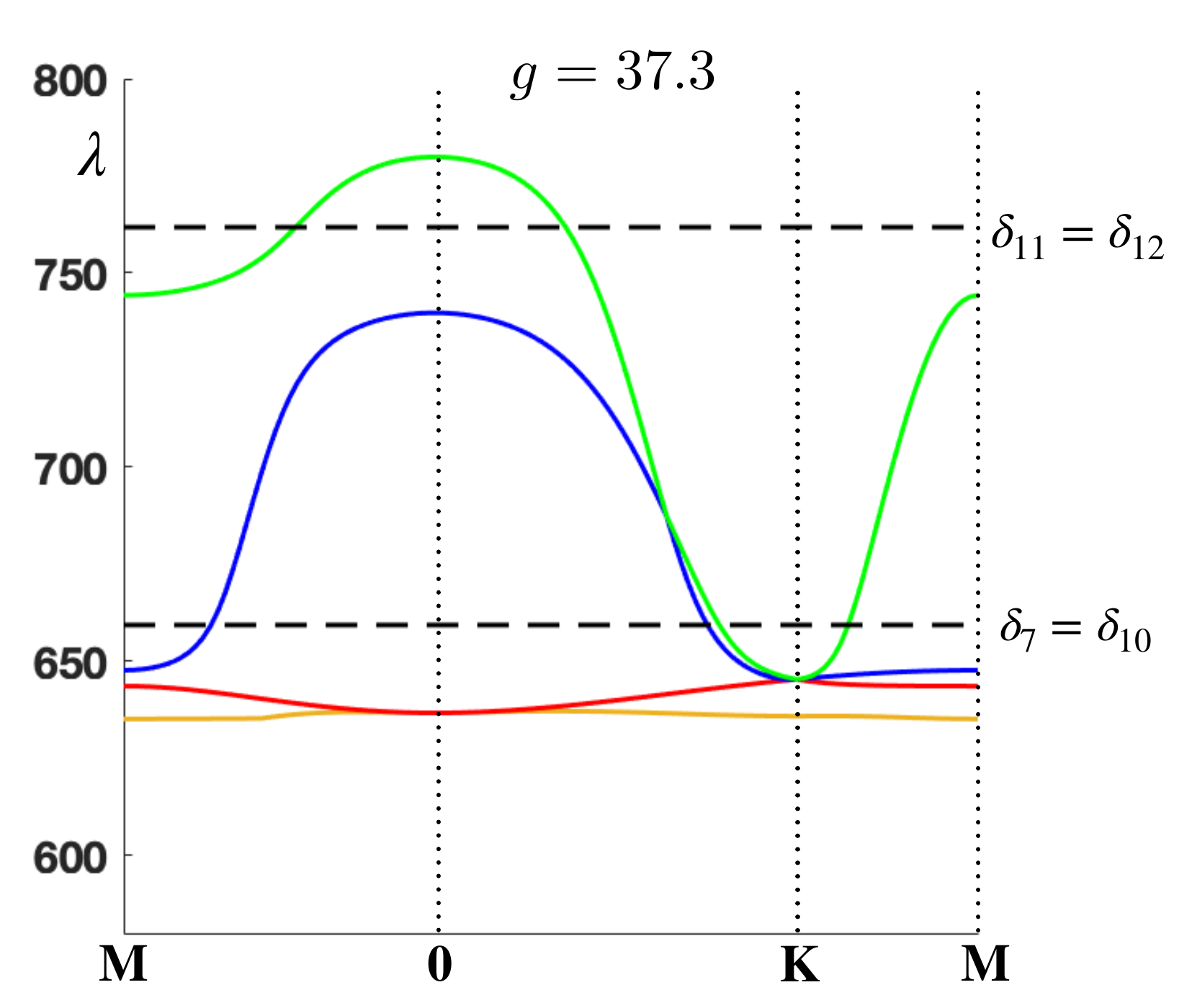}
\hspace{0.1cm}
\includegraphics[width=5.1cm]{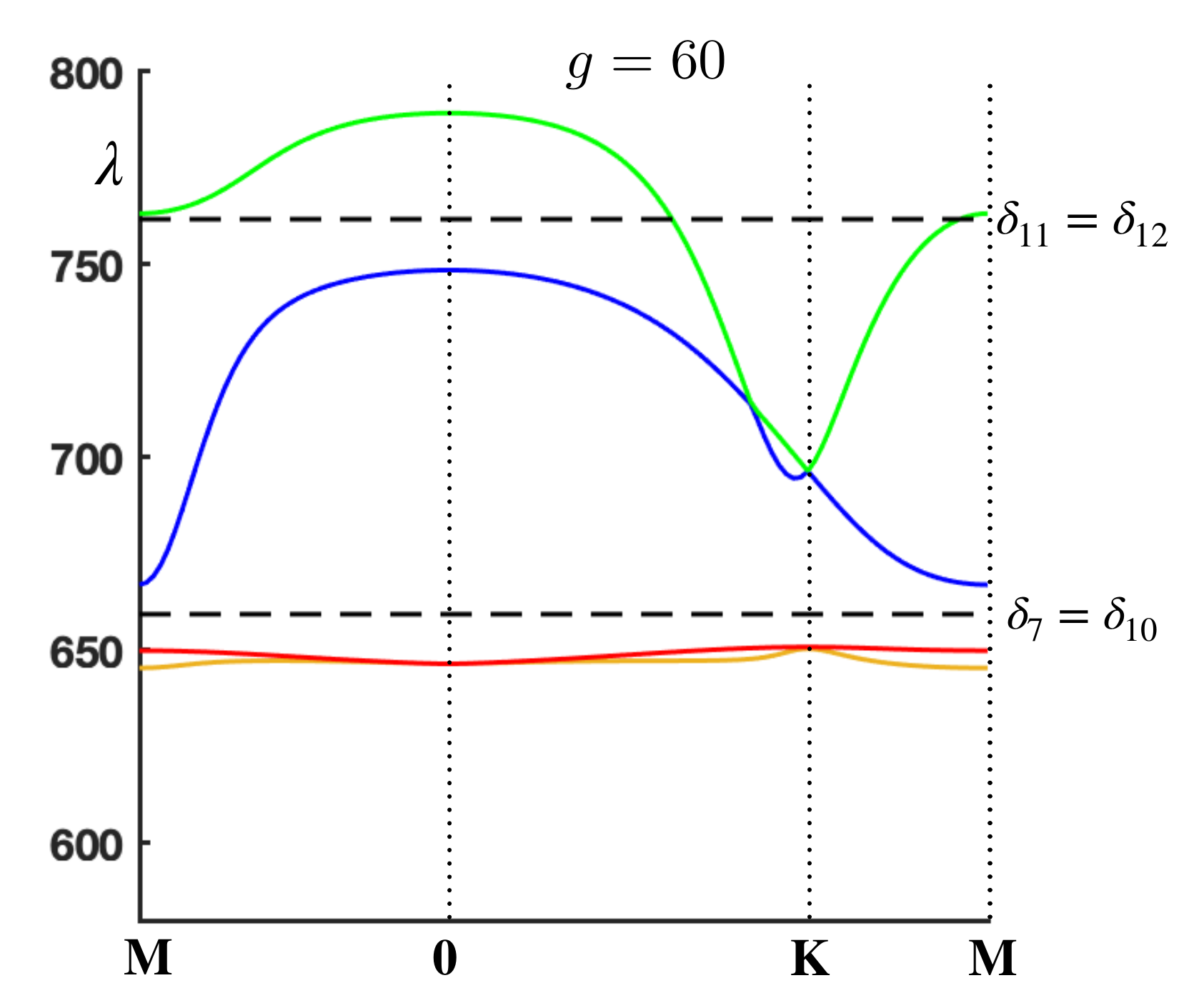}
\end{center}
\begin{center}
\includegraphics[width=5.1cm]{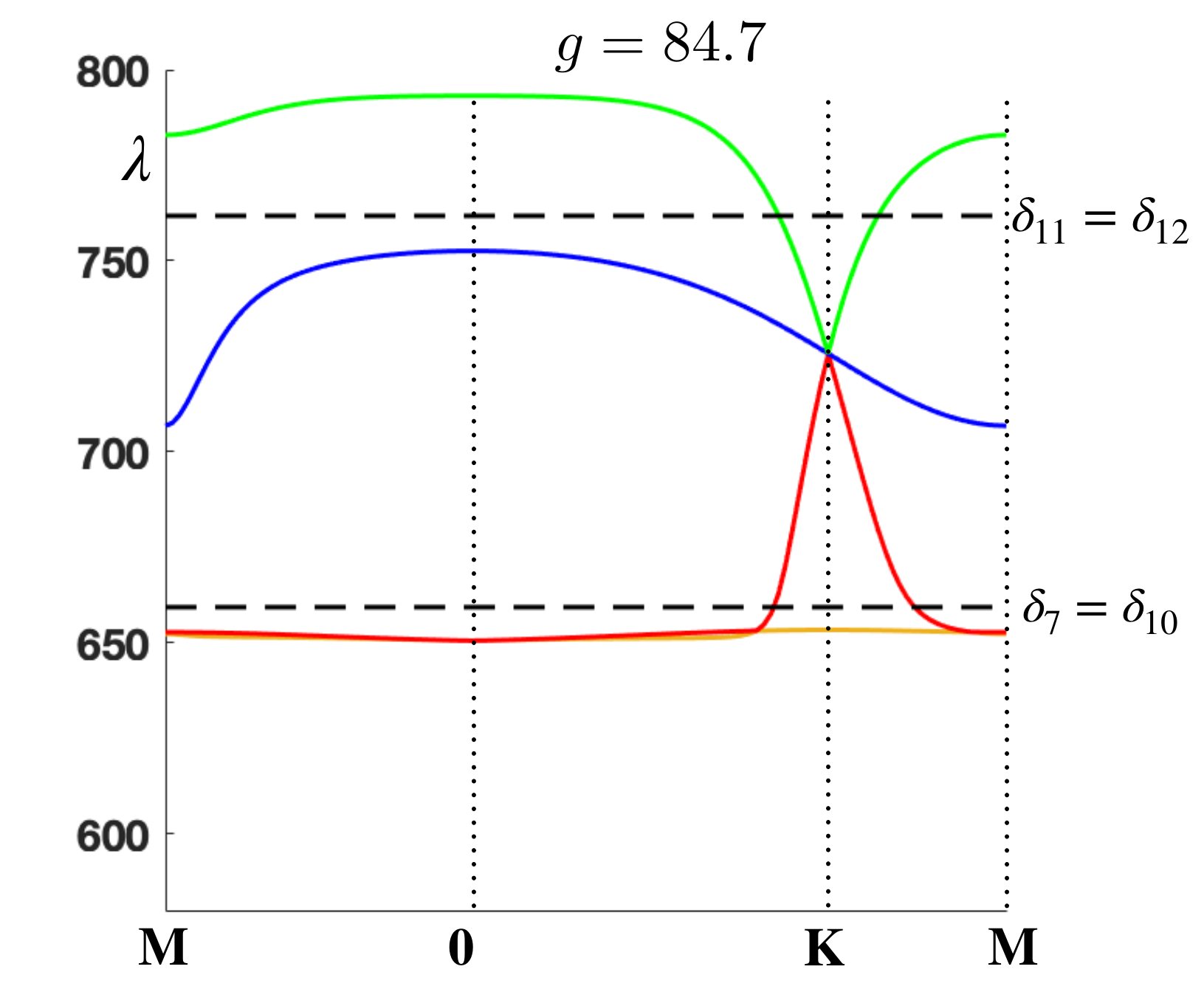}
\hspace{0.1cm}
\includegraphics[width=5.1cm]{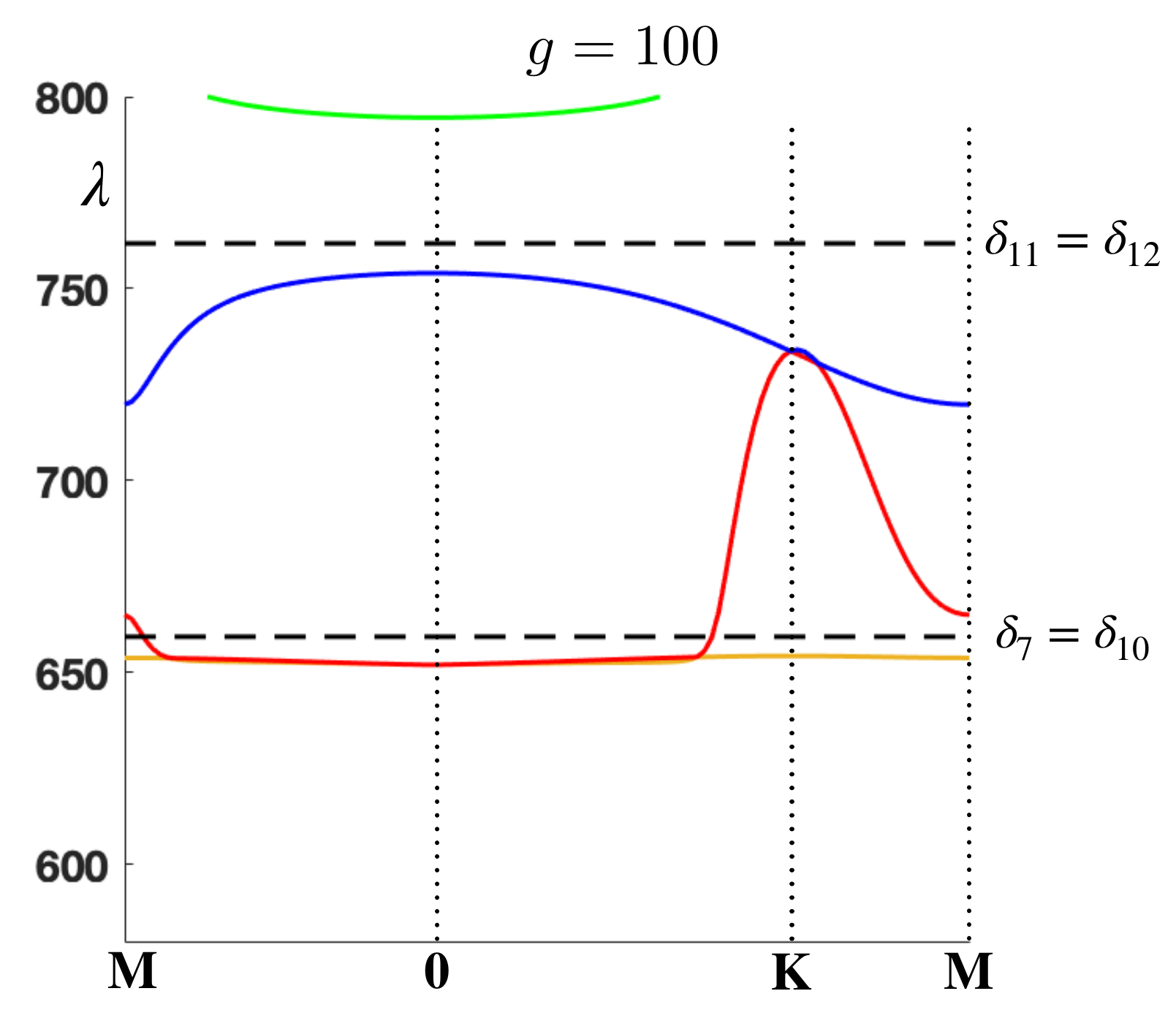}
\hspace{0.1cm}
\includegraphics[width=5.1cm]{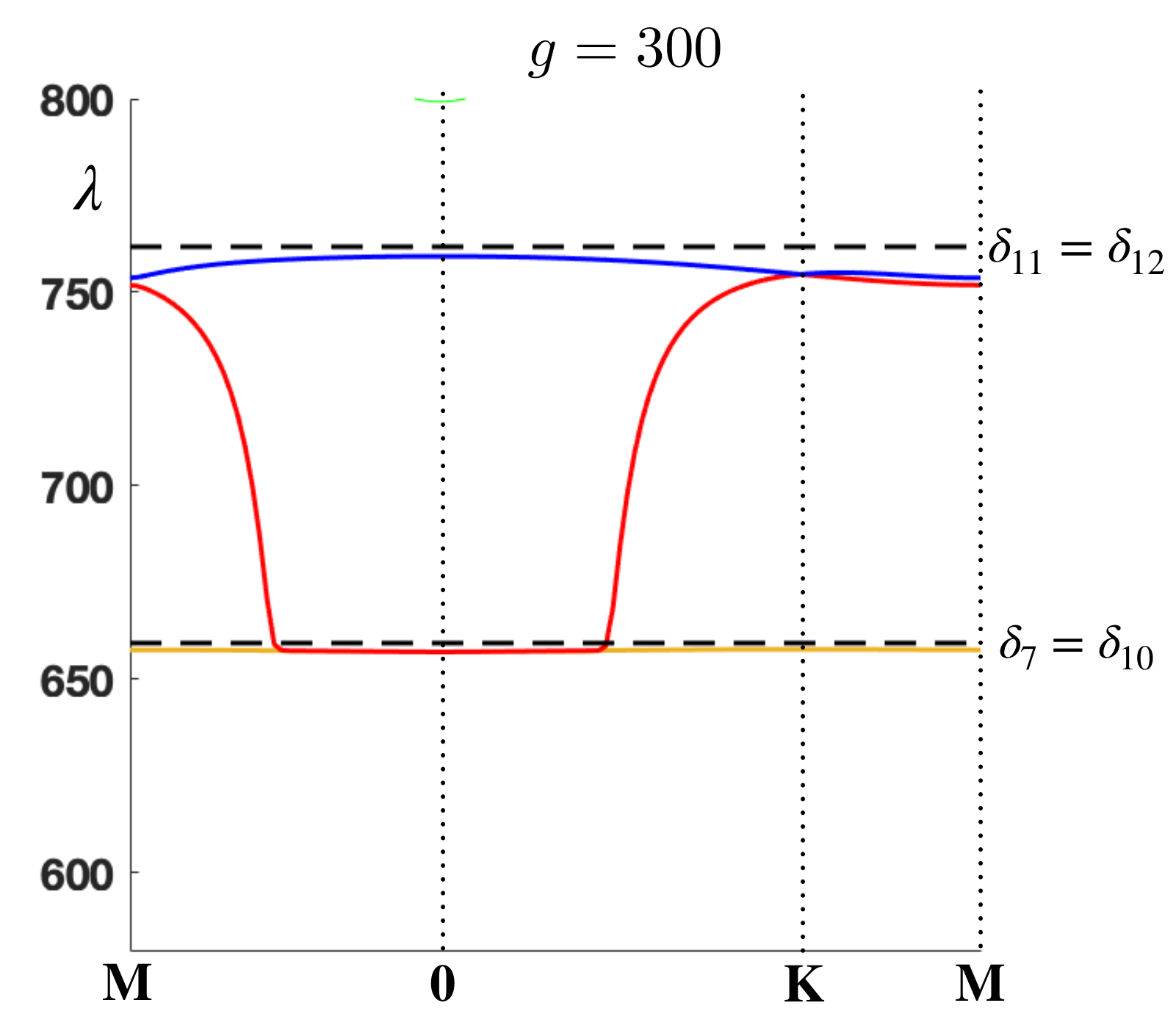}
\end{center}
\caption{ Dispersion maps $\bk\mapsto \lambda_{10}(\bk;g)$, $\lambda_{11}(\bk;g)$, $\lambda_{12}(\bk;g)$,  $\lambda_{13}(\bk;g)$  for the case of two disks $\cell^A$ and $\cell^B$  of radius $R_0=0.2$ per cell. The four maps  are plotted along the contour of ${\bf M}$-$\bf{0}$-${\bf K}$-${\bf M}$ of the  symmetry-reduced Brillouin zone in red (right) for the indicated values of the contrast $\aspar$.} 
\label{fig.3eigenvalbigsecondzero}
\end{figure}

We continue our numerical investigations with the medium described in Section \ref{numerics} with the mesh of $P_2$ periodic finite elements, used to obtain the results of  Figure  \ref{fig.3eigenvalbig}. We are interested here
in the energy bands associated with the Dirichlet eigenvalue $\tilde{\delta}_6=\left(z_{0,2}/R_0\right)^2$ where $z_{0,2}$ is the second positive zero of the Bessel function $J_0(z)$; using relation \ref{eq.speccirc}, the index $n=6$ can be read off Figure \ref{fig.zerosBessel} since the zeros represented by green crosses (resp. red crosses) are associated with eigenvalues of  the operator $\LaplDcellA$ of multiplicity $1$ (resp. 2). By  Proposition \ref{prop.higherband}, since $\tilde{\delta}_6$  is associated with a positive zero of $J_0(z)$, it satisfies the spectral isolation condition $\condS$. Hence, all the results from Sections \ref{sec.condP} to \ref{sec-asympresult} apply to the case  $n=6$, and thus to the associated dispersion surfaces  $\bk\mapsto\lambda_{11}(\aspar;\bk )$ and $\bk\mapsto \lambda_{12}(\aspar;\bk )$.

Figure \ref{fig.3eigenvalbigsecondzero} displays the  dispersion surfaces $\bk\mapsto\lambda_j(g;\bk)$, $ j=10,\, 11,\, 12,\, 13$ over the boundary contour of a symmetry-reduced  Brillouin zone. 
The Dirac point is situated at $\bK$ between the $12-$th and $13-$th bands for $\aspar=60$ and between the $11^{th}$ and $12^{th}$  bands for $\aspar=20$ and for the large contrast values $\aspar=100$ and $300$, 
as predicted by Theorem \ref{th.degeneracy}.
A (transitional) triple degeneracy occurs at $\bK$ for  $\aspar\approx37.3$ and $\aspar\approx 84.7$.
Moreover, for large $\aspar$, as predicted by Theorem \ref{thm.convunifbands} (applied with $n=6$), one observes that:\\ [2pt]
(a) the $11^{th}$ band does not ``become flat''. Indeed, it converges  away from $\bk=0$ to $\tilde{\delta}_6=\delta_{11}=\delta_{12}$ and at $\bk=0$, numerically, one notes that it  converges to $\tilde{\delta}_4=\delta_{7}=\tilde{\delta}_5=\delta_{10}<\delta_{11}$. 
Thus, one can conjecture that  $\nu_{11}=\delta_{10}$ in  Theorem \ref{thm.convunifbands}.  \\[2pt]
(b) the $12^{th}$ band converges  uniformly  to $\tilde{\delta}_6=\delta_{11}=\delta_{12}$ and this band ``becomes flat''.  \\[2pt]
(c) there exists  a gap between the the $12^{th}$th and $13^{th}$  bands. \\[2pt]
Furthermore, since $\delta_7=\delta_8=\delta_9=\delta_{10}$ (see Figure \ref{fig.zerosBessel}), one has by the interlacing property \eqref{eq.interlacing}: $\nu_{10}=\delta_{10}$. Thus,  as predicted  by Proposition \ref{prop.convunif}, the $10^{th}$ band converges uniformly to $\tilde{\delta}_5=\delta_{10}$.
\begin{figure}[h!]
\begin{center}
\includegraphics[width=7.25cm]{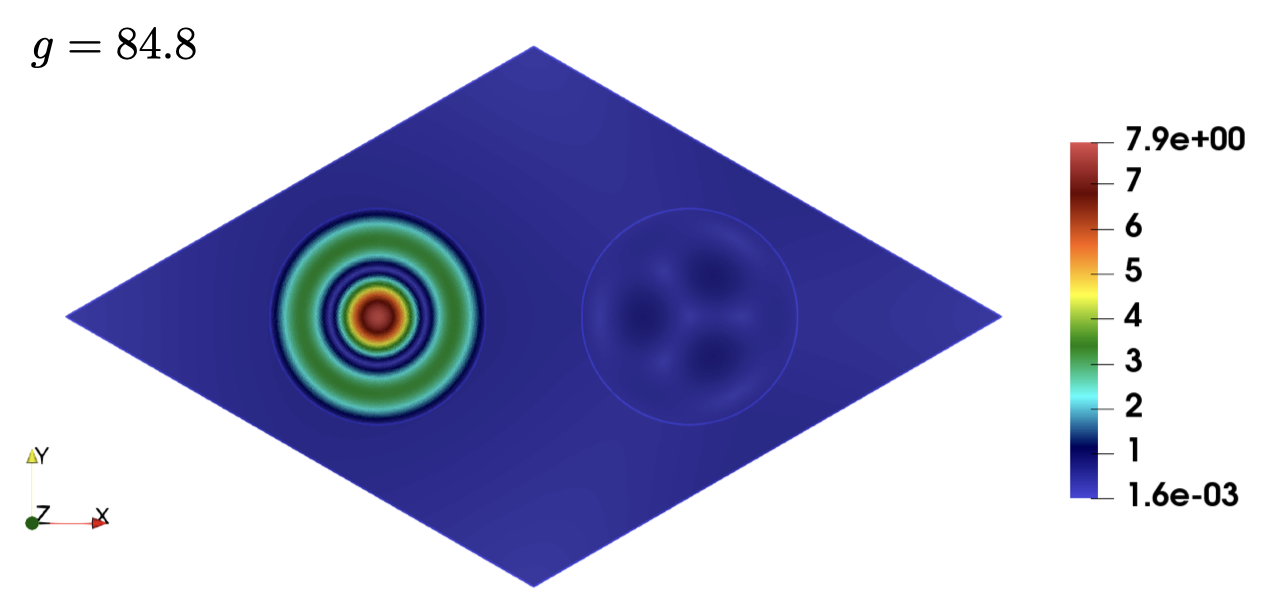}
\includegraphics[width=7.25cm]{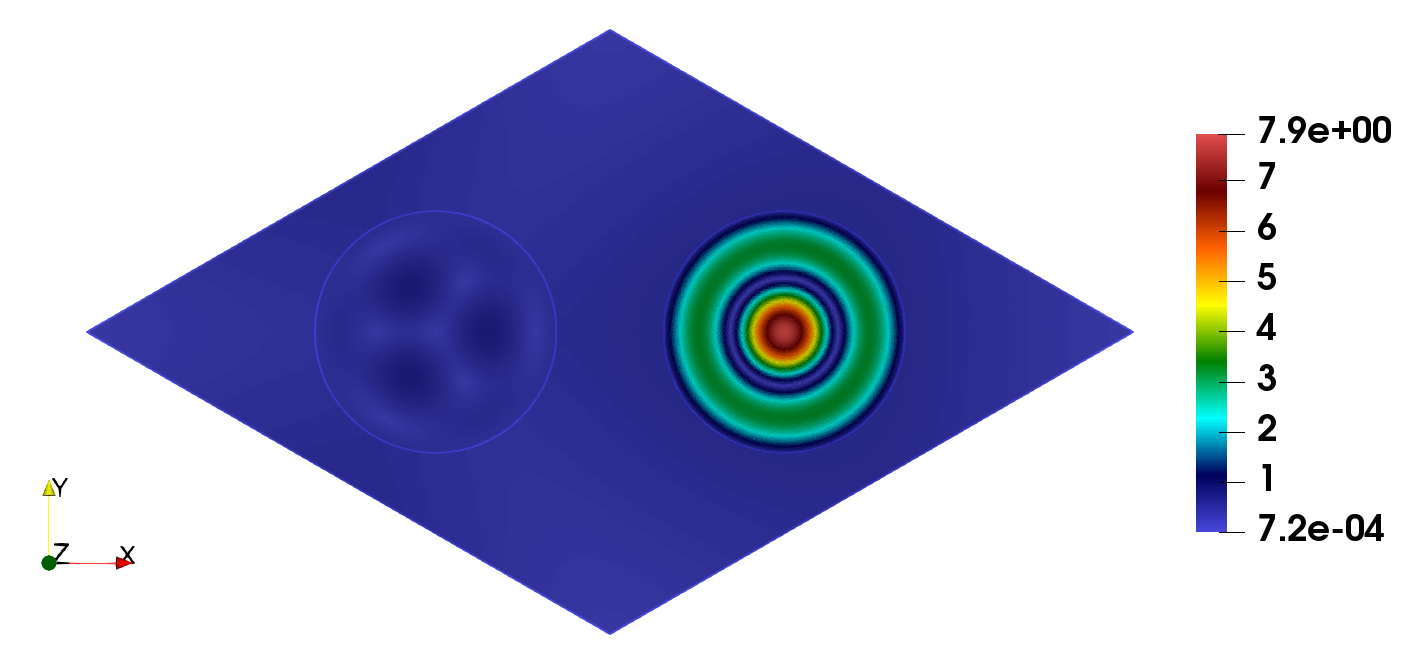}

\includegraphics[width=7.25cm]{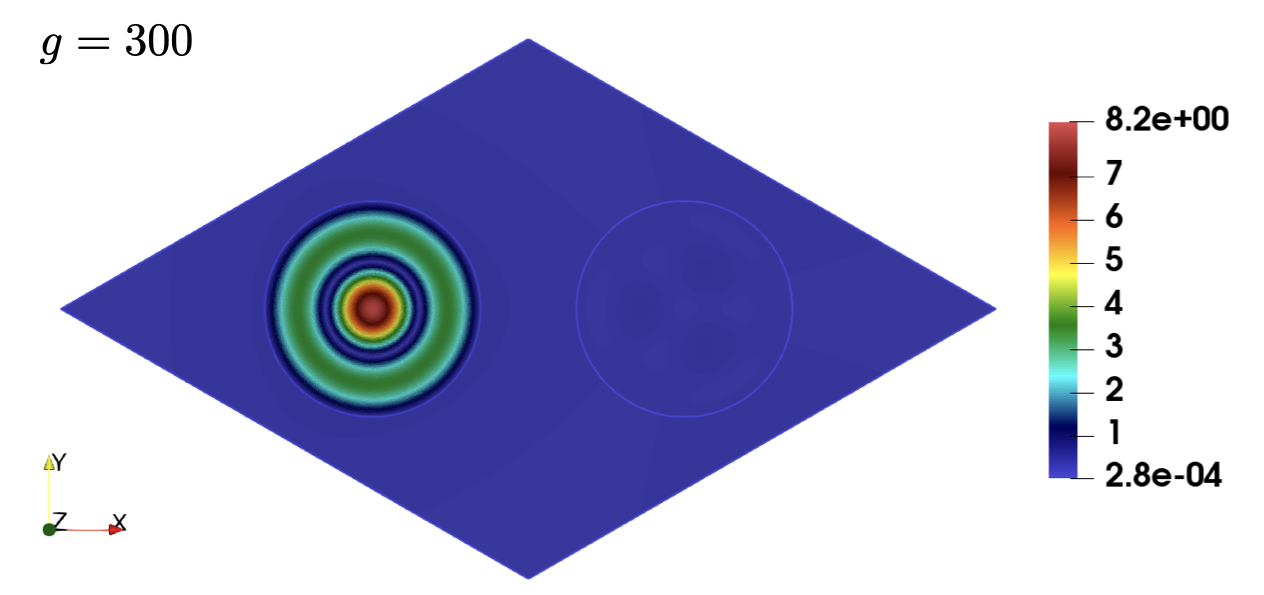}
\includegraphics[width=7.25cm]{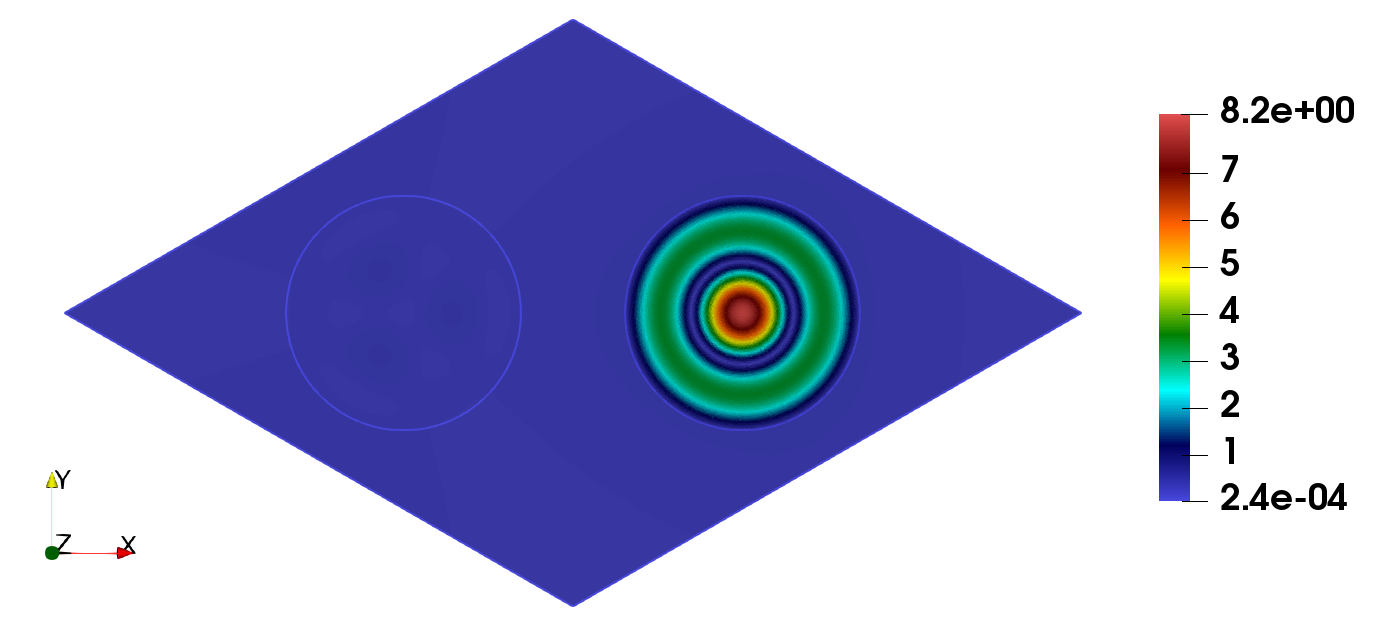}
\end{center}
\caption{$|\Phi_1(\aspar,\cdot)|$ (left)  and $|\Phi_2(\aspar,\cdot)|$  (right)  computed respectively with formula \eqref{eq.defphi1} and \eqref{eq.defphi2} for  $\lambda_D(\aspar)=\lambda_{11}(\aspar;\bK)=\lambda_{12}(\aspar;\bK)$  in the case of disk-shaped inclusions of radius $R_0=0.2$  and for $\aspar=84.8$ (first row) and $\aspar=300$ (second row).}
\label{fig.vecpsecondzero}
\end{figure}

In figure  \ref{fig.vecpsecondzero},  the contrast  values $\aspar=84.8$ and  $300$ are above the occurrence of the last triple (transitioning) point: $\aspar\approx 84.7$ (see Figure \ref{fig.3eigenvalbigsecondzero}). Thus, we are here in the high contrast regime predicted by Theorem \ref{th.degeneracy}:  $\lambda_D(\aspar)=\lambda_{11}(\aspar;\bK)=\lambda_{12}(\aspar;\bK)$ is of multiplicity $2$ and the Dirac points are situated between these two energy bands.  Here, we display $|\Phi_1(\aspar,\cdot)|$  and $|\Phi_2(\aspar,\cdot)|$, computed with formula \eqref{eq.defphi1} and \eqref{eq.defphi2}. 
One observes a localization of $\Phi_1$  (resp. $\Phi_2$) selectively in the  $A-$inclusions (resp. on the $B-$inclusions). Again, 
as predicted by Theorem \ref{th.degeneracy}, this localization is more pronounced as the contrast, $\aspar$, increases.

\section{Proof of Theorem \ref{Thm.weakLyapounovSchmidtred} \large {on conditions for the existence of Dirac points}}\label{sec.diracpointthoerem}
\subsection{Weak form of the Lyapunov-Schmidt / Schur complement reduction}
By symmetry arguments (Remark \ref{rem.Diracpointsequivalence}) it suffices to prove Theorem  \ref{Thm.weakLyapounovSchmidtred} for $\bK_*=\bK$.  Furthermore, relations  \eqref{eq.propcommutdim} of  Corollary \ref{cor.commut},  and assumptions 1 and 2 of  Theorem  \ref{Thm.weakLyapounovSchmidtred}  imply that $\lambda_D(\aspar)$  is an eigenvalue of multiplicity equal to $2$ of  $\bbA_{\aspar,\bK}$. Thus, there exists  $n\geq 1$ such that $\lambda_n(\aspar;\bK)= \lambda_{n+1}(\aspar;\bK)=\lambda_D(\aspar)$.
Hence,  part 1 of  Definition  \ref{def.Diracpoints} of Dirac points has been proved for $(\bK, \lambda_D(\aspar))$.
 It  remains to prove  part 2 of this definition, concerning the conical character of the dispersion surfaces
  near $(\bK,\lambda_D)$.  

We follow the framework developed in \cite{FW:12,FW:14} for the Schroedinger equation and in  \cite{LWZ:18} for divergence form  elliptic operators with smooth coefficients to derive the  asymptotic behavior of the two (Lipschitz) dispersion surfaces $\lambda_{n+1}(\aspar;\bK+\bkg)$ and  $\lambda_{n}(\aspar;\bK+\bkg)$ for $\bkg$ small, which touch at $\lambda_D(g)$   for $\bkg={\bf 0}$.
 Since $\bbA_{\aspar}$ has discontinuous coefficients,  we require a weak formulation of the Lyapunov-Schmidt / Schur reduction of  previous works \cite{FW:12,FW:14} that necessitates many technical 
adjustments.

We seek a non-trivial solution of the  eigenvalue problem:
\begin{equation}\label{eigenvalueprob1}
\bbA_{\aspar,\bK+\bkg} \, \Phi= \lambda\, \Phi,\qquad {D}(\bbA_{\aspar,\bK+\bkg})=\{ u\in H^1_{\bK+\bkg} \mid  \bbA_{\aspar,\bK+\bkg} u\in L^2_{\bK+\bkg}\}\, ,
\end{equation}
where $\lambda$ is near $\lambda_D$
and $\bkg$ is small perturbation of $\bK$. Such solutions $(\lambda,\Phi)$ depend on the asymptotic parameter $\aspar$ but since  $\aspar$ is fixed in this section, suppress this dependence. 
In particular,  as  $\Phi\in {D}(\bbA_{\aspar,\bK+\bkg})$, we have the transmission conditions at the boundary of the inclusions: $[\Phi]=0$ and $[\sigma_{\aspar} \, \nabla \phi \cdot  \bn ]=0$ .
We can reformulate the eigenvalue problem in the space of periodic functions (independent of $\bk$) by
 setting $\Phi(\bx)= \rme^{\rmi (\bkg+\bK) \cdot \bx} \phi(\bx)$. Then,   the eigenvalue problem \eqref{eigenvalueprob1} 
  can be expressed in terms of the operator:
  \[ 
   \tbbA_{\aspar,\bK+\bkg} =\rme^{-\rmi  (\bK+\bkg)\cdot \bx }  \bbA_{\aspar,\bK+\bkg} \,\rme^{\rmi  (\bK+\bkg)\cdot \bx } :=-\nabla_{\bK+\bkg}\cdot \sigma_{\aspar} \nabla_{\bK+\bkg},\quad (\nabla_{\bkg}=\nabla+\mathrm{i}\bkg)\nonumber
\]
acting on the domain:
\[ \rmD(\tbbA_{\aspar,\bK+\bkg})=\{ u\in H^1_{{\bf} 0}=H^1(\mathbb{R}^2/\Lambda) \mid  \tbbA_{\aspar,\bK+\bkg}  \, u \in  L^2_{{\bf} 0}= L^2(\mathbb{R}^2/\Lambda)\}\ .\]
As the operator $\tbbA_{\aspar,\bK+\bkg}$ is unitarily equivalent to $\bbA_{\aspar,\bK+\bkg}$, it is self-adjoint, has a compact resolvent on $L^2(\mathbb{R}^2/\Lambda)$ and the same sequence of eigenvalues (counted with multiplicity) as $\bbA_{\aspar,\bK+\bkg} $.
For $\phi\in \rmD(\tbbA_{\aspar,\bK+\bkg})$,  the transmission boundary conditions are:
\begin{equation}\label{eq.jumpcond}[\phi]=0 \mbox{ and } [\sigma_{\aspar} \, \nabla_{\bK+\bkg} \, \phi \cdot \bn ]=0.
\end{equation}
The eigenvalue problem \eqref{eigenvalueprob1} is equivalent to
\begin{align}\label{eq.eigeprobstrong}
\tbbA_{\aspar,\bK+\bkg} \, \phi &= \lambda\, \phi\ ,\quad \phi\in  \rmD(\tbbA_{\aspar,\bK+\bkg}).
\end{align}

In contrast to  the previous studies  \cite{FW:12,FW:14,LWZ:18}, here $\tbbA_{\aspar,\bK+\bkg}$ has  discontinuous coefficients,  $\sigma_{\aspar}$. In particular, one can not expand the operator $\tbbA_{\aspar,\bK+\bkg}$
as in \cite{FW:12,FW:14,LWZ:18}  since $\tbbA_{\aspar,\bK+\bkg}$ and $\tbbA_{\aspar,\bK}$ do not have the same domain. Indeed, at the boundary of the inclusions, the functions in $ \rmD(\tbbA_{\aspar,\bK+\bkg})$ satisfy the jump conditions  \eqref{eq.jumpcond} whereas functions  $\phi$ in  $ \rmD(\tbbA_{\aspar,\bK})$ satisfy  $[\phi]=0$ and $[\sigma_{\aspar} \, \nabla_{\bK}\, \phi \cdot \bn ]=0$. Nevertheless,
this expansion will be possible via a weak formulation of the Lyapunov-Schmidt / Schur complement reduction.

Multiplication  by a test function $\psi \in H^1(\mathbb{R}^2/\Lambda) $,  integrating over $\cell$, and applying
the divergence theorem, one can show that the eigenvalue problem  \eqref{eq.eigeprobstrong} is equivalent to the following weak formulation: \\  Find a function $\phi  \in H^1(\mathbb{R}^2/\Lambda)$ such that:
\begin{equation}\label{eq.weak}
\int_{\cell}\sigma_{\aspar} \nabla_{\bK+\bkg} \phi \cdot \overline{ \nabla_{\bK+\bkg} \psi}\, \rmd \bx =\ \lambda\ \int_{\cell} \phi \  \overline{\psi}\, \rmd \bx, \ \ \textrm{for all}\  \psi \in H^1(\mathbb{R}^2/\Lambda).
\end{equation}
For $|\bkg|$ small, we seek $(\lambda,\phi)$ with  $\lambda=\lambda_{D}+E_1,\ |E_1|\ll 1$, expecting that $E_1(\bkg)=o(1)$ as $\bkg \to 0$.
We rewrite the problem \eqref{eq.weak} as:
\begin{equation}\label{eq.weakform1}
s_{\bK}(\phi,\psi)=e_{E_1}(\phi,\psi)+b_{\bK,\bkg}(\phi,\psi)+c_{\bk}(\phi,\psi), \ \ \forall \psi \in H^1(\mathbb{R}^2/\Lambda),
\end{equation}
where for all  $u,v\in  H^1(\mathbb{R}^2/\Lambda) $, the continuous sesquilinear forms $s_{\bK}, b_{\bK,\bkg}, c_{\bkg}$  and $e_{E_1}$ are defined by:
\begin{equation}\label{eq.defses}
\begin{array}{lllll}
s_{\bK}(u,v)&=&\langle \bbS^{\bK} u, \overline{v} \rangle&=& \displaystyle \int_{\cell}\sigma_{\aspar} \nabla_{\bK} u \cdot \overline{ \nabla_{\bK} v}\,- \lambda_D\,  u \cdot \overline{  v} \, \rmd \bx, \\
 e_{E_1}(u,v) & = & \langle \bbE^{E_1} u, \overline{v} \rangle &=&   E_1\displaystyle  \int_{\cell} u\cdot  \overline{v} \, \rmd \bx,\\
 b_{\bK,\bkg}(u,v)&=& \langle \bbB^{\bK,\bkg} u, \overline{v} \rangle & =& \displaystyle -\int_{\cell} \sigma_{\aspar} \left[ \rmi \, \bkg u \cdot \overline{\nabla_{\bK} v}+ \nabla_{\bK} u \cdot \overline{ \rmi \, \bkg \, v} \right] \rmd \bx,\\
 c_{\bkg}(u,v) &= & \langle \bbC^{\bkg} u, \overline{v} \rangle& =&\displaystyle  -  |\bkg|^2 \int_{\cell} \sigma_{\aspar} u\cdot  \overline{v} \, \rmd \bx, 
\end{array}
\end{equation}
and $ \langle \cdot, \cdot \rangle $ denotes the duality product between $H^1(\mathbb{R}^2/\Lambda)$ and its dual $H^1(\mathbb{R}^2/\Lambda)^{*}$. The bounded operators $ \bbS^{\bK}$, $\bbE^{E_1}$, $\bbB^{\bK,\bkg}$ and $\bbC^{\bkg} \in \mathcal{B}(H^1(\mathbb{R}^2/\Lambda), H^1(\mathbb{R}^2/\Lambda)^{*})$ are associated to the different continuous sesqulinear forms by the relation \eqref{eq.defses}.
Since $\aspar$ is fixed, we omit the dependence  on  $\aspar$ of these operators and their associated continuous sesquilinear forms.
So the weak formualtion \eqref{eq.weakform1} is equivalent to the following linear equation:
\begin{equation}\label{eq.dual}
\bbS^{\bK} \phi= \bbB^{\bK,\bkg} \phi+ \bbC^{\bkg} \phi+ \bbE^{E_1} \phi
\end{equation}
valued in $H^1(\mathbb{R}^2/\Lambda)^*$ with unknown $\phi\in H^1(\mathbb{R}^2/\Lambda)$.

Let  
\[ \phi_i=\rme^{-i \bK \cdot \bx} \Phi_i,\quad i=1,2,\]
 where $\{\Phi_1,\Phi_2\}$ is the orthnormal basis for the $L^2_{\bK}$
  eigenspace for the eigenvalue $\lambda_D$ in
  Theorem \ref{Thm.weakLyapounovSchmidtred}.
Thus $\{\phi_1,\phi_2\}$ is an orthonormal basis for the $L^2(\mathbb{R}^2/\Lambda)$ of  $\operatorname{ker}(\tbbA_{\aspar,\bK}-\lambda_D \mathrm{I}d)$, i.e. the eigenspace for
 \eqref{eq.eigeprobstrong} with $\bkg=0$. Let $\bbP_{\parallel}$ denote the orthogonal
projection of $L^2(\mathbb{R}^2/\Lambda)$ onto $\mathcal{V}_0=\mathrm{span}\{\phi_1,\phi_2\}$:
$$
\bbP_{\parallel}f=(f,\phi_1)_{ L^2(\mathbb{R}^2/\Lambda)} \phi_1+ (f,\phi_2)_{ L^2(\mathbb{R}^2/\Lambda)} \,\phi_2, \ \mbox{ for any } f\in L^2(\mathbb{R}^2/\Lambda)
$$
and  $\bbP_{\perp}=\mathbb{I}-\bbP_{\parallel}$.
We seek a solution  $(E_1,\phi)$ of \eqref{eq.dual} with
$$
\phi=\phi^{(0)}+\phi^{(1)} 
$$
where 
\begin{equation}\label{eq.defphi0}
\phi^{(0)}=\bbP_{\parallel} \phi =\alpha \, \phi_1+ \beta \,\phi_2 \in H^1(\mathbb{R}^2/\Lambda) \ \ \mbox{  and } \ \ \phi^{(1)}=  \bbP_{\perp} \phi    \in H^1(\mathbb{R}^2/\Lambda).
\end{equation}
Thus our eventual goal is to construct, for all $\bkg$ small: 
$\alpha=(\phi, \phi_1)_{L^2(\mathbb{R}^2/\Lambda)}$, $\beta=(\phi, \phi_2)_{L^2(\mathbb{R}^2/\Lambda)}$,  $\phi^{(1)}$ and $E_1$ such that \eqref{eq.dual}  holds.  

We first remark that the  restriction to $H^1(\mathbb{R}^2/\Lambda) $ of $ \bbP_{\parallel}$ and $\bbP_{\perp}$: $ \bbP_{\parallel}: H^1(\mathbb{R}^2/\Lambda)\to \mathcal{V}_0\subset H^1(\mathbb{R}^2/\Lambda)$ and $\bbP_{\perp}: H^1(\mathbb{R}^2/\Lambda)\to \mathcal{V}_1:=\bbP_{\perp}H^1(\mathbb{R}^2/\Lambda) \subset H^1(\mathbb{R}^2/\Lambda)$ belong to $B(H^1(\mathbb{R}^2/\Lambda))$. Therefore, $\mathcal{V}_0$ and $\mathcal{V}_1$  are closed subspaces of $H^1(\mathbb{R}^2/\Lambda)$ and we have the (non-orthogonal) decomposition $H^1(\mathbb{R}^2/\Lambda)=\mathcal{V}_0\oplus \mathcal{V}_1$.
Thus, for  test functions $\psi\in H^1(\mathbb{R}^2/\Lambda)$ we write: $\psi=\psi^{(0)}+\psi^{(1)}$ where $\psi^{(0)}= \bbP_{\parallel} \psi$ and $\psi^{(1)}= \bbP_{\perp} \psi$.  Furthermore, any operator $\mathbb{M}:H^1(\mathbb{R}^2/\Lambda)\to H^1(\mathbb{R}^2/\Lambda)^*$ is equivalent to an operator $\mathbb{M}: \mathcal{V}_0\oplus \mathcal{V}_1\to \mathcal{V}_0^*\oplus \mathcal{V}_1^*$,  expressed  in the
equivalent block form:
\begin{equation*}
\mathbb{M} \begin{pmatrix} \phi^{(0)}\\   \phi^{(1)}\end{pmatrix} \equiv \begin{pmatrix} \mathbb{M}_{00} & \mathbb{M}_{01}\\ \mathbb{M}_{10} & \mathbb{M}_{11}\end{pmatrix} \begin{pmatrix} \phi^{(0)}\\   \phi^{(1)}\end{pmatrix}  ,
\label{Tform}\end{equation*}
where $\mathbb{M}_{ij}\in  \mathcal{B}(\mathcal{V}_j, \mathcal{V}_i^{*})$ for $i,j=0,1$. We next express 
 \eqref{eq.dual} in the block form $\mathbb{M}(E_1,\bkg)\phi=0$.

First, since $\phi^{(0)}\in \rmD(\bbA_{\bK})$, one can apply the divergence theorem and use that $\phi^{(0)}$ is an eigenfunction of $\tbbA_{\bK}$ for the eigenvalue $E_D$ to obtain:
\begin{equation*}
s_{\bK}(\phi^{(0)},\psi^{(1)})=\langle \bbS^{\bK} \phi^{(0)}, \overline{\psi^{(1)}}\rangle=\int_{\cell}[\nabla_{\bK} \cdot \sigma_{\aspar} \,\nabla_{\bK}\, \phi^{(0)}-E_D \,\phi^{(0)}] \cdot \overline{\psi^{(1)}} \, \rmd \bx=0.
\end{equation*}
It implies that $\bbS_{01}=0$.
In the same way, we find $\bbS_{10}=0$ and $\bbS_{00}=0$. It is also straightforward to see  that: $ \bbE^{E_1}_{01}=0$ and $\bbE^{E_1}_{10}=0$ since $\mathcal{V}_0$ and $\mathcal{V}_1$ are orthogonal subspaces in $L_2(\mathbb{R}^2/\Lambda)$.
Hence the  weak formulation \eqref{eq.dual} of the eigenvalue problem for $\mathbb{A}_{g,\bK+\bkg}$ is equivalent to the system:
\begin{align}
 \begin{pmatrix} \bbE^{E_1}_{00}+\bbB^{\bK,\bkg}_{00}+  \bbC^{\bkg}_{00} & \bbB^{\bK,\bkg}_{01}+ \bbC^{\bkg}_{01} \\
  \bbB^{\bK,\bkg}_{10}+\bbC^{\bkg}_{10}  &   - \bbS^{\bK}_{11} + \bbE^{E_1}_{11}+\bbB^{\bK,\bkg}_{11}+  \bbC^{\bkg}_{11}\end{pmatrix} \begin{pmatrix} \phi^{(0)} \\ \phi^{(1)} \end{pmatrix} \ &=\ \begin{pmatrix} 0 \\ 0 \end{pmatrix}  
  \label{eq.systweak}.
\end{align}
The linear eigenvalue problem \eqref{eq.systweak} is to be solved, for $\bkg$ small,  for $\phi^{(0)}, \phi^{(1)}$ and $E_1$. 
 We shall solve \eqref{eq.systweak} by a Schur complement / Lyapunov-Schmidt reduction strategy. Namely, we first solve 
  the second equation in \eqref{eq.systweak} for $\phi^{(1)}$ as a function of the two parameters $\alpha$ and $\beta$, which specify 
   $\phi^{(0)}=\alpha\phi_1+\beta\phi_2\in \mathcal{V}_0$. We then substitute $\phi^{(1)}[\alpha,\beta]$ into the first equation in \eqref{eq.systweak} to obtain a  two-dimensional homogeneous system of equations  
   $\mathcal{M}(E_1,\kappa;\aspar)\begin{pmatrix}\alpha \\ \beta\end{pmatrix}=\begin{pmatrix}0\\ 0 \end{pmatrix}$, whose 
    $2\times2$ matrix depends nonlinearly on $E_1$ and $\bkg$. The equation
     $\det\mathcal{M}(E_1,\kappa;\aspar)=0$ defines the two dispersion surfaces touching at $(\bK,\lambda_D)$ since it corresponds to solution $E_1=0$ and $ \bkg={\bf 0}$.
    We will show that there exists $\kappa_0\neq 0$  such that for all $|\bkg|<\kappa_0$ (and $\aspar>0$), we can  solve for  $E_1(\bkg;g)$.

\subsection{T-coercivity  and inversion of the operator $\bbS^{\bK}_{11}$}\label{sec.Tcoercivity}

We now proceed with the reduction step. The key is to show that  the $(2,2)$ entry of the operator in \eqref{eq.systweak} is invertible. The main step is to prove the invertiblity of $\bbS^{\bK}_{11}$. One approach is to apply a weak formulation of the Fredholm alternative \cite{Mclean:2000}.  Here, we present an alternative approach based on the notion of {\it T-coercivity}, an explicit reformulation of the $\operatorname{inf-sup}$ theory which generalizes coercivity \cite{Bon:12,Ciar:12}, and can be applied in cases where the standard Fredholm theory cannot be applied, 
e.g.  to invert divergence form elliptic operators with sign-changing coefficients in their principal part \cite{Bon:12}. T-coercivity has the additional appeal of simplicity; it reduces the invertibility of indefinite / non-coercive problems, {\it e.g.} for bands $n\ge2$ where the sesqulinear form is not coercive, to an application of the Lax-Milgram theorem.
The  T-coercivity  approach adopt here has been applied in the context of well-posdness and discretization of Helmholtz  operator  problems in bounded domains \cite{Ciar:12}.

We begin with a brief discussion of  the T-coercivity approach. Let $\mathcal{H}$ be a Hilbert space for which the complex conjugation  is an antiunitary involution. Consider the general  problem of  finding $u\in\mathcal{H}$ such that
\begin{equation}\label{eq.Tcoercivity}
a(u,v)=l(\overline{v}), \quad  \textrm{for all}\  v \in \mathcal{H}, 
\end{equation}
where $a$ and $l$ are respectively a continuous sesquilinear form $\mathcal{H}\times \mathcal{H}$ and a continuous linear form on $\mathcal{H}$. If the form $a$ is not coercive on $\mathcal{H}$, one can not  directly apply the Lax-Milgram theorem to produce  a unique solution $u\in\mathcal{H}$. In such cases, the idea of T-coercivity is to construct (if possible) a bounded isomorphism  $\mathbb{T}$ of $\mathcal{H}$ such that the sesquilinear form:
$a_T(u,v)=a(u,\mathbb{T}v)$ for $u,v \in \mathcal{H}$ is coercive.  One then applies the Lax-Milgram theorem to find a unique solution $u\in\mathcal{H}$ of
\begin{equation}
a_T(u,v)=l(\overline{\mathbb{T} v}),  \quad  \textrm{for all}\ v \in \mathcal{H}.
\label{Tvar}\end{equation}
This solution depends continuously on $l\in  \mathcal{H}^*$, since $\mathbb{T}$ is bounded. 
Furthermore since  $\mathbb{T}:\mathcal{H}\to\mathcal{H}$ is an isomorphism, \eqref{Tvar} is equivalent to \eqref{eq.Tcoercivity}, and thus \eqref{eq.Tcoercivity} has a unique solution.
This is further equivalent  to the bounded operator  $\bbA $ from $\mathcal{H}$ to its dual $\mathcal{H}^*$ defined by 
$\langle \bbA u, \overline{v}\rangle=a(u,v)$ for all $u,v\in \mathcal{H}$ (where $\langle \cdot,  \cdot \rangle $ stands for the duality product between $\mathcal{H}^*$ and $\mathcal{H}$) being invertible.

 To apply the T-coercivity in the present setting, decompose $L^2(\mathbb{R}^2)$ as the orthogonal sum:
\begin{equation}
 L^2(\mathbb{R}^2\setminus \Lambda) = \big(\mathcal{V}_{1,-}\oplus \mathcal{V}_0\big)\oplus \big(\mathcal{V}_{1,-}\oplus \mathcal{V}_0\big)^\perp\ ,
 \label{eq.orthog}
 \end{equation}
 where
 \begin{align*} 
& \mathcal{V}_{0}= \textrm{ ${\rm span}\{\phi_1,\phi_2\}$ }\ \textrm{and}\\
& \mathcal{V}_{1,-}= \textrm{span of a basis of eigenfunctions of the $n-1$ first eigenvalues of $\tilde{\bbA}(\aspar,\bK)$} .
\end{align*}
For   $n=1$,
   we set $\mathcal{V}_{1,-}=\{0\}$.
Associated with the orthogonal decomposition \eqref{eq.orthog} are $3$ projection operators: 
$\bbP_{\perp,-}:L^2(\R^2 /\Lambda)\to\mathcal{V}_{1,-}$, $\bbP_\parallel:  L^2(\R^2 /\Lambda) \to  \mathcal{V}_{0}$ and $ \bbP_{\perp,+}:L^2(\R^2 /\Lambda)\to  (\mathcal{V}_{1,-}\oplus \mathcal{V}_0)^\perp$ satisfying:
$ \bbP_{\perp,-} + \bbP_\parallel + \bbP_{\perp,+} = \mathbb{I}$ and  $ \bbP_{\perp,-}+ \bbP_{\perp,+}= \bbP_{\perp}$. 
The restriction of $\bbP_{\perp,-}$, $\bbP_\parallel$, $\bbP_{\perp,+}$  to $H_1(\mathbb{R}^2\setminus \Lambda)$ belongs the $B(H_1(\mathbb{R}^2\setminus \Lambda))$. We define $\mathcal{V}_{1,+} = \bbP_{\perp,+} H^1(\mathbb{R}^2 / \Lambda)$, a subspace of $\bbP_{\perp,+} L^2(\mathbb{R}^2 / \Lambda)$. This yields the following decomposition of $\mathcal{V}_1$:
\begin{equation*} \mathcal{V}_1=\mathcal{V}_{1,-} \oplus \mathcal{V}_{1,+} ,\label{V1-sum}\end{equation*}
as a non-orthogonal sum of closed subspaces  endowed  with the $H^1-$ inner product and associated norm.
  On the Hilbert space $\mathcal{V}_1$, we define a bounded isomorphism (an involution) $\mathbb{T}$ as follows.  For any $u\in \mathcal{V}_1$, we write the unique  decomposition $u=u_{+}+u_{-}$, where $u_+\in \mathcal{V}_{1,+}$ and $u_{-}\in \mathcal{V}_{1,-}$. Then we define
\begin{align*}
 \mathbb{T} u:= u_{+}-u_{-} .
\label{eq.Tisomorphism}
\end{align*}

\begin{lemma}\label{lem.Tcoerc}
The  operator $\bbS^{\bK}_{11}\in \mathcal{B}\big(\mathcal{V}_1, (\mathcal{V}_1)^{*}\big)$ has a bounded inverse $(\bbS^{\bK}_{11})^{-1}\in \mathcal{B}\big((\mathcal{V}_1)^{*}, \mathcal{V}_1\big)$.
\end{lemma}
\begin{proof}
We show that the invertibility of $\bbS^{\bK}_{11}$ follows from the T-coercivity of the sesquilinear form $s_{\bK}$ on $\mathcal{V}_1$. For the norm and inner product on  $ L^2(\mathbb{R}^2/\Lambda)$ we omit the subscript indicating the function space.
For all $u=u_{+}+u_{-} \ \in \mathcal{V}_1$ with $u_{\pm}\in \mathcal{V}_{1,\pm}$ we have
\begin{eqnarray}\label{eq.coerciv1}
(s_{\bK} \, u, \mathbb{T}u)
&=& (\tilde{\mathbb{A}}_{\bK}^{\frac{1}{2}}u, \tilde{\mathbb{A}}_{\bK}^{\frac{1}{2}} \mathbb{T}u)- \lambda_D (u,\mathbb{T}u) \nonumber \\
&=&\big(\tilde{\mathbb{A}}_{\bK}^{\frac{1}{2}}(u_++u_-), \tilde{\mathbb{A}}_{\bK}^{\frac{1}{2}} (u_+-u_-)\big) -\lambda_D\big(u_++u_-,u_+-u_-\big) \nonumber \\
&=& \big(\tilde{\mathbb{A}}_{\bK}^{\frac{1}{2}}u_+, \tilde{\mathbb{A}}_{\bK}^{\frac{1}{2}} u_+\big)- \lambda_D\|u_+\|^2+\lambda_D \|u_-\|^2- \big(\tilde{\mathbb{A}}_{\bK}^{\frac{1}{2}}u_-, \tilde{\mathbb{A}}_{\bK}^{\frac{1}{2}} u_-\big) .
\end{eqnarray}
The last equality follows  from the orthogonality of  $u_+$ and $u_-$  in $L^2(\bbR^2\setminus \Lambda)$  and  that  $\tilde{\mathbb{A}}_{\bK}^{\frac{1}{2}} u_-$  and  $\tilde{\mathbb{A}}_{\bK}^{\frac{1}{2}} u_+$ are also orthogonal in $L^2(\bbR^2\setminus \Lambda)$ (as an immediate consequence of the spectral theorem).

Since $ \lambda_D(\aspar)=\lambda_{n+1}(\aspar;\bK)= \lambda_{n}(\aspar;\bK)$ has multiplicity $2$, one can choose $\eta>0$ such that $(1-\eta) \, \lambda_{n+2}(\aspar;\bK) -\lambda_D>0$ for $n\geq 1$ and $\lambda_D- (1+\eta) \, \lambda_{n-1}(\aspar;\bK) >0$ if $n>1$. 
Thus, one gets:
\begin{eqnarray}\label{eq.coerciv2}
\big(\tilde{\mathbb{A}}_{\bK}^{\frac{1}{2}}u_+, \tilde{\mathbb{A}}_{\bK}^{\frac{1}{2}} u_+\big)- \lambda_D\|u_+\|^2&=&  \eta\, \big(\tilde{\mathbb{A}}_{\bK}^{\frac{1}{2}}u_+, \tilde{\mathbb{A}}_{\bK}^{\frac{1}{2}} u_+\big)+(1-\eta ) \big(\tilde{\mathbb{A}}_{\bK}^{\frac{1}{2}}u_+, \tilde{\mathbb{A}}_{\bK}^{\frac{1}{2}} u_+\big)-\lambda_D\|u_+\|^2 \nonumber  \\
&\geq&  \eta\, \big(\tilde{\mathbb{A}}_{\bK}^{\frac{1}{2}}u_+, \tilde{\mathbb{A}}_{\bK}^{\frac{1}{2}} u_+\big) +[(1-\eta) \, \lambda_{n+2}(\aspar;\bK)-\lambda_D]  \, \|u_+\|^2 \nonumber  \\
&\geq &  \eta\, \big(\tilde{\mathbb{A}}_{\bK}^{\frac{1}{2}}u_+, \tilde{\mathbb{A}}_{\bK}^{\frac{1}{2}} u_+\big),
\end{eqnarray}
where, in the last inequality, we used  that $(\tilde{\mathbb{A}}_{\bK}^{\frac{1}{2}}u_+, \tilde{\mathbb{A}}_{\bK}^{\frac{1}{2}} u_+)\geq \lambda_{n+2}(\aspar;\bK)   \|u_+\|^2$. This is an immediate consequence of the definition of the orthogonal projector $\bbP_{\perp,+}$, and the expression for $\tilde{\bbA}_{\bK}^{\frac{1}{2}}u_+$ in an orthonormal basis  $L^2(\bbR^2\setminus \Lambda)$ with respect to which  $\tilde{\bbA}_{\bK}$ is diagonal.
Similarly,  for the case  $n>1$:  
\begin{eqnarray}\label{eq.coerciv3}
\lambda_D\|u_-\|^2- \big(\tilde{\mathbb{A}}_{\bK}^{\frac{1}{2}}u_-, \tilde{\mathbb{A}}_{\bK}^{\frac{1}{2}} u_-\big)&=&\lambda_D \|u_-\|^2-(1+\eta) \big(\tilde{\mathbb{A}}_{\bK}^{\frac{1}{2}}u_-, \tilde{\mathbb{A}}_{\bK}^{\frac{1}{2}} u_-\big)+
\eta  \, \big(\tilde{\mathbb{A}}_{\bK}^{\frac{1}{2}}u_-, \tilde{\mathbb{A}}_{\bK}^{\frac{1}{2}} u_-\big) \nonumber  \\
&\geq&  [\lambda_D -(1+\eta) \lambda_{n-1}(\aspar;\bK)] \, \|u_-\|^2  +  \eta\, \big(\tilde{\mathbb{A}}_{\bK}^{\frac{1}{2}}u_-, \tilde{\mathbb{A}}_{\bK}^{\frac{1}{2}} u_-\big) \nonumber \\
&\geq &\eta\, \big(\tilde{\mathbb{A}}_{\bK}^{\frac{1}{2}}u_-, \tilde{\mathbb{A}}_{\bK}^{\frac{1}{2}} u_-\big)\ .
\end{eqnarray}
In the last inequality, we used that  $(\tilde{\mathbb{A}}_{\bK}^{\frac{1}{2}}u_-, \tilde{\mathbb{A}}_{\bK}^{\frac{1}{2}} u_-)\leq \lambda_{n-1}(\aspar;\bK)] \, \|u_-\|^2 $, a consequence of the definition of  $\bbP_{\perp,-}$.
Combining \eqref{eq.coerciv1}, \eqref{eq.coerciv2} and \eqref{eq.coerciv3} (using  that $u=u_+$ and $u_-=0$ for the particular case $n=1$), we have:
\begin{eqnarray*}
(s_{\bK}u, \mathbb{T}u)&\geq &  \eta \Big[ \big(\tilde{\mathbb{A}}_{\bK}^{\frac{1}{2}}u_+, \tilde{\mathbb{A}}_{\bK}^{\frac{1}{2}} u_+\big) + \big(\tilde{\mathbb{A}}_{\bK}^{\frac{1}{2}}u_-, \tilde{\mathbb{A}}_{\bK}^{\frac{1}{2}} u_-\big)\Big]\\
&\geq & \eta \, \big(\tilde{\mathbb{A}}_{\bK}^{\frac{1}{2}}u, \tilde{\mathbb{A}}_{\bK}^{\frac{1}{2}} u\big)\\
&\geq &  \eta    \min (1,\aspar) \, \|\nabla_{\bK} u\|^2 \\
&\geq& C(\bK) \, \eta  \min (1,\aspar)  \, \|u \|^2_{H^1(\mathbb{R}^2\setminus \Lambda)}.
\end{eqnarray*}
Therefore, the continuous sesqulinear form $s_{\bK}$ is T-coercive on $\mathcal{V}_1$. Hence, by the Lax Milgram lemma  $\bbS^{\bK}_{11}\in \mathcal{B}\big(\mathcal{V}_1, (\mathcal{V}_1)^{*}\big)$ has a bounded inverse $(\bbS^{\bK}_{11})^{-1}\in \mathcal{B}\big((\mathcal{V}_1)^{*}, \mathcal{V}_1\big)$.
\end{proof}
\subsection{Reduction to the determinant of a $2\times 2$ matrix}
Using the invertibility of $\bbS^{\bK}_{11}$ (Lemma \ref{lem.Tcoerc}), we rewrite the second equation in \eqref{eq.systweak} as
\begin{equation}\label{eq.invertphi1}
\left(\mathbb{I}-  \Xi(\bkg,E_1)\right) \phi^{(1)}= (\bbS^{\bK}_{11})^{-1} [\bbB^{\bK,\bkg}_{10}+\bbC^{\bkg}_{10} ]\phi^{(0)},
\end{equation}
where  $\Xi(\bkg,E_1):=(\bbS^{\bK}_{11})^{-1}\,[\bbE^{E_1}_{11}+\bbB^{\bK,\bkg}_{11}+  \bbC^{\bkg}_{11}]\in \mathcal{B}(\mathcal{V}_1)$.
From the expression of the sesquilinear forms $e_{E_1}$, $b_{\bK,\bkg}$ and $ c_{\bkg}$, we see that  $\bbE^{E_1}$, $\bbB^{\bK,\bkg}$ and   $\bbC^{\bkg}$  depend  respectively linearly in $E_1$, linearly in $\bkg$ and quadratically in $\bkg$. Since  $(\bbS^{\bK}_{11})^{-1}$ is bounded and independent of $\bkg$, it follows that $\Xi(\bkg,E_1)=O(|\bkg|+|E_1|)$ tends to zero as $(\bkg,E_1) \to 0\in\mathbb{R}^3$.
Thus, for $|\bkg|+|E_1|$ sufficiently small, one has $\| \Xi(\bkg,E_1)\|<1$ and therefore that $\mathbb{I}-\Xi(\bkg,E_1)$ is invertible.
Hence, 
\begin{equation}\label{eq.opsolphi1}
c[\bk,E_1]=\big(\mathbb{I}-\Xi(\bkg,E_1)\big)^{-1}  (\bbS^{\bK}_{11})^{-1} [\bbB^{\bK,\bkg}_{10}+\bbC^{\bkg}_{10} ] \in \mathcal{B}(\mathcal{V}_0,\mathcal{V}_1) 
\end{equation}
is well-defined for $|\bkg|+|E_1|$ sufficiently small. 
Relations  \eqref{eq.invertphi1}, \eqref{eq.opsolphi1}  and \eqref{eq.defphi0} imply, for $|\bkg|+|E_1|$ sufficiently small:
\begin{equation}\label{eq.solvphi_1}
 \phi^{(1)}=  c[\bkg,E_1]\,\phi_1  \, \alpha+  c[\bkg,E_1]\,\phi_2  \, \beta .
\end{equation}
The expression  \eqref{eq.opsolphi1}  can be expanded in a Neumann series. Moreover,  $\bbE^{E_1}$, $\bbB^{\bK,\bkg}$ and $\bbC^{\bkg}$ depend, respectively, linearly in $E_1$, linearly in $\bkg$ and quadratically in $\bkg$. Thus, for $(E_1,\bkg)$  in a  neighborhood, $\mathcal{U}$, of the origin: $(E_1,\bkg) \to c[\bkg,E_1]\ \phi_{j}$,  for $j=1,2$, are analytic as a mapping from  $(E_1,\bkg)\in \mathcal{U}$   into $\mathcal{V}_1$ endowed with the $H^1-$norm (for a discussion of the composition and product of analytic functions defined and valued on a Banach space, see e.g. \cite{Buf:03}).
 We also have
\begin{equation}\label{eq.bound1}
\|c[\bkg,E_1]\ \phi_{j}\|_{H^{1}(\mathbb{R}^2/\Lambda)}\leq  C  |\bkg|\  \mbox{ for }j=1,2,  \mbox{ $(E_1,\bkg) \in \mathcal{U}$  and  some $C>0$.}
\end{equation}
Now substituting the expression of $\phi^{(1)}$ given by \eqref{eq.solvphi_1} into the first equation of \eqref{eq.systweak} yields  a closed equation for $\alpha, \beta$, depending on $(E_1,\bkg)$:
$$
0 =  [ \bbE^{E_1}_{00}+\bbB^{\bK,\bkg}_{00}+  \bbC^{\bkg}_{00}]\ ( \phi_1\, \alpha +\beta \, \phi_2)+ [\bbB^{\bK,\bkg}_{01}+ \bbC^{\bkg}_{01}] \,  \big( c[\bkg,E_1]\,\phi_1  \, \alpha+  c[\bkg,E_1]\,\phi_2 \, \beta \big) 
$$
valued in $(\mathcal{V}_0)^{*}$.
Using that $\{\phi_1,\phi_2\}$ is an orthonormal basis of $\mathcal{V}_0$, the latter equation may be expressed equivalently as a system of 
two homogeneous linear equations for  $\alpha$ and  $\beta$:
\begin{equation}\label{eq.matsystalphabeta}
\mathcal{M}(\bkg,E_1) \begin{pmatrix} \alpha \\ \beta \end{pmatrix}=0,
\end{equation}
where 
\begin{align}\label{eq.defmatMkE1} 
\mathcal{M}(\bkg,E_1) &=E_1 \mathbb{I}_2 +\mathcal{M}_{A}(\bkg)+\mathcal{M}_{B}(\bkg,E_1),\\
\mathcal{M}_{A}(\bkg) &:= \begin{pmatrix}\langle \bbB^{\bK,\bkg}_{00} \phi_{1}, \overline{\phi_{1}}\rangle & \langle \bbB^{\bK,\bkg}_{00} \phi_2, \overline{\phi_{1}}\rangle \label{eq.defmatMA}\\[4pt]
\langle \bbB^{\bK,\bkg}_{00} \phi_{1}, \overline{\phi_2}\rangle & \langle \bbB^{\bK,\bkg}_{00} \phi_2, \overline{\phi_2}\rangle \\ 
  \end{pmatrix},\quad {\rm and}\\
 \mathcal{M}_{B}(\bkg,E_1) &:= \begin{pmatrix}\langle \big[\bbC^{\bkg}_{00} + (\bbB^{\bK,\bkg}_{01}+ \bbC^{\bkg}_{01})  c[\bkg,E_1] \big]\phi_{1}, \overline{\phi_1}\rangle & \hspace{-0.1cm}\langle \big[\bbC^{\bkg}_{00} + (\bbB^{\bK,\bkg}_{01}+ \bbC^{\bkg}_{01})  c[\bkg,E_1] \big]\phi_2, \overline{\phi_{1}} \rangle \\[6pt]
\langle \big[\bbC^{\bkg}_{00} + (\bbB^{\bK,\bkg}_{01}+ \bbC^{\bkg}_{01})  c[\bkg,E_1] \big]\phi_{1}, \overline{\phi_2} \rangle & \hspace{-0.1cm} \langle \big[\bbC^{\bkg}_{00} + (\bbB^{\bK,\bkg}_{01}+ \bbC^{\bkg}_{01})  c[\bkg,E_1] \big]\phi_2, \overline{\phi_2} \rangle \\ 
  \end{pmatrix}.
\label{eq.defmatMB}\end{align}
From  \eqref{eq.defmatMkE1}, the analyticity of $(E_1,\bkg) \mapsto c[\bkg,E_1]\ \phi_{j}$ for $j=1,2$ and the fact that $ \bbE^{E_1}$, $ \bbB^{\bK,\bkg}$ and   $ \bbC^{\bkg}$ depend  respectively linearly in $E_1$, linearly in $\bkg$ and quadratically in $\bkg$, it follows  that $(E_1,\bkg)\mapsto \mathcal{M}(\bkg,E_1)  $ is analytic in a neighborhood $\mathcal{U}$ of the origin in $\mathbb{R}^3$. 
Furthermore from \eqref{eq.defmatMkE1}, we have that
 $E_1\mapsto  E_1 \mathbb{I}_2$ is linear in $E_1$,  $\bkg\mapsto \mathcal{M}_{A}(\bkg)$ is linear in $\bkg$  and from  \eqref{eq.bound1} and \eqref{eq.defmatMB},    the matrix entries of $(\mathcal{M}_{B}(\cdot,\cdot))_{ij}$ satisfy,  for  $i,j=1,2$  and  all $(E_1,\bkg)  \in \mathcal{U}$:
\begin{equation*}\label{eq.bound2}
\mathcal{M}_{B}(\bkg,E_1)_{ij}=g_{ij}(E_1,\bkg)+ h_{ij}(\bkg)  \mbox{ with }  | g_{ij}(E_1,\bkg)| \leq C  |\bkg|^2 \mbox { and } \  | h_{ij}(\bkg)| \leq C |\bkg|^2\ .
\end{equation*}
Here, $C>0$ and $g_{ij}$ and $h_{ij}$ are analytic functions of $(E_1,\bkg)$ and $\bkg$.
The preceding arguments imply a characterization of the dispersion surfaces in a neighborhood of $(\bK,\lambda_D)$.
\begin{proposition}\label{prop.reduc}
For $(E_1,\bkg)\in \mathcal{U}$, a sufficiently small neighborhood of the origin in $\mathbb{R}^3$, 
$\lambda(\aspar;\bK+\bkg)=\lambda_D+E_1$ is an eigenvalue of the eigenvalue problem \eqref{eigenvalueprob1} if and only if 
\begin{equation}\label{eq.determinant}
\operatorname{det} \mathcal{M}(\bkg,E_1)=0.
\end{equation}
Here, the $2\times2$ complex-valued  matrix $\mathcal{M}(\bkg,E_1)$ is given by \eqref{eq.matsystalphabeta}, \eqref{eq.defmatMkE1}, \eqref{eq.defmatMA} and \eqref{eq.defmatMB}. Furthermore,
$(\bkg,E_1)\mapsto  \operatorname{det}   \mathcal{M}(\bkg,E_1)$ is analytic on  $\mathcal{U}$.
\end{proposition}

\subsection{Simplifications of the determinant using symmetries}
We expect that  for $|E_1|+|\bkg|$ small, the leading behavior of $E_1(\bkg,g)$ is given by the eigenvalues
 of the matrix $\mathcal{M}_A(\bkg)$. 
We now use  honeycomb symmetry to simplify  $\mathcal{M}_{A}(\bkg)$.

\begin{lemma}\label{Lem.simplificationdet}
\begin{enumerate}
\item For $\bkg \in \R^2$:
\begin{itemize}
\item $\langle \bbB^{\bK,\bkg}_{00} \phi_{i}, \overline{\phi_{i}}\rangle=0$ for $i=1,2$.
\item $\langle \bbB^{\bK,\bkg}_{00} \phi_{j}, \overline{\phi_{i}}\rangle=\overline{\langle \bbB^{\bK,\bkg}_{00} \phi_{i}, \overline{\phi_j}\rangle}=-2\mathrm{i} \,\bkg\, \cdot \int_{\cell}  \sigma_{\aspar}\Phi_{j} \, \overline{ \nabla \Phi_{i}} \, \rmd \bx \ $  for $i,j=1,2$ and $i\neq j$.
\end{itemize}
\item Let 
\begin{equation}\label{eq.fermveloc}
\tilde{v}_D= -\mathrm{i}  \int_{\cell} \sigma_{\aspar}\Phi_1 \, \overline{\nabla \Phi_2 }\, \rmd \bx \cdot \overline{(1,\rmi)}^{\top}.
\end{equation}
Then,  $v_D=|\tilde{v}_D|$, with $v_D$ defined in \eqref{eq.fermyveloc} and
one rewrites  $\mathcal{M}_{A}(\bkg)$ defined by \eqref{eq.defmatMA} as
$$
\mathcal{M}_{A}(\bkg)=\begin{pmatrix} 0 &  \tilde{v}_D\, (\bkg_1+ \rmi \bkg_2) \\
\overline{\tilde{v}_D}\, (\bkg_1- \rmi \bkg_2)  & 0
  \end{pmatrix}.
$$
\end{enumerate} 
\end{lemma} 
\begin{proof}
Using \eqref{eq.defses}, one has for $i,j=1,2$:
\begin{equation}\label{eq.conjcoeff}
\langle \bbB^{\bK,\bkg}_{00} \phi_j,\overline{ \phi_i}\rangle=-\int_{\cell} \sigma_{\aspar}  [\rmi \, \bkg  \, \phi_j\cdot \overline{ \nabla_{\bK}{\phi_i}} +\nabla_{\bK} \phi_j\cdot \overline{\rmi \, \bkg \, \phi_i}] \, \rmd \bx=\overline{\langle \bbB^{\bK,\bkg}_{00} \phi_i, \overline{\phi_j}\rangle}.
\end{equation}
Since $\phi_i=\Phi_{i} \rme^{-i \bK \cdot \bx}$ and $\nabla_{\bK} \phi_i=\nabla \Phi_{i} \rme^{-i \bK \cdot \bx}$ for $i,j=1,2$ we have:
\begin{equation}\label{eq.ceoffij}
\langle \bbB^{\bK,\bkg}_{00} \phi_j, \overline{\phi_i}\rangle=-\int_{\cell} \sigma_{\aspar
}  [\rmi \, \bkg  \, \Phi_{j}\cdot \overline{ \nabla \Phi_{i}}+\nabla \Phi_{j}\cdot \overline{\rmi \, \bkg  \,\Phi_{i}}] \, \rmd \bx.
\end{equation}
When $i= j\in \{1,2\}$, this latter expression simplifies to:
\begin{equation}\label{eq.coeffmatA11}
\langle \bbB^{\bK,\bkg}_{00} \phi_i, \overline{\phi_i}\rangle=2 \operatorname{Im}\big( \sigma_{g}  \Phi_{i} , \bkg \cdot \nabla \Phi_{i} \big)= 2\operatorname{Im} \Big(\bkg \cdot  \int_{\cell} \sigma_{\aspar }  \Phi_{i}  \overline{\nabla \Phi_i}\,   \rmd \bx \Big).
\end{equation}
Following the ideas developed in the proof of Proposition 4.8 of \cite{LWZ:18}, we show that $\langle \bbB^{\bK,\bkg}_{00} \phi_i, \overline{\phi_i} \rangle$ for $i=1,2$ vanish. We set
that $\nu=\tau$ if $i=1$ and $\nu=\overline{\tau}$ if $i=2$. Then, one has
$$
\int_{\cell} \sigma_{\aspar }  \Phi_{i}  \overline{\nabla \Phi_i}\,   \rmd \bx=\int_{\cell} \mathcal{R} ( \sigma_{\aspar }  \Phi_{i} ) \, \overline{\mathcal{R}\nabla \Phi_i}\,   \rmd \bx.
$$
Since $\mathcal{R}(\sigma_g \Phi_{i})=\sigma_g \, \mathcal{R}\Phi_{i}$, $\Phi_{i}\in L^2_{\bK,\nu}$, $\mathcal{R}  \nabla =  R^{*}\nabla \mathcal{R} $ by  \eqref{eq.commut3} and $|\nu|=1$, it follows that
$$
\int_{\cell} \sigma_{\aspar }  \Phi_{i}  \overline{\nabla \Phi_i}\,   \rmd \bx=\nu   \int_{\cell}\sigma_{\aspar }  \Phi_{i} \,   \overline{R^{*}\nabla \mathcal{R}   \Phi_i}\,   \rmd \bx=\nu \, \overline{\nu}  \int_{\cell}\sigma_{\aspar }  \Phi_{i} \,   \overline{R^{*}\nabla    \Phi_i} \mathrm{d} \bx=R^{\top}  \int_{\cell}\sigma_{\aspar }  \Phi_{i} \,  \overline{\nabla    \Phi_i} \mathrm{d} \bx,
$$
where $R^\top$ denotes the transpose of $R$. Multiplying by $R^{\top}$ on both sides yields:
 \begin{equation}\label{eq.sym}
 R \int_{\cell} \sigma_{\aspar }  \Phi_{i}  \overline{\nabla \Phi_i}\,   \rmd \bx =\int_{\cell}\sigma_{\aspar }  \Phi_{i} \,  \overline{\nabla    \Phi_i} \mathrm{d}\bx=0.
 \end{equation}
The last equality holds since $1$ is not an eigenvalue of $R$. By \eqref{eq.coeffmatA11},  $\langle \bbB^{\bK,\bkg}_{00} \phi_i, \overline{\phi_i}\rangle=0$ for $i=1,2$.

We next simplify the expression \eqref{eq.ceoffij}   using  that $\Phi_{j}=\mathcal{P} \mathcal{C}\, \Phi_{i}$ for $i\neq j$. One has
\begin{eqnarray}\label{eq.coeffBij}
\langle \bbB^{\bK,\bkg}_{00} \phi_j,\overline{\phi_i}\rangle&=&-\int_{\cell} \sigma_{\aspar
}  [\rmi \, \bkg  \, \Phi_{j}\cdot \overline{ \nabla \Phi_{i}}+\nabla  \mathcal{P}\mathcal{C}[\Phi_{i}]\cdot \overline{\rmi \, \bkg  \, \mathcal{P}\mathcal{C}\Phi_{j}}] \, \rmd \bx, \nonumber\\
&=&-\int_{\cell} \sigma_{\aspar
}  [\rmi \, \bkg  \, \Phi_{j}\cdot \overline{ \nabla \Phi_{i}} \, \rmd \bx- \int_{\cell} \nabla  \mathcal{P}\mathcal{C}[\Phi_{i}]\cdot \overline{  \mathcal{P}\mathcal{C} [-\rmi \, \bkg\sigma_{\aspar
}\Phi_{j}}] \, \rmd \bx \ \   \mbox{(as $\mathcal{P}\mathcal{C} (\sigma_{\aspar
} \Phi_j)=\sigma_g \mathcal{P} \,\mathcal{C} \Phi_j$)} ,\nonumber \\
&=& -\int_{\cell} \sigma_{\aspar
}  [\rmi \, \bkg  \, \Phi_{j}\cdot \overline{ \nabla \Phi_{i}} \, \rmd \bx-\int_{\cell} \mathcal{P} \mathcal{C}[- \nabla \Phi_{i}]\cdot \overline{  \mathcal{P}\mathcal{C} [-\rmi \, \bkg\sigma_{\aspar
}\Phi_{j}}] \, \rmd \bx  \mbox{ (as $\nabla \mathcal{P}\mathcal{C}=-\mathcal{P}\mathcal{C}\nabla$,  see \eqref{eq.commut3})} ,\nonumber\\
&=& -\int_{\cell} \sigma_{\aspar}  [\rmi \, \bkg  \, \Phi_{j}\cdot \overline{ \nabla \Phi_{i}} \, \rmd \bx-\int_{\cell } \overline{ \nabla \Phi_{i}}\cdot \rmi \, \bkg \,\sigma_{\aspar
}\Phi_{j}\, \rmd \bx \ \mbox{ (since $\mathcal{P}\mathcal{C}$ is anti-unitary)}, \nonumber\\
&=& -2\mathrm{i} \,\bkg\, \cdot \int_{\cell} \sigma_{\aspar}\Phi_j \, \overline{\nabla \Phi_i }\, \rmd \bx \, . 
\end{eqnarray}
Applying the reasoning for the case $i=j$, but now for $j=2$ and $i=1$, we have that equation \eqref{eq.sym} is replaced by:
$$
R \int_{\cell} \sigma_{\aspar }  \Phi_{2}  \overline{\nabla \Phi_1}\,   \rmd \bx =\tau  \int_{\cell}\sigma_{\aspar }  \Phi_{2} \,  \overline{\nabla    \Phi_1} \mathrm{d}\bx.
$$
It follows that for some $a \in \C$:
\begin{equation}\label{eq.eigenvectau}
\int_{\cell}\sigma_{\aspar }  \Phi_{2} \,  \overline{\nabla    \Phi_1} \mathrm{d}\bx=a   \, \xi,
\end{equation}
where  $\xi\in \operatorname{ker}(R-\tau \mathrm{I}d)$ is displayed  in  \eqref{eq.defroteigenelem}. 
One deduces  that
\begin{equation}\label{eq.eigenvectau2}
 a= a  \, \xi \cdot \overline{\xi}= \int_{\cell}\sigma_{\aspar }  \Phi_{2} \,  \overline{\nabla    \Phi_1} \mathrm{d}\bx  \cdot \overline{\xi}.
\end{equation}
Setting  $\tilde{v}_D=-\sqrt{2}\, \mathrm{i} \, a $ and using \eqref{eq.conjcoeff}, \eqref{eq.coeffBij},   \eqref{eq.eigenvectau} and \eqref{eq.defroteigenelem} we have:
$$
\langle \bbB^{\bK,\bkg}_{00} \phi_2,\overline{\phi_1}\rangle= \overline{\langle \bbB^{\bK,\bkg}_{00} \phi_1,\overline{\phi_2}\rangle}  = -2 \mathrm{i} \,\bkg\, \cdot a  \, \xi=-\sqrt{2}\, \mathrm{i} \, a  \, (\bkg_1+\rmi \bkg_2)=\tilde{v}_D\, (\bkg_1+\rmi \bkg_2).
$$
Finally, by \eqref{eq.eigenvectau2} and \eqref{eq.defroteigenelem}, we have $\tilde{v}_D=-\sqrt{2}\,\mathrm{i} \, a $ is given by \eqref{eq.fermveloc}.
\end{proof}
\subsection{Conical behavior of the dispersion curves}
Under the  condition $v_D\neq 0$, we conclude the locally conical behavior of the dispersion curves near
the Dirac point $(\bK,\lambda_D)$. 
\begin{proposition}\label{Prop.implicit} 
Let $\aspar>0$ and  $v_D(\aspar)$  be defined by \eqref{eq.fermyveloc}. Under the non-degeneracy condition: $v_D(\aspar)\neq 0$, there exists $\kappa_0$ such that for all $|\bkg|<\kappa_0$ the equation \eqref{eq.determinant}: $\operatorname{det} \mathcal{M}(\bkg,E_1)=0$  has two solutions   $\bkg \mapsto E_{1}^{\pm}(\bkg)$ given by:
$$
E_{1}^{\pm}(\bkg)=\pm v_D(\aspar) \, |\bkg|  \,(1+ E_{1,res}^{\pm}(\bkg)),\quad |\bkg|<\kappa_0,
$$
where $\bkg\mapsto E_{1,res}^{\pm}(\bkg )$ are Lipschitz continuous functions which vanish at $\bkg= 0$.
\end{proposition}
\begin{proof}
The proof is the same as for honeycomb Schroedinger operators. We refer to  \cite{FW:12}  for the details.
\end{proof}
%XXXXXXXXXXXXXXXXXXXXXXXXXXXXXXXXXXXXXXXXXXX
\appendix

 \section{$\bbA_{\aspar}$ and transverse electric (TE) modes }\label{TE-Max}
In this section we explain how the elliptic operator $\bbA_{\aspar}$ arises in the study of 
 transverse electric (TE) waves  in electromagnetism.
Let $\varepsilon_0$ and $\mu_0$ denote the vacuum dielectric constant and  vacuum  magnetic  permeability.
Introduce Cartesian coordinates $(x,y,z)\in\mathbb{R}^3$.
 Consider a dielectric (non-conducting) medium with macroscopic permeability, $\mu=\mu_0$,  macroscopic dielectric parameter $\eps(x,y,z)$ and current source density $\bJ(x,y,z)$.
 Let $\bE$ and $\bH$ denote, respectively, the electric and magnetic fields.
 Time harmonic solutions of frequency $\omega$ of Maxwell's equations,  $\left(\bE,  \bH\right)\rme^{-i\omega t}$, are governed by the system:
$$
\rmi\omega\, \eps \, \bE + \nabla\wedge \bH= \bJ
\quad\mbox{and}\quad 
-\rmi\omega\, \mu_0  \,\bH + \nabla\wedge\bE = 0
\quad\mbox{in }\R^3.
$$

We further restrict to the case of a medium, where the current density $\bJ$ and  all medium coefficients depend  on the transverse variables $x$ and $y$, but not on the longitudinal variables, $z$: 
 \[ \bJ=\bJ(\bx)=\bJ(x,y), \quad   \eps=\eps(\bx)=\eps(x,y),\quad \mu=\mu_0, \quad\textrm{where}\  \bx=(x,y).\ \]
 We seek modal solutions (i.e  for $\bJ=0$), which are in a  transverse electric (TE) polarization state: 
\begin{equation*}
 \bE(\bx,z) = \big( \bE_{\perp}(\bx ),0 \big)^{\top} = (E_x(\bx),E_y(\bx),0)^{\top} \,  \mbox{ and } \,
\bH (\bx,z) =  H_z(\bx) \; \boldsymbol{e}_z = (0,0, H_z(\bx))^\top .
\end{equation*}
Introduce the $2-$ dimensional scalar and vector operators:
\[
\nabla^\perp v(x,y) := (\partial_y v, -\partial_x v)^{\top}
\quad\mbox{and}\quad 
\curl \bv := \partial_x v_y -\partial_y v_x \mbox{ where } \bv=(v_x,v_y)^{\top}.
\]
For any vector $\bu(\bx)=(\bu_{\perp}(\bx),0)+u_z(\bx) e_z$, which is independent of $z$, 
 we have identity: $\left(\nabla\wedge \bu\right)(\bx)=\curl \bu_{\perp}(\bx) e_z+(\nabla^\perp u_z(\bx),0)^\top $.
Hence, transverse electric (TE) modes are obtained from solutions  of
\begin{equation}
-\nabla \cdot \varepsilon^{-1} \nabla H_z =\omega^2 \mu_0  \, H_z,
\label{Hz-eqn}\end{equation}
where  $\bE_{\perp}$ can be computed  (for frequencies $\omega\neq 0$) from $H_z$:
\begin{equation*}
\bE_{\perp}=\rmi (\eps \omega)^{-1} \ \nabla^{\perp} H_z \ .
\label{bE-eqn}
\end{equation*}
We take $\eps(x,y)=\eps(\bx)$ to be of the form:
\[
\eps(\bx)\ =\ 
\begin{cases} 
\eps_0\ & \textrm{for $\bx$ in the bulk},\\
 \eps_0 g\ &\textrm{for $\bx$ in the inclusions}\ ,
 \end{cases}
 \]
 a piecewise constant medium with contrast parameter $g$.
Multiplication of \eqref{Hz-eqn} by $\eps_0 g$ yields:
\begin{equation*}\label{eq.TEtoAg}
\bbA_{\aspar} H_z = \lambda  \, H_z \ , \ \mbox{ where } \lambda= \aspar\ \left(\frac{ \omega}{c}\right)^2 
\end{equation*}
and $\bbA_{\aspar}$ is defined by \eqref{eq.defAg} and \eqref{sig-def} and $c=(\mu_0 \eps_0)^{-1/2}$ is vacuum speed of light.
Thus, one observes that one easily translates the key results on the  bands $\lambda_n(\aspar;\bk)$ for $\bbA_{\aspar} $ into 
the original electro-magnetic setting by ``using the change of variable'': $$\omega_{n,\pm}(\bk,g) = \pm c \sqrt{\frac{\lambda_n(\bk,g)}{g}}, \quad \forall n \in \N  \ \mbox{ and  any } \ \bk \in \mathcal{B} .$$

\section{Commutation of symmetry operators}\label{sec-appendixcommutation}
In this section we present results on the commutation properties of the operators: rotation, inversion and $\mathcal{P}\mathcal{C}$. Consider first the Dirichlet Laplacian for a single  inclusion.
\begin{proposition}\label{prop.comutappend}
Let $\cell^{A}$ be a non-empty simply connected bounded open set of $\bbR^2$ with Lipschitz boundary $\partial \cell^A$. Assume further that  $R(\cell_A)=\cell^A$ and  $-\cell^A=\cell^A$. Let 
 $\mathcal{R}_{\cell^A}:L^2(\cell^A) \to L^2(\cell^A)$ and  $\mathcal{P}_{\cell^A}: L^2(\cell^A) \to L^2(\cell^A)$ be the (unitary) rotation operator and  (anti-unitary) inversion  operator defined by \eqref{PCR-def},
and denote using the same notation: $\mathcal{R}_{\cell^A}:  L^2(\cell^A)^2 \to L^2(\cell^A)^2 $ and $ \mathcal{P}_{\cell^A}: L^2(\cell^A)^2 \to L^2(\cell^A)^2$ these 
operators acting on $2$-dimensional vector fields. Then,  $\mathcal{R}_{\cell^A} H^1_0(\cell^A)\subset H^1_0(\cell^A)$ and $\mathcal{P}_{\cell^A}H^1_0(\cell^A)\subset H^1_0(\cell^A)$ and we have the commutation relations:
\begin{align}
& \nabla(\mathcal{R}_{\cell^A}u)= R \,\mathcal{R}_{\cell^A} ( \nabla u)   \mbox{ and } \ \nabla(\mathcal{P}_{\cell^A}u)=- \mathcal{P}_{\cell^A}( \nabla u),  \ \  \forall u \in H^1_0(\cell^A), \label{eq.commut1} \\
& \Delta\, \mathcal{R}_{\cell^A}u=\mathcal{R}_{\cell^A}\, \Delta u    \mbox{ and }   \ \Delta\, \mathcal{P}_{\cell^A}u=\mathcal{P}_{\cell^A}\, \Delta u,  \ \  \forall u \in  H^1_{\Delta,A}=\{v \in H^1(\cell^A)\mid \Delta v \in L^2(\cell^A)\} \label{eq.commut2}.
\end{align}
Finally,  $\LaplDcellA$ commutes  with $ \mathcal{R}_{\cell^A}$ and $ \mathcal{P}_{\cell^A}$. That is,  $D(\LaplDcellA)=H^1_{\Delta,A}\cap H^1_0( \cell^A)$ is stable under  $  \mathcal{R}_{\cell^A} $ and $  \mathcal{P}_{\cell^A} $, and $[\LaplDcellA,\mathcal{R}_{\cell^A}]=[\LaplDcellA,\mathcal{P}_{\cell^A}]=0$  on $D(\LaplDcellA)$.
\end{proposition}
\begin{proof}
We summarize the ideas of the proof. Relations \eqref{eq.commut1} and \eqref{eq.commut2} are first proved for all $\bx\in\cell^A$ and smooth functions $u$ in $D(\cell^A)$ (the space of $C^{\infty}$  functions  compactly supported in $\cell^A$) by  using the chain rule  as in  \cite[Lemma 3.3 and Theorem 3.2]{LWZ:18}. Then, using the  density of $D(\cell^A)$ in $H^1_0(\cell^A)$ (for the $H^1$-norm)   and distribution theory, one extends  \eqref{eq.commut1}   to any functions in $H^1_0(\cell^A)$  and proves the stability of $H^1_0(\cell^A)$ by $\mathcal{P}_{\cell^A}$ and  $\mathcal{R}_{\cell^A}$. 
 Afterwards, one extends \eqref{eq.commut2}  to any functions in $H^1_{\Delta,A}$ by using distribution theory.
Finally,  from the stability of  $H^1_0(\cell^A)$ by $\mathcal{R}_{\cell^A}$ and $\mathcal{P}_{\cell^A}$ and  relation  \eqref{eq.commut2}, one gets that $D(\LaplDcellA) =H^1_0(\cell^A) \cap H^1_{\Delta,A}$  is stable under $\mathcal{R}_{\cell^A}$ and $\mathcal{P}_{\cell^A}$. This last point and \eqref{eq.commut2} imply that $\LaplDcellA$ commutes  with $ \mathcal{R}_{\cell^A}$ and $ \mathcal{P}_{\cell^A}$.
\end{proof}
 
To  deal with the case of elliptic operators with quasi-periodic conditions, 
we introduce for  any $\bk \in \mathcal{B}$ the  following Hilbert spaces on the open sets ${\bf \Omega}^{\pm}$  (defined in \eqref{eq.defOmegapm}):
\begin{align*}\label{eq.defspacesL2H1}
&L^2_{\bk}({\bf \Omega}^{\pm}):=\{ u \mid  u \in  L^2(\calO),\, \mbox{ for all bounded open sets $\calO \subset {\bf \Omega}^{\pm}$}, \  u \mbox{ is $\bk-$quasi-periodic on ${\bf \Omega}^{\pm}$}  \}, \nonumber \\
&H^1_{\bk}({\bf\Omega}^{\pm}):=\{ u\in L^2_{\bk}({\bf \Omega}^{\pm}) \mid u \in  H^1(\calO), \, \forall  \mbox{ bounded open sets $\calO \subset {\bf \Omega}^{\pm}$}  \}.
\end{align*}
These spaces are endowed with their standard inner products:
\begin{equation}\label{eq.innerproduct}
(u,v)_{L^2_{\bk}({\bf \Omega^{\pm}})}=\int_{\Omega^{\pm}} u \, \overline{v} \, \rmd \bx \  \mbox{ and } \ (f,g)_{H^1_{\bk}({\bf \Omega^{\pm}})}=\int_{\cell^{\pm}} f \, \overline{g}\,  \rmd \bx+ \int_{\cell^{\pm}}  \nabla f \cdot \overline{\nabla g}  \, \rmd \bx .
\end{equation} 
\begin{remark}
The functions of $H^1_{\bk}({\bf\Omega}^{-})$ are $\bk$-quasi periodic functions that are $H^1$ on each periodic copy of $\cell^-$  and satisfy a $\bk-$quasi-periodic boundary condition, i. e. the Dirichlet  trace has  no jump across the border of the periodic cells: $\cup_{(m,n)\in\Z^2}\ \partial \cell_{mn}$ and is $\bk-$quasi-periodic. This last condition relies on  the fact that these functions have to be $H^1$ on all bounded open subsets of ${\bf \Omega}^{-}$. Thus their Dirichlet trace has to match at the boundary of the periodic cells. However,  as the inclusions are disjoint, no condition (except $\bk-$quasi-periodicity) is imposed on the Dirichlet  trace on ${\partial  \bf \Omega^{+}}$ for  functions of $H^1_{\bk}({\bf\Omega}^{\pm})$.
\end{remark}
We then introduce  the space  
\begin{equation*}
L^2_{\bk}({\bf \partial {\Omega}^{+}})=\{ f\in L^2_{\mathrm{loc}}({\bf \partial {\Omega}^{+}})\mid  f(\bx+\bv)=\rme^{\rmi  \bk\cdot \bv  } f(\bx), \mbox{ for a.e } \bx\in  {\bf \partial {\Omega}^{+}} \mbox{ and all }  \bv\in \Lambda  \}.
\end{equation*}
Its inner product is defined by replacing $\Omega^{\pm} $ by $\partial \Omega^{+}$ and $\rmd  \bx$ by $\rmd \gamma_{\bx}$ in the left formula of  \eqref{eq.innerproduct}. The Dirichlet trace operator $\gamma_D^{\pm}: C^{1}_{\bk}(\overline{{\bf \Omega^{\pm}}}) \subset  H^1_{\bk}({\bf \Omega^{\pm}}) \to  L^2_{\bk}(\partial { \bf{\Omega}^{+})}$ (where $C^{m}_{\bk}(\overline{{\bf \Omega^{\pm}}}), \, m\in \NO$ stands for the space of $C^m$ smooth $\bk-$quasi-periodic functions $f$ on ${\bf \Omega^{\pm}}$ such that $f$ and all its partial derivatives  of order less than $m$ or equal to $m$ admits a continuous extension on $\overline{\bf \Omega^{\pm}}$) defined by $\gamma_D^{\pm}(u)=u\vert_{\partial  {\bf\Omega^{+}}}$ extends by density to a continuous operator on $ H^1_{\bk}(\bf \Omega^{\pm}) $. Furthermore, its range  is the Sobolev space (see \cite{Girault:1986,Gri:85,Monk:03}):
\begin{equation}\label{eq.defH1/2}
H^{1/2}_{\bk}({\bf \partial \Omega^{+}}):=\{ f\in L^2_{\bk}({\bf \partial {\Omega}^{+}})\mid  \int_{\partial {\Omega}^{+}\times \partial {\Omega}^{+}}\frac{|f(\bx)-f(\by)|^2}{|\bx-\by|^2} \rmd \gamma_{\bx}   \rmd \gamma_{\by}<\infty\}
\end{equation}
endowed with its standard norm defined for all $f\in H^{1/2}_{\bk}({\bf \partial \Omega^{+}})$ by:
\begin{equation}\label{eq.defH1/2-norm}
\| f\|_{H^{1/2}_{\bk}({\bf \partial \Omega^{+}})}^2= \| f\|_{L^{2}_{\bk}({\bf \partial \Omega^{+}})}^2+\int_{\partial {\Omega}^{+}\times \partial {\Omega}^{+}}\frac{|f(\bx)-f(\by)|^2}{|\bx-\by|^2} \rmd \gamma_{\bx}   \rmd \gamma_{\by}.
\end{equation}
The mapping $\gamma_D^{\pm}  : H^1_{\bk}({\bf \Omega^{\pm}})\to H^{1/2}_{\bk}({\bf \partial \Omega^{+}})$ is also continuous (for the $H^{1/2}_{\bk}({\bf \partial \Omega^{+}})$ norm).

For $\bK_*$ a vertex of $\mathcal{B}$, one denotes by $\mathcal{R}_{{\bf \cell^{\pm}}}:L^2_{\bK_*}({\bf\Omega^{\pm}})\to L^2_{\bK_*}({\bf\Omega^{\pm}})$ and $\mathcal{P}\mathcal{C}_{{\bf \cell^{\pm}}}:L^2_{\bK_*}({\bf\Omega^{\pm}})\to L^2_{\bK_*}({\bf\Omega^{\pm}})$, the rotation and $\mathcal{P}\mathcal{C}$ operators defined as \eqref{ed.defrotop}  and  \eqref{eq.definvoppseudoper} but on the sets ${\bf\Omega^{\pm}}$. We also use the notation $ \mathcal{R}_{\bf \cell^{\pm}}:L^2_{\bK_*}({\bf\Omega^{\pm}})^2\to L^2_{\bK_*}({\bf\Omega^{\pm}})^2$ and $\mathcal{P}\mathcal{C}_{\bf \cell^{\pm}}:L^2_{\bK_*}({\bf\Omega^{\pm}})^2\to L^2_{\bK_*}({\bf\Omega^{\pm}})^2$ for the equivalent of  these operators acting on  $2D$ vector fields. Finally,  $\mathcal{R}_{{\bf  \partial \Omega^+}}:L^2_{\bK_*}({\bf \partial \Omega^+})\to L^2_{\bK_*}( {\partial \bf\Omega^{+}})$ and $\mathcal{P}\mathcal{C}_{{\bf \partial \cell^{+}}}:L^2_{\bK_*}( {\bf \partial \Omega^+})\to L^2_{\bK_*}({\bf \partial\Omega^+})$ denote the same rotation and $\mathcal{P}\mathcal{C}$ operators, but for scalar functions defined on ${\bf \partial \Omega^{+}}$. The justification for these definitions is similar to the ones of the operator $\mathcal{R}$ and $\mathcal{P}\mathcal{C}$ (see \eqref{eq.justfRop} and \eqref{eq.definvoppseudoper}). It relies, on one hand, on  ${\bf\Omega^{\pm}}$ and  ${\bf \partial\Omega^+}$ being invariant under $2\pi/3$ clockwise rotation and inversion with respect to the center, $\bx_c$  and, on the other hand, on the stability of $\bK_*-$quasi-periodic boundary conditions under these operators. Finally, to see that the adjoint  of the unitary operators $\mathcal{R}^*_{{\bf \Omega^{\pm}}}$  and  $\mathcal{R}^*_{\bf\partial  \cell^{+}}$ are  defined, interchange the rotation matrix $R$ with its inverse $R^*$ in the definition of $\mathcal{R}_{\bf  \Omega^{\pm}}$ and $\mathcal{R}_{\bf  \partial \Omega^{+}}$.
\begin{lemma}\label{lem.comquasiper}
Let $\bK_*$ be any vertex of $\mathcal{B}$. Then, the space $H^1_{\bK_*}({\bf \cell^{\pm}})$ is stable by $\mathcal{R}_{\bf \cell^{\pm}}$ and $\mathcal{P}\mathcal{C}_{{\bf \cell^{\pm}}}$ and one has the following commutation relations:
\begin{equation}
\nabla( \mathcal{R}_{\bf \cell^{\pm}}u)= R \, \mathcal{R}_{{\bf \cell^{\pm}}}( \nabla u) \quad  \mbox{ and }  \quad   \nabla(\mathcal{P}\mathcal{C}_{{\bf \cell^{\pm}}}u)=- \mathcal{P}\mathcal{C}_{{\bf \cell^{\pm}}}( \nabla u),   \quad \forall u \in H^1_{\bK_*}({\bf \cell^{\pm}}). \label{eq.commut3} 
\end{equation}
Moreover,  the trace operators $\gamma_D^{\pm}$ commute with $\mathcal{R}_{\bf \cell^{\pm}}$ and  $\mathcal{P}\mathcal{C}_{{\bf \cell^{\pm}}}$ in the following sense:
\begin{equation}\label{eq.trace}
 \gamma_D^{\pm}(\mathcal{R}_{\bf \cell^{\pm}}u)=\mathcal{R}_{{ \partial \bf \cell^{+}}}( \gamma_D^{\pm} u )   \quad \mbox{and}  \quad   \gamma_D^{\pm}(\mathcal{P}\mathcal{C}_{\bf \cell^{\pm}}u)=\mathcal{P}\mathcal{C}_{{ \partial \bf \cell^{+}}}( \gamma_D^{\pm} u ) ,  \quad \forall u \in H^1_{\bK_*}({\bf \cell^{\pm}}).
\end{equation}
\end{lemma}
\begin{proof}
We  summarize the idea of the proof of \eqref{eq.commut3}. The relations \eqref{eq.commut3} are first proved for functions in  $C^{1}_{\bK_*}(\overline{{\bf \Omega^{\pm}}})$ using the chain rule as in   \cite[Lemma 3.3, Theorem 3.2]{LWZ:18}. Thus, it follows that $C^{1}_{\bK_*}(\overline{{\bf \Omega^{\pm}}})$ is stable by the operators $\mathcal{R}_{\bf \cell^{\pm}}$ and $\mathcal{P}\mathcal{C}_{{\bf \cell^{\pm}}}$. The extension of the relations \eqref{eq.commut3} to all functions of $H^1_{\bK_*}( {\bf  \Omega^{\pm}})$ and the stability of the space $H^1_{\bK_*}({\bf \cell^{\pm}})$  under $\mathcal{R}_{\bf \cell^{\pm}}$ and $\mathcal{P}\mathcal{C}_{{\bf \cell^{\pm}}}$  are then proved by  distribution theory and  density of $C^{1}_{\bK_*}(\overline{{\bf \Omega^{\pm}}})$ in $H^1_{\bK_*}({\bf \cell^{\pm}})$  for the $H^1_{\bK_*}$-norm.

Let $u\in  H^1_{\bK_*}({\bf \cell^{\pm}})$. To prove  \eqref{eq.trace} for the rotation operator, we use that $\overline{C^{1}_{\bK_*}(\overline{{\bf \Omega^{\pm}}})}=H^1_{\bK_*}( {\bf  \Omega^{\pm}})$ in the $H^1_{\bK_*}$-norm. Thus, there exists  a sequence $(u_n)_{n\in \N}$ of $C^{1}_{\bK_*}(\overline{{\bf \Omega^{\pm}}})$ functions that converge to $u$ for the $H^1_{\bK_*}$-norm. By \eqref{eq.commut3},  $\mathcal{R}_{{\bf \Omega^{\pm}} }u_n\to  \mathcal{R}_{{\bf \Omega^{\pm}} }u$ in  $H^1_{\bK_*}$. By continuity of the Dirichlet trace operator, one has $\gamma_D^{\pm}(\mathcal{R}_{{\bf \Omega^{\pm}} }u_n) \to \gamma_D^{\pm}(\mathcal{R}_{{\bf \Omega^{\pm}} }u)$ and $\gamma_D^{\pm} u_n \to \gamma_D^{\pm} u$ in $H^{1/2}_{\bK_*}({\bf \partial \Omega^{+}})$ (and thus also in $L^{2}_{\bK_*}({\bf \partial \Omega^{+}})$). Moreover,  by continuity of $ \mathcal{R}_{{ \partial \bf \cell^{\pm}}}$, one has $\gamma_D^{\pm}(\mathcal{R}_{{\bf \Omega^{\pm}} }u_n)=\mathcal{R}_{{ \partial \bf \cell^{\pm}}}( \gamma_D^{\pm} u _n) \to  \mathcal{R}_{{ \partial \bf \cell^{\pm}}}( \gamma_D^{\pm} u)$ in $L^{2}_{\bK_*}({\bf \partial \Omega^{+}})$. Therefore, we conclude that $\gamma_D^{\pm}(\mathcal{R}_{{\bf \Omega^{\pm}} }u)=\mathcal{R}_{{ \partial \bf \cell^{\pm}}}( \gamma_D^{\pm} u)$ in $L^{2}_{\bK_*}({\bf \partial \Omega^{+}})$ and also in $H^{1/2}_{\bK_*}({\bf \partial \Omega^{+}})$. The same reasoning applies to the $\mathcal{P}\mathcal{C}$ operator.
\end{proof}

The dual  of $H^{1/2}_{\bk}({\bf \partial \Omega^{+}})$ is denoted  $H^{-1/2}_{\bk}({\bf \partial \Omega^{+}})$  and is equipped with the norm:
$$
\| g\|_{H^{1/2}_{\bk}({\bf \partial \Omega^{+}})}= \sup_{\|f\|_{H^{1/2}_{\bk}({\bf \partial \Omega^{+}})} = 1}|\left\langle g,f \right \rangle _{ H^{-1/2}_{\bk},H^{1/2}_{\bk}}|.
$$
Here, $\left\langle \cdot,\cdot \right \rangle _{  H^{-1/2}_{\bk},H^{1/2}_{\bk}}$ denotes the duality product between $H^{-1/2}_{\bk}({\bf \partial \Omega^{+}})$ and $H^{1/2}_{\bk}({\bf \partial \Omega^{+}})$ (see \cite{Girault:1986,Monk:03}) for which  $\left\langle \overline{u},v \right \rangle _{  H^{-1/2}_{\bk},H^{1/2}_{\bk}}:=\overline{\left\langle u,\overline{v} \right \rangle}_{  H^{-1/2}_{\bk},H^{1/2}_{\bk}}$ for all $u\in   H^{-1/2}_{\bk}$ and all $v\in H^{1/2}_{\bk}$.

Let $\bn$ denote the unit outward normal vector  to  $\bf \Omega^+$. 
The Neumann trace (see e.g. Corollary 2.6 of \cite{Girault:1986}) $[\partial u / \partial \bn]^{\pm}\in   H^{-1/2}_{\bk}({\bf \partial \Omega^{+}})$  is defined for a function $$u \in H^1_{\bk,\Delta}( {\bf  \Omega^{\pm}}):=\{ w\in H^1_{\bk}( {\bf \Omega^{\pm}}) \mid \Delta w \in L^2_{\bk}( {\bf \Omega^{\pm}})   \}$$ via the Green identity as the following continuous linear functional:
\begin{equation}\label{eq.Greenidentity}
\Big\langle \Big[ \frac{ \partial u }{\partial \bn}\Big]^{\pm}  ,\overline{\gamma_D^{\pm}(v)}\Big \rangle _{ H^{-1/2}_{\bk},H^{1/2}_{\bk}}:=\pm \int_{\Omega^{\pm}} \Delta u \cdot \overline{v} + \nabla u \cdot \overline{\nabla v}  \,  \mathrm{d} \, \bx, \ \forall  v\in H^1_{\bk}( {\bf \Omega^{\pm}}),
\end{equation} 
which is well-defined as a linear form since $\gamma_D^{\pm}:H^1_{\bk}({\bf \Omega^{\pm}})\to H^{1/2}_{\bk}({\bf \partial \Omega^{+}})$ is surjective. It is continuous since $u \in H^1_{\bk,\Delta}( {\bf  \Omega^{\pm}})$ and $\gamma_D$ admits a continuous right inverse $\mathscr{E}^{\pm}: H^{1/2}_{\bk}({\bf \partial \Omega^{+}})\to H^1_{\bk}( {\bf \Omega^{\pm}})$ (i.e. a continuous lift operator) such that $\gamma_d(\mathscr{E}^{\pm}w)=w$ for all $w\in  H^{1/2}_{\bk}({\bf \partial \Omega^{+}})$ (see  Proposition 1.1 of \cite{Girault:1986}).

\noindent If $u\in  H^2_{\bk}({\bf \Omega^{\pm}})$ (see  \cite{Girault:1986,Gri:85}) then $[\partial u / \partial \bn]^{\pm} \in L^2_{\bk}({\bf \partial \Omega^{+}})$ and the duality bracket is nothing but the integral $\int_{\partial \Omega^{+}}  [\partial u / \partial \bn]^{\pm}  \, \overline{v } \, \rmd \gamma_{\bx}$. 

One  now defines  $\mathcal{R}_{{\bf  \partial \Omega^+}}$ and $\mathcal{P}\mathcal{C}_{{\bf \partial \cell^{+}}}$ as bounded operators from $H^{-1/2}_{\bK_*}({\bf \partial \Omega^{+}})$ to $H^{-1/2}_{\bK_*}({\bf \partial \Omega^{+}})$ by:
\begin{eqnarray}
\left\langle \mathcal{R}_{\bf  \partial \Omega^+} w ,\overline{v} \right \rangle _{ H^{-1/2}_{\bK_*},H^{1/2}_{\bK_*}}&=&\left\langle w, \overline{\mathcal{R}_{\bf  \partial \Omega^+}^* v}\right \rangle _{ H^{-1/2}_{\bK_*},H^{1/2}_{\bK_*}},  \label{eq.defrotneumtrace}\\
\left\langle \mathcal{PC}_{{\bf  \partial \Omega^+}} w ,\overline{v} \right \rangle _{ H^{-1/2}_{\bK_*},H^{1/2}_{\bK_*}}&=&\left\langle \overline{w}, \mathcal{PC}_{\bf  \partial \Omega^+} v\right \rangle_{ H^{-1/2}_{\bK_*},H^{1/2}_{\bK_*}},  \nonumber
\end{eqnarray}
for all  $ w \in H^{-1/2}_{\bK_*}({\bf \partial \Omega^+}) \mbox{ and all }  v\in H^{1/2}_{\bK_*}({\bf \partial \Omega^{+}})$.
These operators are well-defined since by virtue of \eqref{eq.defH1/2} and \eqref{eq.defH1/2-norm}, it is clear that $\mathcal{R}^*_{{ \partial \bf \cell^{+}}}$ (resp. $\mathcal{PC}_{\bf  \partial \Omega^+}$) is unitary (resp. anti-unitary) on $ H^{1/2}_{\bK_*}({\bf \partial \Omega^{+}})$.
\begin{lemma}\label{lem.comquasiper2}
Let $\bK_*$ be one vertex of $\mathcal{B}$. Then the space $H^1_{\bK_*,\Delta}( {\bf  \Omega^{\pm}})$ is stable by  the operators $\mathcal{R}_{\bf \cell^{\pm}}$ and $\mathcal{P}\mathcal{C}_{{\bf \cell^{\pm}}}$ and one has the following commutation relations:
\begin{equation}
 \Delta\, \mathcal{R}_{\bf \cell^{\pm}} u=\mathcal{R}_{\bf \cell^{\pm}}\, \Delta u \  \   \mbox{ and }   \ \     \Delta\,\mathcal{P}\mathcal{C}_{{\bf \cell^{\pm}}}u=\mathcal{P}\mathcal{C}_{{\bf \cell^{\pm}}}\, \Delta u,  \quad \forall u \in H^1_{\bK_*,\Delta}( {\bf  \Omega^{\pm}}) \label{eq.commut4} .
\end{equation}
Moreover,  the Neumann trace commutes with the rotation and  $\mathcal{P}\mathcal{C}$ operators in the following way:
\begin{equation}\label{eq.tracenormal}
\Big[\frac{\partial \mathcal{R}_{\bf \cell^{\pm}} u}{\partial \bn}\Big]^{\pm}=\mathcal{R}_{\bf \partial \cell^{+}}\Big[\frac{\partial u}{\partial \bn}\Big]^{\pm}       \mbox{ and }  \Big[\frac{\partial \mathcal{P}\mathcal{C}_{\bf \cell^{\pm}} u}{\partial \bn}\Big]^{\pm}=\mathcal{P}\mathcal{C}_{\bf \partial \cell^{+}}\Big[\frac{\partial u}{\partial \bn}\Big]^{\pm}   , \, \forall u \in H^1_{\bK_*,\Delta}( {\bf  \Omega^{\pm}}).
\end{equation}
\end{lemma}
\begin{proof}
We only summarize the idea of the proof of \eqref{eq.commut4}. The relations \eqref{eq.commut4} are first proved  for functions in  $C^{2}_{\bK_*}(\overline{{\bf \Omega^{\pm}}})$ using the chain rule as in \cite[Lemma 3.3, Theorem 3.2]{LWZ:18}. The extension of the relation \eqref{eq.commut4} to all functions of $H^1_{\bK_*,\Delta}( {\bf  \Omega^{\pm}})$ and the stability of  $H^1_{\bK_*,\Delta}( {\bf  \Omega^{\pm}})$ by $\mathcal{R}_{\bf \cell^{\pm}}$ and $\mathcal{P}\mathcal{C}_{\bf \cell^{\pm}}$
are  proved by using  distribution theory and  density of $C^{2}_{\bK_*}(\overline{{\bf \Omega^{\pm}}})$ in the Hilbert space $H^1_{\bK_*,\Delta}( {\bf  \Omega^{\pm}})$ endowed with the norm  $\| \cdot \|_{H^1_{\bK_*,\Delta}( {\bf  \Omega^{\pm}})}$: $ \| u\|_{H^1_{\bK_*,\Delta}( {\bf  \Omega^{\pm}})}^2=\| u\|^2_{H^1_{\bK_*}({\bf \Omega^{\pm}})}+\| \Delta u\|^2_{L^2_{\bK_*}({\bf \Omega^{\pm}}) }$, for all $u\in H^1_{\bK_*,\Delta}( {\bf  \Omega^{\pm}})$.

We show now the relation \eqref{eq.tracenormal}. Let $ u \in H^1_{\bK_*,\Delta}( {\bf  \Omega^{\pm}})$  be fixed and  $v$ be any functions in  $H^1_{\bK_*}( {\bf  \Omega^{\pm}}) $.
First, as the spaces $H^1_{\bK_*}( {\bf  \Omega^{\pm}})$ and $H^1_{\bK_*,\Delta}( {\bf  \Omega^{\pm}})$ are stable by $\mathcal{R}_{{\bf \Omega^{\pm}} }$, one has $\mathcal{R}_{{\bf \Omega^{\pm}} } v\in  H^1_{\bK_*}( {\bf  \Omega^{\pm}}) $ and  $\mathcal{R}_{{\bf \Omega^{\pm}} } u \in H^1_{\bK_*,\Delta}( {\bf  \Omega^{\pm}})$.
Thus, using the Green identity \eqref{eq.Greenidentity}, one has on one hand:  
\begin{equation}\label{eq.tracenormalrotGreen1}
\Big\langle \Big[\frac{ \partial  \mathcal{R}_{{\bf \Omega^{\pm}} }  u}{\partial \bn} \Big]^{\pm} ,\overline{\gamma_D^{\pm}(v)} \Big\rangle_{H^{-1/2}_{\bK_*},H^{1/2}_{\bK_*}}= \pm (\Delta \mathcal{R}_{{\bf \Omega^{\pm}} }  u, v)_{L^2_{\bK_*}({\bf \Omega^{\pm}})} \pm (\nabla \mathcal{R}_{{\bf \Omega^{\pm}} }  u, \nabla v)_{L^2_{\bK_*}({\bf \Omega^{\pm}})}
\end{equation}
and on the other hand:
\begin{eqnarray}\label{eq.tracenormalrotGreen2}
\qquad \quad \Big\langle   \mathcal{R}_{\bf  \partial \Omega^{+}} \Big[ \frac{ \partial  u}{\partial \bn} \Big]^{\pm}  ,\overline{\gamma_D^{\pm}(v)} \Big\rangle_{H^{-1/2}_{\bK_*},H^{1/2}_{\bK_*}}&=& \Big\langle    \Big[ \frac{ \partial  u}{\partial \bn} \Big]^{\pm} , \overline{\mathcal{R}_{{\bf \partial \Omega^{+}} }^{*} \gamma_D^{\pm}(v)} \Big\rangle_{H^{-1/2}_{\bK_*},H^{1/2}_{\bK_*}}\\[4pt]
&=& \pm  (\Delta  u,   \mathcal{R}^*_{\bf \cell^{\pm}}v)_{L^2_{\bK_*}({\bf  \Omega^{\pm}})}  \pm    (\nabla  u, \nabla \mathcal{R}^*_{\bf \cell^{\pm}}v)_{L^2_{\bK_*}({\bf \Omega^{\pm}})} \nonumber
\end{eqnarray}
(we use here that   $H^1_{\bK_*}( {\bf  \Omega^{\pm}})$ is stable by the rotation operator $\mathcal{R}^{*}_{\bf   \Omega^{\pm}} $ and  that $ \gamma_D^{\pm}(\mathcal{R}^*_{\bf \cell^{\pm}}v)=\mathcal{R}^*_{{ \partial \bf \cell^{+}}} \gamma_D^{\pm} (v) $, these properties can be shown in the same way as in the proof of Lemma \ref{lem.comquasiper}  for the rotation operator $\mathcal{R}_{\bf \Omega^{\pm}}$  since $\mathcal{R}^*_{\bf \cell^{\pm}}$ and $\mathcal{R}^*_{\bf\partial  \cell^{+}}$  are defined  by changing the rotation matrix $R$ by its inverse $R^*$ in the definition of $\mathcal{R}_{\bf \Omega^{\pm}}$ and $\mathcal{R}_{\bf  \partial \Omega^{+}}$).

Using  \eqref{eq.commut4} on the first term  of the right hand side of \eqref{eq.tracenormalrotGreen1}, one obtains:
\begin{equation}\label{eq.relationnablarot1}
(\Delta \mathcal{R}_{{\bf \Omega^{\pm}} }  u, v)_{L^2_{\bK}({\bf \Omega^{\pm}})} = ( \mathcal{R}_{{\bf \Omega^{\pm}} } \Delta   u, v)_{L^2_{\bK}({\bf \Omega^{\pm}})} = (  \Delta   u, \mathcal{R}_{{\bf \Omega^{\pm}} }^*v)_{L^2_{\bK}({\bf \Omega^{\pm}})}.
\end{equation}
Then, for the second  term, one has   
\begin{eqnarray*}
 (\nabla \mathcal{R}_{{\bf \Omega^{\pm}} }  u, \nabla v)_{L^2_{\bK_*}({\bf \Omega^{\pm}})}&=&(R^* \nabla( \mathcal{R}_{\bf \cell^{\pm}}u),R^{*} \nabla v)_{L^2_{\bK_*}({\bf \Omega^{\pm}})} \nonumber \\
&=&( \mathcal{R}_{\bf \Omega^{\pm}}  \nabla u, R^{*} \nabla v)_{L^2_{\bK_*}({\bf \Omega^{\pm}})} \ \mbox{{\footnotesize (since by \eqref{eq.commut3}, \, $R^* \nabla( \mathcal{R}_{\bf \cell^{\pm}}u)= \mathcal{R}_{{\bf \cell^{\pm}}}( \nabla u)$),}}  \nonumber \\
&=& ( \nabla u,    \mathcal{R}_{\bf \Omega^{\pm}}^*  \, R^{*} \nabla  \, \mathcal{R}_{{\bf \Omega^{\pm}}}\,  \mathcal{R}_{{\bf \Omega^{\pm}}}^* v)_{L^2_{\bK_*}({\bf \Omega^{\pm}})}.
\end{eqnarray*}
Using again that $R^* \nabla \mathcal{R}_{\bf \cell^{\pm}}=  \mathcal{R}_{{\bf \cell^{\pm}}}\nabla $ on $H^1_{\bK_*}( {\bf  \Omega^{\pm}})$ and that $\mathcal{R}^{*}_{\bf  \Omega^{\pm}} v\in H^1_{\bK_*}( {\bf  \Omega^{\pm}})$  since $ H^1_{\bK_*}( {\bf  \Omega^{\pm}})$ is stable by $\mathcal{R}^{*}_{\bf   \Omega^{\pm}}$ leads to
\begin{eqnarray}\label{eq.relationnablarot2}
(\nabla \mathcal{R}_{{\bf \Omega^{\pm}} }  u, \nabla v)_{L^2_{\bK_*}({\bf \Omega^{\pm}})}&=& ( \nabla u, \mathcal{R}_{{\bf \Omega^{\pm}} }^*  \mathcal{R}_{{\bf \Omega^{\pm}} } \nabla \mathcal{R}_{{\bf \Omega^{\pm}}}^* v)_{L^2_{\bK_*}({\bf \Omega^{\pm}})} \nonumber \\
&=& ( \nabla u, \nabla \mathcal{R}_{{\bf \Omega^{\pm}}}^* v)_{L^2_{\bK_*}({\bf \Omega^{\pm}})} \quad  \mbox{  (since $ \mathcal{R}_{{\bf \Omega^{\pm}} }$ is unitary).}
\end{eqnarray}
Thus, one concludes by \eqref{eq.tracenormalrotGreen1} , \eqref{eq.tracenormalrotGreen2},
\eqref{eq.relationnablarot1} and \eqref{eq.relationnablarot2} that 
$$
\Big\langle \Big[\frac{ \partial  \mathcal{R}_{{\bf \Omega^{\pm}} }  u}{\partial \bn} \Big]^{\pm} ,\overline{\gamma_D^{\pm}(v)} \Big\rangle_{H^{-1/2}_{\bK_*},H^{1/2}_{\bK_*}}=\Big\langle   \mathcal{R}_{\bf  \partial \Omega^{+}} \Big[ \frac{ \partial  u}{\partial \bn} \Big]^{\pm}  ,\overline{\gamma_D^{\pm}(v)} \Big\rangle_{H^{-1/2}_{\bK_*},H^{1/2}_{\bK_*}} ,
$$
which yields the second relation \eqref{eq.tracenormal} (since $\gamma_D^{\pm}$ is surjective). Finally, the equivalent property holds for the $\mathcal{P}\mathcal{C}$  operator  by the same reasoning.
\end{proof}

\begin{proposition}\label{prop.commutoprper}
Let $\bK_*$ be one vertex of $\mathcal{B}$ and $\aspar$ be a positive real number. The operator $\bbA_{\bK_*,\aspar}$ commutes  with the operators $ \mathcal{R}$ and $ \mathcal{P}\mathcal{C}$, i.e. its domain  $D(\bbA_{\bK_*, \aspar})$ is stable under  $  \mathcal{R}$ and $  \mathcal{P}\mathcal{C}$ and the commutators  $[\bbA_{\bK_*, \aspar},\mathcal{R}]$ and $[\bbA_{\bK_*, \aspar},\mathcal{P}]$ vanish on $D(\bbA_{\bK_*, \aspar})$.
\end{proposition}
\begin{proof}
We will show that $ \mathcal{R}$ and $\bbA_{\bK_*, \aspar}$ commute. The proof that $ \mathcal{P}\mathcal{C}$ and $\bbA_{\bK_*, \aspar}$ commute is similar.
One has $D(\bbA_{\bK_*, \aspar})=\big\{ u \in H^1_{\bK_*}\mid  -\nabla \cdot \sigma_{\aspar }  \nabla u \in L^2_{\bK_*}\big\}$. Therefore,  $u\in D(\bbA_{\bK_*, \aspar})$ is equivalent to $u\in H^1_{\bK_*,\Delta}( {\bf  \Omega^{+}})$, $u\in H^1_{\bK_*,\Delta}( {\bf  \Omega^{-}})$ and $u$ satisfies the following transmission conditions on ${\bf \partial \Omega^{+}}$: $\gamma_D^{+}u=\gamma_D^{-}u$ and $\aspar  \, [  \partial  u/\partial \bn ]^{-}=[ \partial  u/\partial \bn]^{+}$. 

Let $u\in  D(\bbA_{\bK_*, \aspar})$. From Lemma \ref{lem.comquasiper2}, $H^1_{\bK_*,\Delta}( {\bf  \Omega^{\pm}})$ is stable by $\mathcal{R}_{\bf  \Omega^{\pm}}$. Thus, the restriction of $\mathcal{R}u$  to ${\bf  \Omega^{\pm}}$ belongs to $ H^1_{\bK_*,\Delta}( {\bf  \Omega^{\pm}})$. The continuity on the Dirichlet trace of $u$ and   \eqref{eq.trace} implies that $\gamma_D^{+}(\mathcal{R}_{\bf \cell^{+}}u)= \gamma_D^{-}(\mathcal{R}_{\bf \cell^{-}}u)$.
For the Neumann trace, as $ \aspar  [  \partial  u/\partial \bn ]^{-}=[ \partial  u/\partial \bn]^{+}$, the relation \eqref{eq.tracenormal}   implies that $ \aspar  [  \partial \mathcal{R}_{\bf \Omega^-} u/\partial \bn ]^{-}=[ \partial  \mathcal{R}_{\bf \Omega^+} u/\partial \bn]^{+}$.  Hence, $\mathcal{R}u\in D(\bbA_{\bK_*, \aspar})$ and thus $D(\bbA_{\bK_*, \aspar})$ is stable by $\mathcal{R}$.

Finally, let us prove  that  $\bbA_{\bK_*, \aspar}\mathcal{R} u=\mathcal{R} \bbA_{\bK_*, \aspar} u$.  Using the Green identity and the definition of $\sigma_{\aspar}$, we have for all $v\in H^1_{\bK_*}$:
\begin{equation}\label{eq.vanishcomutant1}
(\bbA_{\bK_*, \aspar}\mathcal{R} u,v)_{L^2_{\bK_*}}=(\sigma_{\aspar} \nabla \mathcal{R} u, \nabla v)_{L^2_{\bK_*}}=(\aspar  \nabla \mathcal{R} u, \nabla v)_{L^2_{\bK_*}({\bf  \Omega^{-}})}+( \nabla \mathcal{R} u, \nabla v)_{L^2_{\bK_*}({\bf  \Omega^{+}})}.
\end{equation}
On the other hand, for all $v\in H^1_{\bK_*}$:
\begin{eqnarray}\label{eq.vanishcomutant2}
(\mathcal{R} \bbA_{\bK_*, \aspar} u,v)_{L^2_{\bK_*}}&=&( \bbA_{\bK_*, \aspar} u, \mathcal{R}^*v)_{L^2_{\bK_*}} \nonumber\\
&=&(   \sigma_{\aspar} \nabla u, \nabla \mathcal{R}^*v)_{L^2_{\bK_*}}  \nonumber\\\
&=&(\aspar  \nabla u, \nabla \mathcal{R}_{\bf\Omega^-}^* v)_{L^2_{\bK_*}({\bf  \Omega^{-}})}+( \nabla  u, \nabla  \mathcal{R}_{\bf\Omega^+}^*v)_{L^2_{\bK_*}({\bf  \Omega^{+}})}.
\end{eqnarray}
Therefore, by \eqref{eq.vanishcomutant1},  \eqref{eq.vanishcomutant2} and \eqref{eq.relationnablarot2}, we have  $(\bbA_{\bK_*, \aspar}\mathcal{R} u,v)_{L^2_{\bK_*}}=(\mathcal{R} \bbA_{\bK_*, \aspar} u,v)_{L^2_{\bK_*}}$ for all $v\in H^1_{\bK_*}$. Since $H^1_{\bK_*}$ is dense in $L^2_{\bK_*}$, we have $\bbA_{\bK_*, \aspar}\mathcal{R} u=\mathcal{R} \bbA_{\bK_*, \aspar} u$  and thus $[\bbA_{\bK_*, \aspar},\mathcal{R}]$ vanishes on
$D(\bbA_{\bK_*, \aspar})$.
\end{proof}

\section{From quasi-modes to genuine modes}\label{app.Appendixquasimode}

To prove the asymptotic expansions of the Floquet-Bloch eigenpairs in Section \ref{sec-asympresult}, we use a corollary of the following
theorem on quasi-modes.
\begin{theorem}\label{th.quasimode}
Let $\bbA:\mathcal{H}\to\mathcal{H}$ be a linear compact self-adjoint positive   operator on a Hilbert space $\mathcal{H}$. Let $u\in \mathcal{H}$ with $\| u \| _{\mathcal{H}}=1$, $\lambda \in \mathbb{C}$ with  $\operatorname{Re}(\lambda)>0$ and $\eta>0$ such that
$$
\| \bbA u- \lambda u \| _{\mathcal{H}} \leq  \eta ,
$$
(such a pair $(u,\lambda)$ is usually referred to as a quasi-mode of accuracy $\eta$ of $\bbA$).
Then, there exists an eigenvalue $\lambda_n\in \sigma(\bbA)$ satisfying
$$|\lambda_n-\lambda|\leq \eta.$$
Furthermore, if there exists $\eta_*>\eta$ such that $\overline{B(\lambda,\eta_*)}\cap \sigma(\bbA)=\{\lambda_n\}$, there exists $u_n$ an eigenfunction, associated with $\lambda_n$, such that $\| u_n \| _{\mathcal{H}}=1$ and
$$
\|u-u_n\|_{\mathcal{H}}\leq  \frac{2\, \eta}{\eta_*}.
$$
\end{theorem}

The proof of Theorem \ref{th.quasimode} can be found for $\lambda >0$ in \cite{Ole:92} (Lemma 1.1 p age 264), but it is easy to check that their proof  holds  also for $\operatorname{Re}(\lambda)>0$.
We now reformulate this result for  unbounded self-adjoint operators $\bbA: \rmD(\bbA)\subset \mathcal{H}\to\mathcal{H}$ that are positive definite and have a compact resolvent. In such setting, one can introduce the closed sesquilinear form $a$ associated to the operator $\bbA$ defined by:
$$
a(u,v)=(\bbA^{\frac{1}{2}}u,\bbA^{\frac{1}{2}}v)_{\calH}, \quad \forall u,v \in \rmD(a)=\rmD(\bbA^{\frac{1}{2}}).
$$
As $\bbA$ is positive definite, one can introduce the inner product $(\cdot,\cdot)_{a}$ defined on the Hilbert space  $\rmD(a)$ by
$$
(u,v)_{a}=a(u,v), \quad \forall u,v \in \rmD(a),
$$
and $\|\cdot\|_{a}$ its associated norm.  Furthermore, as $\bbA$ is a self-adjoint operator with a compact resolvent,  $\rmD(a)=\rmD(\bbA^{\frac{1}{2}})$ is  compactly embedded in $\mathcal{H}$ and dense in $\mathcal{H}$.
Then, using the Riesz representation theorem, one defines a bounded injective operator $\OpcompApp: \rmD(a)\to \rmD(a)$
 by
\begin{equation}\label{eq.defOpcompact}
(\OpcompApp u, v)_a=(u,v)_{\mathcal H}, \ \forall u,v \in \rmD(a).
\end{equation}
One shows easily that
the operator $\OpcompApp$ is self-adjoint, positive and compact. Moreover, using  \eqref{eq.defOpcompact}, it is straightforward to prove that: $(\lambda_n,e_n)$ is an eigenpair for $\bbA$ if and only if $(\lambda_n^{-1},e_n)$
is an eigenpair for $\OpcompApp$. Hence,  $\sigma(\OpcompApp)=\sigma(\bbA^{-1})$.
The following two corollaries reformulate Theorem \ref{th.quasimode} in this current setting and provide eigenvalue and eigenfunction estimates. One one hand,  they extend  the notion of quasi-mode for functions with less regularity (namely, functions that do not belong to $\rmD(\bbA)$ but to the domain of $\rmD(\bbA^{\frac{1}{2}})$). On the other hand, they allow for estimation of  the error in the quasi-mode approximation of eigenfunctions of $\bbA$ in a stronger norm: $\|\cdot\|_{a}$ than the norm
$\|\cdot\|_{\mathcal{H}}$. 
See, for example,  \cite{Del-16,Del:2017}  for the eigenvalue estimate and  \cite{Cak-14} for the eigenfunction estimate  in the context of a Dirichlet Laplacian. We present  the details of the  eigenfunction estimate in a more general setting adapted to our problem.
\begin{corollary}(Eigenvalue estimate)\label{cor.eigvalest}
Assume that there exists $u\in \rmD(a)\setminus\{0\}$, $\lambda\in \mathbb{C}$ with $\operatorname{Re}(\lambda)>0$ and   $0<\eta < (|\lambda|+1)^{-1}  \|u\|_a$  such that
\begin{equation}\label{eq.quasimodeweakestimate}
|a(u,v)-\lambda (u,v)_{\mathcal{H}}|\leq \eta  \,  \|v\|_a ; \ \forall v\in  \rmD(a)\ . 
\end{equation}
Then there exists $\lambda_n\in \sigma(\bbA)$ such that:
\begin{equation*}\label{eq.eigest}
|\lambda-\lambda_n|\leq (|\lambda|+1)  \, \frac{\eta}{ \|u\|_a}.
\end{equation*}
\end{corollary}

\begin{proof}
The proof consists of rewriting the weak ``quasimode estimate''  \eqref{eq.quasimodeweakestimate} with the identity \eqref{eq.defOpcompact}. Then, one concludes by applying Theorem \ref{th.quasimode} to the operator $\OpcompApp$  and the quasimode $(u/\|u\|_a, \lambda^{-1})$ (since $\operatorname{Re}(\lambda^{-1})>0$) by using the fact that $(\lambda_n,e_n)$ is an eigenpair of $\bbA$ if and only if $(\lambda_n^{-1},e_n)$ is an eigenpair of $\OpcompApp$. The details are presented
in  the proof of Proposition 15 in \cite{Del-16}. 
\end{proof}
With a good estimate of the eigenvalues, one can use  the following result to estimate the eigenfunctions. Let $\operatorname{dis}(\cdot,\mathcal{O})$ denote the distance function to a set $\mathcal{O}$, and let $\lambda_1>0$ denote the smallest eigenvalue of $\bbA$.
\begin{corollary}(Eigenfunction estimate)\label{cor.eigenaprox}
Let $\lambda_n\in \sigma(\bbA)$. Assume that there exists $u\in \rmD(a)\setminus\{0\}$ with  $\|u\|_{\mathcal{H}}=1$ and $\eta\in (0,1)$ satisfying $$0< \eta <   \frac12\lambda_n \, \|u\|_a \operatorname{dis}(\lambda_n^{-1},\sigma(\bbA^{-1})\setminus\{\lambda_n^{-1}\})$$ such that
\begin{equation}\label{eq.quasimodeweakestimateeigenfunc1}
|a(u,v)-\lambda_n (u,v)_{\mathcal{H}}|\leq \eta  \,  \|v\|_a \, , \ \forall v\in  \rmD(a).
\end{equation}
Then, there exists  $u_n$, an eigenfunction of $\bbA$ associated to $\lambda_n$,  such that   $\| u_n \| _{\mathcal{H}}=1$ and
$$
\|u -u_n \|_{a}\leq  \widetilde{C} \eta .
$$
The positive constant $\widetilde{C}$  is  given  explicitly by:
\begin{equation}\label{eq.constante}
\tilde{C}=C_1+\lambda_n^{-\frac{1}{2}} + \lambda_n^{\frac{1}{2}}   \lambda_1^{-\frac{1}{2}} \ \mbox{ with } \ C_1=4 \, \big(\lambda_n \, \operatorname{dis}(\lambda_n^{-1},\sigma(\bbA^{-1})\setminus\{\lambda_n^{-1}\})\big)^{-1}.
\end{equation}
\end{corollary}
\begin{proof}
Using the definition \eqref{eq.defOpcompact} of the operator $\OpcompApp$, one can rewrite the inequality \eqref{eq.quasimodeweakestimateeigenfunc1} as:
$$
|(u - \lambda_n \OpcompApp u,v)_a|\leq  \eta   \|v\|_a ,  \ \forall v\in  \rmD(a).
$$
As, $\lambda_n> 0$ and $u\neq0$, it leads  to  
\begin{equation*}\label{eq.estT1}
\Big\|\OpcompApp\, \frac{u}{ \|u\|_a }- \frac{1}{\lambda_n} \frac{u}{ \|u\|_a } \Big\|_a \leq\frac{\eta}{\lambda_n \,   \|u\|_a} .
\end{equation*}
Furthermore,  as $\sigma(\OpcompApp)=\sigma(\bbA^{-1})$, by choosing $\eta_*=1/2 \, \operatorname{dis}(\lambda_n^{-1},\sigma(\bbA^{-1})\setminus\{\lambda_n^{-1}\})$, one has $\sigma(\OpcompApp)\cap \overline{B(\lambda_n^{-1},\eta_*)}=\{\lambda_n^{-1}\}$ and by assumption $\eta \, (\lambda _n\, \|u\|_a)^{-1}<\eta_*$. Therefore,  the Theorem \ref{th.quasimode} applied to the operator $\OpcompApp$ (using $\eta \, (\lambda _n\, \|u\|_a)^{-1}$  for $\eta$ in this Theorem) implies that there exists an eigenfunction $\tilde{u}_n$ associated to $\lambda_n$ for $\bbA$ (and to $1/\lambda_n$ for $\OpcompApp$)
such that $\| \tilde{u}_n \| _{a}=1$  and
\begin{equation}\label{eq.estu1}
\Big\| \frac{u}{ \|u\|_a }-\tilde{u}_n \Big\|_{a}\leq  \frac{C_1 \eta}{ \|u\|_a}.
\end{equation}
with $C_1>0$ defined by \eqref{eq.constante}. It remains now to renormalize this last inequality with respect to the  norm $\|\cdot\|_{\mathcal{H}}$. To this aim, we introduce the vector $u_n$ defined by $u_n=\tilde{u}_n/\| \tilde{u}_n \|_{\mathcal{H}}$.
Then, one has:
\begin{eqnarray}\label{eq.estimationun}
\| u-u_n\|_a & \leq& \| u-\|u\|_a  \tilde{u}_n \|_a+  \| \|u\|_a  \tilde{u}_n-u_n\|_a \nonumber \\
&\leq &C_1 \eta + \| \|u\|_a  \tilde{u}_n-u_n\|_a  \quad \mbox{(using \eqref{eq.estu1})}.
\end{eqnarray}
Then, we estimate the second term of the right hand side of \eqref{eq.estimationun}. To this aim, one uses that
$$
\| \tilde{u}_n\|_{\mathcal{H}}^2=  ( \OpcompApp \tilde{u}_n, \tilde{u}_n)_{a}=  \lambda_n^{-1} (\tilde{u}_n,\tilde{u}_n)_{a}= \lambda_n^{-1}
$$
and it follows that
\begin{equation}\label{eq.constantest1}
\left\| \|u\|_a  \tilde{u}_n-u_n\right\|_a=\left\| (\|u\|_a -\| \tilde{u}_n\|_{\mathcal{H}}^{-1}) \tilde{u}_n \right\|_a=\left|\|u\|_a-\lambda_n^{\frac{1}{2}}\right|\leq \lambda_n^{-\frac{1}{2}}\, \left|\|u\|_a^2-\lambda_n\right| .
\end{equation}
Using the estimate \eqref{eq.quasimodeweakestimateeigenfunc1} for $v=u$  and the fact that $\|u\|_{\mathcal{H}}=1$  yields
\begin{equation}\label{eq.ineqnormL2interm}
 \left|\|u\|_a^2-\lambda_n\right| \leq \eta \|u\|_a.
\end{equation}
\noindent The last point is to dominate $ \|u\|_a$. Using  the estimate \eqref{eq.ineqnormL2interm}, the ``Poincar\'e type inequality'': $\|u\|_{\mathcal{H}}/ \|u\|_a\leq 1/\sqrt{ \lambda_1}$, $\|u\|_{\mathcal{H}}=1$ and that $\eta<1$ leads to:
\begin{equation}\label{eq.constantest2}
\|u\|_a=\frac{\|u\|_a^2}{\|u\|_a} \leq  \eta +\frac{\lambda_n}{\|u\|_a}  \leq  \eta+ \lambda_n \frac{\|u\|_{\mathcal{H}}}{\|u\|_{a}} \leq 1+ \frac{\lambda_n}{\sqrt{\lambda_1}} .
\end{equation}
Finally, one concludes from  \eqref{eq.estimationun}, \eqref{eq.constantest1}, \eqref{eq.ineqnormL2interm} and \eqref{eq.constantest2} that:
$$
\| u-u_n\|_a \leq \tilde{C}\eta  \ \mbox{ with } \
  \tilde{C}=C_1+ \lambda_n^{-\frac{1}{2}} + \lambda_n^{\frac{1}{2}}   \lambda_1^{-\frac{1}{2}}.
$$
\end{proof}
\section{Extensions  to a larger class of elliptic operators}\label{sec-extensionresults} 
We mention here that our approach and  results  extend easily  to a more general class of honeycomb self-adjoint elliptic divergence form operators with anisotropic and spatially heterogeneous coefficients.  
Our more general class of honeycomb operators is the operator  $\bbA_{\aspar}$, where $\sigma_{\aspar}$  in \eqref{eq.defAg}  is now the matrix valued-function 
 $\tilde{\sigma}_{\aspar}$ given by
\begin{equation}
\tilde{\sigma}_{\aspar}(\bx)\ =
\begin{cases} \label{eq.sigmaanisotrop}
a_1(\bx) \mathrm{I}_2+b_1(\bx) \sigma_2   &\textrm{for $\bx$ in the inclusions}\\
\aspar\,  (a_2(\bx) \mathrm{I}_2+ b_2(\bx) \sigma_2  ), \, \aspar>1 & \textrm{for $\bx$ in the bulk,}\\
 \end{cases}
 \end{equation}
with  $\sigma_2$ the  Pauli matrix:
 $$
 \sigma_2:=\begin{pmatrix} 0 &-\mathrm{i} \\[5pt] \mathrm{i}  &0\end{pmatrix}.
 $$
Here, $a_i$ and $b_i$, $i=1,2$ are assumed to be $\Lambda$-periodic bounded, scalar real-valued functions with
$a_1$ and $b_1$ (restricted to $\cell^+=\cell^A\cup \cell^B$) invariant under translation by $\bv_B-\bv_A$ (from the inclusion $\cell^A$ to the inclusion $\cell^B$).
Thus,  $\tilde{\sigma}_{\aspar}$, given by \eqref{eq.sigmaanisotrop}, extends  to give a Hermitian periodic matrix-valued function: 
 $\tilde{\sigma}_{\aspar}\in L^{\infty}(\mathbb{R}^2/\Lambda)$. 
The scalar functions $a_i$ and $b_i$ are required to satisfy the symmetry relations:
 \begin{eqnarray}\label{eq.symheterocoeff}
 && a_i(\bx_c+R\, (\bx-\bx_c))=a_i(\bx) ,\  a_i(2\bx_c-\bx)=a_i(\bx)  , \,  \mbox{ for } i=1,2  \\ &&  b_i(\bx_c+R\, (\bx-\bx_c))=b_i(\bx) , \ b_i(2\bx_c-\bx)=-b_i(\bx),  \,  \mbox{ for } i=1,2\nonumber ,
 \end{eqnarray}
 where $R$ is the $2\pi/3-$ clockwise  rotation matrix , $\bv_A$ the center of the inclusion $A$ is  the origin and $\bx_c$ is the reference point indicated in Figure \ref{fig.funcell} and defined in Section \ref{sec-not}.
Finally,  the matrix-valued functions $a_i(\bx) \mathrm{I}_2+b_i(\bx) \sigma_2$ for $i=1,2$
 are required to be uniformly positive definite.  
 This more general class of  honeycomb operators, $\bbA_\aspar$, which models a class of magneto-optic  materials and bi-anisotropic meta-materials, commutes with the required symmetry operators: $\mathcal{R}$ and $\mathcal{PC}$; see \cite{LWZ:18}. We can also apply the variational approach of \cite{Hem-00} to study its high contrast behavior of dispersion surfaces. Furthermore, the asymptotic expansions of the $L^2_{\bK}-$ Floquet-Bloch eigenpairs, their justification via the weak formulation of the quasi-mode approach (Section \ref{sec-asympresult}), and the Lyapunov-Schmidt / Schur complement reduction scheme of Section \ref{sec.diracpointthoerem} all extend easily to this setting. Hence, all the results and proofs of this paper  (except for those in Section \ref{sec.higherbandscirccase}) can be adapted to this context.

In particular, one replaces the Dirichlet Laplacian: $-\Delta$ on the inclusions $\cell^A\cup \cell^B $ by  the  
strictly elliptic operator $-\nabla \cdot (a_1(\bx) \mathrm{I}_2+ b_1(\bx) \sigma_2)  \nabla$ with Dirichlet boundary conditions on  $\cell^A\cup \cell^B $. One introduces the spectral isolation condition $\condS$ (see Definition \ref{Def.condS}) relative to this operator and $\condS$ holds  at least in this more general setting for the first eigenvalue if $b_1=0$ (see  e.g. \cite{Henrot:2006} pages $14$-$15$).  Finally, we point out that the $\Lambda-$periodicity of $a_1$ and $b_1$, their invariance by translation by $\bv_B-\bv_A$ and the symmetry relations \eqref{eq.symheterocoeff} imply that $a_1(R\, \bx)=a_1(\bx)$, $a_1(-\bx)=a_1(\bx)$ on $\cell^A$ and $b_1(R\, \bx)=b_1(\bx)$, $b_1(-\bx)=-b_1(\bx)$ on $\cell^A$. These latter relations are used to prove the commutation  $-\nabla \cdot (a_1(\bx) \mathrm{I}_2+b_1(\bx) \sigma_2) \nabla$ (equipped with Dirichlet boundary conditions on $\partial \cell^A$) with the symmetries operators $\mathcal{R}_{\cell_A}$ and $\mathcal{P}_{\cell^A}\mathcal{C}$ defined on the single inclusion $\cell^A$ (which corresponds in this more general setting to the relation \eqref{commute} for the Dirichlet Laplacian).

\bibliographystyle{plain}
\bibliography{high_contrast}

\end{document}